\documentclass[11pt]{article} 

\usepackage{amsmath,amsthm,latexsym,amssymb,amsfonts,epsfig}

\newtheorem*{Definition}{Definition}

\newtheorem*{Proposition}{Proposition}

\newcommand{\be}{\begin{equation}}
\newcommand{\ee}{\end{equation}}
\newcommand{\ba}{\begin{eqnarray}}
\newcommand{\ea}{\end{eqnarray}}

\title{{\sf Quantum Field Theory of}\\
{\sf Black Hole Perturbations with Backreaction:}\\
{\sf I. General framework}} 
\author{
{\sf T. Thiemann}$^1$\thanks{{\sf 
thomas.thiemann@gravity.fau.de}}\\
\\
{\sf $^1$ Institute for Quantum Gravity, FAU Erlangen -- N\"urnberg,}\\
{\sf Staudtstr. 7, 91058 Erlangen, Germany}\\
}
\date{{\small\sf \today}}

\makeatletter
\@addtoreset{equation}{section}
\makeatother

\begin{document} 

\maketitle

{\sf

\begin{abstract}
In a seminal work, Hawking showed that natural states for free 
quantum matter fields on classical spacetimes that solve 
the spherically symmetric vacuum Einstein equations are KMS states 
of non-vanishing temperature. Although Hawking's calculation does 
not include backreaction of matter on geometry, it is more than plausible 
that the corresponding Hawking radiation leads to black hole evaporation
which is in principle observable.

Obviously, an improvement of Hawking's calculation including backreaction
is a problem of quantum gravity. Since no commonly accepted quantum field 
theory of general relativity is available yet, it has been difficult to 
reliably derive the backreaction effect. An obvious approach is to
use black hole perturbation theory of a Schwarzschild black hole 
of fixed mass and to quantise those perturbations. But it is not 
clear how to reconcile perturbation theory with gauge invariance beyond 
linear perturbations.

In a recent work we proposed a new approach to this problem that 
applies when the physical situation has an approximate symmetry, 
such as homogeneity (cosmology), spherical symmetry (Schwarzschild) or 
axial symmetry (Kerr). The idea, which is surprisingly feasible, 
is to first construct the non-perturbative physical (reduced) Hamiltonian of 
the reduced phase space of fully gauge invariant observables and only 
then to apply perturbation theory directly in terms of observables. 
The task to construct observables 
is then disentangled from perturbation theory, thus allowing to unambiguosly 
develop perturbation theory to arbitrary orders.    

In this first paper of the the series
we outline and showcase this approach for spherical symmetry
and second order in the perturbations for Einstein-Klein-Gordon-Maxwell
theory. Details and generalisation to other matter and symmetry 
and higher orders 
will appear in subsequent companion papers.
\end{abstract}

\section{Introduction}
\label{s1}

Black holes are fascinating objects. Not only are black hole binaries 
among the most important sources for gravitational radiation 
\cite{1} and are supermassive black holes good candidates for 
active galactic nuclei (AGN) connected with a rich astrophysical phenomenology
\cite{1a}, they are also the source of numerous 
debates and speculations in classical and quantum gravity \cite{2}.
Indeed, as summarised in the seminal singularity theorems 
by Penrose (Nobel prize 2020) \cite{HE}, black holes clearly indicate that 
General Relativity is an incomplete theory and must be supplemented 
by quantum considerations, thus 
resolving the classical singularity. 
For instance the classical black hole area theorem \cite{HE} combined with
quantum field theory on Schwarzschild spacetime \cite{4} leads to 
the speculation that black holes carry an intrinsic entropy measured by the 
area of the event horizon \cite{5} which hints at a deep connection between
quantum field theory, classical general relativity and thermodynamics. 
The fact that the entropy is apparently 
measured by a two dimensional rather than 
three dimensional region in spacetime leads to the speculation that 
general relativity is a holographic theory \cite{6} which delivers a strong
motivation for holographic approaches to quantum gravity such as 
the modern string theory approach based on the AdS/CFT conjecture \cite{7}.  

While Hawking's original calculation only considered free quantum matter on 
a classical vacuum Schwarzschild spacetime of fixed, time independent mass, 
which violates the Einstein equations as the corresponding energy momentum 
tensor is obviously non-vanishing,
the presence of the corresponding black body radiation makes it more than
plausible that the black hole loses mass. This is a pure quantum effect 
forbidden by the classical area theorem. It leads to sthe so-called 
information paradox \cite{8} which maybe sketched as follows:
We imagine that the Hilbert space of the entire 
system made from geometry and matter can be considered as a tensor 
product where one factor corresponds to the observables located 
in the spacetime region behind 
the event horizon (black hole region) and the other to the observables 
located in the region
outside of it (asymptotic region). Given an initial pure state we can 
form its partial trace with respect to the black hole region which leads to 
a mixed state for the outside algebra. 
This mixed state should correspond to the KMS (temperature)
state discovered by Hawking and its von Neumann entanglement entropy 
should correspond to the black hole entropy (or information) 
\cite{9}. We use the 
Heisenberg picture and describe the dynamics by unitary evolution of
operators acting on the the total algebra of observables while the 
state remains unchanged and in particular pure for the total 
algebra. 

As the black hole shrinks due to 
Hawking radiation, the outside region and and thus the algebra 
of operators located therein of outside 
observables grows while the algebra of inside observables shrinks. 
The total algebra of all observables remains the same during the entire
process since the whole system is closed. 
When the black hole is gone, the outside algebra becomes the 
total algebra again. If the semiclassical consideration that leads to 
the Hawking radiation picture which was derived for large black holes 
remains valid also for small black holes, at the end of the process 
we have a portion of spacetime isometric to Minkowski spacetime 
which means that not all of spacetime (namely the black hole region)
can be reconstructed from the data available at future null infinity.
Therefore information available at past null infinity (e.g. multipole 
mements of ingoing radiation) 
was lost (e.g. the outgoing radiation only carries information about 
mass, charge and angular momnentum), entropy was created, the state 
therefore is not pure i.e. we     
still have a KMS state. Thus this state would be truly mixed but 
{\it now for the total algebra}. 
However, in the 
Heisenberg picture the state does not change and remains pure which 
is a contradiction. Note that 
the representation of the $^\ast-$algebra of 
observables that derives from a pure state via the GNS construction 
is irreducible while that for a mixed state it is properly reducible,
i.e. there are non-trivial invariant subspaces with corresponding 
projectors \cite{10}. This means that there 
are drastic differences between the two situations, somehow the whole 
representation of the algebra has changed. 

To resolve the apparent contradiction one has several possibilities of which 
we just mention two, see \cite{2} for some of the speculations on possible 
mechanisms. 
The first possibility is that the semiclassical picture breaks down at 
some point and that indeed at the end of the evaporation the state is 
pure, the von Neumann entropy vanishes, information was preserved. 
For instance the evaporation could be incomplete 
leaving a ``remnant'' so that there remains an inside region and the 
outside state can stay mixed. 
The second possibility is that the evaporation 
is complete and ends with a mixed state. Then the corresponding projections
on invariant subspaces  
must arise by a dynamical mechanism, i.e. the quantum dynamics cannot
be unitary. The von Neumann entropy is non-vanishing, information was lost. 

The reason why we repeat here this well known discussion is to highlight two 
facts:\\ 
First, that black holes are ideal 
laboratories for quantum gravity. Any candidate theory of quantum gravity
must pass the test that consists, among other things, in 1. explaining 
the microscopic origin of the Bekenstein-Hawking entropy, 2. resolving
the black hole information paradox, 3. deriving the end product of 
black hole evaporation and the fate of the black hole singularity, 4.
describing Hawking radiation including backreaction from first 
principles and 5. determining the truth value of the (weak) cosmic censorship
conjecture \cite{12a} (i.e. that singularities cannot communicate
with future null infinity of an asymptotically flat spacetime 
so that predictability from a Cauchy surface holds in classical GR).\\
Second, that one needs a sufficiently reliable framework in order to 
turn these speculations into precise statements. 
For instance there has been considerable progress 
concerning the microscopic origin of the black hole entropy 
of stationary black holes both 
in String Theory (ST) \cite{13} and in Loop Quantum Gravity (LQG) \cite{14}.
These considerations have already lead to a more useful local definition 
of black
hole horizon which, in contrast to the event horizon, does not rely on 
knowing the entire spacetime \cite{14a}.
As far as black hole radiation from dynamical black holes is concerned 
important insights have been delivered by the two dimensional exactly 
solvable model \cite{15}, the Vaidya metric model \cite{16} corresponding
to null radiation and the spherically symmetric scalar boson star model
\cite{17}. Concerning singularity avoidance of exactly spherically symmetric
black holes see \cite{LQG-BH,SF-BH}.
 
However, it is certainly fair to say that we are very far from understanding
all aspects of quantum black hole physics. As the discussion reveals, many 
puzzles about black holes have their origin in the attempts to extend
Hawking's calculation, which was performed without taking backreaction into 
account, to treating the case with backreaction. The simplest of such 
considerations uses the Stefan-Boltzmann law to relate the power of the 
black hole, i.e. minus the time derivative of its mass, 
to its temperature which in turn is a function of the mass. This leads to 
(suppressing numerical factors)
\be \label{1.1}
M(t)=[M(0)^3-\frac{t}{t_P}\;M_P^3]^{1/3}
\ee
where $M(0),M_P,t_P$ are initial mass, Planck mass and Planck time 
respectively. Clearly (\ref{1.1}) can at best be an approximation 
because $\dot{M}$ diverges at the {\it evaporation time}
$t_E=[M(0)/M_P]^3\; t_P$ which is of the order of the age of the universe for 
sufficiently small primordial black holes.
Thus precisely in the last stages of the evaporation
process do we expect significant deviations from (\ref{1.1}) that lead 
to a resolution of the singularity and to settling the question whether there 
is a remnant thus contributing to dark matter \cite{19}. Such 
deviations are therefore smoking guns for quantum gravity finger prints 
provided that primordial black abundances are sufficiently large
\cite{FermiLat}.\\
\\
In order to compute the actual time dependence 
of the black hole mass 
and the 
deviation from (\ref{1.1})
we need a proper quantum gravity calculation from first principles. 
This is a very complex task within any approach to quantum gravity. 
For instance, in LQG \cite{20} one would need to first find exact solutions to 
the quantum Einstein equations (Wheeler DeWitt equation) \cite{21}
which can 
be interpreted as black hole states and then study their relational dynamics
in terms of quantum Dirac observables. Such a programme is indeed 
conceivable in the reduced phase space approach \cite{22}.  
This {\it non-perturbative} 
programme is still under development and currently 
non-perturbative renormalisation methods \cite{23}
are being applied to fix point 
the details of the corresponding physical (reduced) Hamiltonian.

To make progress before this step is completed, in this series of works 
we take a {\it perturbative} route. The idea is to separate the degrees of 
freedom into (spherically) symmetric background and non-symmetric 
perturbations and to expand all relevant quantities of the theory such as 
constraints with respect to that split. This is of course well known 
and there is a rich literature on the subject starting with the seminal work
by Regge, Wheeler and Zerilli \cite{24} in the Lagrangian setting and Moncrief 
\cite{25} in the Hamiltonian setting. See also \cite{26} for a modern 
account. Common to these works however is, that the background is 
not considered as a dynamical entity. Therefore the phase space of the 
system is coordinatised just by the perturbative degrees of freedom.
The constraints of the system truncated to first order in the perturbations 
generate gauge transformations on this ``frozen'' phase space and lead 
to a notion of first order gauge invariant objects. To the best of our 
knowledge, higher orders have not been considered, partly because 
the second order frozen constraints do not close under Poisson brackets
in contrast to the first order ones.

Unfortunately 
this kind of analysis is inappropriate precisely when we want to 
take the dynamical interplay between background and perturbations into 
account, that is, backreaction. A similar question also arises in cosmology 
and in \cite{Gomar} it was shown how to extend the action of the first order 
constraints to the full phase space such that they still close and define 
a notion of ``unfrozen'' gauge invariance to first order. In \cite{Gomar}
only a partial reduction of the gauge invariance was performed, i.e. 
there is still a single second order constraint that was left unreduced.
This remaining single constraint trivially commutes with itself so that 
a consistent quantum theory can be defined \cite{Hybrid}. 

To the best of our knowledge, the formalism of \cite{Gomar}
has not been extended yet to the question of   
gauge invariance at higher 
orders or to backgrounds which are not homogeneous.
In \cite{pa129} we developed a general framework to precisely do this.
It turns out that the approach of \cite{Gomar} is embedded in a general
reduction algorithm that can be performed to any order and for 
any Killing symmetry. It directly
computes the physical (reduced) Hamiltonian perturbatively and 
including backreaction 
with respect to (Dirac) observables which are gauge 
invariant to all orders. The reason why it computes the fully reduced 
Hamiltonian and not the partially reduced constraints is that 
partial reduction combined with perturbation theory does not lead to
a consistent algebra of perturbed left over constraints when there is 
more than one which is the case for example when we have spherical
symmetry rather than homogeneity. This way a notion of n-th order 
perturbative gauge invariance including backreaction, for which no 
consistent definition is known to the best knowledge of the author, 
is never necessary, perturbation theory and non-perturbative gauge 
invariance are disentangled. An important ingredient of \cite{pa129} 
is a symplectic chart consisting of four sets of canonical pairs:
Symmetric versus non-symmetric and gauge versus true. The (non-)symmetric 
gauge degrees of freedom are adapted to a similar split of the constraints 
into those that generate gauge transformations that do (not) preserve 
the symmetry and are used to reduce them.
Thus the idea to make progress on the question of backreaction
in Hawking radiation is to use the framework of \cite{pa129} and 
to compute the reduced phase space and Hamiltonian to the desired order and 
with the desired matter content. This reduced phase space strategy is 
then the basis to develop the quantum theory.

Some of the immediate questions that arise when approaching black hole 
perturbation theory with backreaction are the following:\\
1.\\
Which black hole symmetry (or solution) should be used? \\
2.\\
When considering a dynamical background, should one only allow for 
a dynamics of the parameters of the symmetric background solution 
to the Einstein (vacuum?) equations or should 
one allow for a dynamics of all background fields 
compatible with the symmetry?\\     
3.\\
As we want to study the late stages of the evaporation process and thus 
want to ``look into the singularity'' how can we make sure that 
we can explore both the interior and the exterior of the black hole 
when the location of the quantum horizon becomes fuzzy?\\ 
\\
Concerning the first question, assuming that the black hole no hair 
theorem and suitable energy conditions on the energy momentum tensor of 
admissable matter holds \cite{HE}, it suffices to consider black holes 
of either axi-symmetry or spherical symmetry. As shown in \cite{Page},
semiclassical considerations suggest that charge and angular momentum are 
radiated off much faster than mass. Hence, starting from a primordial 
black hole created close to the big bang and evaporating today, it has 
reached spherical symmetry long before evaporation. Thus, as we are 
interested in the late stages of the evaporation process, it appears to be 
well motivated to consider spherical symmetry.
Notice, however, that the 
methods of \cite{pa129} are immediately applicable also to axial symmetry
for which much of the tedious work has been performed by 
Teukolsky \cite{Teukolsky}. Hence to check whether these semiclassical 
considerations are justified, one can repeat the steps outlined in the 
present paper also for the Kerr case. 

To explain the relevance of the second question, recall that
according to Birkhoff's theorem a spherically
symmetric vacuum black hole is uniquely defined by a single parameter (the 
mass) up to diffeomorphisms. On the other hand a spherically symmetric 
spacetime is uniquely defined by four functions that 
depend on a radial and time coordinate. Hence one can either assume 
that the backreaction just affects the mass parameter making 
it time and radially dependent or one can assume that the backreaction 
affects all four functions. Of course there is some redundancy in the second 
description due to the presence of residual radial and temporal 
diffeomorphisms (as well as gauge transformations generated by the primary 
constraints)  
that act on those four functions but how these should be
consistently accomodated in the presence of perturbation theory is a priori
not clear. It is also not clear whether one should use the first description
at all when one considers matter that gives rise to a non vanishing 
spherically symmetric energy momentum tensor (such as scalar matter) so 
that the vacuum solution is not an exact solution to the system.\\
The framework of \cite{pa129} gives an unambiguous answer to this question:   
One has to start from the description in terms of the four functions. The 
matter content and the interplay between the residual gauge transformations
that preserve the symmetry and those gauge transformations that do not, 
automatically dictate the precise reduced form of the metric. 

Finally concerning the third question, we must choose a coordinate system
such that the metric is regular across any potential horizon, in particular
the coordinate system must not depend on any dynamical parameters. 
In the ideal case it should simplify the quantisation of the 
reduced Hamiltonian as much as possible. Since the reduction process involves 
splitting the degrees of freedom into gauge and true degrees of freedom, 
these requirements select suitable gauges. For spherical symmetry we 
will therefore impose the Gullstrand-Painleve\'e gauge \cite{26} that 
the spatial spherically symmetric metric be the flat Euclidean one 
which indeed simplifies the quantisation of the reduced Hamiltonian since 
e.g. spatial curvature terms vanish and the Laplacian becomes the flat 
one with explicitly known spectrum. In particular, the spacetime manifold 
has the topology of $\mathbb{R}^4$ for each asymptotic 
end. This moves the information 
about the non-trivial 4-curvature into the extrinsic curvature.  
Note that GP coordinates cover 
either the advanced or retarded Finkelstein charts of Schwarzschild 
spacetime and therefore in the classical theory 
provide a convenient chart to explore both one asymptotic end 
and either the black hole or white hole region of the full 
Kruskal spacetime up until the singularity. Whether the classical 
singularity is resolved and replaced for example by  
a black hole -- white hole transition is then a question to be answered by the 
quantum theory which can be accomodated in this framework by working 
with two copies of GP manifolds. \\
\\
The focus of the present paper is on the general structure of this 
programme and we use the showcase of second order perturbations of 
the physical Hamiltonian in the presence of scalar and Maxwell matter. 
We will not write out the details of the corresponding expressions which 
are reserved to our companion papers \cite{NT-SS, NT-RN, N-SA}.
Generalisations to 
fermionic (in particular neutrino) matter, 
axial symmetry and higher orders
will be subject of future 
publications. The architecture of this article is as follows:\\
\\

In section \ref{s2} we basically introduce our notation and 
apply \cite{pa129} to spherical symmetry. In particular 
we review spherical tensor harmonics \cite{STH} and how they give 
rise to constraints $C$ and $Z$ respectively which preserve or do 
not preserve the symmetry respectively. Adapted to those we have gauge 
canonical pairs $(q,p)$ and $(x,y)$ respectively which are symmetric and 
non-symmetric respectively. These are complemented by true (observable) 
canonical pairs $(Q,P)$ and $X,Y$ respectively which are also
symmetric and non-symmetric respectively. 

In section \ref{s0} we include a conceptual overview over many interrelated 
topics associated with black hole evaporation such as 1. whether the 
vacuum black hole mass parameter $M$ in fact can change at all during 
time evolution or rather has the status of an integration constant such as 
$M(0)$ in (\ref{1.1}), 2. what replaces $M$ as a measure of evaporation 
in case it is presereved, a natural candidate being the area 
of the apparent horizon,
3. how to describe black hole white hole transitions within the perturbative 
framework, 4. what kind of Fock structures of black hole perturbative 
QFT are selected by the use of GP coordinates, 5. which kind of 
Hawking effects are to be expected and 6. how singularity resolution would 
manifest itself.  

In section \ref{s4} we construct the exact reduced or physical 
Hamiltonian based on the GP gauge for $q$ and the trivial 
gauge $x=0$ for $x$ following \cite{pa129} which is natural in view 
of the algebraic structure of the constraints. This requires 
the full machinery of decay conditions on the background and 
perturbation fields and the solution of the constraints. 
We also construct the exact physical lapse and shift following from 
the stability analysis of the gauge condition 
using the asymptotic structure at spatial infinity. This 
distinguishes gauge diffeomorphisms from symmetry diffeomorphims and  
opens access to the full spacetime metric.
The resulting expressions obtained,  
while non-perturbative, necessarily are implicit. To 
obtain eplicit expressions, perturbation methods must be invoked. This
entails two steps: 1. standard non-gauge invariant 
perturbative expansions of the constraints
and 2. assembling different bits and pieces of those into gauge invariant 
contributions. 

In section \ref{s3} we perform the first step  
and determine the general perturbative form of the classical 
constraints to all orders, that is, we simply expand the constraints in their 
polynomial form into spherical tensor harmonics. While tedious, this 
step is straightforward and just involves the recoupling theory of 
angular momentum or equivalently harmonic analysis on the sphere. In 
polynomial form, the gravitational contribution 
to the constraints is of degree ten which is also the top degree 
of the perturbative epansion. Performing the sphere 
integral returns therefore an expression which is exact.  

In section \ref{s5} we perform the second step and perturbatively solve the 
constraints $C=0,Z=0$ respectively for $p,y$ in terms of 
$q,x,Q,P,X,Y$ and imposing gauges on $q,x$ following the algorithm 
of \cite{pa129}. While the formulae provided in \cite{pa129} cover all
orders, we detail out the concrete expression only up to second order. 
As explained 
in \cite{pa129} certain degrees of freedom (e.g. core mass and charge)
encoded in $q,p$ are retained among the $(Q,P)$ due to the presence of 
boundary terms in the constraints. These also serve to perform a second 
Taylor expansion in order to explicitly solve the differential equations
that occur when solving for $p,y$.

In section \ref{s6} we give a brief introduction to various notions 
of black hole horizons and argue that in the present situation with a 
distinguished notion of time defined by the free falling GP observers 
the apparent horizon and its area is an important quantity that 
captures important information about the degree of black hole evaporation.
We show that the pertubative scheme developed in previous sections 
extends to all orders also to the apparent horizon and its area. 

In section \ref{s7} we enter the quantum regime. Since second 
order perturbation theory of the reduced Hamiltonian 
reproduces the Regge-Wheeler and Zerilli 
free Hamiltonians for the perturbations in a GP spacetime 
of given mass, its Fock quantisation is the starting point of 
perturbative QFT for the perturbations with the higher order terms in 
the reduced Hamiltonian considered as interaction terms. While we do not
complete this step in the present paper we sketch all steps towards 
this goal, i.e. we formulate QFT in GP spacetime. This includes a discussion
of mode functions that are valid thoughout the black or white 
hole region and 
an asymptotic region for each asymptotic end of the spacetime 
(of which there are two in the case of a black hole white hole 
transition). The mathematical challenge is to gain sufficient control 
over those mode functions in GP spactime and to formulate 
junction conditions in the transition region. As a regularising 
method we consider an Einstein-Rosen type bridge of gluing a {\it past}
ingoing GP spacetime with a {\it future} outgoing GP spacetime which 
is foliated by proper GP time Cauchy surfaces. Once established,
the Fock quantisation 
can then be applied to the apparent horizon area and its perturbation theory.

In section \ref{s8} we touch upon the  question of backreaction,
i.e. interaction between symmetric and non-symmetric true degrees of 
freedom. This can be non-trivial already within Einstein-Maxwell theory
when the mass $M$ can change dynamically due to the details of
imposing the the GP gauge. In this case one can use space adiabatic perturbation
theory \cite{SAPT}. We used this already in applcation to cosmology 
\cite{ST}. When the matter content goes beyond that 
of Maxwell fields then new challenges arise because then the  
symmetric ``slow'' sector is also a field theory with infinitely many degrees 
of freedom, not only the ``fast'' non-symmetric sector.  

In section \ref{s9} we construct the conserved Noether current 
that follows from the reduced Hamiltonian and which can be used to 
compute the classical energy flux. It can also be used to construct 
the analog of grey body factors for the corresponding Hawking radiation in 
the quantum theory.

In section \ref{s10} we summarise and give an outlook into further work 
under development, in particular the contact with phenomenology.

In appendix \ref{sa} we review in a simple setting the Hamiltonian 
distinction between symmetry and gauge and how decay behaviours of fields 
and constraint smearing functions as well as concepts of variational 
analysis come into play.

In appendix \ref{sb} we apply this to vacuum black holes and show that 
next to the mass there exists a second Dirac observable. The difference 
between symmetry and gauge diffeomorphisms 
explains why this is not in contradiction to 
Birkhoff's theorem. This is relevant because if the second variable 
shows up in the spacetime metric (depending on the details of the 
GP gauge condition) then the black hole mass $M$ is not 
a constant of motion as soon as gravitational perturbations are 
present.

In appendix \ref{sd} we complement the GP description of appendix 
\ref{sb} by the Kantowski-Sachs description which has recently received 
much interest in black hole singularity resolution scenarios 
and how they are matched.

In appendix \ref{se} we include elements of the analysis of generalised GP 
coordinates and free falling observers and foliations in black hole 
white hole transition spacetimes and regularised versions thereof. It 
contains also a consistent mechanism that reconciles the existence of 
the second Dirac observable without it appearing in the reduced Hamiltonian.

\section{Spherical tensor harmonics, symmetry and gauge degrees of freedom}
\label{s2}

In the first subsection we summarise the relevant information on spherical 
tensor harmonics \cite{STH}. These guide our notation and serve to 
identify the gauge and true degrees of freedom as well as the symmetric and 
non-symmetric degrees of freedom. In the second subsection 
we interpret these in terms 
of the notation applied in the general framework of \cite{pa129}. In 
the third 
we show how to perform the perturbative expansion of all constraints
in closed form.

\subsection{Spherical tensor harmonics}
\label{s2.1}

Let $\theta^1:=\theta\in [0,\pi],\; \theta^2:=\varphi\in [0,2\pi)$ be 
spherical polar coordinates on $S^2$, 
\be \label{2.1}
\Omega_{AB}:=
\delta^1_A\; \delta^1_B+\sin^2(\theta)\; \delta^2_A \delta^2_B; \;\;A,B,C,..
=1,2;\;\;
\omega:=\sqrt{\det(\Omega)};\;\; d\mu:=\frac{\omega}{4\pi}\; d^2\theta;\;\;
\eta_{AB}:=\omega\; \epsilon_{AB} 
\ee
respectively round metric on $S^2$, its associated scalar density of weight 
one, corresponding normalised measure and skew pseudo-metric of density 
weight zero where $\epsilon_{12}=+1$. Let $D_A$ be 
the torsion free $\Omega$ compatible covariant differential. 
The corresponding Riemann tensor is easily computed to be 
$R_{ABCD}=\eta_{AB}\; \eta_{CD}$ with Ricci tensor 
$R_{AB}=\Omega_{AB}$ and Ricci scalar $R=2$. All indices 
are moved with $\Omega$ or its inverse where $\Omega^{AC}\;\Omega_{CB}=
\delta^A_B$. We define the Laplacian $\Delta=D_A\; D^A$. 

Let 
$L_{l,m},\; l=0,1,2,...; m=-l, -l+1,..,l$ be the real valued orthonormal
basis of $L_2(S^2,d\mu)$ defined by Legendre polynomials.
In terms of the usual complex valued $Y_{l,m}$ with $\bar{Y}_{l,m}=
Y_{l,-m}$ we have 
$L_{l,0}:=Y_{l,0},\; \sqrt{2}\;L_{l,m}=Y_{l,m}+Y_{l,-m};\; m>0,\; 
i\sqrt{2}\;L_{l,m}=Y_{l,-m}-Y_{l,m};\; m<0$. 
The scalar harmonics are simply the $L_{l,m}$. We define the 
even ``e'' and odd ``o'' vector harmonics for $l>0$ by 
\be \label{2.2}
\sqrt{l(l+1)} \; L_{A;e,l,m}:=D_A\; L_{l,m}:=     
\sqrt{l(l+1)} \; L_{A;o,l,m}:=\eta_A\;^B\; D_B\; L_{l,m}
\ee
which are orthonormal with respect to the inner product
\be \label{2.3}
<L_{\alpha,l,m},L_{\beta,l',m'}>_{L_2^2}
:=\int_{S^2}\; d\mu\; 
\overline{L_{A;\alpha,l,m}} \;\Omega^{AB}\; L_{B;\beta,l',m'}
=\delta_{\alpha\beta}\;\delta_{l,l'}\; \delta_{m,m'}
\ee 
with $\alpha,\beta\in\{e,o\}$. The terminology used in the literature 
referring to ``even'' and ``odd'' is not entirely consistent. 
A better qualifier would be ``polar'' and ``axial'' i.e. $L_{A;e,l,m}$ 
does not involve the pseudo tensor $\eta_{AB}$ while $L_{A;o,e,l,m}$ does.  
An equivalent characetrisation is that under reflection
$\theta\mapsto \pi-\theta,\;      
\varphi\mapsto \varphi+\pi$ the 1-form (or the corresponding vector field)
$L_{A;e/o,l,m}$ has the opposite/same intrinsic 
parity as $L_{l,m}$ which is $(-1)^l$.
Thus ``even,odd'' should not be confused with the intrinsic parity 
of $L_{l,m}$ which is defined to be even/odd when $l$ is even/odd.
With the understanding of ``even, odd'' as polar, axial the scalar 
perturbations are all even.
The $L_{A;\alpha,l,m}$ are complete i.e. every one form that is square
integrable in the sense of (\ref{2.3}) can be expanded in terms of them
and that expansion converges to it with respect to the $L_2^2$ norm. 
This will be shown below. 

Next consider for $l\ge 0$ the horizontal 
``h'' and for $l\ge 2$ the 
even ``e'' respectively  odd ``o'' symmetric (with respect to tensor indices)
2-tensor harmonics  
\ba \label{2.4}
\sqrt{2}\;L_{AB;h,l,m} &:=& \Omega_{AB}\; L_{l,m}
\nonumber\\
\sqrt{2(l^2-1)(l+1)(l+2)}\;L_{AB;e,l,m} &:=& (D_A\; D_B-\frac{1}{2}
\Omega_{AB}\Delta)\;L_{lm}
\nonumber\\
\sqrt{2(l^2-1)(l+1)(l+2)}\;L_{AB;o,l,m} &:=& D_{(A}\; \eta_{B)}\;^C D_C
\;L_{lm}
\ea
where both the horizontal and even 2-tensor have parity $(-1)^l$ and the odd 
2-tensor has parity $(-1)^{l+1}$. The motivation for the term ``horizontal''
will become clear only in the next section. In the literature one 
refers to both the horizontal and even tensors as ``even'' since they are 
both polar while the ``odd'' tensors are axial. Note that 
the even and odd tensors in contrast to the horizontal tensors 
are tracefree with respect to $\Omega$. 
 
The tensors (\ref{2.4}) are orthonormal
with respect to the inner product
\be \label{2.5}
<L_{\alpha,l,m},L_{\beta,l',m'}>_{L_2^4}
:=\int_{S^2}\; d\mu\; 
\overline{L_{AC;\alpha,l,m}} \;\Omega^{AB}\;\Omega^{CD}\; 
L_{BD;\beta,l',m'}
=\delta_{\alpha\beta}\;\delta_{l,l'}\; \delta_{m,m'}
\ee 
with $\alpha,\beta\in \{h,e,o\}$.
The $L_{AB;\alpha,l,m}$ are complete i.e. every 2-tensor that is square
integrable in the sense of (\ref{2.5}) can be expanded in terms of them
and that expansion converges to it with respect to the $L_2^4$ norm. 

Tensor harmonics for tensors of higher rank can be constructed analogously 
using the building blocks $\Omega, D, L_{l,m}$. The orthonormality 
can be established by exploiting that $D$ and $\tilde{D}=\eta\cdot D$
are anti-self adjoint as  
operators 
$D:\; L_2\to L_2^2,\; L_2^2\to L_2^4$ (derivative),
$D:\; L_2^2\to L_2,\; L_2^4\to L_2^2$ (divergence)
and similar for $\tilde{D}$. The completeness can be 
established relying on the completeness of the scalar harmonics as well 
as the fact that $-\Delta L_{l,m}=l(l+1)\; L_{l,m}$.
For instance we have for a given vector field $v^A$ with divergence 
$d=D_A v^A$ and curl $c=\tilde{D}_A v^A$ 
\ba \label{2.6}
&& \Delta v^A= 
D_B\; D^B \; v^A 
=D_B\;[ (D^B \; v^A-D^A\; v^B)+D^A v^B]
=D_B\;(\eta^{BA}\;c)
+ (D^B D^A-D^A D^B) v_B + D^A \; d
\nonumber\\
& =& \tilde{D}^A\;c+ D^A d+R^{BA}\;_{BC} v^C
=\tilde{D}^A\;c+ D^A d+v^A
\ea
As $d,c$ can be expanded into scalar harmonics and $D,\tilde{D}$ annihilate 
the $l=0$ contributions we find that $-(-\Delta+1)v$ can be expanded into 
vector harmonics. However the operator $-\Delta+1$ is positive definite 
whence 
\be \label{2.7}
v^A=-(-\Delta+1)^{-1}[D^A d + \tilde{D}^A c]
\ee
and since 
$D_A \Delta f=(\Delta-1) D_A f,\;
\tilde{D}_A \Delta f=(\Delta-1) \tilde{D}_A f$ we find that after expanding 
$d,c$ into scalar harmonics labelled by $l\not=0$ 
we can simply replace $(-\Delta+1)^{-1}$
by $[l(l+1)]^{-1}$. 

In our application to perturbation theory the Hilbert spaces 
$L_2,L_2^2, L_2^4$ appear naturally in first order and 
some terms of second order. 
In other terms of second order and those of higher order one 
encounters higher order contractions of the 
$L_{l,m},\;L_{A,e/o,l,m},\; L_{AB,h/e/o,l,m}$ that are integrated 
over $S^2$ with measure $\mu$. These can be computed by combining 
Clebsch-Gordan decomposition 
\be \label{2.8}
L_{l,m}\; L_{l',m'}
=\sum_{|l-l'|\le \tilde{l}\le l+l'; \tilde{m}=m+m'}\; 
c_{l,m;l',m';\tilde{l},\tilde{m}}\; L_{\tilde{l},\tilde{m}'} 
\ee
with expressing $D_A$ in terms of the angular momentum operators 
$L_\mu,\; \mu=1,2,3$ which act diagonally or as ladder operators on the 
$L_{l,m}$.   

Some useful identities are 
\ba \label{2.8a}
&& D_A \Omega_{BC} = D_A\; \omega = 0
\nonumber\\
&& D_A\; L^A_{\alpha,l,m}=-\sqrt{l(l+1)}\;\delta_\alpha^e\; L_{l,m}
\nonumber\\
&& D_A\; L^{AB}_{h,l,m} = \sqrt{l(l+1)/2}\; L^B_{e,l,m}
\nonumber\\
&& D_A\; L^{AB}_{\alpha,l,m} = -\frac{1}{2}\sqrt{(l-1)(l+2)/2}\; 
L^B_{\alpha,l,m};\;\;\alpha=e,o
\nonumber\\
&& D_A L_{l,m}=\sqrt{l(l+1)} \; L_{A;e,l,m}
\nonumber\\
&& D_A L_{B;e,lm} 
=\sqrt{2(l-1)(l+2)} \; L_{AB;e,l,m}-\sqrt{l(l+1)/2}\; L_{AB;h,l,m}
\nonumber\\
&& D_A L_{B;o,lm} 
=\sqrt{2(l-1)(l+2)} \; L_{AB;o,l,m}
+\eta_{[A}\;^C\; D_{B]} \; L_{C;e,l,m}
\ea

\subsection{Classification of symmetry and gauge degrees of freedom}
\label{s2.2}

As mentioned, for the purpose of concrete illustration we focus 
on the following matter content: A charged scalar field $\Phi$ 
with potential $V$ which may serve to build a boson star and the 
Maxwell field $A$. In the canonical setting \cite{12a} assuming 
global hyperbolicity the spacetime manifold $M$ is diffeomorphic to 
$\mathbb{R}\times \sigma$ where $\sigma$ is a three manifold and can 
be foliated by Cauchy surfaces $\Sigma_t$ labelled by $t\in \mathbb{R}$.
What follows can be done in any spacetime dimension, we consider 
the case of four dimensions.\\
\\ 
We thus have 
the following ingredients ($\mu,\nu,\rho,...=1,2,3$ are spatial tensor 
indices with respect to coordinates $x^\mu$ on the manifold $\sigma$):\\
1. Gravitational degrees of freedom:\\
$(S^0,W_0),\;(S^\mu,W_\mu),\;(m_{\mu\nu},W^{\mu\nu})$ where $S^0,S^\mu$  
are called    
lapse and shift functions parametrising the embeddings $\sigma\to \Sigma_t$
and $m$ is the intrinsic metric of $\sigma$. The $W_0,W_\mu,W^{\mu\nu}$ 
are the 
respective conjugate momenta. We denote by $m^{\mu\nu}$ the inverse 
of $m_{\mu\nu}$ and $R[m]$ the Ricci scalar of $m$ constructed from
the torsion free covariant differential $\nabla_\mu$ compatible 
with $m_{\mu\nu}$.\\
2. Scalar degrees of freedom:\\
$(\Phi,\Pi)$ where $\Pi$ is the conjugate momentum of the scalar field
$\Phi$ 
on $\sigma$ which we take as a real valued SO(2) dublett.\\
3. Electromagnetic degrees of freedom:\\
$(S_0, W^0),\; (A_\mu,E^\mu)$ where $S_0$ is the temporal
component of the 4-connection and $A_\mu$ is its spatial component. Again 
$W^0,E^a$ are the conjugate momenta and $B^\mu=\epsilon^{\mu\nu\rho}\; 
\partial_\nu\; 
A_\rho,\;E^\mu$ are referred to as magnetic and electric fields respectively.
We also refer to $F_{\mu\nu}:=2\partial_{[\mu} A_{\nu]}$ 
as the curvature of $A_\mu$.\\
4. Primary constraints:\\
$B_0:=W_0,\; B_\mu:=W_\mu,\; B^0:=W^0$\\
5. Secondary constraints:\\
\ba \label{2.9}
V_0 &:=& V_0^E + V_0^{KG} + V_0^M
\nonumber\\
V_0^E &:=& \frac{1}{\sqrt{\det(m)}}\;
[m_{\mu\rho} \; m_{\nu\lambda}-\frac{1}{2} m_{\mu\nu}\; 
m_{\rho\lambda}]\; W^{\mu\nu}\; W^{\rho\lambda}
-\sqrt{\det(m)}\; R[m]
\nonumber\\
V_0^{KG} &:=& 
\frac{||\Pi||^2}{2\sqrt{\det(m)}}+\frac{1}{2}
\sqrt{\det(m)}\;m^{\mu\nu}\; [D_\mu\Phi]^T\; [D_\nu \Phi]
+\sqrt{\det(m)}\; U(||\Phi||^2)
\nonumber\\
V_0^M &:=& \frac{m_{\mu\nu}}{2\sqrt{\det(m)}}\;[E^\mu\; E^\nu+B^\mu \; B^\nu]
\nonumber\\
V_\mu &:=& V_\mu^E + V_\mu^{KG} + V_\mu^M
\nonumber\\
V_\mu^E &:=& -2\; \nabla_\nu W^\nu\;_\mu
\nonumber\\
V_\mu^{KG} &:=& \Pi^T\; \Phi_{,\mu}
\nonumber\\
V_\mu^M &:=& E^\nu (\partial_\mu A_\nu)-(E^\nu\; A_\mu)_{,\nu}
%=F_{\mu\nu} E^\nu-A_\mu\; G
\nonumber\\
G &:=& G^M +G^{KG}:= \partial_\mu\; E^\mu+\Pi^T\;\epsilon\; \Phi
\ea
where $\epsilon$ is the skew matrix in 2 dimensions, $D=d+\epsilon A$, 
$||\Phi||^2=\Phi^T\Phi$ of which the potential $U$ is a polynomial
(e.g. mass term) and $\nabla$ is the covariant differential 
determined by $m$.
Here $V_0$ is referred to as the Hamiltonian constraint, $V_\mu$ as the 
spatial diffeomorphism constraints and $G$ as the Gauss constraint.
We have labelled the respective contributions due to 
Einstein, Klein-Gordon and Maxwell fields respectively by $E,KG,M$ 
respectively.\\
6. Unreduced Hamiltonian ($v^0,v^\mu,v_0$ are velocities that remain
undeterminded by the Legendre transform)\\
\be \label{2.10}
H=\int_\sigma\; d^3x\;
[v^0\; B_0+v^\mu\; B_\mu+v_0\; B^0+S^0 \; V_0+S^\mu\; V_\mu+S_0\; G]
\ee
7. Symplectic potential (we normalise by the unit sphere area and 
$d$ is the exterior differential on fioeld space)
\be \label{2.11}
4\pi \;\Theta=\int_\sigma\; d^3z\; 
[W_0 \; dS^0+W_\mu \; dS^\mu+W^0 \; dS_0+
W^{\mu\nu}\; dm_{\mu\nu}+\Pi^T \; d\Phi+E^\mu \; dA_\mu]
\ee
The Hamiltonian density is a linear combination of all consraints as 
it is always true for generally covariant field theories. As there is a 
``boundary'' at spatial infinity, one has to add boundary terms to 
(\ref{2.10}) ensuring that (\ref{2.10}) continues to be differentable 
and convergent also when the functions $S^0,S_\mu,S_0$ do not vanish at 
infinity but we will not display them here as we will automatically
encounter them in a later stage of the analysis. The 
velocities $v^0,v^\mu, v_0$ are arbitrary ``Lagrange multipliers'' that 
one could not solve for when performing the Legendre transform. 
The stabilisation
(preservation in time) of the primary constraints $B_0,B_\mu,B^0$ implies 
the secondary constraints $D_0,D_\mu,G$. Their stabilisation leads to no new 
constraints because their Poisson algebra closes: The hypersurface 
deformation algebra generated by the $V_0,V_\mu$ closes by itself while 
$V_0,G$ are invariant under $G$ and the $V_\mu,G$ close among themselves.
All primary constraints obviously close among each other (they are 
Abelian since they involve only momenta $W_0,W_\mu,W^0$) and they have 
vanishing 
Poisson brackets with the secondary constraints because these do not 
involve the variables $S^0,S^\mu,S_0$.

Before proceeding it is convenient to get rid off the primary constraints 
$B_0, B_\mu, B^0$ by gauge fixing the variables $S^0,S^\mu,S_0$ conjugate 
to $W_0, W_\mu, W^0$. This is accomplished by imposing suitable gauge 
conditions $K=(K^0, K^\mu, K_0)$ on the variables 
$(m_{\mu\nu}, W^{\mu\nu}), (\Phi,\Pi),(A_\mu, 
E^\mu)$. In order that these conditions are stable under the gauge flow,
the variables $S^0,S^\mu,S_0$ become given functions 
$S_\ast=(S^0_\ast, S^\mu_\ast,
S_0^\ast)$ of 
$(m_{\mu\nu}, W^{\mu\nu}), (\Phi,\Pi),(A_\mu, E^\mu)$ and are thus also 
fixed. The reduced Hamiltonian acts only on the subset of true degrees
of freedom, i.e. those among 
$(m_{\mu\nu}, W^{\mu\nu}), (\Phi,\Pi),(A_\mu, E^\mu)$ not determined by 
$K=0$ and $V=(V_0,V_\mu,G)=0$ and is determined by first computing the Poisson 
bracket between functions of the true degrees of freedom and the unreduced 
Hamiltonian and then evaluating the result at $S=S_\ast, K=V=0$. It follows 
that the terms proportional to $B=(B_0, B_\mu, B^0)$ can be ignored from 
the outset. The correspondingly simplified  
Hamiltonian coincides with (\ref{2.10}) except that the terms 
proportional to $B_0,B_\mu,B^0$ are dropped and that $S^0,S^a,S_0$ acquire 
now the status of arbitrary smearing functions of the 
constraints that appear (\ref{2.10}) on that reduced phase space which 
will later be fixed by the stability requirement for 
the gauge fixing condition.\\
\\   
We now make the contact with the general framework \cite{pa129}. To do 
this we must introduce two splits of the degrees of freedom into i. symmetric 
ones and non-symmetric ones on the one hand and 
ii. gauge and observable (true) ones on the other. 
The symmetry split is essentially dictated by the decomposition of the 
fields into tensor harmonics. In particular it induces a symmetry 
split among the fields $S^0, S^\mu, S_0$. Their symmetry split induces a 
symmetry split 
of the secondary constraints $V_0,V_\mu,G$ which then suggests a natural
additional gauge split on all configuration and momentum variables. 

Accordingly we begin with the spherical harmonics decomposition 
using the polar coordinate system $z^A:=\theta^A,\;A=1,2$ (see the previous 
subsection) and one radial coordinate 
$z^3:=r$ for each asymptotic end where\\ 
$z^\mu=r(\sin(\theta)\cos(\varphi),\sin(\theta)\sin(\varphi),\cos(\theta))$
are the ususal Cartesian coordinates at spatial infinity. We have 
the following symmetry split of the functions $S^0,S^\mu,S_0$
\ba \label{2.12}
S^0 &=:& f^v+\sum_{l>0,|m|\le l} \; g^{v,l,m}\; L_{l,m} 
\nonumber\\
S^3 &=:& f^h+\sum_{l>0,|m|\le l} \; g^{h,l,m}\; L_{l,m}
\nonumber\\
S^A &=:& 0+\sum_{l>0,|m|\le l,\;\alpha\in \{e,o\}} \; 
g^{\alpha,l,m}\; L^A_{\alpha,l,m},\;\;A=1,2 
\nonumber\\
S_0 &=:& f^M+\sum_{l>0,|m|\le l}\; g^{M,l,m}\; L_{l,m} 
\ea
where the labels ``v'' and ``h'' mean ``vertical'' and 
``horizontal'' respectively 
and capture the fact that the Hamiltonian and spatial diffeomprphism 
constraints respectively generate spacetime diffeomorhisms 
transversal (vertical) and tangential (horizontal) to the Cauchy surfaces.
The functions $f^v,\; g^{v,l,m},\;f^h,\;g^{h,l,m},\; g^{\alpha,l,m},\;
f^M, g^{M,l,m}$ depend only on the radial coordinate $r$. The contributions 
to $S^0, S^3, S_0$ given by the functions 
$f^v, f^h, f^M$ are theorefore spherically symmetric and thus are referred 
to as ``symmetric'' smearing functions while the contributions 
to $S^0, S^3, S^A, S_0$ defined by the functions 
$g^{v,l,m}, g^{h,l,m}, g^{\alpha,l,m}, g^{M,l,m}$ 
are not spherically symmetric and 
thus are referred to as ``non-symmetric'' smearing functions.

We perform an analogous split for the fields 
$m_{\mu\nu},\;W^{\mu\nu},\;\Phi,\Pi,A_\mu,E^\mu$ 
\ba \label{2.13}
m_{33} &=:& q^v_E + \sum_{l>0,|m|\le l}\; x^{v,l,m}_E\; L_{l,m}
\nonumber\\
m_{3A} &=:& 0 + \sum_{l>0,|m| \le l,\alpha\in \{e,o\}}\; 
x^{\alpha,l,m}_E\; L_{A;\alpha,l,m}
\nonumber\\   
m_{AB} &=& q^h_E\;\Omega_{AB}
+\sum_{l>0,|m|\le l}\; 
x^{h,l,m}_E\; L_{AB;h,l,m}
+\sum_{l>1,|m| \le l,\alpha\in \{e,o\}}\; 
X^{\alpha,l,m}_E\; L_{AB;\alpha,l,m}
\nonumber\\
W^{33} &=:& \omega\;[p_v^E+\sum_{l>0,|m|\le l}\; y_{v,l,m}^E\; L_{l,m}]
\nonumber\\
W^{3A} &=:& \frac{\omega}{2}\;[0+\sum_{l>0,|m|\le l,\alpha\in \{e,o\}}\; 
y_{\alpha,l,m}^E\; L^A_{\alpha,l,m}]
\nonumber\\
W^{AB} &=:& \omega\;[\frac{p_h^E}{2} \;\Omega^{AB}+
\sum_{l>0,|m| \le l}\; 
y_{h,l,m}^E\; L^{AB}_{h,l,m}
+\sum_{l>1,|m|\le l,\alpha\in \{e,o\}}\; 
Y_{\alpha,l,m}^E\; L^{AB}_{\alpha,l,m}]
\nonumber\\
\Phi &=:& Q_{KG}+\sum_{l>0,|m| \le l}\; X^{l,m}_{KG}\; L_{l,m}
\nonumber\\
\Pi &=:& \omega\;[P^{KG}+\sum_{l>0,|m| \le l}\; Y_{l,m}^{KG}\; L_{l,m}]
\nonumber\\
A_3 &=:& q_M + \sum_{l>0,|m| \le l}\; x^{l,m}_M\; L_{l,m}
\nonumber\\
A_C &=:& 0 + \sum_{l>0,|m| \le l,\alpha\in\{e,o\}}\; X^{\alpha,l,m}_M\; 
L_{C;\alpha,l,m}
\nonumber\\
E^3 &=:& \omega\;[p^M + \sum_{l>0,|m| \le l}\; y_{l,m}^M\; L_{l,m}]
\nonumber\\
E^C &=:& \omega\;[0 + \sum_{l>0,|m| \le l,\alpha\in\{e,o\}}\; 
Y_{\alpha,l,m}^M\; 
L^C_{\alpha,l,m}]
\ea
We have payed attention to the fact that the momenta conjugate to the 
respective configuration variables are respective dual tensors 
that carry density weight one rather than zero and thus have pulled 
out a factor of $\omega=\sqrt{\det(\Omega)}$ (see the previous section).
The labels E, KG, M mean again Einstein, Klein-Gordon, Maxwell degrees of 
freedom.
We have grouped the coefficient functions of the tensor harmonics 
that appear in (\ref{2.13}) and which only depend on the radial coordinate 
$r$ into the following four groups following the general notation 
of \cite{pa129}:\\
1. symmetric gauge\\
$\{(q^a,p_a)\}:=\{(q^v_E,p_v^E),(q^h_E,p_h^E),(q_M,p^M)\}$\\
2. symmetric true:\\
$\{(Q^A,P_A)\}:=\{(Q_{KG},P^{KG}\}$\\
3. non-symmetric gauge ($l>0,|m|\le l,\; \alpha\in \{v,h,e,o\}$):\\
$\{(x^j,y_j)\}:=\{(x^{\alpha,l,m}_E,\;y_{\alpha,l,m}^E),(x^{l,m}_M,\;
y_{l,m}^M)\}$\\
4. non-symmetric true ($\alpha\in \{e,o\}$ and $l>0$ for KG,M
while $l>1$ for E):\\
$\{(X^J,Y_J)\}:=\{(X^{\alpha,l,m}_E,\; Y_{\alpha,l,m}^E),\;
(X^{l,m}_{KG},\;Y_{l,m}^{KG}),\;
X^{\alpha,l,m}_M,\; Y_{\alpha,l,m}^M)\}$.\\
The labels $a,b,c,..;\; A,B,C,..;\; j,k,l,..;\;J,K,L,..$ take the 
corresponding values which includes the value of the coordinate $r$.
The unfortunate doubling of the labels $A,B,C,..=1,2$ of components of tensors 
on $S^2$ with the range of $A,B,C,..=r\in [0,\infty)$ for $Q^A,P_A$ 
does not cause confusion 
because we will get rid of the spherical harmonics right away so 
that they play no role any more below.  

When plugging the decomposition (\ref{2.13}) into the symplectic potential
(\ref{2.11}) we find due to $\omega\; d^3z=4\pi\; dr\; d\mu$ and the 
normalisation of the spherical harmonics with respect to $\mu$ (see
the previous subsection)
\ba \label{2.14}
&& \Theta = p_a\; dq^a+P_A\; dQ^A+y_j\; dx^j+Y_J\;dX^J  
\\
&=& \int_0^\infty\; dr\;
\{
[p_v^E\; dq^v_E+p_h^E\; dq^h_E+p^M\; dq_M]
+[(P^{KG})^T\; dQ_{KG}]\}
\nonumber\\
&& +
\sum_{l>0,|m|\le l}
[(\sum_{\alpha\in \{v,h,e,o\}} \;y_{\alpha,l,m}^E\; dx^{\alpha,l,m}_E)
+y_{l,m}^M\; dx^{l,m}_M]
\nonumber\\
&& +\sum_{l>0,|m|\le l}\; (Y_{l,m}^{KG})^T\;dX^{l,m}_{KG}]
+
\sum_{\alpha\in \{e,o\}}\;
[\sum_{l>1,|m|\le l}\; Y_{\alpha,l,m}^E\;dX^{\alpha,l,m}_E
+\sum_{l>0,|m|\le l}\; Y_{\alpha,l,m}^M\;dX^{\alpha,l,m}_M]
\}
\nonumber
\ea
which shows that the pairs listed in the decomposition 
$(q,p), (Q,P), (x,y), (X,Y)$ are indeed conjugate so that we have 
for instance 
\be \label{2.15}
\{y_{\alpha,l,m}^E(r),x^{\alpha',l',m'}_E(r')\}=\delta^{(1)}(r,r')\;
\delta_\alpha^{\alpha'}\;\delta_l^{l'}\;\delta_m^{m'}
\ee
etc. (in particular Poisson brackets between fields from different 
species E,KG,M vanish) where $\delta^{(1)}(r,r')$ is the $\delta$ distribution
on the positive real line if we consider one asymptotic end.

\subsection{Perturbative decomposition}
\label{s2.3}

When plugging the decompositions (\ref{2.12}), (\ref{2.13})
into the Hamiltonian (\ref{2.10}) one would like to integrate out the 
angle dependence. This is immediately possible for the contributions 
to $H$ from spatial diffeomorphism and Gauss constraint 
$V_\mu, G$ (recall that we could already delete the 
piece depending on $v^0,v^\mu,v_0$) because these are homogeneous 
polynomials of degree three and and two respectively in the 
fields (\ref{2.12}) and (\ref{2.13}) and just requires to apply
the normalisation
of the spherical harmonics and the Clebsch-Gordan decomposition. However,
the contribution to $H$ from $V_0$ is non-polynomial in the 
metric field (all other degrees of freedom enter polynomially). 
While 
one could in principle try to integrate out the angular dependence 
of the non-polynomial constraints directly 
and non-perturbatively, it is not known whether one can   
actually do this in closed form. Moreover,
as we are interested in perturbation theory with respect to the 
non-symmetric degrees of freedom $x,y,X,Y$, we may as well perform 
the perturbative expansion before the angular integration. One 
then obtains a perturbation series which is infinite but only due 
to field $m_{\mu\nu}$, the series is finite as far as the other fields are 
concerned and each term in that series can 
be integrated in closed form using again Clebsch-Gordan theory.  

One can avoid this infinite series as follows: The Hamiltonian 
constraint $V_0$ depends on $\sqrt{\det(m)}^{\pm 1}$ in order that 
each term has net density weight unity. If we multiply $V_0$ by 
$\sqrt{\det(m)}$ then only $\det(m)$ appears which is a cubic polynomial
in $m$. Then $\sqrt{\det(m)}\; [V_0^M+V_0^{KG}$ and the piece of 
$\sqrt{\det(m)}\; V_0^E$ not involving the Ricci scalar is already 
polynomial. The Ricci scalar contains a term linear that is 
derivated and a term quadratic but without derivatives  
in the Christoffel symbol 
$\Gamma^\mu_{\nu\rho}=m^{\mu\lambda}\;\Gamma_{\lambda\nu\rho}$ 
where $\Gamma_{\lambda\nu\rho}$ is homogeneously linear and 
these terms are contracted with the inverse metric. 
Thus $\tilde{V}_0:=\sqrt{\det(m)}^5\; V_0$ is polynomial in all variables.
It is quadratic in all momenta $W^{\mu\nu},\; \Pi,\; E^\mu$, quadratic
in $A_\mu$ and quadratic in $\Phi$ if $U$ is just a mass term, otherwise 
higher if $U$ is a higher order polynomial. 
Since $\det(m)\; m^{\mu\nu}$ is a polynomial of degree two, in $\tilde{V}_0^E$
all terms are of order eight in   
$m_{\mu\nu}$, in $\tilde{V}_0^{KG}$ the term independent of respectively 
dependent on $\Pi$ has 
degree eight (or higher if there is a non quadratic potential) respectively six
and $\tilde{V}_0^M$ has 
degree seven in $m_{\mu\nu}$. Thus $\tilde{V}_0$ is polynomial in all 
fields. Since the smearing function $S^0$ is arbitrary, we can absorb 
$\sqrt{\det(m)}^{-5}$ into it, thererby defining 
$\tilde{S}^0=[\det(m)]^{-5/2} S^0$ which has density weight minus five.
Then we can still use the first formula in (\ref{2.12}) with 
$S^0$ replaced by $\tilde{S}^0$ if we multiply its right hand side by 
$\omega^{-5}$. 
 
To see that this is allowed, note while the ``rescaling'' of $S^0,V_0$ 
by $\gamma:=\det(m)]^{5/2}>0$ does have a non-trivial effect on both 
the constraint and the smearing function, it has absolutely no effect 
on the reduced Hamiltonian which is what we are interested in. 
To see the latter, recall that given gauge fixing conditions say of the form 
$G=K-\tau$ ($K$ are functions on the phase space, $\tau$ are some 
time dependent coordinate conditions) for constraints $C$ with smearing
functions $f$ the reduced Hamiltonian $H_r$ acting on functions $F$ 
of the true degrees of freedom is computed by the formula
$\{H_r,F\}=\{C(f),F\}_{C=G=f-\hat{f}=0}$ where $f=\hat{f}$ solves 
$\{C(f),G\}=\dot{\tau}$. Since $C,\tilde{C}=\gamma C$ have the same 
zeroes, the gauge fixing is the same and the matrices $\gamma$
cancel in $H_r$ because $\hat{\tilde{f}}=\gamma^{-1}\hat{f}$ when 
$C=0$.  
 
With this understanding, all constraints are finite polynomials in all 
fields (of top degree twelve e.g. for a charged scalar 
field and for at most quadratic potential otherwise
of degree nine plus the degree of the potential) 
and plugging in
(\ref{2.12}), (\ref{2.13}) allows us to carry out all angular integrals in 
closed form, thereby yielding an {\it exact} closed expression for 
the Hamiltonian written as a polynomial of degree two in all $p,P,y,Y$; 
of degree two in $x_M,X_M$; of degree two or the degree of the potential 
in $Q_{KG}, X_{KG}$; of degree at most nine in $q_E,x_E,$. 
Thus, we obtain an expression
for $H$ that is known {\it non-perturbatively}. However, note that this 
perfectly allowed reformulation of $H$ in terms of $\tilde{V}_0$ 
(relying on the assumed non-degeneracy of $\det(m)$) which enables us 
to carry out all angular integrals, does not prevent the {\it exact} 
reduced or physical Hamiltonian from being non-polynomial in $Q,P,X,Y$.
This is because in its computation \cite{pa129} we must solve for the 
momenta $p,y$ which appear non-linearly (namely quadratically) 
in $\tilde{V}_0$ and thus their solution leads to square roots. 
The algorithm of \cite{pa129} computes that square root perturbatively
in $X,Y$ which thus involves again an infinite series. Nevertheless 
the computational effort when working with $\tilde{S}^0,\tilde{V}_0$
is significantly smaller than when working with 
$S^0,V_0$ because 1. the number of necessary Clebsch-Gordan
decompositions required  is finite and can 
be performed in closed form, hence there are no perturbative 
errors at this stage and 2.
since $\tilde{V}_0$ does not involve an infinite 
series while $V_0$ does, the 
perturbative solution of of $\tilde{V}_0$ 
in terms of $x,y,X,Y$ is tremendously simplified. Of course,  
whether one works with 
$D_0$ or $\tilde{V}_0$, the perturbative solution to both constraints 
including the angular integrals is the same, it is just that $\tilde{V}_0$
is significantly more convenient.    

With these preparations and dropping the tilde in $\tilde{S}^0,\tilde{D}_0$
again we can thus write the Hamiltonian in the 
form 
\be \label{2.16}
H= f^a \;C_a+g^j \;Z_j        
:=\int_{\mathbb{R}^+_0}\; dr\;\{
[f^v \; C_v+f^h\; C_h+f^M\; C_M]+\sum_{l>0,|m|\le l}\;
[\sum_{\alpha\in \{v,h,e,o,M\}}\; g^{\alpha,l,m}\; Z_{\alpha,l,m}]
\;\}
\ee   
where ($\alpha\in \{e,o\},\; l>0$)
\ba \label{2.17}
&& C_v:=<1,V_0/\omega^6>_{L_2},\;
C_h:=<1,V_3/\omega>_{L_2},\;
C_M:=<1,G/\omega>_{L_2},\;
\\
&& Z_{v,l,m}=<L_{l,m}, V_0/\omega^6>_{L_2},\;
Z_{h,l,m}=<L_{l,m}, V_3/\omega>_{L_2},\;
\nonumber\\
&& Z_{\alpha,l,m}=<L_{\alpha,l,m}, V_{.}/\omega>_{L_2^2};\;\alpha\in\{e,o\},\;
Z_{M,l,m}=<L_{l,m}, G/\omega>_{L_2}
\nonumber
\ea
where we refer to $C_a,Z_j$ respectively as the symmetric and non-symmetric 
constraints respectively. Note however that each of them depends on all
degrees of freedom and thus Poisson brackets of $C_a,Z_j$ respectively 
also affect $x,y,X,Y$ and $q,p,Q,P$ respectively.            

To compute the inner producs (\ref{2.17}) we expand the polynomials 
$C_a, Z_j$ into its homogeneous pieces $C_{a(n)}, Z_{j(n)}$ respectively 
where the notation means that e.g. $C_{a(n)}$ is a homogeneous polynomial 
of degree $n\ge 0$ in $x,y,X,Y$. Then the integrals over $S^2$ in 
(\ref{2.17}) involve $n$ or $n+1$ tensor harmonics respectively for 
the contribution $C_{a(n)}, Z_{j(n)}$ respectively. Since 
$<1,L_{l,m}>_{L_2}=0,\; l\not= 0$ 
it follows immediately the simple but very powerful
observation that 
\be \label{2.18}
C_{a(1)}=Z_{j(0)}=0
\ee
which turns out to be crucial in order for the 
perturbative construction algorithm for 
the physical Hamiltonian to work. On the other hand, the 
contributions from $C_{a(0)},C_{a(n)};n\ge2,\;,Z_{j(n)};n\ge 1$ 
are in general 
not vanishing. One computes them explicitly using Riemann tensor calculus 
and harmonic analysis on the sphere as well as Clebsch-Gordan decomposition.
Explicitly one finds for say vanishing Klein Gordan potential $U$ that 
\ba \label{2.19}
&&
C_{v(n)},\;1\not=n\le 12;\; C_{h(n)},\; 1\not=n\le 2;\; C_{M(n)}; n=0,2,\;
Z_{v,l,m(n)},\; 0<n\le 12;\;         
\nonumber\\
&& Z_{\alpha,l,m(n)},\;\alpha\in \{h,e,o\}, 0<n\le 2;\;         
Z_{M,l,m(n)},\; n=1,2
\ea
are non-vanishing and can be computed in closed form.

\section{Concepts of quantum black hole perturbation theory}
\label{s0}

The purpose of this section is to review several concepts of (quantum) black 
hole perturbation theory in a non-technical fashion before we 
go into technical details in later sections. The aim is  
to explain these notions and their interrelations in order to erect a 
consistent 
conceptual picture.

\subsection{Observables, backreaction and black hole evaporation}
\label{s0.1}

By backreaction we understand the interaction between the 
spherically symmetric and spherically asymmetric true degrees of freedom. 
The symmetric and asymmetric true degrees of freedom respectively 
are basically the $l=0$ and $l\ge 1$ multipole
moments respectively of the various observable (or true) fields 
(geometry and matter) as described in the previous section. With this 
understanding of backreaction, the following issue arises: Consider 
first the Einstein-Maxwell sector, i.e. there is no additional 
e.g. scalar matter ``hair''.
Then by Birkhoff's theorem \cite{HE}, 
the symmetric sector (no radiation, i.e. no multipoles)
is uniquely described by two 
variables, namely the black hole mass $M$ and charge $Q$. These arise 
as integration
constants when solving the perturbative Einstein equations also when the 
multipoles are non-vanishing. The physical Hamiltonian then will be 
a functional of the asymmetric true degrees of freedom denoted by $X,Y$ 
for the Einstein-Maxwell sector and a function of $M,Q$, say $H[M,Q;X,Y]$.
By construction, $M,Q$ have vanishing Poisson brackets with $X,Y$
and among themselves. Therefore,
$M,Q$ would be constants of motion with respect to $H$ and while there is 
gravitational and electromagnetic radiation described by $X,Y$, 
certainly $M,Q$ would be unchanged
by the dynamics described by $H$, in particular, $M,Q$ could not evaporate.
This is in contrast to the situation in cosmology where the physical 
Hamiltonian, e.g. deparametrised with respect to the homogeneous 
mode of a scalar field $\phi$, does depend 
on the scale factor $a$ {\it and} its conjugate momentum $p_a$ 
so that there {\it is} backreaction in the above sense between $a,p_a$ and 
the inhomogeneous modes of both matter and geometry \cite{ST}. In spherical
symmetry, the situation for $M$ changes e.g. when introducing neutral scalar 
matter and for $Q$ it changes when e.g. introducing charged scalar matter.
In that case, $M,Q$ are simply absorbed into true symmetric (in this 
case) scalar matter degrees of freedom with which there is non-trivial 
backreaction.  

Thus it would seem that in Einstein-Maxwell theory $M,Q$ could not evaporate,
{\it not even in the quantised theory} which would include a quantisation
of $H[M,Q;X,Y]$.
While for $Q$ this is expected as photons do not carry charge, for $M$ 
this is non-trivial: In the classical theory it is a manifestation 
of the black hole area theorem since Maxwell matter obeys the weak 
energy condition. However in the quantum theory where the weak energy 
condition is typically violated locally, one would not expect that $M$ 
does not evaporate. In particular, in Einstein-Maxwell theory the above 
argument suggests that even including backreaction {\it there is no 
dynamical mechanism for black hole evaporation} which is in contradiction 
to the semiclassical argument that the existence of Hawking radiation 
predicts black hole evaporation irrespective of the matter species.
\\
Two ways out of that conclusion suggest themselves:\\ 
1.\\ 
The first way out is based
on the observation that Birkhoff's theorem treats {\it all} spacetime 
diffeomorphisms as gauge transformations. However, in the Hamiltonian 
framework one makes a finer distinction between diffeomorphisms that 
generate non observable gauge transformations and those that are 
observable symmetry transformations. If one adopts that Hamiltonian 
point of view which is consequential within this purely Hamiltonian 
treatment of black holes, then additional observables, namely 
momenta $P_M,P_Q$ conjugate to $P,Q$ are unlocked. If the spacetime 
metric depends at least on both $(M,P_M)$ then $M$ 
is no longer conserved even in pure Einstein-Maxwell theory and can 
possibly evaporate in the quantum theory. We will see that 
$P_Q$ does not enter the spacetime metric while $P_M$ does or does 
not, depending on the way that the expression that defines $P_M$ is 
made compatible with the chosen gauge fixing condition.\\
2.\\
The second way out is to accept the absence of $P_M$ from 
the reduced Hamiltonian and consists in interpreting $M$ not as the dynamical 
mass but just as an integration constant, namely the initial mass 
$M(0)$ in (\ref{1.1}). The role of the dynamical mass must then be 
played by another object. It cannot be the ADM mass which is 
basically the reduced Hamiltonian and which is therefore also conserved.
The natural notion of dynamical mass is the square root of the area of the 
apparent horizon with respect to the foliation selected by the gauge fixing 
conditions (equivalent to the selection of an observer congruence), which 
coincides with the notion of irreducible mass for the case that apparent and 
event horizon coincide.\\ 
\\
In the following two subsections we will spell out some of the details 
of these two possibilities. In the main part of the paper we 
adopt the second point of view as it appears to be less 
sensitive to the choice of gauge fixing condition but keep the first 
point of view in mind for potential future applications. Some of the 
possible technical implementations of the first viewpoint can be found 
in appendices \ref{sa}, \ref{sb}, \ref{sd}.

\subsubsection{Dirac observable conjugate to the mass}
\label{s0.1.1}

In \cite{Kastrup,KucharSS} it was observed 
that the reduced phase space of a 
spherically symmetric vacuum black hole 
is not described by just the a
mass $M$ but also its conjugate momentum $P_M$ and a scaling 
parameter $\kappa$. Both $M,P_M$ are Dirac observables,
that is,
functionals of the canonical variables of a vacuum black hole (or the 
symmetric degrees of freedom) that  
have vanishing Poisson brackets with the constraints. Moreover, 
$M,P_M$ are canonically conjugate. The number $\kappa$ enters 
the physical Hamiltonian $H=\kappa\; M$. In appendices \ref{sa}, \ref{sb} 
we explain in non-technical terms why this happens: Essentially, one 
can perform a canonical transformation to conjugate variables $m,p_m$ 
and the constraints impose that $m'=0$. This leaves an integration 
constant $m=M$ as solution. One then shows that $P_M:=\int\;dr\; p_m$
is gauge invariant, thus cannot be gauged away and is conjugate to $M$.
Finally, the transformations that stabilise any choice of gauge consistent 
with the value $P_M$ have the generator $\kappa M$ where $\kappa$ is 
arbitrary. The existence of $\kappa$ can be understood from the 
fact that $C(r)=m'(r)=0$ is equivalent to 
$\tilde{C}(r)=m(r)-m(0)=0$ which identically satisfied 
at $r=0$ thus the constraints $C(r)$ are redundant. Equivalently, when 
solving the stability condition for transformations preserving a gauge 
compatible with $P_M$ one must solve a differential equation for 
the smearing function of the constraint which has a free integration 
constant which is $\kappa$. All of this 
has to be done with due care paying attention to decay conditions, boundary
terms, finiteness of both symplectic structure and constraint integrals 
and  functional differentiability of the constraints,  see appendix 
\ref{sa}.    

In \cite{Kastrup} it 
is explained why the existence of $P_M$ is not in conflict with Birkhoff's
theorem: Indeed, Birkhoff's theorem says that in asymptotically flat regions 
of spherically symmetric spacetimes one can pass to coordinates in which 
$P_M$ vanishes. By carefully investigating the required temporal 
diffeomorphism, one observes that this diffeomorphism is asymptotically 
non-trivial and thus is to be considered as a {\it symmetry} transformation 
rather than a {\it gauge} transformation \cite{HRT} in the Hamiltonian 
setting. Thus the existence of $P_M$ comes about due to the different notions 
of gauge in the Lagrangian (as used in Birkhoff's theorem) and Hamiltonian 
setting respectively. To make our exposition self-contained, we will 
review this subtle difference of gauge in appendices \ref{sa}, \ref{sb}.
A similar observation was made in \cite{AshtekarSamuel} in the context 
of cosmological models.

Note that the existence of a two dimensional rather than one 
dimensional reduced phase space of the 
sperically symmetric vacuum sector is also natural from the point of view 
of symplectic reduction of phase spaces with respect to first class
constraints which always returns an even dimensional reduced phase space 
at least in the case of finite dimensional unreduced phase spaces. 
The idea would then be that the reduced Hamiltonian depends on both 
$M,P_M$ and 
$X,Y$ so that $H[M,P_M;X,Y]$ mediates an interaction between both 
types of degrees freedom. In the quantum theory, that interaction 
can then be treated e.g. using the methods of 
space adiabatic perturbation theory \cite{SAPT} (SAPT), a generalisation 
of the Born Oppenheimer approximation scheme that can deal with the situation 
that the interaction depends on both slow degrees of freedom $M,P_M$ rather 
than just one, similar to the analysis performed for 
cosmology \cite{ST}. The SAPT scheme then 
would produce an effective Hamiltonian for $M,P_M$ that would take this notion 
of backreaction into account.           

To make this work, one has to supply a missing ingredient to the works 
\cite{Kastrup,KucharSS}, namely to explain how the 
functionals $M[q,p],P_M[q,p]$ of the spherically symmetric 
intrinsic metric $q$ and its conjugate
momentum $p$ can give rise to prescribed values $M,P_M$ when one solves
the constraints for $p$ with some choice of gauge for $q$ installed.
Since the constraints can be solved for $p(r)=p(r;M,q)$ for general 
$q$ where $M$ is an integration constant, apparently   
the only solution to this problem appears to be that the gauge for $q$ 
must depend on both $M,P_M$ or at least on $P_M$. 
But then the following puzzle arises:
Since the physical Hamiltonian is essentially the black hole mass $M$,
while $M$ is a constant of physical motion, $P_M$ is not. This would 
mean that the metric $q$, even in absence 
of the perturbations $X,Y$ and outside the horizon 
is potentially not static and the effect is not necessarily 
small as $\dot{P}_M=O(1)$ and thus $|P_M|$ is unbounded in time. 
Below we offer three possible resolutions 
to this puzzle which all rest on the fact that the formal expression 
for $P_M[q,p]$ is actually an ill-defined integral, which requires a 
more careful definition,
for the standard choices of gauge.
The three resolutions differ in the way that this integral is regularised 
given a choice of gauge. The choice of gauge we employ respectively 
will be closely related to the (generalised) Gullstrand Painlev\'e gauge 
(GPG) \cite{26} which we review in appendix \ref{sd}. 
This gauge choice is motivated by the fact that it is both 
adapted to the spherical symmetry of the problem and natural from the 
point of view of QFT in curved spacetimes (Hawking radiation) and black hole 
-- white hole transitions as we explain further below.  
\begin{itemize}
\item[A.] 
Since $\tilde{P}_M[q,p]:=P_M[q,p]-f(M)$ is still conjugate to $M$ 
for any function $f$ of $M$, it is possible to obtain a finite expression 
by choosing the exact Gullstrand-Painlev\'e gauge (GPG) 
\cite{26} which is independent
of both $M,P_M$, except for an 
arbitrarily small neighbourhood of the origin where the metric is singular
anyway. The coordinate size $L$ of that neighbourhood does not grow with time 
and can be chosen to be of at most 
Planck size so that this deviation is hidden 
behind the horizon even for Planck size black holes. Yet, the deviation
can bechosen to depend on $L,M,P_M$ in such a way that the prescribed value 
for $P_M$ is obtained from $\tilde{P}_M[q,p]$ on the reduced phase space 
(i.e. both constraints and this gauge are installed). This will be described 
in appendix \ref{sd}. A variant of this is to glue two asymptotic ends 
along the cylinder $r=l<2M$. Then it turns out that the corresponding 
vacuum solution reaches expnentially fast the Einstein-Rosen bridge 
solution $l=2M$, see appendix \ref{sb}.
\item[B.] 
Another way to regularise the integral is to take principal values of 
the integral which has singularities at $r=M,r=\infty$. We consider 
the generalised GPG \cite{26} which depends on an additional parameter
$e$ corresponding the energy of a timelike radial 
geodesic observer on which more
will be said below. It is then possible to regularise the integral such that 
$c\;P_M={\rm arth}(e/e_0)$ for some fixed numerical value $c,e_0$. Then as 
$|P_M|$ grows, $e$ approaches a constant value $e_0$, the more rapidly the 
larger $|c|$, e.g. $e_0=1$ which 
is the exact GPG. Thus while $e$ is not a constant of physical motion, it 
quickly reaches a quasi constant value $e_0$.
\item[C.] We can pick the exact GPG and still regularise the integral such 
that the given value of $P_M$ results. 
\end{itemize}
Option C is the simplest and while $P_M$ exists it does not show up 
in the gauge fixed $q$, its value is simply a regularisation ambiguity.
For the electric charge this is automatically the case, i.e. the momentum 
$P_Q$ conjugate to $Q$ does not enter the gauge fixed 
metric and therefore the reduced Hamiltonian. Option B is quasi-equivalent 
to option C rapidly in time. Option A in the frirst variant 
is almost (locally in space) 
the exact GPG, it has the disadvantage 
to make perturbative calculations more complicated. Common to these options 
is that $P_M$ either does not show up in $q$ at all (option C) or 
is quasi absent either with time (option B) or spatially locally 
(option A, first variant) or the interior of the black hole is removed
so that $P_M$ becomes a function of $M$ 
(option A, second variant). 
In that sense, backreaction on $M$ via the interaction between $M,P_M,X,Y$ 
is either exactly or quasi absent.

Note that the generalised
GPG with parameter $e_0$ are spacetime diffeomorphic to the standard SS 
coordinate solution for $r>2M$ but the required temporal diffeomorphism
involves in all options a rescaling of the time coordinate (i.e. the 
lapse is asymptotically different from unity). Hence it 
is asymptotically non-trivial and should therefore not be considered a 
gauge transformation. 

To see which of these three options is preferred one may invoke the 
mathematical argument to have a match with the Kantowski-Sachs (KS)
picture which underlies most 
of the treatments of the quantum black hole with LQG methods (LQBH) 
\cite{LQG-BH,SF-BH} building on the huge amount of experience with 
the treatment of quantum cosmology with LQG methods (LQC) \cite{LQC}.
We will review this in appendix \ref{se}. Note 
however two caveats with that argument: First, the equivalence
of both pictures is due to the fact that the GPG vacuum solution is 
stationary while the KS vacuum solution is homogenous so that one can 
just switch the roles of time and space. This is no longer true 
with matter coupling unless the 
spherically symmetric 
matter sector also leads to only stationary solutions 
(such as TOV spacetimes with 
perfect fluid energy momentum tensors and suitable equations of state between 
pressure and energy density \cite{12a}). Second, the mathematical 
equivalence between the KS picture and the generalised GPG picture 
rests on considering all spactime diffeomorphisms as gauge transformations 
and thus is insensitive precisely on the issue about the status of 
existence of $P_M$ which relies on a finer classification 
of diffeomorphisms. 

Nevertheless, it is of interest to understand the correspondence between 
the two pictures in the Hamiltonian setting:\\
In the KS picture one considers the Hamiltonian analysis of 
homogeneous, spherically 
symmetric but anisotropic 
KS cosmologies which depend on two scale factors $A,B$ and their conjugate 
momenta $p_A,p_B$. The spatial diffeomorphism constraint vanishes 
identically and 
there is only one remaining Hamiltonian constraint. Hence this is a mechanical 
system with a 4d phase space and one (necessarily first class) constraint 
so that the counting of degrees of freedom is very simple: The reduced phase 
is 2d and there are exactly two algebraically independent and canonically
conjugate Dirac observables. One can develop both the relational observables 
and true degrees of freedom picture and compute the physical respectively 
reduced Hamiltonian. The solutions of the physical equations of motion 
show that on shell the scale factors and lapse are exactly those of the  
Schwarzschild interior solution with the switch between radial and temporal 
coordinate 
understood, up to one exception: While the first observable is associated 
with the mass, the second corresponds to a rescaling of the time coordinate
and can be interpreted as a ``clock ticking rate''. It is exactly 
the same scaling parameter $\kappa$ that occurs also in the above GPG 
picture (which covers both interior and exterior of the black hole).
Accordingly, in the KS picture the possibility for $P_M$ to show up 
in the physical metric never occurs. It follows that the two pictures agree 
if we pick for the GPG picture 
option C i.e. we install the exact GPG and ascribe $P_M$ to a 
regularisation freedom.

It is quite interesting to see how this happens: In the KS picture we 
start from a finite dimensional phase space with a very transparent 
counting of degrees of freedom. On the other hand in the generalised GPG 
picture we start with a field theory i.e. an infinite dimensional phase 
space before the constraint treatment and there a complicated set of 
issues such as spatial fall-off conditions, boundary terms, solving 
spatial differential equations etc. arise. Because of this the 
construction of the reduced phase spaces and physical Hamiltonians is 
quite different. In the cosmological KS picture the physical Hamiltonian 
simply results from a gauge fixing conditions and the effective 
equations of motion for the corresponding true degrees of freedom while 
in the GPG picture it is the boundary term that must be added to make 
the constraints functionally differentiable which drives the dynamics of 
the true degrees of freedom. While in the KS picture we have 2 true 
degrees of freedom or Dirac observables $M,\kappa$, in the GPG picture 
we have 2 observables $M,P_M$ and an additional integration constant 
$\kappa$ which arises due to constraint redundancy while in the KS 
picture there can be no such redundancy. In option C above we can 
discard $P_M$ as it does not show up anywhere in the metric and then 
both pictures match as far as the surviving parameters, namely 
$M,\kappa$, is concerned. However, while $M,\kappa$ have non vanishing 
brackets in the KS picture, in the GPG picture $\kappa$ is considered a 
phase space independent constant. The resolution of this apparent 
contradiction is as follows: In the KS picture, the physical Hamiltonian 
and $M,\kappa$ are {\it explicitly time dependent} with respect to KS 
time $T$. However, on solutions of the equations of motion they become 
constants of motion. By contrast, in the exact GPG picture the physical 
Hamiltonian is conservative, i.e. like $M,P_M$ not explicitly time 
dependent with respect to GPG time $\tau$. On solutions, also $M$ is a 
constant of motion and $\kappa$ was a time independent constant from the 
outset. Note that since essentially the KS time $T$ is the radial 
variable, $T$ independence in the KS picture translates into radial 
independence which again brings both pictures into congruence.\\ \\ To 
summarise: For the rest of this paper we will follow option C when 
considering non vacuum spacetimes. This means that the mass parameter 
$M$ is a constant of physical motion and can therefore be called the 
``remnant mass'' because the physical Hamiltonian reduces to it when 
perturbations and matter are absent. There is no dynamical mechanism 
that can change $M$ because while the second Dirac observable $P_M$ 
exists it does not enter the physical Hamiltonian. We work in the exact 
GPG. Still we may use the above introduced cut-off $l=r$ as a 
regularisation method when we compute the black hole white hole 
transition mode functions, see section \ref{s7}.

\subsubsection{Apparent horizon area}
\label{s0.1.2}

Consider then $M$ as an integration constant. It 
happens to coincide in the exactly 
spherically symmetric case with many different definitions of mass in 
general spacetimes that have been discussed in the literature such as 
the ADM mass \cite{HRT} or the irreducible mass 
\cite{12a,HE}. The ADM mass is expected to be 
the leading term of the physical Hamiltonian $H$ for small 
perturbations and thus is in particular 
preserved (togther with $M$ if $P_M$ is absent in the metric), 
hence ``evaporation of $H$'' 
is again not possible. However, the irreducible mass with respect to 
the apparent horizon serves as a 
suitable more direct measure of backreaction. Recall that given a 
foliation of $(M,g)$ by spacelike hypersurfaces $\tau\mapsto\Sigma_\tau$ a 
compact 2-surface $S_\tau\subset \Sigma_\tau$ without boundary is called outer 
marginally trapped if $\theta_+=0,\;\theta_-<0$ where $\theta_\pm$ are 
the expansions of the null normals $l_\pm=n\pm s$ with $g(n,n)=-1, 
g(s,s)=+1,\; g(n,s)=0$ with future oriented timelike unit normal $n$  
to $\Sigma_\tau$ and 
outward oriented spacelike unit normal $s$ to $S_\tau$ 
and tangential to $\Sigma_\tau$. 
The outermost trapped surface 
(2-dimensional) is called the apparent horizon $\mathfrak{H}_\tau$ at $\tau$ 
and the union of apparent horizons as $\tau$ varies is called 
trapped horizon (3-dimensional) $\mathfrak{H}$. 
Finally, the irreducible mass $M_\tau$ at $\tau$ is the square root of the 
area $A_\tau$ of $\mathfrak{H}_\tau$ (up to a constant factor; abusing 
terminology, the actual definition refers to the event horizon rather 
than the apparanet horizon). 

Following \cite{PertApparentHorizon} it 
is not difficult to show, see section \ref{s6},
that indeed one can uniquely determine $A_\tau$ to 
arbitrary order in perturbation theory directly in the Hamiltonian 
setting and thus obtain a notion of time dependent mass that can 
possibly evaporate. This notion of mass is also operationally preferred 
as an astrophysicist would recognise a black hole as a marginally outer 
trapped region (which is the condition that the light rays leaving 
$S_\tau$ orthogonally are marginally converging). Now the following issue 
arises: Under a combination of the usual assumptions, 
namely that the classical Einstein
equations hold, that the energy conditions for the energy momentum tensor 
are satisfied, the validity of cosmic censorship 
and global hyperbolicity, classical GR makes two 
predictions: First, that the existence of trapped surfaces implies the 
existence of a singularity and second, if that singularity is a black 
hole (rather than a naked one forbidden by cosmic censorship)
then the trapped surface is within the black hole region. Therefore,
if the astronomer is outside the black hole region, she will indeed 
measure the event horizon as the sphere of no escape. From that point 
of view the apparent horizon appears to be of no use except that one 
expects it to be a good approximation (or at least a lower bound) of the 
event horizon. The apparent horizon has the advantage of being 
less teleological than the event horizon but it has the disadvantage of 
being foliation dependent. We do not share that criticism but rather 
accept that the notion of irreducible 
mass {\it is} observer dependent and here the
foliation is selected by the (generalised) Gull-strand-Painlev\'e gauge
which we motivate in the next subsection.
We refer the reader to the rich literature on apparent horizons 
and its specialisations (e.g. dynamical and isolated horizons) \cite{14a}. 
 
Now again by classical 
GR the area of the event horizon cannot decrease, therefore in classical 
GR there is no evaporation possible. We need quantum theory to have 
evaporation by violating at least one of the assumptions
of classical GR. The most obvious one is the violation of the 
energy conditions, recalling 
similar quantum violations of classical (energy) inequalities in QFT in flat 
spacetime \cite{Fewster}. 
In that case one can have trapped regions which evolve 
dynamically while there is no event horizon at all if for example 
the singularity is resolved as the existence of trapped regions no 
longer implies the existence of a singularity (assuming that 
the classical reasoning can be applied at all, at least in a semiclassical
sense). Even if there still is an event horizon, the violation 
of the energy conditions now no longer implies the apparent horizon 
to lie in the black hole region, it can even lie outside. This picture 
is confirmed in exactly solvable 2d models such as the CGHS black hole 
solutions including matter \cite{15}.

Now in perturbation theory the zeroth order 
is a spherically symmetric vacuum black hole if there is no spherically 
symmetric matter hair and that spacetime does have a singularity and 
an event horizon. The second order describes the perturbations as 
propagating on that singular spacetime and can e.g. serve to start 
a Fock quantisation of the system. Therefore it appears strange to perturb 
a singular spacetime with an event horizon if one expects that in 
quantum theory the non-perturbative spacetime is in fact free of 
singularities and therefore does not have an event horizon, that is, 
it seems that black hole perturbation methods are unable to capture the 
actual non-singular nature of the quantum theory. Our point of view here 
is that the causal structure of the spacetime including its singularities and 
horizons is itself subject to perturbation theory. Hence even at second 
order one would compute e.g. the expectation value with respect to an 
initial state (using the Heisenberg picture) of the (perturbed)
metric tensor, curvature tensor, reduced Hamiltonian, apparent horizon
and its area etc. using their classical expressions and subsitute it by 
the corresponding operator valued distributions. Altogether this 
describes a new effective spacetime metric to which we may apply the 
usual classical GR 
definitions of singularities and horizons to calculate a quantum
corrected Penrose diagramme (at least in regions where that effective 
metric has small quantum fluctuations). This process can 
be iterated at higher orders of perturbation theory and presents 
a drastic form of backreaction. It goes beyond the semiclassical 
Einstein equations in which one defines the Einstein tensor as the 
expectation value of the matter energy momentum tensor 
$G(g):=<T(g)>$ and tries to find a self-consistent metric $g$ solving 
this equation \cite{RunAway}
because here the (perturbations of) the metric are also 
quantised in the Heisenberg picture (the 
Heisenberg equations follow from the reduced Hamiltonian, to any 
order in perturbation theory, using the Gell-Mann Low formula if 
we use Fock quantisation). A possible mechanism for violation of energy 
conditions is due to normal ordering prescriptions of $<T(g)>$. 

If the singularity disappears in this 
process, what forms and evaporates is then not the event horizon 
but the apparent horizon, see also \cite{14b}. Therefore perhaps one    
should in fact set $M=0$ at zeroth order i.e. start with Minkowski space 
and follow the above process from formation to evaporation although 
it is unclear whether starting from $M=0$ and Minkwoski space Fock 
spaces for the perturbations and an initial coherent state 
peaked at gravitational collapse initial data one really obtains
a collapse - evaporation process in the quantum theory. 
For the sake of generality in this paper we handle the general 
$M>0$ case but motivated by the quantum theoretical considerations
feel free to regularise the singular $M>0$ background spacetime when needed.

It is remarkable that the innocent looking integration constants 
$M,Q$ have such a tremendous impact on the whole quantisation process.
Namely they decide whether at second order perturbation theory we
consider Fock representations on singular or non-singular spacetimes, 
with or without horizons.\\  
\\ 
We close the discussion by mentioning the following 
observations: \\
A.\\
The components of the spacetime metric $g$ are specific functions 
of $M,X,Y$. These arise
as follows: One imposes the GPG, fixing the components of 
the spatial metric $m$ different from 
$X$, solves the constraints for the components of $p$ different from $Y$ 
and solves for lapse $S^0$ and shift $S^a$ 
using the stability condition of the imposed 
gauge under gauge transformations. The irreducible mass then is also a 
specific function $m=m[M,X,Y]$ of these true degrees of freedom. As we will
show in section \ref{s6} we have $m=M+m_2[M,X,Y]+m_3[M,X,Y]+..$ where 
$m_n$ is of $n-$th order in $X,Y$. Now even if radiation described 
by $X,Y$ is {\it produced} only in a compact spacetime region $R$, since
$X,Y$ have to obey wave equations that radiation is generically non vanishing 
in the entire causal future $J_+(R)$ of that region (the causal 
future is of course also influenced by the amount of radiation 
present as it perturbs the metric). Given a 
foliation of $(M,g)$ by Cauchy surfaces $\Sigma_\tau$ let $\tau_0$ be the 
latest foliation parameter such that $\Sigma_\tau\cap R\not=\emptyset$. 
Then still $\Sigma_\tau\cap J_+(R)\not=\emptyset$ for all 
$\tau\ge \tau_0$. Thus a timelike observer with eigentime $\tau$
will eventually enter 
$J_+(R)$, however for sufficiently large timelike distances from $R$
the signal described by $X,Y$ will be weak. Thus the spatial metric 
$m$ returns to almost strict GPG for sufficiently large $\tau$ because 
it is a spatially {\it local} function of $X$, namely 
$q-g^{{\sf GPG}}=X$. On the other hand, the solution 
of $p,N,N^a$ at given $\tau$ is a spatially 
{\it non-local} function of $X,Y$ involving integrals over the entire 
hypersurface $\Sigma_\tau$. This is because the constraint and stability 
equations are PDE's and not algebraic equations. Since 
$\Sigma_\tau\cap J_+(R)$ becomes larger in volume the later 
$\tau$, these integrals can counter balance the decay of $X,Y$ and lead
to strong deviatiations of lapse and shift from their pre-radiation 
values which are $S^0=1$ and $S^a=\delta^a_3 \sqrt{2M/r}$ for all 
$\tau\ge \tau_0$ in the causally allowed region of spacetime.
As a measure of this deviation we may introduce the 
effective mass by $\sqrt{2m_{{\sf eff}}(r,\tau)/r}:=\int_{S^2}\;d\Omega(y)
S^3[M,X,Y](r,y,\tau)$ which therefore can deviate from $M$ for all 
$\tau\ge \tau_0$ and can potentially vanish therefore describing 
the evaporation effect. Since the spatial 
integral over a fixed $\tau$ hypersurface captures non-linear contributions 
from the ``gravition'' fields
$X,Y$ related to their data in $R$ by corresponding retardation,
this may be considered as an instance of a non-linear memory effect
\cite{Memory}.\\ 
B.\\
The apparent horizon at $\tau$ is defined by a radial profile function 
$\rho:\;S^2\to\mathbb{R}_+;\;$ which depends on $X,Y$ which become quantum 
fields. In that sense the coordinate location of the apparent horizon becomes 
quantised, subject to quantum fluctuations. This on the one hand is very 
similar to the construction of quantum reference frames \cite{QRF} and 
on the other hand intutively explains why the black hole area theorem can 
be violated in the quantum theory: Even if event and apparent 
horizon coincide, in suitable states the fluctuations can 
be very large so that the location of the apparent horizon becomes fuzzy.\\
C.\\
In the classical theory, if the metric does not 
depend on the observable $P_M$ conjugate to $M$, and possibly also the 
quantum theory it is 
conceivable that significant evaporation, apart from miniscule quantum 
fluctuations, only arise if we take interactions into account. 
These arise only beyond second order perturbation theory due 
to either self-interactions of $X,Y$ or interactions between $X,Y$ and 
the matter degrees of freedom. The reason is that at second order 
geometry and matter fluctuations decouple and all field species effectively 
propagate on a GPG background with fixed $M$. As that background is
GP time independent, each mode function that solves the corresponding 
classical equations of motion is $\propto e^{i\omega\tau}$ 
for some $\omega\in \mathbb{R}$ and thus 
periodic in GP time with periodicity determined by $\omega$. 
If the classical or quantum field is only excited for a finite number 
of such $\omega$ then 
all notions of mass dpeending on the fluctuations 
will be (quasi-)periodic rather than decaying functions of time which 
would rather 
require a superposition (integral) of an infinite number of modes. 
At second 
order there can be interaction between matter and geometry fluctuations 
if the electric charge does not vanish (or if the Klein Gordon potential 
has linear term) but such a quadratic interaction can be decoupled 
by a canonical transformation and the time dependence would still be 
quasi-periodic.
This indicates that having a manifestly gauge invariant formalism at one's 
disposal that allows to unambiguously compute the effects of higher order 
perturbations of the true degrees of freedom is probably 
very crucial in order that significant evaporation effects are turned on
even if only a finite number of modes are excited.     

\subsection{Foliations and Hawking radiation}
\label{s0.2}

The physical Hamiltonian can be expanded to any order in $X,Y$ and in 
quadratic order suggests a Fock representation of $X,Y$ corresponding 
to a free field on a spacetime whose metric depends  
parametrically on $M$.
%$M,P_M$. 
This requires that the $\tau=$const. surfaces 
of the foliation defined by the gauge fixing conditions are actually complete 
Cauchy surfaces, i.e. they have a timelike normal and every inextendible 
causal curve must cross each $\Sigma_\tau$ precisely once. Moreover, as we wish 
to explore the fate of the singularity and the possibility of black to white
hole transitions, the spacetime covered cannot only be an asymptotic region 
but must also deal with the interior of the black hole parametrised by 
$M$ and possibly another exterior region.    

All of this rules out to use standard Schwarzschild time as the Killing field
$\partial_t$ orthogonal to the $t=const.$ slices is spacelike for $r<2M$ and 
there is a coordinate singularity at $r=2M$. 
As a physical selection criterion for the gauge fixing and the 
corresponding foliation, 
we use the principle of general relativity, that is, the equivalence 
principle: The Fock vacuum (zero particle vector state) 
selected by the Hamiltonian should be the one 
of an observer in geodesic motion since this observer comes as close as 
possible to an inertial observer in flat spacetime. Since the spacetime 
parametrised by $M$ 
%$M,P_M$ 
is spherically symmetric we employ that symmetry property  and 
consider radial unit 
timelike geodesics adapted to spherical symmetry.
These carry two parameters $e,\rho$ in addition to the directional 
angles. Here $e:=-g(\xi,u)\in \mathbb{R}_+$ 
is the asymptotic energy per mass where 
$\xi=\partial_t$ is the Killing vector field corresponding to 
Schwarzschild time $t$ and $u=\partial_\tau$ is the unit timelike tangent along
the affinely parametrised geodesics with eigentime parameter $\tau$, i.e.
$g(u,u)=-1,\; \nabla_u u=0$. The parameter $\rho\in\mathbb{R}$ 
labels the geodesic and has  
the meaning that at eigentime $\tau=\rho$ the geodesic hits the singularity 
$r=0$ where $r$ is the Schwarzschild radial coordinate. The parameter $e$
which is the same for the whole geodesic congruence 
is a function of $M,P_M$ as detailed in appendix \ref{sd} which
would be of interest if we would use option B of the previous 
subsection and is therefore a Dirac 
observable as one would expect from its geometrical meaning. For $e\not=1$
the corresponding coordiantes are called generalised GP coordinates while 
$e=1$ corresponds to the standard or exact GPG. As motivated in the previous
subsection we will use $e=1$ in what follows but briefly comment on the 
case $e\not=1$ for possible future use. Since the geodesics are complete 
($\tau$ has infinite range) only for $|e|>1$ it is sufficient to consider 
only the case $e\ge 1$. 

At each fixed $e$, the family of ingoing 
geodesics labelled by $\rho,\theta,\phi$
is a geodesic congruence which covers the Schwarzschild (SS) and black hole
(BH)  
portion of the Kruskal extension of the spacetime and 
the family of outgoing 
geodesics labelled by $\rho,\theta,\phi$
is a geodesic congruence which covers the mirror Schwarzschild (BSS) 
and white hole (WH) 
portion of the Kruskal extension of the spacetime. It turns out that these 
geodesics intersect the spacelike surface $r=0$ orthogonally when 
$e=1$ and that 
the $\tau=const.$ slices 
$\Sigma_\tau$ which 
carry coordinates $\rho,\theta,\varphi$ are spacelike surfaces that intersect 
the singularity tangentially when $e=1$. 
If we restrict these surfaces to SS and BH or 
MSS and WH only, then they are not Cauchy surfaces for these portions because 
they end at the singularity.  
We can turn them into Cauchy surfaces in two ways: The first 
possibility is to restrict to say 
SS and BH portion and use part of the singularity to complete $\Sigma_\tau$ in 
a $C^1$ manner to a Cauchy surface. This spacetime has topology 
$\mathbb{R}\times\mathbb{R}_+\times S^2$ covered by $\tau,r,\theta,\varphi$ 
coordinates. However, to use part of the singularity surface $r=0$ as Cauchy
surface is problematic as the metric is singular on an
entire  
3d submanifold of that Cauchy surface and because the so extended 
Cauchy surfaces actually overlap at the singularity surface. Thus one would 
rather extend the Cauchy surface into the BH part 
slightly off the $r=0$ surface after hitting the 
singularity and such that they do not overlap (there is of course 
considerable freedom in doing so).   
The second possibility is to glue a SS and BH portion 
belonging to a ``past universe'' to an MSS and WH portion of a 
``future universe''. We do this by gluing an ingoing geodesic labelled by 
$\rho$ in SS+BH to an outgoing geodesic labelled by $\rho$ in MSS+WH. 
In this way we can foliate the entire set SS+BH+WH+MSS by ``free falling
Cauchy surfaces'' describing a black hole white hole transition between 
two universes. This spacetime has topology $\mathbb{R}^2\times S^2$ 
covered by $\tau,\rho,\theta,\varphi$ 
coordinates. All $\rho=const.$ geodesics start at past timelike infinity 
of the past 
universe as $\tau\to-\infty$ and end in future timelike infinity of the 
future universe as $\tau\to +\infty$. All $\tau=const.$ 
Cauchy surfaces end in the two asymptotic ends, i.e. the spatial  
infinity of SS in the past universe and of MSS in 
the future universe. They intersect the singularity in a single point only.
Accordingly, it is mathematically preferred to 
use the entire SS+BH+WH+MSS portion. It is the common domain of dependence
of all free falling Cauchy surfaces and thus this spacetime region 
is globally hyperbolic if we do not exclude the singularity $r=0$. 
In appendix \ref{sd} we have 
collected the background material associated with this construction.

Note that due to the singular behaviour of the metric at $r=0$ there 
is a priori no reason to assume that the mass $M$ is the same in the future and 
past universe respectively and for the same reason the clock ticking rate 
$\kappa$ is not necessarily the same. For reasons of continuity of the geodesic 
$\rho=$const. we consider only the equal value case for both 
$M,\kappa$.\\
\\ 
In SS+BH+WH+MSS the metric can be described using global coordinates 
$\tau,\rho,\theta,\varphi$ except for the 
singularity at $\rho=\tau$ and it remains to 
be investigated whether certain observables are nevertheless singularity 
free across $\rho=\tau$ in the quantum theory. For instance, the mode 
functions that enter the construction of the Fock representation 
have to solve a stationary equation of Schr\"odinger type involving a 
potential that is singular at $r=0$. However, such a situation is 
common in quantum mechanics (e.g. the hydrogen atom)
and not necessarily an obstacle to solve 
the corresponding stationary Schr\"odinger equation.
 
Note that the construction can be repeated 
indefinitely to the future and the past by gluing these transition spacetimes 
labelled by $I\in \mathbb{Z}$ along the various horizons 
$r_I=2M,\bar{r}_I=2M$ where $r_I$ is the radial coordinate in the I-th 
SS+BH and $\bar{r}_I$ in the I-th MSS+WH part. However, that extended 
spacetime is no longer globally hyperbolic even when ignoring 
the singularities because while each transition block is the common domain 
of dependence of all its free falling leaves, for instance the timelike 
geodesics start and end in the past and future timelike infinity of that 
block and do not enter other blocks. Thus boundaries between the 
blocks are Cauchy horizons. Alternatively, one can complete 
such a transition block along the $r=2M$ and $\bar{r}=2M$ boundaries by 
two Minkowski space Penrose diagramme triangles in the past and the future
where the vertical long side of the triangles represent $r=0$ and $\bar{r}=0$
respectively before/after formation/evaporation of the black/white hole,
see appendix \ref{sd} for the details. That completed spacetime continues
to be globally hyperbolic.
    
Returning to one block, given such a foliation, 
one has to determine the solutions of the classical 
equations of motion for $X,Y$ dictated by the quadratic part of the 
Hamiltonian (mode functions). 
For the case $e=1$ and $r>M$ these are known as Heun functions 
\cite{Heun} and there are techniques available to extend them to $0<r\le M$ 
\cite{Perlick}. Thus in principle we can construct the mode functions of 
the Fock representation selected by the radially free falling observers. 
This works for either the SS+BH or the MSS+WH portions seprately. If 
these two descriptions can be meaningfully joined we can discuss black hole 
-- white hole transitions as discussed in the next subsection, otherwise 
we have to restrict to only one of these portions keeping in mind 
the necessity to extend the free falling
equal proper time surfaces to Cauchy surfaces in this case. In both cases we 
call the corresponding vacuum the {\it geodesic} vacuum. As a congruence 
of geodesic observers defines a Riemann normal coordinate system wrt which
the metric is locally Minkowski, we expect the two-point function of the 
geodesic vacuum of Hadamard form \cite{Fulling} when we 
pick the asymptotic form of the mode functions at the 
spatial infinities to correspond to a flat space Hadamard state.\\
\\
One can then describe two kinds of Hawking effects in the usual way:
The first type is to use the t=const. foliation of the SS portion to define the 
Fock structure with respect to an asymptotically static observer, express 
the quantum fields $X,Y$ restricted to SS with respect to the mode functions 
of the $\tau=const.$ and $t=const.$ foliations respectively and derive 
the Bogolubov coefficients from the equality of the two expansions. 
This is similar 
to the Unruh effect with the role of the inertial respectively accelerated 
obserer of Minkowski spacetime played by the $\tau$ respectively
$t$ foliation observer in the curved spacetime parametrised by $M$,
%$M,P_M$, 
i.e.
one uses two {\it different} foliations in a portion of spacetime. 
The second type is to use the observation that in $\rho,\tau$ coordinates 
the vector field $\partial_\tau$  is not Killing but everywhere timelike
and orthogonal to the foliation 
(the metric depends only on the combination $\rho-\tau$ so that the Killing
vector field is $\partial_\rho+\partial_\tau$, however it is not everywhere 
timelike). Accordingly one can use the notion of adiabatic vacua 
\cite{Parker,Fulling} familiar 
from cosmology to describe particle production between different 
$\tau=const.$ slices within the {\it same} foliation of spacetime.

\subsection{Black Hole -- White Hole Transition and singularity resolution}
\label{s0.3}

The metric expressed in $\tau,\rho$ coordinates covers a spherically symmetric 
vacuum 
spacetime which is singular at $\rho=\tau$ where the leaves of the geodesic 
foliation intersect the singularity. Accordingly, the corresponding 
wave equations for $X,Y$ are also singular at $\rho=\tau$. If they
can be meaningfully continued across $\rho=\tau$ we can discuss black hole
white hole transitions. In the literature on quantum black holes
inspired by LQG 
\cite{LQG-BH, SF-BH} one argues that the singularity is removed as follows:
The BH and WH portions of the spacetime are described by a Kantowski-Sachs
cosmology 
in suitable coordinates i.e. the metric is spatially homogeneous and 
described by two scale factor functions $A,B$ of BH resp. WH 
``time'' $r,\bar{r}$ joined at $r=\bar{r}=0$  
(recall that the radius is timelike in the interior; one can consider 
the time coordinate $T:=-r$ in BH and $T:=\bar{r}$ in WH to work with 
a single ``time'' coordinate) subject to the condition
$A^2=1-2M/B,\; B^2=r^2,\bar{r}^2$ respectively. Instead of imposing these 
conditions we can consider a phase space with canonical pairs 
$(A,p_A),\;(B,p_B)$ and a Hamiltonian constraint $C$ such that the symplectic 
reduction of that constrained system recovers the above form of $A,B$ where 
$2M$ plays the role of an integration constant. See appendix 
\ref{se} for some of the details of this construction.

Then one quantises the unconstrained phase space using a 
Narnhofer-Thirring type of representation \cite{Narnhofer-Thirring} 
of the corresponding Weyl algebra inspired by LQG
by the same logic applied in LQC \cite{LQC}. Then one must impose 
$C$ as a quantum constraint which in this representation is 
only possible if one modifies $C$ by replacing $A,B$ by suitable 
linear combinations of Weyl elements 
which are not all strongly continuous in this representation e.g. 
$A$ becomes $\sin(\lambda A)/\lambda$ for small $\lambda$ in the simplest 
proposal. One finds that the singularity $B^2=0$ is resolved and replaced by 
a minimal positive value of Planck area order. That is, the quantum
metric becomes regular. If we would simply take over 
those results and replace the singular metric by that modified regular 
one, we could in principle straightwardly extend $X,Y$ between 
the BH, WH portions and discuss BH -- WH transitions in a singularity free 
manner although solving the corresponding mode functions 
would become quite involved.

Besides those options, here we explore a less radical 
possibility: As we work in a reduced phase 
space context and GPG coordinates, there is no room for such modifications 
in the symmetric sector and the unperturbed metric 
stays singular because the reduction is done prior to 
quantisation. However, we may use the orthonormal basis
of the one particle Hilbert space developed in \cite{TTNewBasis}
that allow to deal with potentials that contain 
arbitrarily negative powers of $r,\bar{r}$ in order 
to define a dense domain of the corresponding Schr\"odinger type 
operator and to meaningfully analyse the possibility of a BH to WH transition
and a singularity resolution.

\section{Choice of gauge condition and associated reduced Hamiltonian}
\label{s4}

In the first subsection we motivate the choice of gauge condition adapted 
to spherical symmetry. In the second we specify the decay behaviour 
of the fields with respect to chosen system of coordinates. This is 
somewhat different from the usual decay behaviour in terms of 
the coordinates of an asymptotic observer at rest as the shift function
approaches zero more slowly at spatial infinity.
In the third we show that the chosen gauge can 
be locally installed modulo the usual global issues. In the fourth 
we solve the constraints {\it non-perturbatively} 
in a neighbourhood of the chosen gauge cut in implicit form, that is, 
modulo explicitly solving a system of ordinary differential equations,
however, we provide an iteration method for solving it. In the fifth 
we compute the corresponding 
reduced Hamiltonian implicitly but non-perturbatively by solving 
the stability conditions for lapse and shift functions thereby 
obtaining the dependence of the full spacetime metric on the 
true degrees of freddom and exploiting the explicit decay conditions.
We will follow the general programme outlined in \cite{pa129}.

While all formulae of this section are implicit only, they provide the 
fundamental starting point for the explicit pertubative scheme that 
is developed in the subsequent sections.

\subsection{Exact and generalised Gullstrand-Painlev\'e Gauge}
\label{s4.1}

As outlined in \cite{pa129} it is important to impose gauge conditions 
which involve only configuration coordinates on the phase space. This is 
because otherwise the choice of gauge would not be disentangled from the 
(perturbative) solution of the constraints for the momentum variables 
which in turn determines the reduced Hamiltonian. Therefore, the gauge 
condition must not depend on the mass of the black hole. Next,
as we 
wish to be independent of the matter content, this forces us to impose 
conditions on the three metric $m_{\mu\nu}$. 
Furhermore, as we wish to explore 
the black hole interior, we should impose gauge conditions on $m_{\mu\nu}$ 
which ensure that the metric is regular across any possible horizons. 
Therefore the gauge condition must be regular on the entire three manifold
$\sigma$ which we choose to be $\sigma=\mathbb{R}^3$ for each 
asymptotic component. At the same time the 
gauge condition should of course not be in conflict with the 
possible presence of 
a black hole (non vanishing mass) and must be consistent with the available 
gauge freedom, i.e. the gauge must not eliminate physical degrees of 
freedom that cannot be removed by a true gauge transformation. This 
is a subtle point which for completeness is reviewed in appendices 
\ref{sa}, \ref{sb} and \ref{sd}. Finally, there are practical 
considerations, which prefer 
gauge conditions which simplify the computations of covariant derivatives
and curvature associated with $m_{\mu\nu}$ as much as possible because 
$m_{\mu\nu}$ features prominently into all couplings between matter and 
geometry in the Hamiltonian constraint.

These guidelines motivate the {\bf (Generalised) 
Gullstrand-Painlev\'e Gauge} (GGPG)
$G_\mu=0,\; \mu=0,1,2,3$ where \cite{26}
\be \label{4.1}
G_3:=q_3-e^{-2},\;G_A:=q_A,\;G_0:=q_0-r^2
\ee    
which have proved very powerful for exactly spherically symmetric
classical and quantum LTB spacetimes \cite{GTT,WE}. Here 
$e^2\ge 1$ is a parameter that cannot be removed by a Hamiltonian 
gauge transformation (it can be by a Lagrangian one, see appendix 
\ref{sd}). These 
coordinates and their relations to timelike geodesic congruences 
and simultanity foliations are reviewed in appendix \ref{sd}.
We refer to the exact Gullstrand Painlev\'e gauge (GPG) as the 
one corresponding to $e^2=1$.  

The notation is as follows:
We have chosen Cartesian coordinates $x^\mu,\; \mu=1,2,3$ on 
$\sigma=\mathbb{R}^3$ with corresponding 
radial coordinates $z^3:=r$ and angular 
coordinates $z^A,\; A=1,2;\;z^1=\theta,\; z^2=\varphi$ and the components 
$q_{33}, q_{3A}, q_{AB}$ are with respect to $z^3,z^A$. Then
\be \label{4.2}
q_3:=m_{33},\; q_A:=m_{3A}, \; q_0:=\frac{1}{2}\Omega^{AB}\; m_{AB}
\ee
The latter is the trace with respect to the background metric $\Omega$ on 
$S^2$.
The GGPG implies that 
the dynamical configuration degrees of freedom are 
\be \label{4.3}
q_{AB}^\perp=X_{AB}=q_{AB}-\Omega_{AB}\; q_0\
\ee
The conjugate momenta are 
\be \label{4.4}
P^3:=W^{33},\;P^A:=2\;W^{3A},\;
P^0:=\Omega_{AB}\;W^{AB},\;
P^{AB}_\perp=Y^{AB}:=W^{AB}-\frac{1}{2}\Omega^{AB}\; P^0
\ee
Note that the non-vanishing Poisson brackets are 
\ba \label{4.5}
&& \{P^3(z),q_3(z')\}=
\{P^0(z),q_0(z')\}=\delta(z,z'),\;\;
\{P^A(z),q_B(z')\}=\delta^A_B\;\delta(z,z'),\;\;
\nonumber\\
&& \{P^{AB}_\perp(z),q_{CD}^\perp(z')\}=
[\delta^A_{(C} \; \delta^B_{D)}-\frac{1}{2}\Omega^{AB}\Omega_{CD}]
\;\delta(z,z')
\ea
The GGPG suggests to solve the constraints for the momenta 
$P^3,P^A,P^0$ to which we turn in the next subsection.

For completeness we repeat here the argument why the GGPG can always 
be installed by a spacetime diffeomorphism 
when the metric is spherically symmetric (but not 
necessarily stationary), that is, of the 
form 
\be \label{4.6}
ds^2=-f\; d\hat{T}^2+g\;dX^2+2\;h\;d\hat{T}\;dX+k^2\;\;d\Omega^2,\;
d\Omega^2:=\Omega_{AB}\; dz^A\;dz^B
\ee
where $f,g,h,k$ are functions of $\hat{T},X$ and $g,k>0,f+h^2\;g^{-1}>0$
in order that the metric be regular and of signature $-1,+1,+1,+1$. 
We first set $T:=\hat{T},\;r:=k(\hat{T},X)$ and invert $\hat{T}=T,\;
X=K(T,r)$. Then upon substitution (\ref{4.6}) acquires the form 
\be \label{4.7}
ds^2=-F\; dT^2+G\;dr^2+2\;H\;dT\;dr+r^2\;\;d\Omega^2,\;
\ee
where now $F,G,H$ are functions of $T,r$. Next we set 
$T=T(\tau,r)$ and 
find 
\be \label{4.7a}
ds^2=-F\;\dot{T}^2\; d\tau^2
+[G-F\;(T')^2+2\;H\;T']\; \;dr^2
+2\;[H\; T'-F\; \dot{T}\; T']\; d\tau\; dr
+r^2\;\;d\Omega^2,\;
\ee
where $\dot{T}=dT/d\tau,\; T'=dT/dr$.
The GGPG can now be completed by solving the non-linear ODE (at fixed 
$\tau$)
\be \label{4.8}
G-F\;(T')^2+2\;H\;T'=e^{-2}
\ee
where $G,F,H$ are functions of $T(t,r)$ and $r$. The solution of 
(\ref{4.8}) is unique up to a sign and 
to addition of a function $T_0(\tau)$. Plugging in that solution we find 
\ba \label{4.9a}
ds^2 &=& -A\; d\tau^2+2\;B\; \; d\tau\;dr+e^{-2}\; d\vec{x}^2
+r^2(1-e^{-2}) d\Omega^2
\nonumber\\
&=&
-[A+B^2 e^2]\; d\tau^2\;+e^{-2}\;
\;\delta_{\mu\nu}\;[dx^\mu+B\;e^2\;n^\mu d\tau]\;
[dx^\nu+B\; e^2\;n^\nu d\tau]
+r^2(e^{-2}-1) d\Omega^2
\ea
where $n^\mu=x^\mu/r$
which shows that the metric is regular and of Lorentzian signature if 
$A+B^2/e^2>0$ which can be ensured using the freedom of choosing $T_0$. 

It is illustrative \cite{26} to transform the Schwarzschild form of the 
line element for $r>2M$ ($2M$ is the Schwarzschild radius) to GGP form 
which can be done using $T_0(t)=t$. One finds (see appendix \ref{sd})
that 
$A=e^{-2}[1-\frac{2M}{r}],\;B=\mp
e^{-2}\;\sqrt{e^2-1+\frac{2M}{r}}$ i.e. $A+B^2\;e^2=1$. It follows that 
lapse and shift are $N=1,\; N^3=e^2\; B,\; N^A=0$. The sign 
corresponds to a congruence of outgoing/ingoing timelike geodesics for 
which $\tau$ is the eigentime. 
    
The GGP form can be analytically extended 
beyond $r>2M$ and covers the advanced ($B>0$) and retarded ($B<0$) Finkelstein 
portions of the Kruskal extension and thus makes it a suitable 
coordinate system for our purposes to explore the interior of 
the black hole. In fact we will explore the possibility to glue 
two such spacetimes to a Black Hole White Hile Transition (BHWHT) spacetime
constructed more explicitly in appendix \ref{sd} and for which we 
require two asymptotic ends covered by two different radial coordinates 
$r,\bar{r}$ for the part of the spacetime containing the black and white
hole respectively.
 
An important feature of the exact GPG (i.e. $e^2=1$) is that in the foliation 
defined by the 
GP time $T$ the spatial sections are {\it flat}. This tremendously 
simplifies all subsequent calculations in pertubation theory and the 
canonical quantisation of the system. The complete 
information about the non-vanishing 
4-curvature therefore does not reside in the intrinsic 3-curvature but 
the extrinsic curvature, e.g. 
$K_{33}=-\frac{1}{2N}\;[{\cal L}_{\vec{N}} q]_{33}=-B'$.
The price to pay is that the line element is 
stationary but not static in these coordinates. Of course the GP or 
GGP coordinates 
do not remove the singularity at $r=0$. This cannot be achieved by the 
classical theory but potentially by the quantum theory, which is what we 
wish to explore.

That the GGPG can always be installed when the 3-metric is spherically symmetric
by a coordinate transformation that preserves spherical symmetry (i.e. 
does not depend on angular coordinates) does not 
show that the GGPG can also be always installed when the 3-metric is not 
spherically symmetric. In what follows we will show that this is nevertheless
the case where coordinate transformations that violate spherical symmetry 
can be exploited to achieve just that. We will do this directly in the 
Hamiltonian formalism. As a prerequisite this requires to carefully state 
the decay behaviour of the fields at spatial infinity because these 
decide which kind of transformations are to be considered as gauge and 
which as symmetry. See appendix \ref{sa} for an illustrative toy model
that exhibits this phenomenon in a mathematically transparent setting and 
its implications already for the exactly spherically symmetric vacuum sector 
in appendix \ref{sb}.

\subsection{Decay behaviour of the fields at spatial infinity} 
\label{s4.2}

We consider two asymptotic ends with radial coordinates $r,\bar{r}$ glued 
at $r=\bar{r}=0$. Thus $r=0$ is not a boundary of the spatial slices 
$\tau=$const. and no boundary conditions need to be stated there.\\
\\
In the gravitational sector 
we have to state the decay behaviour of the canonical pairs 
$(m_{33}, W^{33})$, $(m_{3A}, W^{3A})$, $(m_{AB}, W^{AB})$. 
As before 
we split (dropping the subscript $_E$ for notational simplicity)
\ba \label{4.100}
&& m_{33} = q^v+x^v,\; m_{3A}=q_A+x_A,\; m_{AB}=q^h\; \Omega_{AB}
+X_{AB}
\nonumber\\
&& \frac{W^{33}}{\omega}=p_v+y_v,\; 
\frac{W^{3A}}{\omega}=\frac{1}{2}\;[p^A+y^A],\; 
\frac{W^{AB}}{\omega}=\frac{1}{2}\Omega^{AB}\;[p_h+y_h]+Y^{AB}
\ea
with $q_A=p^A\equiv 0$. The fields $q^v, p_v, q^h, p_h$ coordinatise 
the purely spherically symmetric sector and do not carry any angular 
dependence. By contrast, the fields 
$x^v, y_v, x_A, y^A, h^h, y_h, X_{AB}, Y^{AB}$ capture the total angular 
dependence and have no spherically symmetric $l=0$ modes when expanded into 
scalar, vector and tensor harmonics. The symplectic potential 
reads (for two asymptotic ends, $z=\theta(z)r-\theta(-z)\bar{r}$)
\be \label{4.101}
\Theta=\int_{-\infty}^\infty\; dz\; [
\sum_{\alpha\in \{v,h\}}\; p_\alpha\;dq^\alpha
+\sum_{\alpha\in\{v,h,e,o\},1\le l\ge |m|}\; 
y_{\alpha,l,m}\;dx^{\alpha, l,m} 
+\sum_{\alpha\in\{e,o\},2\le l\ge |m|}\; 
Y_{\alpha,l,m}\;dX^{\alpha,l,m}]
\ee
It is customary to state the fall-off conditions in tandem with 
parity conditions \cite{BS} but in (\ref{4.101}) there are no 
parity conditions to state any longer as all functions 
displayed depend only on $r$ i.e. they are parity invariant. Since all
functions displayed are independent of each other, every single term must 
be convergent. Therefore, since we wish not to rely on possible cancellation 
effects from the two infinities as the background solution does not display 
such an effect, the weakest condition that we can impose is that
the individual terms of the form  
$p\;dq,\; y\; dx,\; Y\; dX$ decay stronger than $r^{-1}$, say as 
$r^{-1+\epsilon},\; \epsilon>0$. 

In the purely spherically symmetric sector reviewed in appendix \ref{sb} 
we find that the general solution of the constraints is given 
in terms of $\gamma^2:=q^v,\; \delta^2:=q^h,\; p_\gamma:=2\gamma p_v,
p_\delta=2\; \delta\; p_h$ by
\be \label{4.102}
\{\frac{p_\gamma^2}{16\delta}+\delta[1-(\frac{\delta'}{\gamma})^2]\}'=0,\;\;
p_\delta=\gamma\; p'_\gamma/\delta'
\ee
where $\gamma, \delta$ are still arbitrary functions of $r$ 
provided that $\gamma\not=0, \delta'=d\delta/dr>0$. Denoting the 
integration constant by $m=2M$ the solution of (\ref{4.102}) in the 
GPGG const.$=\gamma^2=e^{-2}\le 1,\; \delta=r$ is given by 
\be \label{4.103}
p_\gamma=\pm 4\sqrt{m\;r+[\gamma^{-2}-1]\;r^2},\;
p_\delta=\pm 2\gamma 
\frac{m+2[\gamma^{-2}-1] r}{\sqrt{m\;r+[\gamma^{-2}-1]\;r^2}}
\ee
which for the extact GPG $\gamma\equiv 1$ simplifies to 
\be \label{4.104}
p_\gamma=\pm 4\sqrt{m\;r},\; p_\delta=2\sqrt{\frac{m}{r}} 
\ee
Thus the choice between GGPG and exact GPG makes a drastic difference for the 
decay behaviour of the momenta in the vacuum case: For generic 
$0<\gamma^2<1$ we have $p_\gamma=O(r),\; p_\delta=O(1)$ while for 
$\gamma^2\equiv 1$ we have  
$p_\gamma=O(r^{1/2}),\; p_\delta=O(r^{-1/2})$. Likewise, as explained 
in appendix \ref{sd.6}, the choice between GGPG and GPG has a drastic 
consequence for the dynamics of $M$ in the presence of metric perturbations 
and or matter: While in the GGPG $M$ can change dynamically, in the 
exact GPG it cannot.\\ 
\\
Since in the GGPG the fields $p_v, p_h$ decay much slower and because 
$d\;q_v=\propto d\;\gamma^2=O(1)$ even does not decay at all to leading 
order (as $\gamma^2$ is dynamical in this case), the 
definition of the reduced phase space consistent with the GGPG 
is much more delicate than in the case of the strict GPG. In the present 
paper we will therefore consider only the decay behaviour in the strict 
GPG. We will return to the GGPG case in a forthcoming manuscript.\\
\\
As in the strict GPG the sperically symmetric spatial background metric 
is the flat Euclidan metric $dr^2+r^2\;d\Omega^2$ we can motivate the 
choice of decay behaviour from the usual decay behaviour stated in standard 
Cartesian coordinates at spatial infinity
\be \label{4.105} 
m_{ab}^{{\sf Cart}}=\delta_{ab}+\frac{f^e_{ab}(\Omega)}{r}+
\frac{f^o_{ab}(\Omega)}{r^2},\;
W^{ab}_{{\sf Cart}}=\frac{F^{ab}_o(\Omega)}{r^2}+
\frac{F^{ab}_e(\Omega)}{r^3},\;
\ee
where $f^{e/o}_{ab}(\Omega),\;
F^{ab}_{e/o}(\Omega)$ are even/odd parity tensors depending only on 
the angular coordinates. Then with $\vec{x}=r \vec{n}(\Omega),
z^1=\theta, z^2=\phi, z^3=r$        
\ba \label{4.106} 
m_{\mu\nu}^{{\sf Sph}} &=&
\frac{\partial x^a}{\partial z^\mu}
\;\frac{\partial x^a}{\partial z^\mu}
m_{ab}^{{\sf Cart}}
\nonumber\\
W^{\mu\nu}_{{\sf Sph}} &=& |\det([\frac{\partial x}{\partial z}])|\;
\frac{\partial z^\mu}{\partial x^a}
\;\frac{\partial z^\nu}{\partial x^b}
W^{ab}_{{\sf Cart}}
\nonumber\\
m_{33} &=& 1
+\frac{f^e_{33}(\Omega)}{r}+\frac{f^o_{33}(\Omega)}{r^2}
\nonumber\\
m_{3A} &=& 0
+f^e_{3A}(\Omega)+\frac{f^o_{3A}(\Omega)}{r}
\nonumber\\
m_{AB} &=& r^2\;\Omega_{AB}
+r\;f^e_{AB}(\Omega)+f^o_{AB}(\Omega)
\nonumber\\
\frac{W^{33}}{\omega} &=&
F^{33}_o(\Omega)+\frac{F^{33}_e(\Omega)}{r}
\nonumber\\
\frac{W^{3A}}{\omega} &=&
\frac{F^{3A}_o(\Omega)}{r}+\frac{F^{AB}_e(\Omega)}{r^2}
\nonumber\\
\frac{W^{AB}}{\omega} &=&
\frac{F^{AB}_o(\Omega)}{r^2}+\frac{F^{AB}_e(\Omega)}{r^3}
\ea
Note that this transformation between Cartesian coordinates $x^a$ at fixed 
Schwarzschild time $t$ is to be supplemented by the transition to GP time 
$\tau$ which is given by $\tau=t\mp f(r),\; f'=\sqrt{R/r}\;
[1-\frac{R}{r}]^{-1},\;R=2M$ which means that $g_{tt}\to 
g_{\tau\tau}=g_{tt},\; g_{tr}\to g_{\tau r}=g_{tr}-g_{tt}\; f',\;
g_{rr}\to g_{rr}+g_{tt} [f']^2$. Then if 
$g_{tt}=-1+O(r^{-1}), g_{tr}=O(r^{-1}),g_{rr}=1+O(r^{-1})$ we have 
$g_{\tau\tau}=-1+O(r^{-1}), g_{\tau r}=O(r^{-1/2}),g_{rr}=1+O(r^{-1})$ 
which means the lapse and shift decay $N=1+O(r^{-1}),N^3=O(r^{-1})$ has 
changed to $N=1+O(r^{-1}),N^3=O(r^{-1/2})$. Thus the decay behaviour of 
$m_{\mu\nu}$ in (\ref{4.106}) is not affected by the switch between 
the two time coordinates, however, the changed shift behaviour affects 
the decay behaviour of the momenta which will be accounted for below.
The fact that only the spherically symmetric part of the shift is 
affected motivates to change the decay behaviour of only the 
spherically symmetric part of the momenta.  

Thus we first translate the decay behavior of the 
$x^{\alpha, l, m},\;y_{\alpha,l,m},\; X^{\alpha,l,m},\;
Y_{\alpha,l,m}$ unchanged. 
We could in principle choose this individually mode 
by mode $(l,m)$. But to simplify this, we will write e.g.
for $\alpha\in \{v,h,o\}$
\be \label{4.107}
x^\alpha=
[\sum_{|m|\le l;l/2\in \mathbb{N}} x^{\alpha,l,m} L_{\alpha,l,m}]
+[\sum_{|m|\le l;(l+1)/2\in \mathbb{N}} x^{\alpha,l,m} L_{\alpha,l,m}]
=:x^\alpha_e+x^\alpha_o
\ee
and we impose the decay behaviour only on 
the linear combination of all even/odd parity terms rather than individually
(for $\alpha=e$ the restrictions on $l$ in the two sums in 
(\ref{4.107} are switched).
Then
\ba \label{4.108}
&&
x^v=\frac{x^v_e}{r}+\frac{x^v_o}{r^2},\;
x_A=x^e_A+\frac{x^e_A}{r},\;
x^h=x^h\;r+x^h_o,\;
X_{AB}=X_{AB}^e\; r+X_{AB}^o
\nonumber\\
&&
\frac{y_v}{\omega}=y_v^o+\frac{y_v^e}{r},\;
\frac{y^A}{\omega}=\frac{y^A_o}{r}+\frac{y^A_e}{r^2},\;
\frac{y_h}{\omega}=\frac{y_h^o}{r^2}+\frac{y_h^e}{r^3},\;
\frac{Y^{AB}}{\omega}=\frac{Y^{AB}_o}{r^2}+\frac{Y^{AB}_e}{r^3},\;
\ea
In this way in the integral definining $\Theta$, the leading order 
terms are of the form 
\be \label{4.109}
\int\; dr\;(\int_{S^2} d\Omega [y^o\; d x_e])\;\frac{1}{r}=0 
\ee
as $[y^o \delta x_e]$ is an odd parity scalar on the sphere. The subleading 
terms do contribute but decay as $r^{-2}$ and thus converge. Of course 
in (\ref{4.108}) we could also allow for a slower decay of the momenta,
we could multiply the right hand sides of above equations for 
$y_v, y^A, y_h, Y^{AB}$ by $r^{1-\epsilon}, \; \epsilon>0$, a freedom 
to keep in mind when discussing the solutions of the physical equations 
of motion.  

As just motivated, for the sperically symmetric sector 
we use the vacuum solution as an input and require that 
\be \label{4.110}
q^v= 1+\frac{k^v}{r^{5/2}},\; q^h=r^2+\frac{k^h}{r^{1/2}},\;
p_v=i_v r^{1/2},\; p_h=i_h\; r^{-3/2}
\ee
where $k^v,\; k^h,\;i_v,\; i_h$ are functions of $r$ only which 
approach constants or decay at spatial infinity. 
Then $p_v dq^v, p_h dq^h$ decay at 
least as $r^{-2}$. 

Altogether (\ref{4.108}) and (\ref{4.110}) provide a consistent set of 
decay conditions compatible with the vacuum solution in GP coordinates. 
We take this as the {\it definition} of the phase space which then may be  
translated into any other frame.

Next we consider the constraints. The variation of the spatial 
diffeomorphism constraint
\ba \label{4.111}    
d[V_\parallel[S_\parallel]] 
&=& \int\; d^3x\; S^\rho\; d[W^{\mu\nu}\; m_{\mu\nu,\rho}
-2 (m_{\rho\mu}\; W^{\mu\nu})_{,\nu}]
\\
&=&
\int\; d^3x\; 
\{[dW^{\mu\nu}]\; [{\cal L}_S m]_{\mu\nu}
-[{\cal L} W]^{\mu\nu}\; [d m_{\mu\nu}]\}
+\int\; d\Sigma_\rho \{S^\rho\; W^{\mu\nu}\; dm_{\mu\nu}
-2 S^\nu \;d[ W^{\rho\mu} m_{\mu\nu}\}
\nonumber\\
&=&
\int\; d^3x\; 
\{[dW^{\mu\nu}]\; [{\cal L}_S m]_{\mu\nu,\rho}
-[{\cal L} W]^{\mu\nu}\; [d m_{\mu\nu}]\}
+\int\; \frac{d\Omega}{\omega} \;S^3 W^{\mu\nu}\; dm_{\mu\nu}
-2\;d[\int\;d\Omega S^\nu \;W^{3\mu} m_{\mu\nu}]
\nonumber
\ea
where $d\Sigma_\rho =\frac{1}{2}\epsilon_{\mu\nu\rho} 
d z^\mu \wedge dz^\nu=\delta_\rho^3 \frac{d\Omega}{\omega}$
and the surface 
integral is taken at $r=\infty$ for each asymptotic 
region. Since by construction 
$W^{\mu\nu} dm_{\mu\nu}$ to leading order is $r^{-2}$ for the symmetric 
sector and $r^{-1}$ odd plus $r^{-2}$ even for the non symmetric sector, 
the first boundary term in (\ref{4.111}) 
vanishes if the even part of $S^3$ grows at most 
linearly while the odd part approaches a constant. In the decomposition of 
the previous section this means 
\be \label{4.112}
S^3=f^h+g^h,\; f^h=O(r),\; g^h=g^h_e\; r+g^h_o
\ee
where $g^h_{e/o}$ are even/odd functions on the sphere. It follows that 
\be \label{4.113}
H_\parallel[S_\parallel]=V_\parallel[S_\parallel]+
B_\parallel(S_\parallel),\; B_\parallel[S_\parallel]=
2\int\;d\Omega\; S^\nu \;W^{3\mu} m_{\mu\nu}
\ee
has well defined variational derivatives corresponding to the bulk term 
in (\ref{4.111}).

Next we consider the variation of Hamiltonian constraint 
$d[V_\perp[S_\perp]]$.
This picks up a boundary term coming purely from the spatial curvature term. 
Since the decay behaviour of the spatial metric is unaffected by 
the use of GP coordinates, we may take over the standard result 
that (see e.g. \cite{BS} the second reference in \cite{20}
and references therein)
\ba \label{4.114}
H_\perp[S_\perp] &=& V_\perp[S_\perp]+B_\perp[S_\perp]
\nonumber\\
B_\perp[S_\perp] &=& 
-\int\; \sqrt{\det(m)}\; S^0\; m^{\mu\nu}\;
[d\Sigma_\mu\; 
(\Gamma^\rho_{\rho\nu}-(\Gamma^{{\sf ND}})^\rho_{\rho\nu})
-d\Sigma_\rho\; 
(\Gamma^\rho_{\mu\nu}-(\Gamma^{{\sf ND}})^\rho_{\mu\nu})
]
\nonumber\\
&& +\int\; \sqrt{\det(m)}\; [\nabla_\mu S^0]\; m^{\mu\nu}\;m^{\rho\sigma}
(d\Sigma_\nu [m_{\rho\sigma}-m^{{\sf ND}}_{\rho\sigma}]
-d\Sigma_\rho [m_{\nu\sigma}-m^{{\sf ND}}_{\nu\sigma}])
\ea
where $m^{{\sf ND}}$ is the non-dynamical part of $m_{\mu\nu}$ which in the 
present case is just $\delta_{\mu\nu}$ in the Cartesian frame and $\nabla$ 
is the covariant differential compatible with $m$. Also $\Gamma^{{\sf ND}}$ 
is the Christoffel symbol of the non-dynamical part. This term is missing 
in the usual treatment in which one implicitly assumes a flat 
Cartesian frame at infinity.
However the GPG is not a Cartesian frame and the Christoffel symbol is not 
a tensor, hence subtraction of that term is necessary in general 
(in the derivation of the boundary term \cite{BS} or second reference 
of \cite{20}
only the variation $d \Gamma$ 
enters which {\it is} a tensor). As in (\ref{4.111}), 
in order that $d V_\perp[S_\perp]$ can be written as 
$d\;H_\perp[S_\perp]-dB_\perp[S_\perp]$ we must assume that 
\be \label{4.115}
S^0=f^v+g^v,\; f^v=O(1),\; g^v=g^v_e+g^v_o\; r
\ee 

\subsection{Installation of the GPG}
\label{s4.3}

As usual, we consider those $S_\parallel, S_\perp$ respectively for which
the boundary terms 
$B_\parallel[S_\parallel],B_\perp[S_\perp]$ vanish, a gauge transformation,
those for which these do not vanish, a symmetry transformation.
We must show that it is possible to install the exact GPG by picking 
suitable $S_\parallel, S_\perp$ corresponding to gauge transformations.

Having made sure that the functionals $H_\parallel, H_\perp$ have 
well defined variational derivatives they 
generate the following transformations on the 3-metric
\ba \label{4.8a}
\delta m_{\mu\nu}(z) &=& 
%\{\int\; d^3y\; [S^3\; V_3+S^A V_A+S^T\; V_T](y),
\{H_\parallel[S_\parallel]+H_\perp[S_\perp], m_{\mu\nu}(z)\}
\\
&=& [{\cal L}_{\vec{S}} m]_{\mu\nu}
+S^0 \;(2 W_{\mu\nu}-m_{\mu\nu}\; W)](x);\;\; 
P_{\mu\nu}:=(m_{\mu\rho}\;m_{\nu\sigma})\;W^{\rho\sigma},\;\;
W:=m_{\mu\nu}\;W^{\mu\nu}
\nonumber
\ea
where $\cal L$ denotes the Lie derivative. Given arbitrary values of 
$q_3=q^v+y^v,q_A=y_A,q_0=q^h+y^h$ 
consistent with the imposed decay behaviour of the previous 
subsection we want to show that we can find $S^3,S^A,S^0$ 
corresponding to a gauge transformation such 
that $G_\mu+\delta G_\mu=G_\mu+\delta q_\mu=0$ for $\mu=3,A,0$
where $G_\mu=q_\mu-q_\mu^{{\sf GPG}}$.
Decomposing with respect to tensorial structure on $S^2$ we find explicitly
\ba \label{4.8b}
\delta q_3 &=& 
2\; [q_3 \; (S^3)'+q_A (S^A)']
+[S^3\; q_3'+S^A\; D_A q_3]
+S^0\;[2 W_{33}-q_3\; W]
\nonumber\\
\delta q_A &=& [q_{AB} (S^B)'+q_A (S^3)']
+[S^3\; q_A'+q_3 \; D_A S^3+S^B\; D_B q_A+q_B \; S_A S^B] 
+S^0\;[2 W_{3A}-q_A\; W]
\nonumber\\
\delta q_0 &=& \Omega^{AB}\;\{
[S^3 q'_{AB}+2q_{(A} \; D_{B)} S^3+
S^C\; D_C q_{AB}+2\; q_{C(A} \; D_{B)}\; S^C]
+S^0\;[2 P_{AB}-q_{AB}\; P]
\}
\ea
where $q_{AB}=q_0\Omega_{AB}+X_{AB}$.
Assuming that $\Omega^{AB}[2\;W_{AB}-q_{AB}\;W]\not=0$ (otherwise conduct 
another gauge transformation on the system first so that this quantity is 
non vanishing to begin with which is always possible as it is not gauge 
invariant) we can solve the equation $\delta q_0+G_0=0$ algebraically for 
$S^0$ and introduce that solution, which depends linearly 
on $S^3, D_A S^3, S^A, D_B S^A$, into the equations for 
$\delta q_3+G_3=0, \delta q_A+G_A=0$. Due to 
the Euclidian signature of $m_{\mu\nu}$ we have 
$q_3>0$ and that $q_{AB}$ is non-degenerate and thus the resulting set of three 
equations can be cast into the form
\be \label{4.8c}
[S^\mu]'+F^\mu_\nu\; S^\nu+K^{\mu A}_\nu D_A S^\nu=-
q^{\mu\nu}\;G_\nu
\ee
for certain coefficient functions $F^\mu_\nu,\;K^{\mu A}_\nu$ and 
$(.)'=\frac{d}{dr}(.)$.
Thus (\ref{4.8c}) is an inhomogeneous linear, infinite
system of ODE's in the variable 
$r$ 
and the unknowns $S^\alpha_y$ with $S^\alpha_y(r)=S^\alpha(r,y),\; 
y=(y^1,y^2)$ and thus 
can be solved by the method of variation of constants. It remains to show 
that the decay behaviour of the solution so obtained indeed corresponds
to a gauge transformation. We assume this to be the case by making 
use of the choice of integration constants when solving the system 
(\ref{4.8c}).

\subsection{Solution of the constraints in the GPG}
\label{s4.4}

We will only treat the gravitational constraints 
because other constraints corresponding to Yang-Mills type of gauge 
transformations can be treated independently, see next section. Accordingly
we will write the constraints of GR as 
\ba \label{4.9}
V_\mu &:=& -2\;\nabla_\nu\; W^\nu\;_\mu+V^m_\mu,\;
\nonumber\\
V_0 &:=& [m_{\mu\rho}\;m_{\nu\sigma}-\frac{1}{2}
m_{\mu\nu}\;m_{\rho\sigma}]\;W^{\mu\nu}\;W^{\rho\sigma}
+V_0^{c,m}
\ea
Here $V^m_\mu$ denotes the matter contribution to the spatial diffeomorphism 
constraint and $V_0^{c,m}$ the spatial curvature and matter and contribution 
to the Hamiltonian constraint multiplied by $\sqrt{\det(q)}$. The precise 
form of these matter and curvature contributions are displayed in the next 
section but will not be important for the purposes of the present section.
The Levi Civita differential compatible with $m_{\mu\nu}$ is denoted
by $\nabla$. 

The first step is the decomposition of (\ref{4.9}) with respect to the 
canonical chart (\ref{4.2})-(\ref{4.4})
\ba \label{4.10}
V_\mu^m &=& 
2(m_{\mu\nu}\; W^{\nu\rho})_{,\rho}
-m_{\nu\rho,\mu}\;W^{\nu\rho}
\nonumber\\
&=&      
% 2(q_{\alpha 3}\; P^{33})_{,3}
%+2(q_{\alpha 3}\; P^{3B})_{,B}
%+2(q_{\alpha B}\; P^{B 3})_{,3}
%+2(q_{\alpha B}\; P^{BC})_{,C}
%-q_{33,\alpha}\;P^{33}
%-2 q_{3B,\alpha}\;P^{3B}
%-q_{BC,\alpha}\;P^{BC}
2(m_{\mu 3}\; P^3)'
+(m_{\mu 3}\; P^B)_{,B}
+(m_{\mu B}\; P^B)'
\nonumber\\
&& +2(m_{\alpha B}\; W^{BC})_{,C}
-m_{33,\mu}\;P^3
-m_{B,\mu}\;P^B
-m_{BC,\mu}\;P^{BC}
\nonumber\\
V_3^m &=&
2(q_3\; P^3)'
+(q_3\; P^B)_{,B}
+(q_B\; P^B)'
+2(q_B\; P^{BC})_{,C}
\nonumber\\
&& -q_3'\;P^3
-q_B'\;P^B
-q_{BC}'\;P^{BC}
\nonumber\\
V_A^m &=& 2(q_A\; P^3)'
+(q_A\; P^B)_{,B}
+(q_{AB}\; P^B)'
+2(q_{AB}\; P^{BC})_{,C}
\nonumber\\
&& -q_{3,A}\;P^3
-q_{B,A}\;P^B
-q_{BC,A}\;P^{BC}
\nonumber\\
-V_0^{c,m} &=& 
[q_3 P^3]^2+2\;[q_3 P^3]\;[q_A\; P^A]
+\frac{1}{2}\; [q_A P^A]^2
+\frac{1}{2}\; q_3\; q_{AB}\; P^A\; P^B
\nonumber\\
&& +2\; q_A\; q_B\; P^{AB}\; P^3
+2 q_A\; q_{BC}\; P^C\; P^{AB}
+q_{AC}\; q_{BD}\; P^{AB}\; P^{CD}+
\nonumber\\
&& -\frac{1}{2}
[q_3\; P^3+q_A\; P^A+q_{AB}\; P^{AB}]^2
\ea
where we have written $(.)':=\frac{\partial}{\partial r}(.)$. 

Remarkably, the constraints (\ref{4.10}) display the following features:\\
1. All momenta $P^3, P^A, P^0$ appear only polynomially in $V_0$.\\
2. The momentum $P^0$ does not enter $V_3, V_A$ with radial derivatives.\\
3. The momenta $P^3,P^A$ do enter $V_3, V_A$ with radial derivatives.\\
4. All momenta $P^3, P^A, P^0$ appear in $V_3,V_A$ with angular derivatives.\\
This suggests the following solution strategy:\\
1.\\
We solve $V_0$ algebraically for $P^0$. In fact, since $V_0$ is a quadratic 
polynomial in $P^0$ we may write 
\be \label{4.11}
V_0=g_0(q)\;[P^0+f^0_+(P^3,P^A;P_\perp;q,V_0^{c,m})]\;
[P^T+f^0_-(P^3,P^A;P_\perp;q,V_0^{c,m})]
\ee
where $g_0(q)$ depends polynomially on $q=(q_3,q_A,q_{AB})$ 
and $f^0_\pm$ are 
the two possible real roots which depend algebraically on 
$P^3,P^A, P^{AB}_\perp$
(i.e. no derivatives enter)
and algebraically on $q$ and $V_0^{c,m}$ where 
the latter depends on the gravitational degrees of freedom 
algebraically only through $q$ and all its first and 
second spatial derivatives. 
The explicit calculation reveals 
\ba \label{4.11a}
f^0_\pm &=& -\frac{A}{B}\;\{1\pm\sqrt{1-\frac{B}{A^2}\; V_0^{c,m,\perp}}\}
\nonumber\\
A &=& r^2\; P^3 -q_{AC}^\perp\; q_{BD}^\perp\; \Omega^{AB}\; P^{CD}_\perp
-r^2\; q_{AB}^\perp\; P^{AB}_\perp \; 
\nonumber\\
B &=& q_{AC}^\perp\; q_{BD}^\perp\; \Omega^{AB}\; \Omega^{CD}
\nonumber\\
V_0^{c,m,\perp} &=&
V_0^{c,m}+\frac{1}{2}\;[P^3-q_{AB}^\perp P^{AB}_\perp]^2
+\frac{1}{2}\; q_{AB}\; P^A\; P^B
-\; [q_{AB}^\perp\; P^{AB}_\perp]^2
\nonumber\\
&& +r^4\; \Omega_{AC}\; \Omega_{BD}\; P^{AB}_\perp\;P^{CD}_\perp
+2 r^2 q_{AC}^\perp\; \Omega_{BD}\; P^{AB}_\perp\;P^{CD}_\perp
+q_{AC}^\perp\; q_{BD}^\perp\; P^{AB}_\perp\;P^{CD}_\perp
\ea
The choice of the sign in (\ref{4.11a}) is in fact unique if 
we impose that (\ref{4.11a}) has a regular limit $B\to 0$ as we 
approach an exactly spherically symmetric solution which 
selects the solution $P^0=-f^0_-$.\\
2.\\
We write $V_3,V_A$ as 
\be \label{4.12}
2\; q_3 \; [P^3]'+ q_A\; [P^A]'+\tilde{f}_3(P^3,P^A,P^0,q)=V^m_3,\;
2\; q_A \; [P^3]'+ q_{AB}\;\; [P^B]'+\tilde{f}_A(P^3,P^C,P^0,q)=V^m_A,\;
\ee
where $\tilde{f}_3, \tilde{f}_A$ depend on $P^3,P^A, P^0$ only linearly 
and either with no or at most first angular derivatives while 
$q$ enters linearly and with at most first radial and angular derivatives. 
Taking linear combinations these can be decoupled and written as
\be \label{4.13}
[P^3]'+f^3(P^3,P^A,P^0,q,V^m)=0,\;\;
[P^A]'+f^A(P^3,P^B,P^0,q,V^m)=0
\ee
where $f^3,f^A$ are still linear in $P^3,P^B,P^0$, depend at most 
on angular derivatives and are non polynomial with respect to 
$q$ but linear in $V_3^m, V_B^m$. 
That (\ref{4.12}) can be written as (\ref{4.13}) close 
to the GPG $q_3=1, q_A=0,q_{AB}=r^2\Omega_{AB}+q_{AB}^\perp$ and 
for the perturbation $q_{AB}^\perp$ sufficiently small is by inspection.\\
3.\\
We substitute the root $P^0=-f^0_-$ of (\ref{4.11}) into
(\ref{4.13}) thereby obtaining the constraints
\ba \label{4.14}
\tilde{V}^3 &:=&[P^3]'+\hat{f}^3(P^3,P^A,P_\perp,q,V^m)=0,\;
\tilde{V}^A:=[P^A]'+\hat{f}^A(P^3,P^B,P_\perp,q,V^m)=0,\;
\nonumber\\
\tilde{V}^0 &:=&=P^0+f^0_-(P^3,P^A,P_\perp,q,V_0^{c,m})
\ea
where $\hat{f}^3,\hat{f}^A$ 
still depend only linearly on angular derivatives of $P^3,P^A$ 
but no longer linearly on $P^3,P^A$, in fact they depend non-polynomially 
on $P^3,P^A$ due to the substitution of the square root. \\
4.\\
The system (\ref{4.14}) can be considered as a coupled (inifinite) system of 
ODE's in the variable $z=r$ for the unknowns $P^\mu_y(r)$ 
where ($\mu=1,2,3;\; y:=(y^1,y^2)$) is considered as a compound 
label for these 
unkonwns. Therefore formal solutions exist and are 
unique given initial values. 
They can be found using the Picard-Lindel\"of iteration for any 
$-\infty\le z <\infty$ (recall that we work with two asymptotic ends 
and $z=r\theta(z)-\theta(-z)\bar{r})$ 
\be \label{4.15}
P^\mu_y(z)=P^\mu_y(-\infty)+\int_{-\infty}^z\; ds\; 
\hat{f}^\mu(P_y(s),q_y(s),V^m(s))
\ee
where $P^\alpha_y(-\infty)$ are integration constants.
Note that the integrand not only depends on $P^\beta_y(s)$ but also 
on $\partial_y P^\beta_y(s)$ and thus the iteration does not decouple 
with respect to $y$. 
Note also that equation (\ref{4.15}) is identically 
satisfied for $z=-\infty$. Thus there is an up to 
$3 \times S^2$ worth of degeneracy 
in the constraints and accordingly as many conjugate variables 
such as $Q_\alpha^y=\int_{-\infty}^\infty dz q_\alpha(z,y)$ 
cannot be gauge fixed but 
must be counted as belonging to the set of true degrees of freedom 
in company with the integration constants $P^\alpha_y(-\infty)$. 
In appendix \ref{sa} we demonstrate this phenomenon for a lower dimensional
field theory. The degree of degeneracy is reduced by the boundary conditions
at infinity.

Let $-h^\alpha_y(r)$ be the solution so obtained. Then the constraints can 
be written in the equivalent 
form 
\be \label{4.16}
\hat{V}^\mu_y(r)=P^\mu(y,r)+h^\mu_y(r),\;
\hat{V}^0_y(r)=P^0(y,r)+h^0_y(r)
\ee
where $h^0=[f^0_-]_{P^\mu=-h^\mu}$. The functionals $h^3,h^A,h^0$
depend on the degrees of freedom $q_3,q_A,q_0$, the integration 
constants $P^\alpha_y(-\infty)$, as well  
as the union $R$ of the collection of true matter degrees of 
freedom with the collection of the true gravitational
degrees of freedom $q_{AB}^\perp, P^{AB}_\perp$. By construction
(\ref{4.16}) is an identity for $r=-\infty$ and thus the set of constraints 
(\ref{4.16}) for $r>-\infty$ is strictly equivalent to the set of constraints 
(\ref{4.9}) for $r\ge -\infty$. Since the degrees of freedom 
$P^\mu_y(-\infty)$ are thus left unconstrained and since to gauge fix 
(\ref{4.16}) for $r\not=-\infty$ does not require us to impose the GPG for 
all $-\infty\le r\le \infty$,
following \cite{pa129}   
we consider the set of canonical pairs 
$(Q_\mu(y)=\int\; dz\; q-\mu(z,y0,\;,P^\mu(r=-\infty,y))$ 
as part of the true degrees of freedom and adjoin them to $R$ (again
the boundary conditions reduces this set). 
The so extended set of true degrees of freedom is denoted as $\hat{R}$.   
Note that by virtue of the Picard-Lindel\"of integration involved, the 
functionals 
$h^3,h^A,h^0$ are non-local with respect to $r$ consisting of nested 
radial integrals.

\subsection{Reduced Hamiltonian}
\label{s4.5}

We first review the general theory of how 
to construct a reduced Hamiltonian in the presence 
of boundaries as  introduced in \cite{HRT} and show that this requires a 
non-trivial generalisation. Indeed the same authors mention in 
\cite{HRT1} that beyond the linearised gravity setting a non-trivial 
generalisation of \cite{HRT} is necessary. In proposition 
\ref{prop4.1} below we state a sufficient condition which when satisfied 
allows to construct the reduced Hamiltonian. In the second part of 
this subsection we then analyse the details of the gauge stability 
conditions for the the present system. In the third part of this subsection
we confirm that the sufficient condition stated in the proposition 
holds for the solution of the stability condition constructed 
and then provide the reduced Hamiltonian.

\subsubsection{Reduced dynamics in the presence of boundaries}
\label{s4.5.1}

As we have seen in section \ref{s4.5},
in the presence of boundaries the constraints $V(S)$ are not automatically 
functionally differentiable which poses a problem when computing Poisson
brackets. The problem is displayed by writing the variation of $V(S)$ 
with respect to the canonical variables $q_{ab}, p^{ab}, ..$ as a sum of a 
bulk and boundary contribution
\be \label{4.a.1}
d\;V(S)= 
[d\; V(S)]_\sigma+[d\; V(S)]_{\partial\sigma}
\ee
where the first term is a volume integral over the scalar density 
$D_S^{\mu\nu}\; [d m_{\mu\nu}]+D^S_{\mu\nu} [d W^{\mu\nu}]+..$ while the second 
term is a boundary integral over the vector density  
$J_S^{\mu\nu\rho} [d \; m_{\mu\nu}+D^{S\rho}_{\mu\nu} [d\;W^{ab}]+...$ 
for certain coefficient tensor densities depending on $S$. The bulk term 
yields well defined functional derivatives (the three dimensional delta
distribution is integrated out), the second does not (a one dimensional 
delta distribution is left over). As explicitly 
shown in the previous subsection, the idea to remove that contribution 
is to impose fall-off conditions on $S,m,W,..$ such that one can write 
the boundary contribution as an exact differential 
$[\delta V(S)]_{\partial\sigma}=-\delta B(S)$ for a suitable boundary 
functional $B(S)$ and then to define 
$H(S):=V(S)+B(S)$. Then by construction
\be \label{4.a.2}
\delta H(S)= [\delta V(S)]_\sigma
\ee
is functionally differentiable. 

The constraints are still defined by 
$V(S)=0$ for all $S$ such that $H(S)=B(S)$ on the constraint surface which 
can be non-vanishing if $S$ does not decay sufficiently fast at the boundary.
Accordingly, one interprets transformations generated by $H(S)$ with $S$ 
for which 
$B(S)=0$ as gauge transformations while those with $B(S)\not=0$ are 
considered as symmetry transformations. We subdivide the degrees of freedom 
$m_{\mu\nu}, W^{\mu\nu}$ into two subsets of canonical pairs 
$q^\alpha, p_\alpha$
and $Q^A, P_A$ and impose gauge fixing conditions 
$G^\alpha=q^\alpha-k^\alpha=0$ on the $q^\alpha$ where $k^\alpha$ are 
certain fixed functions on $\sigma$ independent of the foliation time $\tau$
and without dependence on the phase space coordinates.
We set $q_\ast^\alpha:=k^\alpha$. 
We also solve the constraints $V_\alpha=0$ for $p_\alpha$ when 
$q^\alpha=q^\alpha_\ast$ 
which yields solutions $p_\alpha=p_\alpha^\ast$. 

The gauge $G^\alpha=0$ is supposed to be reachable from any point within 
the constraint surface $V_\alpha=0$ and once it is reached the residual 
transformations allowed are those that preserve them
\be \label{4.a.3}
\{H(S),G^\alpha\}_{q=q_\ast,p=p^\ast}=0
\ee
These stability conditions can be solved for $S^\alpha=S^\alpha_\ast$
in terms of $q_\ast,p^\ast$ and in general are symmetry transformations 
rather than gauge transformations.
Consider now a functional $F$ depending on $Q,P$ only. Then the 
reduced Hamiltonian $E$ on the reduced phase space coordinatised by $Q,P$,
if it exists, is 
supposed to give the same equations of motion as $H(S)$ when we restrict 
to the fixed quantities $q_\ast,p^\ast, S_\ast$, that is,
\be \label{4.a.4}
\{E,F\}=\{H(S),F\}_{q=q_\ast,p=p^\ast,S=S_\ast}
\ee
Now, being the boundary value of a volume integral variation that arises due 
to one or several integrations by parts, the boundary term has the form
\be \label{4.a.5}
B(S)=\int_{\partial\sigma}\; 
d\Sigma_\mu\;[
S^\alpha\; j^\mu_\alpha +S^\alpha_{,\nu}\; j^{\mu\nu}_\alpha+...]
\ee
for some ``currents'' $j^\mu_\alpha,\; j^{\mu\nu}_\alpha,\;$. 
Using integrations by 
parts on $\partial \sigma$ and exploiting $\partial^2 \sigma=\emptyset$ 
we can assume w.l.g. that $j^{\mu\nu}_\alpha,..=0$ by redefining 
$j^\mu_\alpha$. 
Similarly, the bulk term has 
the form 
\be \label{4.a.6}
V(S)=\int_{\sigma}\; 
d^3x\;
S^\alpha\; V_\alpha
\ee
for some ``densities'' $V_\alpha$. 
\begin{Proposition} \label{prop4.1} ~\\
Let $d\Sigma_\mu=d^2z\; N_\mu(z)$ where 
$N_\mu(y)=\frac{1}{2}\epsilon^{AB}\epsilon_{\mu\nu\rho} x^\nu_{,A} x^\rho_{,B}$ 
is the corresponding co-normal of the embedding 
$S^2\; \to \; \partial \sigma;\; y\; \mapsto x$. Suppose that
there exists a real valued functional $\chi$ of currents $j_\alpha$ on 
$\partial\sigma$ such that on $\partial\sigma$ (i.e. the functional 
derivative is with respect to the coordinate dependence on 
$\partial\sigma$)  
\be \label{4.a.7}
S^\alpha_\ast=[\frac{\delta\chi}{\delta j_\alpha}]_{j=j^\ast},\;
j_\alpha^\ast := N_\mu\; j^{\ast\mu}_\alpha,\;
j^{\ast\mu}_\alpha := [j^\mu_\alpha]_{q=q_\ast,p=p^\ast}
\ee
Then $E=\chi[j_\ast]$.
\end{Proposition}
\begin{proof}
We simplify the notation and denote $z=(q,p), z_\ast=(q_\ast,p^\ast)$
Then on the one hand
\be \label{4.a.8}
\{B(S)_{z=z_\ast},F\}_{S=S_\ast}
=\int\; d^2y\; S^\alpha_\ast\; \{j^\ast_\alpha,F\}
\ee
Note that $S$ is set to $S_\ast$ only {\it after} the Poisson bracket has been
taken, the Poisson bracket is computed with $S$ treated as being
independent of the phase space coordinates. On the other hand 
using the identity $B(S)_{z=z_\ast}
=H(S)_{z=z_\ast}$ for all $S$ we have
\ba \label{4.a.9}
&& \{B(S)_{z=z_\ast},F\}_{S=S_\ast} 
=\{H(S)_{z=z_\ast},F\}_{S=S_\ast}
\nonumber\\
&=& \{H(S),F\}_{z=z_\ast,S=S_\ast}
+\int\;d^3x\; (
[\frac{\delta H(S)}{\delta q^\alpha(x)}]_{z=z_\ast,S=S_\ast}\; 
\{q^\alpha_\ast(x),F\}
+[\frac{\delta H(S)}{\delta p_\alpha(x)}]_{z=z_\ast,S=S_\ast}\; 
\{p_\alpha^\ast(x),F\}
)
\nonumber\\
&=& \{E,F\} +
\int\;d^3x\;
\{H(S),G^\alpha(x)\}_{z=z_\ast,S=S_\ast} \; 
\{p_\alpha^\ast(x),F\}
\nonumber\\
&=& \{E,F\} 
\ea
where in the second step we split the Poisson bracket into a contribution
which acts on the explicit dependence of $H(S)$ on $P,Q$ and a contribution 
on the implicit dependence of $H_\ast(S)$ on $P,Q$ through 
$q^\alpha_\ast,p_\alpha^\ast$ and
used that $H(S)$ is functionally differentiable, in the third we used 
that $q^\alpha_\ast$ does not depend on the phase space coordinates 
and rewrote the second contribution as a Poisson bracket with 
the gauge fixing condition and in the last we used that by construction 
of $S_\ast$ that Poisson bracket vanishes, that is, (\ref{4.a.3}).

Thus 
\be \label{4.a.10}
\{E,F\}=\{B(S)_{z=z_\ast},F\}_{S=S_\ast}
\ee
This is different from the result quoted in \cite{HRT}. This is because
$E\not=B(S_\ast)_{z=z_\ast}$ unless $S_\ast$ evaluated on 
$\partial\sigma$ does not depend on the coordinates $P,Q$ as it 
is implicitly assumed in \cite{HRT}. This will in 
general not be the case because the solution $S_\ast$ of the stability 
condition involves solving differential equations and thus depends 
non-locally on $P,Q$ via integrals on all of its bulk values, therefore 
$\{S^\alpha_\ast(x),F\}\not=0$ even if $x\in \partial\sigma$ and $F$ is 
localised with respect to $Q,P$ in the bulk. Comparing with (\ref{4.a.8})
we see that the only chance to write the right hand side of (\ref{4.a.10}) 
as a Poisson bracket is that $S^\ast_\alpha$ is the functional defrivative 
on $\partial\sigma$ with respect to a functional $\chi$ which depends on 
the current $j_\alpha$ that appears in (\ref{4.a.8}).
\end{proof}
The case that $S^\alpha_\ast$ is a constant on the phase space such that 
$E=B(S_\ast)_{z=z_\ast}$ as it is considerd in \cite{HRT}
is included and corresponds to $\chi[j]$ being a linear functional of $j$.
In order that $\{E,F\}$ is well defined, the dependence of $j_\ast$ on 
$P,Q$ should be non local i.e. involving radial integrals in order that 
the Poisson brackets $\{E,F\}$ which use functional derivatives on $\sigma$
rather than $\partial\sigma$ are well defined. This is in fact conceivable for 
the gravitational system because the stability conditions involve the 
$p_\alpha$ explicitly which in turn are to be solved for the $P,Q$ 
using the constraints and 
for this we need to perform radial integrations as the constraints depend on 
radial derivatives. Now $j_\alpha$ corresponds to the ADM energy and momentum 
currents and these again involve $p_\alpha$ explicitly. Then to see whether 
$j_\alpha^\ast, S^\alpha_\ast$ satisfy all the requirements critically depends 
on the fall-off behaviour of the fields as these decide which terms 
in $S^\alpha_\ast \{j^\ast_\alpha,F\}$ survive as we take the limit 
$r\to\infty$. It is only with respect to these limiting surviving terms that
the assumptions of proposition \ref{prop4.1} have to hold. Fortunately there 
is a substantial amount of flexibility in the choice of those fall-off 
conditions and one may also take advantage of the fact that the solution of 
the constraints for $p=p^\ast$ and of the stability conditions for 
$S=S_\ast$ both at $q=q_\ast$ involve ``integration constants'' in the 
form of functions on $S^2$, as we have seen 
in the previous section, because we have to solve differential equations 
with respect to the radial coordinates and one can try to use 
the freedom to choose those free functions on the 
sphere in order to meet the conditions of 
proposition \ref{prop4.1}.

\subsubsection{Solution of stability conditions}
\label{s4.5.2}

We now proceed to construct the gauge fixed values $S^3_\ast, S^A_\ast,
S^0_\ast$. This has two purposes: First, the leading order behaviour 
at spatial infinity dictates the analytic form of the reduced Hamiltonian,
i.e. the global mass or energy.
Second, the explicit bulk behaviour determines the physical lapse and 
shift and therefore contains information about the local mass.\\
\\ 
\\
{\bf Asymptotics:}\\
\\
For a general transformation of $m_{\mu\nu}$ induced by the 
functionally differentiable version (\ref{4.113}), (\ref{4.114}) 
of the constraints we have 
\be \label{4.a.11} 
[\delta m]_{\mu\nu}=\{H_\parallel[S_\parallel]+H_\perp[S_\perp],m_{\mu\nu}\}
=[{\cal L}_{S_\parallel} m]_{\mu\nu}
+2\tilde{S}_0\;[W_{\mu\nu}-\frac{1}{2}\;W\; m_{\mu\nu}]
\ee
where $\tilde{S}^0=\frac{S^0}{\sqrt{\det(m)}}, 
W_{\mu\nu}=m_{\mu\rho} m_{\nu\sigma} W^{\rho\sigma},\;W=m_{\mu\nu}
W^{\mu\nu}$. The gauge fixed values $S^\cdot_\ast$ are determined by the 
stability conditions $\delta G^\cdot=0$ where 
$G^3=m_{33}-1,\; G^A=m_{3A},\;G^0=\Omega^{AB} m_{AB}-2 r^2$. This 
needs to hold only at $G^3=V_3=G^A=V_A=G^0=V_0=0$
i.e. at 
\be \label{4.a.12}
m_{33}=q_3^\ast=1,\;
m_{3A}=q_A^\ast=0,\;
\Omega^{AB} m_{AB}=q_0=q_0^\ast =2 r^2,\;
P^3=P^3_\ast,\; P^A=P^A_\ast, P^0=P^0_\ast
\ee
where we used the notation  
(\ref{4.2}), (\ref{4.3}), (\ref{4.4}) and the values 
$P^\cdot_\ast=-h^\cdot_\ast$ are 
given implicitly by (\ref{4.16}) where $h^\cdot_\ast$ is $h^\cdot$ 
evaluated at $G^\cdot=0$.

We find 
\ba \label{4.a.13}
0 &=& \delta G^3 = [{\cal L}_{S_\parallel} q]_{33}
+2\tilde{S}^0\;[P_{33}-\frac{1}{2}\; P\; q_{33}] 
\nonumber\\
0 &=& \delta G^A = [{\cal L}_{S_\parallel} q]_{3A}
+2\tilde{S}^0\;[P_{3A}-\frac{1}{2}\; P\; q_{3A}] 
\nonumber\\
0 &=& \delta G^0 = \Omega^{AB}\;([{\cal L}_{S_\parallel} q]_{AB}
+2\tilde{S}^0\;[P_{AB}-\frac{1}{2}\; P\; q_{AB}] 
\ea
Using the GPG we have (we suppress the super/subscript $\ast$ in 
$q_\cdot^\ast,P^\cdot_\ast$ 
for notational simplicity) using the decomposition 
$P^3=W^{33}=p_v+y_v, P^A=y^A=2W^{3A}, W^{AB}=
P^{AB}=[p_h+y_h]\Omega^{AB}/2+Y^{AB}$   
and $\Omega^{AB} X_{AB}=\Omega_{AB} Y^{AB}=0$
\ba \label{4.a.14}
P_{33} &=& m_{3\mu}\; m_{3\nu} W^{\mu\nu}=P^3=p_v+y_v
\nonumber\\
P_{3A} &=& m_{3\mu}\; m_{A\nu} W^{\mu\nu}
=[r^2 \Omega_{AB}+X_{AB}]\; P^B/2=[r^2 \Omega_{AB}+X_{AB}]\; y^B/2
\nonumber\\
P_{AB} &=& 
m_{A\mu}\; m_{B\nu} W^{\mu\nu}
%=[r^2 \Omega_{AC}+X_{AC}]\;
%[r^2 \Omega_{BD}+X_{BD}]\;P^{CD},\;
[r^2 \Omega_{AC}+X_{AC}]\;
[r^2 \Omega_{BD}+X_{BD}]\;[\frac{1}{2}(p_h + y_h)\Omega^{CD}+Y^{CD}]
\nonumber\\
P&=& m_{\mu\nu} W^{\mu\nu}
=P^3+ [r^2\Omega_{AB}+X_{AB}] P^{AB}
=[p_v+y_v]+r^2\;[p_h+y_h]+X_{AB} Y^{AB}
\ea
Next with $S^3=f_h+g_h, S^A=g^A, S^0=f_v+g_v$
\ba \label{4.a.15}
&& [{\cal L}_{S_\parallel} m]_{33}=
S^\mu m_{33,\mu}+2\; m_{\mu 3} S^\mu_{,3}
=2 S^3_{,3}=2(f_h+g_h)'
\\
&& [{\cal L}_{S_\parallel} m]_{3A} 
=S^\mu m_{3A,\mu}
+m_{\mu 3} S^\mu_{,A}
+m_{\mu A} S^\mu_{,3}
=S^3_{,A}+q_{AB} S^{B\prime}=g_{h,A}+[r^2 \Omega_{AB}+X_{AB}]\;g^{B\prime}
\nonumber\\
&& [{\cal L}_{S_\parallel} m]_{AB} 
=S^\mu m_{AB,\mu}+2\; m_{\mu (A} S^\mu_{,B)}
=S^3 q_{AB}'+S^D q_{AB,C}+2 m_{C (A} S^C_{,B)}
=S^3 q_{AB}'+[{\cal L}_{S_\parallel\perp} q]_{AB} 
\nonumber
\ea
where $S_{\parallel\perp}^A:=S^A$ has only angular non-vanishing components. 
Thus we have explicitly
\ba \label{4.a.16}
0 &=& S^{3\prime}+\tilde{S}^0\;[P^3-\frac{1}{2} P]
\nonumber\\
0 &=& S^3_{,A}+
q_{AB}\; [S^{B\prime}+\tilde{S}^0 P^B]
\nonumber\\
0 &=& \Omega^{AB}\;\{
S^3 q_{AB}'+[{\cal L}_{S_\parallel\perp} q]_{AB} 
+2 \tilde{S}^0[q_{AC} q_{BD} P^{CD}-\frac{1}{2} q_{AB} P] \}
\ea
We now recall the decay conditions derived in section \ref{s4.3}
\ba \label{4.a.17}
&& p_v=O(r^{1/2}),\; p_h=O(r^{-3/2})
\nonumber\\
&& 
y_v=y_v^o+\frac{y_v^e}{r},\;
y^A=\frac{y^A_o}{r}+\frac{y^A_e}{r^2},\;
y_h=\frac{y_h^o}{r^2}+\frac{y_h^e}{r^3},\;
\nonumber\\
&&
X_{AB}=X_{AB}^e \;r+X_{AB}^0,\;
Y^{AB}=\frac{Y^{AB}_o}{r^2}+\frac{Y^{AB}_e}{r^3}
\nonumber\\
&&
S^3=f_h+g_h,\; f_h=O(r),\; g_h=g_h^e \; r+ g^h_o;\;\;
S^0=f_v+g_v,\; f_v=O(1),\; g_v=g_v^e + g_h^o\; r;\;\;
\ea
while the decay behaviour of $S^A=f^A+g^A,\; f^A\equiv 0$ has not 
been fixed yet. The various functions displayed have no $l=0$ modes 
except for $p_v,p_h, f_v, f_h$ which are 
pure $l=0$ modes and $e,o$ refers to their parity behaviour, i.e. 
not the polar/axial character. 
The power of $r$ with respect to which these functions decay is 
at most $O(r^0)$ but can be lower. The same applies to $f_h, f_v$,
i.e. $O(r^n)$ means decay with at most power $r^n$, it can be faster
e.g. $O(r^{-1})$ allows a decay with power $r^{-(1+m)},\; m\ge 0$. \\
\\
In what follows we content ourselves with solving the stability conditions 
with respect to the highest non vanishing order in $r$ and only in as much 
detail as necessary to construct the reduced Hamiltonian. We also 
content ourselves with constructing one particular solution 
$S^3, S^A, S^0$ and leave 
it for further investigation whether that solution is unique.\\
\\
We consider first the second relation in (\ref{4.a.16}) 
\be \label{4.a.18}   
g_{h,A}+q_{AB}\;[g^{B\prime}+\tilde{S}^0 y^B]=0
\;\;
\Leftrightarrow
\;\;
g^{A\prime}=-[q^{AB} g_{h,B}+\tilde{S}^0 y^A]
\ee
We assume that 
\be \label{4.a.19}
g_v=O(1)
\ee
which can be achieved by assuming that 
$g_v^o r$ is $O(1)$. Then $S^0=O(1)$ and thus $\tilde{S}^0=O(r^{-2})$ because 
\be \label{4.a.20}
\det(q)=\det(r^2\Omega+X)=\frac{1}{2} \epsilon^{AC}\epsilon^{BD}   
[r^2\Omega+X]_{AB}\; [r^2\Omega+X]_{CD}
=r^4 \omega^2+r^2\;\omega^2 \Omega^{AB} X_{AB} +\det(X)
=r^4 \omega^2+\det(X)
\ee
where we used 
\be \label{4.a.21}
\epsilon^{AC} \epsilon^{BD} \Omega_{CD}=\omega^2 \Omega^{AB},\;
\omega^2=\det(\Omega)
\ee
and $\Omega^{AB} X_{AB}\equiv 0$ by definition. Since $\det(X)=O(r^2)$ 
it follows that 
$\det(q)=r^4\omega^2[1+O(r^{-2})],\;
[\det(q)]^{-1/2}=r^{-2}\omega^{-1}[1+O(r^{-2})]$. Thus 
$\tilde{S}^0 y^A=O(r^{-3})$.   

Furthermore we assume that 
\be \label{4.a.22}
g_h=O(r^{-1})
\ee
which can be achieved by assuming that $g_h^e\; r, g_h^o$ decay as $r^{-1}$.
Then it follows from (\ref{4.a.18}) that 
\be \label{4.a.23}
S^A=g^A=O(r^{-2})
\ee
It follows that 
\be \label{4.a.24}
[{\cal L}_{S_\parallel\perp} q]_{AB}
=r^2\;[{\cal L}_{S_\parallel\perp} (\Omega+r^{-2} X)]_{AB}
=O(1)
\ee
in leading order. 

The first relation in (\ref{4.a.16}) gives 
\ba \label{4.a.25}
0 &=&
2\; S^{3\prime}+2\;\tilde{S}^0[P^3-\frac{1}{2}[P^3+r^2 P^0+X_{AB} Y^{AB}]]
=2\; S^{3\prime}+\tilde{S}^0[P^3-r^2 P^0-X_{AB} Y^{AB}]
\nonumber\\
&=& 
2(f_h+g_h)'+\tilde{S}^0\;[p_v-r^2 p_h]-\tilde{S}^0\;[y_v-r^2 y_h-X_{AB} Y^{AB}]
\ea
Since $y_v, r^2 y_h=O(1), X_{AB} Y^{AB}=O(r^{-1})$ and 
$\tilde{S}^0=O(r^{-2})$ the third term in (\ref{4.a.25})
is $O(r^{-2})$. It can thus be 
cancelled by $2\;g_h'$ which by (\ref{4.a.22}) is also $O(r^{-2})$. 
The second term in (\ref{4.a.25}) 
on the other hand is $O(r^{-2})\;O(r^{1/2})=O(r^{-3/2})$ and then must be 
cancelled by $2 f_h'$. Hence 
\be \label{4.a.26}
f_h=O(r^{-1/2})
\ee
and we must have 
\be \label{4.a.27}
f_h'+\frac{\tilde{S}^0}{2}[p_v-r^2 p_h]=0
\ee  
to order $r^{-3/2}$ i.e. (\ref{4.a.27}) decays at least as $r^{-2}$. 

The third relation in (\ref{4.a.16}) gives 
\ba \label{4.a.28}
0 
&=& 
4\;r\; S^3 + \Omega^{AB}\;[{\cal L}_{S_\parallel\perp} q]_{AB} 
+2 \tilde{S}^0[ \Omega^{AB}\; q_{AC} q_{BD} P^{CD}-r^2\; P]
\\
&=&
4\;r\; S^3 + \Omega^{AB}\;[{\cal L}_{S_\parallel\perp} q]_{AB} 
\nonumber\\
&& +2 \tilde{S}^0[ 
(r^4 \Omega_{CD}+2 r^2\; X_{CD}+\Omega^{AB} X_{AC} X_{BD})\;
(\frac{P^0}{2}\Omega^{CD}+Y^{CD})
-r^2(P^3+r^2\; P^0+X_{AB} Y^{AB})]
\nonumber\\
&=&
4\;r\; S^3 + \Omega^{AB}\;[{\cal L}_{S_\parallel\perp} q]_{AB} 
+2 \tilde{S}^0\;\{r^4\; P^0 +\frac{1}{2} P^0
\Omega^{AB}\Omega^{CD} X_{AC} X_{BD}\;
\nonumber\\
&& +2 r^2\; X_{AB} Y^{AB}+\Omega^{AB} X_{AC} X_{BD}) \;Y^{CD}
-r^2(P^3+r^2\; P^0+X_{AB} Y^{AB})\}
\nonumber\\
&=&
[4\;r\; S^3 -2 r^2\; \tilde{S}^0 \; P^3] 
+ [\Omega^{AB}\;[{\cal L}_{S_\parallel\perp} q]_{AB}] 
\nonumber\\
&& +2 \tilde{S}^0[\frac{1}{2} P^0
\Omega^{AB}\Omega^{CD} X_{AC} X_{BD}\;
+r^2\; X_{AB} Y^{AB}+\Omega^{AB} X_{AC} X_{BD}) \;Y^{CD}]
\nonumber
\ea
The second square bracket in (\ref{4.a.28}) is $O(1)$ as follows from 
(\ref{4.a.24}). In the third square bracket of (\ref{4.a.28}) the first 
term is $O(r^{-3/2})$, the second is $O(r^{-1})$ and the third is $O(r^{-2})$.
In the first square bracket of (\ref{4.a.24})
the second term is $O(r^{1/2})$ in leading order
therefore the order of $r\; S^3 $ must not exceed $O(r^{1/2})$ 
as there is no term 
to compensate this. Thus we conclude again (\ref{4.a.26}) and 
\be \label{4.a.29}
f_h -\frac{1}{2} r\; \tilde{S}^0 \; p_v=0
\ee
to leading order $r^{-1/2}$ i.e. (\ref{4.a.29}) decays at least as 
$O(r^{-1})$. Together with  
\be \label{4.a.30}   
f_v=O(1)
\ee
we see that the leading order decay behaviour of our solution 
of the stability condition is now fixed and falls into the allowed 
class (\ref{4.112}), (\ref{4.115}).

Finally we consider the constraint 
\ba \label{4.a.31}
V_3-V_3^m &=& W^{\mu\nu} m_{\mu\nu}'-2[W^{\mu\nu} m_{\mu 3}]_{,\nu}
=P^{AP} q_{AB}'
-2\; P^A_{,A}-2\; P^{3\prime}
\nonumber\\
&=& 
(\frac{P^0}{2} \Omega^{AB}+Y^{AB})\;[r^2 \Omega_{AB}+X_{AB}]'
-2\; P^A_{,A}-2\;P^{3\prime}
\nonumber\\
&=& 
2\;r\;P^0+Y^{AB}\; X'_{AB}
-2\; P^A_{,A}-2\;P^{3\prime}
\nonumber\\
&=& 
2\;[r\;(p_h+y_h)-(p_v+y_v)'-y^A_{,A}]+Y^{AB}\; X'_{AB}
\ea
We have $r y_h=O(r^{-1}), y_v'=O(r^{-2}), y^A_{,A}=O(r^{-1}), 
X'_{AB} Y^{AB}=O(r^{-2})$ which means that these terms can cancel among each 
other in (\ref{4.a.31}). On the other hand $r p_h, p_v'=O(r^{-1/2})$ and 
these are the only terms of this type if we assume 
that the matter term $V_3^m$ decays faster than this. Therefore
\be \label{4.a.32}
p_h-\frac{1}{r} p_v'=0
\ee
to leading order $r^{-3/2}$ i.e. (\ref{4.a.32}) decays al least as 
$r^{-2}$. 

Combining (\ref{4.a.27}), (\ref{4.a.29}) and (\ref{4.a.32}) we find 
\be \label{4.a.33}
\frac{f_h'}{f_h}
=-\frac{p_v-r^2 p_h}{r\; p_v}
=-\frac{p_v-r p_v'}{r\; p_v}
=-\frac{1}{r}+\frac{p_v'}{p_v} 
%\ln(f_h\; r/p_v)=const.
\ee
to leading order which leads to the solution
\be \label{4.a.34}
f_h=\kappa\; \frac{p_v}{r}
\ee
where $\kappa$ is an integration constant. This 
correctly reproduces $f_h=O(r^{-1/2})$ if $p_v=O(r^{1/2})$.  
Moreover from (\ref{4.a.29}) to leading order
\be \label{4.a.35}
\tilde{S}^0=\frac{f_v+g_v}{\sqrt{\det(q)}}
=2\frac{f_h}{r\;p_v}=2\kappa \;\;\Rightarrow\;\;
f_v=2\kappa
\ee
as $g_v$ has no $l=0$ mode. If one wants the lapse to equal unity at infinity 
then $\kappa=\frac{1}{2}$.\\
\\
We summarise: The solution $S^\cdot_\ast$ 
of the stability conditions constructed 
displays the decay behaviour (\ref{4.a.19}), (\ref{4.a.22}), 
(\ref{4.a.23}), (\ref{4.a.26}), (\ref{4.a.30}) and the 
relations (\ref{4.a.27}), (\ref{4.a.29}) and (\ref{4.a.32}) with 
$p_h, p_v$, namely to the order displayed
\ba \label{4.a.36}
S^3_\ast &=& f^\ast_h+g^\ast_h,\; 
f^\ast_h=\kappa \frac{p_v}{r}=O(r^{-1/2}),\; g^\ast_h=O(r^{-1})
\nonumber\\
S^A_\ast &=& g^A_\ast =O(r^{-2})
\nonumber\\
S^0_\ast &=& f^\ast_v+g^\ast_v,\; f^\ast_v=2\kappa=O(1), g^\ast_v=O(1)
\nonumber\\
p_h &=& \frac{p_v'}{r}
\ea
where the last relation holds when the constraints are used i.e. 
they should also carry a label $\ast$ which we dropped for convenience.\\
\\
\\
{\bf Bulk solution}\\
\\
We note that (\ref{4.a.16}) is a system 
of four PDE's in four variables $S^3, S^A, S^0$. One may therefore obtain 
an exact and non-perturbative solution in principle as follows:\\
We solve the fourth equation algebraically for $S^0$ and insert its 
solution into the first three eqations. This remaining system of 
three PDE's which is linear in $S^3, S^A$
can then be solved for the radial derivatives i.e. it can be 
written in the form ($\mu,\nu=1,2,3$)
\be \label{4.a.36a}
S^{\mu\prime}(r,y)=F^\mu_\nu(r,y)\; S^\nu+G^{\mu A}_\nu\; S^\nu_{,A}=:
\int\;d^2y'\; K^\mu_\nu(r;y,y')\; S^\nu(r,y')
\ee
where the right hand side is linear $S^\mu$ and its angular derivatives
which we have written in terms of an integral kernel. 
Hence the solution can be written
\be \label{4.a.36b}
S^\mu_y(r)=[{\cal P}(\exp(\int_{-\infty}^r\;ds K(s))\cdot 
\hat{S}(-\infty)]^\mu_y   
\ee
where the path ordering symbol $\cal P$ orders the highest radial value 
to the outmost left and $\hat{S}^\mu_y(-\infty)$ are initial values on 
the sphere $\bar{r}=-z=\infty$ where $(\mu,y)$ is considered a 
compound label in order to write (\ref{4.a.36a}) as a matrix 
equation. This yields in 
paticular lapse $S^0$ and shift $S^\mu$ for the GPG and therefore 
an entire spacetime metric. The initial values have to be adjusted
to the boundary conditions stated. 

Instead of a non-perturbative solution, a perturbative solution of 
(\ref{4.a.36a}) maybe obtained as follows: The coefficient 
functions $F^\mu_\nu,\; G^{\mu A}_\nu$ are under perturbative 
control i.e. they have known expansions in terms of say $X,Y$ in the 
GPG thanks to the possibility to solve for $P^3=p_v+y_v,\; P^A=y^A,\;
P^0=p_h+y_h$ as we will indicate in the next section. By the notation 
$A(n)$ we mean the n-th order contribution in $X,Y$ in the expansion 
of a quantity $A$. The perturbative scheme implies that 
the $l=0$ modes $p_v,p_h$ have no first order contribution i.e.
$p_v(1)=p_h(1)=0$ while the $l>0$ modes $y_v,y^A,y_h$ have no zeroth 
order contribution, i.e. $y_v(0)=y^A(0)=y_h(0)=0$. We will solve 
(\ref{4.a.36a}) under the same premise for the dcomposition
$S^3=f_h+g_h,\; S^A=g^A$ i.e. that the $l=0$ mode obeys 
$f_h(1)=0$ while the $l>0$ modes obey $g_h(0),g^A(0)=0$. We now 
expand all quantities in (\ref{4.a.36a}) in powers of $X,Y$ and obtain 
a hierarchy of equations for the functions $f_h(n);\; 1\not=n\ge 0,\;
g_h(n),\; g^A(n); \;n\ge 1$. 

The structure of these equations is given in more detail by 
\ba \label{4.a.36c}
0 &=& [S^3]'+[\alpha\; S^3+\beta_A\; S^A+\gamma_A^B\; S^A_{,B}]\; F^3
\nonumber\\
0 &=& [S^A]'+q^{AB} S^3_{,B}+
[\alpha\; S^3+\beta_A\; S^A+\gamma_A^B\; S^A_{,B}]\; F^A
\ea
where all $\alpha,\beta_A,\gamma_A^B,F^3$ have $n=0$ contributions 
while $F^A=y^A$ starts at order $n=1$. 

Since $f_{h,A}=0$ it is easy to see that the second equation 
of (\ref{4.a.36c}) is identically satisfied at order $n=0$ while 
the first equation at $n=0$ gives a first order ODE in $r$ 
for $f_h(0)$ whose 
solution is exactly the asymptically leading behaviour of $f_h$ displayed 
in (\ref{4.a.36}). 

For $n=1$, because $F^A$ is already of first order and because we have 
$f_h(1)=0$, the second equation can be solved explicitly for $g^A(1)$
by decomposing  $g_h(1), g^A(1), F^A(1)$ into scalar and vector harmonics 
respectively 
which yields a first order linear inhomogeneous 
ODE in $r$ for $g_{\alpha,l,m}(1),\; \alpha=e,o$ 
that relates it to $g_{h,l,m}(1),F_{\alpha,l,m}(1)$ 
and can be solved by quadrarture.
When inserted into the first equation of (\ref{4.a.36c}) 
at $n=1$ we obtain an
integro differential equation for $g_{h,l,m}(1)$ or equivalently a 
inhomogeneous linear second order ODE which can be solved by holonomy 
and variation of constant 
methods by transforming it to a system of two homogeneous linear first 
order equations. 

Proceeding inductively,
at order $n\ge 2$ we see that the second equation in (\ref{4.a.36c})
takes the form 
\be \label{4.a.36d}
0=[g^A(n)]'+r^{-2}\Omega^{AB} [g_h(n)]_{,B}+J^A(n)
\ee
where $J^A(n)\propto F^A$ contains at most (n-1)th order contributions 
of $f_h,g_h,g^A$ and thus (\ref{4.a.36d})
can be solved for $g^A(n)$ by quadrature.
For the first equation in (\ref{4.a.36c}) 
we note that the $g^A(n)$ contribution 
to $\beta_A S^A+\gamma_A^B\; S^A_{,B}$ is proportional to $2r^3 D_A g^A(n)$
and hence can be written
\be \label{4.a.36e}
0=[f_h(n)+g_h(n)]'+\alpha(0)[f_h(n)+g_h(n)]+K(0)\; D_A g^A(n) +J^3(n)
\ee
where $J^3(n)$ contains at most (n-1)th order contributions of 
$f_h,g_h,g^A$ and $K(0)$ is a computable zeroth order 
function. Thus dividing by $K(0)$ and 
taking the radial derivative of (\ref{4.a.36d})   
we find 
\be \label{4.a.36f}
\{\frac{1}{K(0)}\;
([f_h(n)+g_h(n)]'+\alpha(0)[f_h(n)+g_h(n)]+J^3(n))\}'
=D_A \{r^{-2}\Omega^{AB} [g_h(n)]_{,B}+J^A(n)\}
\ee
We decompose into scalar harmonics and obtain a second order, 
linear, inhomogeneous system among the $f_{h,l,m}(n), g_{h,l,m}(n)$ which can 
be transformed to a first order system and solved by variation of constant 
and holonomy methods.\\
\\
To construct for instance $S^3$ perturbatively has the following significance:
For spherically symmetric vacuum GR we have $2M=[S^3]^2\; r$ for any 
value of $r$. With matter and perturbations we may {\it define a local 
effective mass function}
\be \label{4.a.36g}
2M_{{\rm eff}}(x):=[S^3(x)]^2\; r       
\ee

\subsubsection{Evaluation of the boundary terms at the solution
of the stability conditions}
\label{s4.5.3}

The reason why it was sufficient to determine the asymptotic 
leading decay order of $S^\cdot_\ast$ is because in the boundary 
integrals over the asymptotic spheres we take the limit $r\to \infty$.
Thus, the only terms that survive this limit are those that display the 
leading decay behaviour, the subleading terms drop out. We now 
determine the corresponding boundary values 
$B_\parallel[S^\ast_\parallel],\;
B_\perp[S^\ast_\perp]$.\\
\\
We have (we consider only one asymptotic end for simplicity)
\ba \label{4.a.37}
B_\parallel[S^\ast_{\parallel}]
&=& 2\lim_{r\to\infty}\;\int\; d\Sigma_\mu\; S_\ast^\rho\; 
W^{\mu\nu}\; m_{\rho\nu} 
\nonumber\\
&=& 2\lim_{r\to\infty}\;\int\; d\Omega\;\omega^{-1} \;
[S_\ast^3\; W^{3\nu}\; m_{3\nu} 
+S_\ast^A\; W^{3\nu}\; m_{A\nu}] 
\nonumber\\
&=& 2\lim_{r\to\infty}\;\int\; d\Omega\;
[S_\ast^3\; P^3  
+S_\ast^A\;\frac{1}{2} P^B\; (r^2 \Omega_{AB} +X_{AB})] 
\nonumber\\
&=& 2\lim_{r\to\infty}\;\int\; d\Omega\;
[(f^\ast_h+g^\ast_h)\; (p_v+y_v)  
+g^A_\ast\;\frac{1}{2} y^B\; (r^2 \Omega_{AB} +X_{AB})] 
\ea
The term $g^A_\ast y^B=O(r^{-3})$ while $r^2 \Omega_{AB}+X_{AB}=O(r^2)$ 
hence the second term in (\ref{4.a.37}) drops out. 
Next as $g_h^\ast=O(r^{-1}))$ while $p_v=O(r^{1/2}), y_v=O(1)$ it follows 
that the term proportional $g_h^\ast$ vanishes. Finally as 
$f_h^\ast=O(r^{-1/2})$ it follows that the $f_h^\ast y_v$ term vanishes.
Accordingly 
\be \label{4.a.38}
B_\parallel[S^\ast_{\parallel}]
=2\lim_{r\to\infty}\;\int\; d\Omega\; f_h^\ast \; p_v
\ee
Since $f_h^\ast=\kappa \frac{p_v}{r}$ we see that 
\be \label{4.a.39}
B_\parallel[S^\ast_{\parallel}]
=2\;\kappa\;\lim_{r\to\infty}\;\int\; d\Omega\; \frac{p_v}{r}\; \; p_v
=:\int_{\partial \sigma}\; d^2y\; s_\ast\; j_\ast
\ee
with $j_\ast=[\frac{p_v}{r^{1/2}}]_{r=\infty}$ and $s_\ast=2\kappa\;j_\ast$. 
Hence we can apply proposition \ref{prop4.1} with 
\be \label{4.a.40}
\chi[j]:=\kappa\; \int_{\partial\sigma} \; d\Omega\; j^2
\ee
The contribution to the reduced Hamilton from $B_\parallel$ 
is therefore given per asymptotic end by 
\be \label{4.a.41}
H_\parallel=\kappa\; \lim_{r\to\infty}\;\int_{S^2}\; d\Omega\;
\frac{p_v(r)^2}{r}
\ee
Note that the naive prescription to use $B_\parallel[S^\ast_\parallel]$
would have resulted in {\it twice} $H_\parallel$.\\
\\
Now we consider $B_\perp[S^\ast_\perp]$ which only depends on 
$F:=S^0_\ast=f_v^\ast+g_v^\ast=O(1)$. We have explicitly
\ba \label{4.a.42}
B_\perp[S^\ast_\perp] 
&=& 
-\int\; \sqrt{\det(m)}\; F\; m^{\mu\nu}\;
(d\Sigma_\mu\; [\Gamma^\rho_{\rho\nu}-\Gamma^{{\sf ND}\rho}_{\rho\nu}]
-d\Sigma_\rho\; [\Gamma^\rho_{\mu\nu}-\Gamma^{{\sf ND}\rho}_{\mu\nu}])
\\ &&
+\int\; \sqrt{\det(m)}\; [\nabla_\mu F]\; m^{\mu\nu}\;m^{\rho\sigma}
(d\Sigma_\nu [m_{\rho\sigma}-m^{{\sf ND}}_{\rho\sigma}]
-d\Sigma_\rho [m_{\nu\sigma}-m^{{\sf ND}}_{\nu\sigma}])
\nonumber\\
&=& 
-\int\; \omega^{-1}\;\sqrt{\det(m)}\;d\Omega\;
F\; 
(m^{3\nu}\; [\Gamma^\rho_{\rho\nu}-\Gamma^{{\sf ND}\rho}_{\rho\nu}]
-m^{\mu\nu}\; [\Gamma^3_{\mu\nu}-\Gamma^{{\sf ND}3}_{\mu\nu}]
\nonumber\\ &&
+\int\;\omega^{-1}\; d\Omega\; \sqrt{\det(m)}\; 
[\nabla_\mu F]\; 
(m^{\mu 3}\;m^{\rho\sigma}\;[m_{\rho\sigma}-m^{{\sf ND}}_{\rho\sigma}]
-
m^{\mu\nu}\;m^{3 \sigma}\;
[m_{\nu\sigma}-m^{{\sf ND}}_{\nu\sigma}])
\nonumber\\
&=& 
-\int\; \sqrt{\det(q)/\omega^2}\;d\Omega\;
F\; 
([\Gamma^\rho_{\rho 3}-\Gamma^{{\sf ND}\rho}_{\rho 3}]
-[\Gamma^3_{33}-\Gamma^{{\sf ND}3}_{33}]
-q^{AB}\; [\Gamma^3_{AB}-\Gamma^{{\sf ND}3}_{AB}]
)
\nonumber\\ &&
+\int\;d\Omega\; \sqrt{\det(q)/\omega^2}\; 
([\nabla_3 F]\; m^{\rho\sigma}\;[m_{\rho\sigma}-m^{{\sf ND}}_{\rho\sigma}]
-[\nabla_\mu F]\; m^{\mu\nu}\;
[m_{\nu 3}-m^{{\sf ND}}_{\nu 3}])
\nonumber\\
&=& 
-\int\; \sqrt{\det(q)/\omega^2}\;d\Omega\;
F\; 
([\Gamma^A_{A 3}-\Gamma^{{\sf ND}A}_{A3}]
-q^{AB}\; [\Gamma^3_{AB}-\Gamma^{{\sf ND}3}_{AB}])
\nonumber\\ &&
+\int\;d\Omega\; \sqrt{\det(q)/\omega^2}\; 
([\nabla_3 F]\; q^{AB}\;[q_{AB}-q^{{\sf ND}}_{AB}]
\nonumber\\
&=& 
-\int\; \sqrt{\det(q)/\omega^2}\;d\Omega\;
F\; 
([\Gamma^A_{A 3}-\Gamma^{{\sf ND}A}_{A3}]
-q^{AB}\; [\Gamma^3_{AB}-\Gamma^{{\sf ND} 3}_{AB}])
\nonumber\\ &&
+\int\;d\Omega\; \sqrt{\det(q)/\omega^2}\; 
([\nabla_3 F]\; q^{AB}\;[q_{AB}-q^{{\sf ND}}_{AB}]
\nonumber\\
&=& 
-\int\; \sqrt{\det(q)/\omega^2}\;d\Omega\;
F\; 
([\Gamma^A_{A 3}-\Gamma^{{\sf ND}A}_{A3}]
-q^{AB}\; [\Gamma^3_{AB}-\Gamma^{{\sf ND}3}_{AB}])
+\int\;d\Omega\; \sqrt{\det(q)/\omega^2}\; 
F'\; q^{AB}\;X_{AB}
\nonumber
\ea
where we used $d\Sigma_\mu =\omega^{-1}\;d\Omega \delta^3_\mu$ and that 
$m_{\mu\nu}$ is block diagonal i.e. $m_{33}=1, m_{3A}=0, 
m_{AB}=q_{AB}=r^2 \Omega_{AB}+X_{AB}$ in the GPG so that $m^{33}=1,
m^{3A}=0, m^{AB}=q^{AB},\; q^{AC} q_{CB}=\delta^A_B$
with 
$m^{{\sf ND}}_{33}=1,\;m^{{\sf ND}}_{3A}=0,\; 
m^{{\sf ND}}_{AB}=r^2\Omega_{AB}$. 

Consider first the second term in the last line of (\ref{4.a.42}) which 
has integrand 
\ba \label{4.a.43} 
&& [\frac{\det(q)} \omega^2]^{-1/2}\; 
F'\; \det(q)\; q^{AB} X_{AB} =[\frac{\det(q)} \omega^2]^{-1/2}\; F'\; 
\epsilon^{AC} \epsilon^{BD} [r^2 \Omega_{CD}+X_{CD}]\;X_{AB} 
\nonumber\\ 
&=& [\frac{\det(q)} \omega^2]^{-1/2}\; F'\; 
[r^2\det(\Omega) \Omega^{AB} X_{AB}+\det(X)] =[\frac{\det(q)} 
\omega^2]^{-1/2}\; F'\; \det(X) 
\ea 
Since $F=O(1)$ we have $F'=0$ at 
$r=\infty$ since $F=f_v+g_v,\; f_v=c+d\; r^{-n},\; 
g_v=C(\Omega)+D(\Omega) r^{-N},\;n,N>0$. Moreover $\det(X)=O(r^2)$ while 
$\det(q)^{-1/2}=O(r^{-2})$. Hence the second term has vanishing limit 
$r\to\infty$.

To evaluate the first term in 
(\ref{4.a.42}) we require the Christoffel symbols
$2\Gamma_{\mu\nu\rho}=2m_{\mu(\nu,\rho)}-m_{\nu\rho,\mu},\;
\Gamma^\mu_{\nu\rho}=m^{\mu\sigma} \Gamma_{\sigma\nu\rho}$
\ba \label{4.a.44}
&& \Gamma_{333}=\Gamma_{A33}=\Gamma_{3A3}=0
\nonumber\\
&& \Gamma_{3AB}=-\frac{1}{2} q_{AB}'=-\Gamma_{A3B}
\nonumber\\
&& \Gamma_{ABC}=r^2\Gamma^\Omega_{ABC}+\Gamma^X_{ABC}
\nonumber\\
&& \Gamma^3_{33}=\Gamma^3_{A3}=\Gamma^A_{33}=0
\nonumber\\
&& \Gamma^3_{AB}=\Gamma_{3AB}, \Gamma^A_{3B}=q^{AC}\Gamma_{C3B}
\nonumber\\
&& \Gamma^A_{BC}=q^{AD}\; \Gamma_{DBC}
\ea
since $q_{3\mu}=\delta^3_\mu$ is constant and we have used 
again block diagonality.  

It follows that the non-vanishing symbols have at most one index $\mu=3$ 
and $\Gamma_{3AB}, \Gamma_{A3B}=O(r),\; \Gamma_{ABC}=O(r^2),\;
\Gamma^3_{AB}=O(r),\; \Gamma^A_{3B}=O(r^{-1}), \Gamma^A_{BC}=O(1)$.
In particular the combination required in (\ref{4.a.42}) 
is 
\ba \label{4.a.45}
&& [\det(q)\omega^{-2}]^{1/2}\;
([\Gamma^A_{A 3}-\Gamma^{{\sf ND}A}_{A3}]
-q^{AB}\; [\Gamma^3_{AB}-\Gamma^{{\sf ND}3}_{AB}])
\nonumber\\
&=&
[\det(q)\omega^2]^{-1/2}\;\det(q)\;
(q^{AB}\; [\Gamma_{BA3}-\Gamma_{3AB}]
-q^{AB}_{\sf ND}\;[\Gamma^{{\sf ND}}_{BA3}-\Gamma^{{\sf ND}}_{3AB}]
)
\nonumber\\
&=&
[\det(q)\omega^2]^{-1/2}\;
(\epsilon^{AC}\epsilon^{BD} \;
[r^2 \Omega_{CD}+X_{CD}]\;[2r \Omega_{AB}+X'_{AB}]
-\det(q)\;r^{-2}\Omega^{AB}\;[2r \Omega_{AB}])
\nonumber\\
&=&
[\det(q)\omega^2]^{-1/2}\;
([r^2\det(\Omega)\;\Omega^{AB}+\epsilon^{AC}\epsilon{BD}X_{CD}]\;
[2r \Omega_{AB}+X'_{AB}]
-4[r^4\omega^2+\omega^2\Omega^{AB}X_{AB}+\det(X)]\; r^{-1})
\nonumber\\
&=&
[\det(q)\omega^2]^{-1/2}\;
(4\;r^3\omega^2+[\det(X)]'-4[r^4\omega^2+\det(X)]\; r^{-1})
\nonumber\\
&=&
[\det(q)\omega^2]^{-1/2}\;([\det(X)]'-4\det(X)\; r^{-1})
\ea
which is $O(r^{-1})$. Since $F=O(1)$ it follows 
\be \label{4.a.46}
B_\perp[S^\ast_\perp]=0
\ee
Altogether, the reduced Hamiltonian is therefore given by (taking both 
asymptotic ends into account)
\be \label{4.a.47}
H=\kappa\;
[\lim_{r\to\infty}\int_{S_2}\; d\Omega\;\frac{[P^3_\ast(r,\Omega)]^2}{r}
+\lim_{\bar{r}\to\infty}\int_{S_2}\; d\Omega\;
\frac{[P^3_\ast(\bar{r},\Omega)]^2}{\bar{r}}]
\ee
where $P^3_\ast$ is the value of $P^3$ obtained by solving all the constraints 
and by imposing the GPG. Due to the decay $y_v=O(r^{-1})$ we may replace 
$P^3=p_v+y_v$ by $p_v^\ast$ in (\ref{4.a.47}).\\ 
\\
\\
Several remarks are in order:\\
0.\\
In the presentation so far we have omitted the prefactor $1/k,\;
k=16\pi G$ of the action where $G$ is Newton's constant and we use
units with $c=1$. That prefactor propagates into the constraints and
thus the boundary term and therefore into $H$. We check that $H=M$ for
spherically symmetric vacuum GR when (\ref{4.a.47}) is multiplied by
$1/k$. To see this we use with $S^0=N$, the relations
$W^{\mu\nu}=\sqrt{\det(m)}[m^{\mu\nu} m^{\rho\sigma}
-m^{\mu\rho} m^{\nu\sigma}] K_{\rho\sigma}$ and
$K_{\rho\sigma}=\frac{1}{2N}[\dot{m}_{\rho\sigma}
-[{\cal L}_{\vec{S}} m]_{\rho\sigma}]$
in GPG. This gives with $\omega=\sqrt{\det(\Omega)}$ the
identity $p^v=P^{33}/\omega=r^2[K_{33}-m^{ab} K_{ab}]=-\Omega^{AB} K_{AB}$
and $K_{AB}=-\frac{1}{2N}[{\cal L}_{\vec{S}} m]_{AB}=
-\frac{1}{2N}[2r S^3]\Omega_{AB}=-\frac{r S^3}{N}\Omega_{AB}$ and therefore
$p_v=2r \frac{S^3}{N}$. If we compare with (\ref{4.a.34}) and use
$f_h=S^3, f_v=2\kappa=N$ we obtain exact match. It follows, performing
the angular integral in (\ref{4.a.47})
$H=\frac{4\pi}{k}\;\kappa\;4 r\; [S^3]^2/N^2$. Since $S^3=\sqrt{2GM/r}$
in GPG it follows $H=2\kappa M/N^2$ for one asymptotic end 
which equals $M$ for $N=1$ i.e.
$\kappa=1/2$. Again, had we wrongly identified $H$ with the value of
the boundary term we would have obtained $H=2M$ i.e. twice the ADM
mass. The difference arises because in Schwarzschild gauge the ADM
momentum vanishes exactly and the stabilising lapse is a constant on the
phase space, therefore in this case $H$ is simply the boundary term and yields
correctly $H=M$ with the same prefactor $1/k$ (see e.g. the second
reference in \cite{20}).\\
1.\\
By contrast to the usual computation, the Hamiltonian results from the 
boundary term of the ADM momentum rather than the ADM Hamiltonian. This 
can be traced back to the fact that in the GPG the information about the 
mass is not encoded in the three metric but rather in the extrinsic 
curvature while in the Schwrzschild gauge the roles are switched 
(the extrinsic curvature vanishes in the static presentation of the 
metric).\\ 
2.\\
In Cartesian coordinates, the Christoffel symbols vanish exactly 
for the GP background metric (which is flat) so that the ADM energy
term vanishes exactly for that background, not only to leading order. 
The relation between the Schwarzschild and GP coordinates 
involves a boost (see appendix \ref{sd}) so that the resulting Hamiltonian 
becomes now a component of the ADM momentum.\\
3.\\
Remarkably, the expression (\ref{4.a.47}) is {\it positive definite}
no matter what the concrete expression for $P^3_\ast$ is in terms of 
the truee degrees of freedom. For the present model these 
are gravitational mass $M$, electric charge $Q$, gravitational
tracefree (wrt $\Omega$) angular perturbations $X_{AB}, Y^{AB}$ (only
$l\ge 2$ modes), 
electromagnetic angular perturbations 
$A_C,\;E^A$ (only $l\ge 1$ modes) and Klein-Gordon field 
$\Phi,\Pi$ (all $l\ge 0$ modes are observables but only $l>0$ modes
are perturbations) when we solve the Gauss constraint for $E^3$.\\    
4.\\
At exact spherical symmetry in vacuum we have $P^3_\ast=\sqrt{2Mr}$ hence 
$H=2\kappa M$ (with prefactor $k$ included). 
The freedom $\kappa$ is therefore the same that arises 
in the purely spherically symmetric sector (see appendix \ref{sb}) 
or in the Kantowski-Sachs reformulation (see appendix \ref{se}).\\ 
5.\\
When expanding $P^3_\ast$ perturbatively in terms of the perturbations 
just mentioned we obtain schematically 
$P^3_\ast=P^3_\ast(0)+P^3_\ast(1)+P^3_\ast(2)+..$ where
the three terms are respectively independent, linear and quadratic in 
the perturbations. Now $P^3=p_v+y_v$ and by the general theory \cite{pa129}
$p_v(1)\equiv 0$ while $y_v(n),\; n\ge 1$ decays too fast to be visible in 
(\ref{4.a.47}). 
Therefore to second order in the perturbations for one asymptotic end
\be \label{4.a.48}
H=\frac{\kappa}{k}\lim_r \frac{1}{r}\; \int\; d\Omega\;
\{
[p_v^\ast(0)]^2+2\; p_v^\ast(0)\; p_v^\ast(2)\}
\ee
(the angular integral gives just $4\pi$ because $p_v$ has only zero modes). 
As we will confirm in our companion papers, (\ref{4.a.48})
reproduces the Regge-Wheeler and Zerilli Hamiltonian \cite{24,25} in GP 
coordinates. \\
6.\\
The real virtue of (\ref{4.a.47}), however, is that it is a non-perturbative 
result. It provides a formula for the physical Hamiltonian entirely 
expressed in terms of the true degrees of freedom and in that sense 
is gauge invariant to all orders that one may want to expand it into.
For instance we now have access to a non-ambiguous Hamiltonian that 
includes cubic (``Non-Gaussian'') (self-)interactions of the perturbations.
In particular in absence of scalar matter the system can be interpreted 
as a black hole formed due to collapse of gravitational and/or 
electromagnetic waves due to self-interactions mediated by gravity.

\section{Perturbative structure of the constraints}
\label{s3}

While in the previous section we have derived the non-perturbative definition
of the reduced Hamiltonian, it is given only implicitly. To be practically
useful in particular for quantisation, we need an explicit formula. This 
can be provided at least perturbatively, which will be the task of the 
present and next section.

In this section we discuss subsequently the perturbative structure 
of Gauss, spatial diffeomorphism and Hamiltonian constraint. 
Since the Gauss and spatial diffeomorphism constraint are first and second 
order homogeneous polynomials in all fields, we will be able to 
exhibit all perturbative orders explicitly where we leave 
the evaluation of integrals of contractions of three spherical 
tensor harmonics in terms of Clesch-Gordan coefficients for 
our companion papers. For the the Hamiltonian constraint,
which we treat in its polynomial form, we collect the full non-perturbative 
structure of all sub-polynomials from which it is assembled but 
then just keep the orders up to two. The explicit computation of the 
finite number of higher orders (up to ten in vacuum GR) will be subject 
of our companion papers.

\subsection{Reduction of the Gauss constraint}
\label{s3.1}

Since the Gauss constraint Poisson commutes with all other constraints 
(we replace $V_\mu^M+V_\mu^{KG}$ by $F_{\mu\nu} E^\nu+\Pi^T [D_\mu \Phi]$ 
by subtracting $A_\mu G=A_\mu (G^M+G^{KG})$ which can be done by 
redefining the $S^\mu,S_0$) we may reduce the theory with respect to 
the corresponding gauge degrees of freedom before entering the reduction 
with respect to Hamiltonian and spatial diffeomorphism constraint. 
In the case of charge (i.e. a scalar dublett) it is most convenient
to use a ``unitary gauge'' otherwise (i.e. a scalar singlett) a radial 
``axial'' gauge.

\subsubsection{Unitary gauge}
\label{s3.1.1}   

As the Gauss constraint generates rotations of $\Phi=(\phi_1,\phi_2)$, 
a perfect gauge is $\phi_2=0$ and we can solve $G=0$ algebraically for 
$\pi_2=J:=\phi_1^{-1}\partial_a E^a$ so that no decay properties of 
the fields come into play. Then $(\phi=\phi_1,\;\pi=\pi_1),\;
(A_\mu,E^\mu),\;\mu=1,2,3$ are the true degrees of freedom as far 
as reduction of the Gauss constraint is concerned. Then the only 
task to do is to perform the following relacements in the 
Klein Gordon contributions to the SDC and HC respectively
\be \label{3.3a0} 
V_\mu^{KG}\;\to\; \pi\phi_{,\mu},\;\;
2\;V_0^{KG} \;\to\; \frac{\pi^2+J^2}{\sqrt{\det(m)}}
+\sqrt{\det(m)}m^{\mu\nu}[\phi_{,\mu}\phi_{,\nu}+A_\mu A_\nu \phi^2
+2V(\phi^2)]
\ee

\subsubsection{Axial gauge}
\label{s3.1.2}   

We now assume that there is just one KG field $\phi$ with 
conjugate momentum $\phi$, hence the term 
$\Pi^T\epsilon\Phi$ in $G$ is missing and $D_\mu \Phi\to 
\partial_\mu\phi,\; \Pi\to \pi$ in both $V_\mu^{KG}, V_0^{KG}$. 

In view of the subtle difference between gauge and symmetries 
that arises for constraints that depend on spatial derivatives of 
the fields exemplified in appendix \ref{sa}, we 
need to specify the decay behaviour of the fields at the spatial 
infinities which we take to be $r^{-2}$ in an asymptotic Cartesian 
frame for both electric and magnetic fields so that 
the electromagnetic energy density decays as $r^{-4}$.
This allows the vector potential 
to decay as $r^{-1}$ or faster. 
In terms of the spherical frame 
$A_\mu(z)=(\partial x^a/\partial y^\mu) A_a(x),\; 
E^\mu(z)=|\det(\partial x/\partial z)|\;(\partial z^\mu/\partial x^a) 
E^a(x)$ this translates into $E^3=O(1),\; E^C=O(r^{-1}),\; 
A_3=O(r^{-2}),\;A_C=O(r^{0})$. However, we must require that the 
dynamical part of the Maxwell connection decays as 
$A_3=O(r^{-2}),\;A_C=O(r^{-1})$. 
This makes sure that $E^3\; [dA_3],\;E^B\; [dA_B]$ decay 
as $r^{-2}$ which then makes the symplectic potential converge.
The functionally differential form of the Maxwell contribution 
to the Gauss constraint is (the Klein Gordon cointribution is 
functionally differentiable as it is) 
\be \label{3.3a1}
H[S_0]=G[S_0]+B[S_0]=-\int\;d^3z\; E^\mu S_{0,\mu},\;
B[S_0]=-\int\;d^2y\; [E^3\; S^0]_{r=-\infty}^{r=\infty}
\ee
The solution of the constraint $(E^3)'+E^A_{,A}=0$ is 
\be \label{3.3a2}
E^3_\ast=P^M-\int_{-\infty}^r\;ds\; E^A_{,A}
\ee
where $P^M$ is any function on $S^2$. The boundary term becomes on such a 
solution 
\be \label{3.3a2a}
B[S_0]=
\int\; d^2y;
\{-P^M\;[S_0(\infty)-S_0(-\infty)]+[\int_{-\infty}^\infty\;dr\;E^A_{,A}
]\; 
S_0(\infty)\}
\ee
which vanishes for all $E^A,p^M$ iff $S_0(\infty)=S_0(-\infty)=0$. 
Such $S_0$ correspond to a gauge transformation. The general transformation 
of $A_\mu$ is $\delta A_\mu=-S_{0,\mu}$. We cannot gauge $A_3$ completely to 
zero because for general $A_3$ this would require that 
$S_0=\int_{-\infty}^r\;ds A_3(s)$ which does not necessarily 
vanish at $r=\infty$. Let $w$ be a function of $r$ only with
$\int_{-\infty}^\infty dr w=1,\; w=O(r^{-2})$, say 
$w=\frac{1}{\pi\;r_0 (1+r^2/r_0^2)}$. Then we can gauge $A_3$ to 
$A_3^\ast=w\; Q_M,\;Q_M:=\int_{-\infty}^\infty\; dr\; A_3$ using 
$S_0=\int_{-\infty}^r \; dr\; [A_3-Q_M\; w]$ which now is a gauge 
transformation. Under a gauge transformation $Q_M$ is an invariant.
The transformations that stabilise the gauge $A_3-Q_M w$ satisfy
$-S_0'+w\;[S_0]_{-\infty}^\infty=0$ i.e. 
$S_0^\ast=S_0^\ast(-\infty)+[\int_{-\infty}^r w]\; 
K,\; K=[S_0^\ast]_{-\infty}^\infty$. 
which does not vanish at $r=\infty$ and thus corresponds to a symmetry 
transformation. The symplectic structure pulled back to $E^3_\ast, A_3^\ast$
is given by 
\ba \label{3.3a3}
\Theta_M &=& 
\int\;dr\; d^2y\;\{[P^M-\int_{-\infty}^r\;ds\; E^B_{,B}(s)]\;
w(r)\; [d Q_M]+E^B \;[dA_B]\;]
\nonumber\\
&=&\int\; d^2y\; P^M\; [dQ_M]+
\int\;dr\; d^2y\; E^B\;[d(A_B+[\int_r^\infty\; ds \; w(s)]\;Q_{M,B}]
\ea
and displays the canonical pair $(Q_M,P^M)$ on the sphere 
and the bulk canonical pair $(E^B,\hat{A}_B=A_B+
[\int_r^\infty ds w(s)] Q_{M,B})$. Under a symmetry transformation
$E^B,\hat{A}_B$ are both invariant. For the magnetic field we have
$B^3=\epsilon^{BC} A_{C,B}=\epsilon^{BC} \hat{A}_{C,B}$ and 
$B^A=-\epsilon^{AB}[A_B'-Q_{M,B} w)=-\epsilon^{AB}\;\hat{A}_B'$ 
i.e. it just depends on the invariant $\hat{A}_B$. Thus neither 
the magnetic nor the electric field depend explicitly on the $Q_M$ which 
thus acquire the same invisibility as the momenta $Q$ conjugate to the 
gravitational mass if we follow the startegy of section 
\ref{sd.6}. The electric field does depend on $P^M$ of which
the $l=0$ mode is just the electric charge. \\
\\
We now perform the same analysis more explicitly in terms of spherical 
harmonics. The  
coorsponding symmetric and non-symmetric Gauss constraint is easily found 
to be using ($(.)':=\frac{d}{dr}(.)$)
\ba \label{3.1}
G &=&\nabla_\mu E^\mu
=\omega[p^{M\prime}+\sum_{l>0}\;[y_{l,m}^{M\prime} L_{l,m}
+\sum_{\alpha\in \{o,e\}} Y_{\alpha,l,m}^M\; D_A\; L^A_{\alpha,l,m}]
\nonumber\\
&=&
=\omega[p^{M\prime}+\sum_{l>0}\;[y_{l,m}^{M\prime}
-\sqrt{l(l+1)} Y_{e,l,m}^M]\; L_{l,m}
\ea
Thus
\be \label{3.2}
C_M(r)=C_{M(0)}(r)=p^{M\prime}(r),\; 
Z_{M,l,m}(r)=Z_{M,l,m(1)}(r)=y_{l,m}^{M\prime}-\sqrt{l(l+1)} Y^M_{e,l,m}(r)
\ee
with general solution (we relabel $z=\theta(z)r-\theta(-z)\bar{r}$ by $r$
which has no range on the real axis)
\be \label{3.3}
p^M(r)=p^M(0),\;
y^M_{l,m}(r)=p^M_{l,m}(0)+\sqrt{l(l+1)}\int_{-\infty}^r\; ds \; Y^M_{e,l,m}(s)
\ee
We plug these into the symplectic structure for the Maxwell field and 
obtain up to a total phase space differential
\ba \label{3.3a}
\Theta_M &=&\int_{-\infty}^\infty\; dr\; 
\{
p^M(r)\; dq_M(r)+\sum_{l>0,|m|\le l}\;[y_{l,m}^M\; dx^{l,m}_M(r)
+Y_{e,l,m}^M(r)\; dX^{e,l,m}_M(r) 
+Y_{o,l,m}^M(r)\; dX^{o,l,m}_M(r) 
\}
\nonumber\\
&=&P^M_h\;d[\int_{-\infty}^\infty\; dr\; q_M(r)]
+\sum_{l>0,|m|\le l} Y^M_{h,l,m}\; d[\int_{-\infty}^\infty \; 
dr\; x^{l,m}_M(r)]
+\sum_{l>0,|m|\le l}\;
\int_{-\infty}^\infty\;dr\;
\{Y_{e,l,m}^M(r)\; d[X^{e,l,m}_M(r)
\nonumber\\
&& +\sqrt{l(l+1)}\int_r^\infty\; ds\;
x^{l,m}_M(s)] 
+Y_{o,l,m}^M(r)\; dX^{o,l,m}_M(r) 
\}
\nonumber\\
&=:& P_h^M \; dQ^h_M+
\sum_{l>0,|m|\le l} \; Y_{h,l,m}^M \; dX^{h,l,m}_M
+
\int_{-\infty}^\infty\;dr\;[Y^M_{e,l,m}\; d\tilde{X}^{e,l,m}_M
+Y^M_{o,l,m}\; d\tilde{X}^{o,l,m}_M](r) 
\ea
where we have set $P_h^M:=p^M(0), Y_{h,l,m}^M:=p^M_{l,m}(0)$. We 
see that the reduced symplectic structure only depends
on the Dirac observables (with respect to the 
Gauss constraint) $(P^M_h,Q_M^h),(Y^M_{h,l,m},X^{h,l,m}_M)$ which 
which are 
independent of $r$ 
(they are the harmonic modes of $P^M,Q_M$ above) 
and the Dirac observables 
$(Y_{e,l,m}^M,\tilde{X}^{e,l,m}_M) 
(Y_{o,l,m}^M,\tilde{X}^{o,l,m}_M)$ which do depend on $r$
(they are the harmonic modes of $E^B,\hat{A}_B$ above). 

For the magnetic fields we find 
\ba \label{3.3b}
&& B^3=\epsilon^{AB}\; D_A A_B
=\omega\sum_{l>0,|m|\le l,\alpha\in\{e,o\}}
X^{\alpha,l,m}_M\; \; \eta^{AB}\; D_A L_{B;\alpha,l,m}
=\omega\sum_{l>0,|m|\le l}\;\sqrt{l(l+1)} X^{o,l,m}_M \; L_{l,m}
\nonumber\\
&& B^A=-\omega\; \eta^{AB}\;(A_B'-D_B A_3)
=-\omega\;\sum_{l>0,|m|<l}\;
[(X^{e,l,m\prime}_M-\sqrt{l(l+1)} x^{l,m}_M)\;L^A_{o,l,m}
-X^{o,l,m\prime}_M \;L^A_{e,l,m}]
\nonumber\\
&& =-\omega\;\sum_{l>0,|m|<l}\;
[\tilde{X}^{e,l,m\prime}\;L^A_{o,l,m}
-X^{o,l,m\prime}_M \;L^A_{e,l,m}]
\ea
which of course also only depends on these Dirac observables.    

Thus we can drop $A_3,E^3$ from the list of independent 
variables in $\Theta$ 
and denote $\tilde{X}^{e,l,m}_M$ by $X^{e,l,m}_M$ again. 
We drop the term proportional to $G$ from the Hamiltonian, note that 
$V_0, V_\mu$ then only depend on electric and magnetic fields and  
write these in terms of the degrees of 
freedom $P_h^M,P_{h,l,m}^M,X^{\alpha,l,m}_M,Y_{\alpha,l,m}^M$ 
explicitly
\ba \label{3.3c}
E^3 &=& \omega \; \{
P_h^M+\sum_{l>0,m}\;[Y_{h,l,m}^M+\sqrt{l(l+1)}\;
\int_{-\infty}^r\; ds\; Y_{e,l,m}^M(s)]\; L_{l,m}\}
=:
P_h^M+\sum_{l>0,m}\; \tilde{Y}_{h,l,m}^M\; L_{l,m}
\nonumber\\
E^A &=& 
\omega\;\sum_{l>0,m,\alpha=e,o}\; Y_{\alpha,l,m}\; L^A_{\alpha,l,m}
\nonumber\\
B^3 &=& \omega\;\sum_{l>0,m} \; \sqrt{l(l+1)}\; 
X^{o,l,m}_M\;L_{l,m}
=:\sum_{l>0,m} \tilde{X}^{h,l,m}_M \; L_{l,m}
\nonumber\\
B^A &=&-\omega\;\sum_{l>0,m,\alpha,\beta=e,o} \;
\epsilon_{\alpha\beta}\;
X^{\alpha,l,m\prime}_M\; L^A_{\beta,l,m}  
\ea
where $\tilde{Y}_{h,l,m}^M,\; \tilde{X}^{h,l,m}_M$ are 
just abbreviations for the above linear functions of $Y_{e,l,m}^M,\;
X^{o,l,m}_M$ respectively and $\epsilon_{\alpha\beta},\;
\alpha,\beta=e,o$ is the skew symbol with $\epsilon_{eo}=1$.  Note 
that 
\be \label{3.3d}
\eta_{AB} L^B_{\alpha,l,m}=-\sum_\beta\;  \epsilon_{\alpha\beta}
L_{A;\beta,l,m}
\ee
In this way, the Hamiltonian $H$ remains polynomial of degree two 
in these variables. The field components $E^3,B^3$ have merely the status 
of an abbreviation for the r.h.s. of the first line and third 
line of (\ref{3.3c}). \\

\subsection{General perturbative structure of the spatial diffeomorphism
constraint}
\label{s3.2}

It is simplest to start with the form of the constraint in which its 
geometric meaning becomes most transparent (dropping the boundary term
for the moment)
\be \label{3.4}
V_\mu[S^\mu]=\int\; d^3x\; \{
W^{\mu\nu}\; [{\cal L}_{\vec{S}} m]_{\mu\nu}
+E^\mu\; [{\cal L}_{\vec{S}} A]_\mu
+\pi^T\; [{\cal L}_{\vec{S}} \phi]
\}
\ee
The boundary term picked up by a variation has been discussed at 
length in the previous section as far as the gravitational degrees of freedom 
are concerned. The boundary term picked up with respect to the matter 
fields is 
$\int\; d\Sigma_\mu S^\mu [E^\nu\delta A_\nu+\pi \delta\phi] 
=\int\; d^2y\;S^3 [E^\nu\delta A_\nu+\pi \delta\phi]$ 
which vanishes identically since by construction the term in the square 
bracket which enters the symplectic potential is $O(r^{-2})$. For $\pi,\phi$
this is achieved if $\pi,\phi$ decay as $O(r^{-2})$ in Cartesian coordinates 
or $\pi$ as $O(1)$ and $\phi$ as $O(r{-2})$ and $\phi$ as $O(r^{-2}$ in 
spherical coordinates. The decay of $A,E$ was specified in the previous 
subsection.  

We have explicitly ($(.)'=\frac{d}{dr}(.)$)
\ba \label{3.5}
&& {\cal L}_{\vec{S}} m]_{33}
=S^3\; m_{33}'+S^A\;D_A m_{33}+2 m_{33}\;S^{3\prime}+2\; m_{3A} \;S^{A\prime}
\nonumber\\
&& [{\cal L}_{\vec{S}} m]_{3A}
=S^3\; m_{3A}'+S^B\;D_B m_{3A}+ 
m_{33}\;D_A S^3+
m_{B3}\;D_A S^B+
m_{3A}\;S^{3\prime}+
m_{BA}\;S^{B\prime}
\nonumber\\
&& [{\cal L}_{\vec{S}} m]_{AB}
=S^3\; m_{AB}'+S^C\;D_C m_{AB}+2 m_{C(A} D_{B)} S^C+ 2\;m_{3(A} D_{B)}S^3 
\nonumber\\
&& [{\cal L}_{\vec{S}} A]_3= S^3 \; A_3'+S^B\; D_B A_3+ A_3 \; S^{3\prime}
+A_B S^{B\prime}
\nonumber\\
&& [{\cal L}_{\vec{S}} A]_B= S^3 \; A_B'+S^C\; D_C A_B+ A_3 \; D_B S^3
+A_C\; D_B S^C
\nonumber\\
&& [{\cal L}_{\vec{S}} \phi]=S^3\; \phi'+S^A\; D_A\phi
\ea
Plugging the expansions (\ref{2.12}), (\ref{2.13}) into
(\ref{3.5}) we can carry out the covariant differentials and write 
(\ref{3.5}) in terms of contracted quadratic monomials of tensor 
harmonics. Then contracting with the momenta in (\ref{3.4}) we obtain 
a sum of contracted cubic monomials of tensor harmonics that are 
being integrated with $d\Omega$. These integrals can be performed in closed 
form.
Of interest to us in this work is not their explicit form but the 
qualitative structure. We note that we are interested in the 
coefficients of the smearing functions $f^a, g^j$ which are immediately 
available from (\ref{3.5}) modulo an integration by parts with respect
to the radial variable. These integrations by parts generate radial 
derivatives of momenta. We see from (\ref{3.5}) that both $S^{3\prime}$ and 
$S^{A\prime}$ occur only in  
$[{\cal L}_{\vec{S}} m]_{33},
[{\cal L}_{\vec{S}} m]_{3A},
[{\cal L}_{\vec{S}} A]_{3}$ therefore the only radial derivatives of 
momenta that appear are those of $W^{33},W^{3A},E^3$ (also the Gauss 
constraint contains only radial derivatives of $E^3$). 
Since the Hamiltonian constraint does not contain momentum derivatives,
it follows that the only radial derivatives acting on the 
radial coefficient functions of the tensor harmonics for the 
momenta that occur are those
of $p_v^E,y_{v,l,m}^E,y_{\alpha,l,m}^E,p^M,y_{l,m}^M$ with 
$l>0,\alpha\in\{e,o\}$. Since in the construction algorithm 
we want to solve jointly the constraints $C_a$ 
for $p_a$ and the constraints $Z_j$ for $y_j$ it follows that we can 
solve algebraically for $p_h^E, y_{h,l,m}^E$ but have to solve radial 
differential equations for the other $p,y$ type momenta just listed.
Note that in particular that the true momenta $P_A, Y_J$ occur without 
radial derivatives.

Performing integrations by parts explicitly to recover the actual
constraint without derivatives on the smearing function we find 
\ba \label{3.6}
&& V_\mu[S^\mu]=\int\; d^3x\; \{
\nonumber\\
&& S^3\;[
W^{33}\; m_{33}'-2(W^{33}\; m_{33})'
+2(W^{3A}\; m_{3A}'-2(W^{3A}\; m_{3A})')
+W^{AB}\; m'_{AB}-2\;D_A(W^{3A}\;m_{33}+W^{AB}\; m_{3B})
\nonumber\\
&&+\pi\;\phi'+ E^3\; A_3'-(E^3 A_3)'+E^B\;A_B'-D_B(E^B A_3)]
\nonumber\\
&& + S^A\;[
W^{33} \; D_A\; m_{33}-2(W^{33} \; m_{3A})'+2\; W^{3B}\; D_A m_{3B}
-2\;D_A(W^{3B} m_{3B})-2(W^{3B} m_{AB})'+W^{BC}\; D_A m_{BC}
\nonumber\\
&& -2\;D_B(W^{BC} m_{CA})]
\nonumber\\
&& +\pi\; D_A\phi + E^3\; D_A A_3-(E^3 A_A)'+E^B \; D_A A_B-D_B(E^B A_A)]
\}
\ea
Using the result of the previous subsection we may in fact exploit that
$G=0$ is identically satisfied so that the contribution from the 
Maxwell field can be simplified to
\ba \label{3.7}
&& S^\mu [E^\nu \partial_\mu A_\nu-\partial_\nu(E^\nu A_\mu)]
=2\;S^\mu \partial_{[\mu} A_{\nu]}\;E^\nu
=2\;S^\mu \epsilon_{\mu\nu\rho}\; E^\nu B^\rho
\nonumber\\
&=& 2\;[S^3 \; \epsilon_{BC}\; E^B\; B^C+S^A \; \epsilon_{AC}\;
(E^C\; B^3-E^3\; B^C)]
\ea
We extract the coefficients of $S^3, S^A$ with respect to the decomposition
(\ref{2.12}), see (\ref{2.16}) and (\ref{2.17}). The symmetric constraint is
\ba \label{3.7a}
&& f^h C_h=f^h\; <1,V_3/\omega>_{L_2}
= f^h\{
[p_v^E\; (q^v_E)'-2\;(p_v^E\; q^v_E)'+p_h^E\; (q^h_E)'+P^{KG}\; Q_{KG}']
\nonumber\\
&& +
\sum_{l>0,m} [(\sum_{\alpha=v,e,o}
y_{\alpha,l,m}^E\; (x^{\alpha,l,m}_E)'-
2\; (y_{\alpha,l,m}^E\; x^{\alpha,l,m}_E)')
+y_{h,l,m}^E\;(x^{h,l,m}_E)'
+Y_{l,m}^{KG}\; (X^{l,m}_{KG})'
\nonumber\\
&& + 
\sum_\alpha\; \epsilon_{\alpha\beta} 
Y_{\alpha,l,m}^M\; X^{\beta,l,m\prime}]_M
]
\}
\nonumber\\
&=:& f^h
\{C_{h(0)}((p^E,q_E),(P^{KG},Q_{KG}))+
C_{h(2)}((y^E,x_E),(Y^{KG},X_{KG}),(Y^M,X_M))
\}
\ea
The non-symmetric constraints are 
\ba \label{3.7b}
&& g^{h,l,m} \; Z_{h,l,m}= g^{h,l,m} <L_{l,m},V_3/\omega>_{L_2}
\nonumber\\
&=& g^{h,l,m}\;\{
[y_{v,l,m}^E \;(q^v_E)'+p_v^E\; (x^{v,l,m}_E)'-
2(y_{v,l,m}^E \;q^v_E+p_v^E\; x^{v,l,m}_E)'
\nonumber\\
&& +p_h^E\; (x^{h,l,m}_E)'
+y_{h,l,m}^E\; (q^h_E)'
+\sqrt{l(l+1)}\;y_{e,l,m}^E\; q^v_E
+P^{KG}\; (X^{l,m}_{KG})' + Y_{l,m}^{KG} \; (Q_{KG})'
]
\nonumber\\
&& +
\sum_{l',m',\tilde{l},\tilde{m}} \; 
[
<L_{l,m},L_{l',m'}\; L_{\tilde{l},\tilde{m}}>_{L_2}
\;\{
y_{v,l',m'}^E\; (x^{v,\tilde{l},\tilde{m}}_E)'
-2\;(y_{v,l',m'}^E\; x^{v,\tilde{l},\tilde{m}}_E)'
\}
\nonumber \\ &&
+\sum_{\alpha,\beta=o,e}\;
<L_{l,m},L^A_{\alpha,l',m'}\; L_{A;\beta,\tilde{l},\tilde{m}}>_{L_2}\;
\{
y_{\alpha,l',m'}^E\; (x^{\beta,\tilde{l},\tilde{m}}_E)'
-2\;(y_{\alpha,l',m'}^E\; x^{\beta,\tilde{l},\tilde{m}}_E)'
\}
\nonumber \\ &&
+
\sum_{\alpha,\beta=o,e}\;
<L_{l,m},L^{AB}_{\alpha,l',m'}\; L_{AB;\beta,\tilde{l},\tilde{m}}>_{L_2}\;
Y_{\alpha,l',m'}^E\; (X^{\beta,\tilde{l},\tilde{m}}_E)'
\nonumber \\ &&
+
\sum_{\alpha=o,e}\;
<L_{l,m},L^{AB}_{\alpha,l',m'}\; L_{AB;h,\tilde{l},\tilde{m}}>_{L_2}\;
Y_{\alpha,l',m'}^E\; (x^{h,\tilde{l},\tilde{m}}_E)'
\nonumber \\ &&
+
\sum_{\beta=o,e}\;
<L_{l,m},L^{AB}_{h,l',m'}\; L_{AB;\beta,\tilde{l},\tilde{m}}>_{L_2}\;
y_{h,l',m'}^E\; (X^{\beta,\tilde{l},\tilde{m}}_E)'
\nonumber \\ &&
+
<L_{l,m},L^{AB}_{h,l',m'}\; L_{AB;h,\tilde{l},\tilde{m}}>_{L_2}\;
y_{h,l',m'}^E\; (x^{h,\tilde{l},\tilde{m}}_E)'
\nonumber \\ &&
+
2\;\sqrt{l(l+1)}\;\sum_{\alpha=o,e}\;
<L_{A;e,l,m},L^A_{\alpha,l',m'}\; L_{\tilde{l},\tilde{m}}>_{L_2^2}\;
y_{\alpha,l',m'}^E\; x^{v,\tilde{l},\tilde{m}}_E
\nonumber \\ &&
+
2\;\sqrt{l(l+1)}\;\sum_{\alpha,\beta=e,o}\;
<L_{A;e,l,m},L^{AB}_{\alpha,l',m'}\; L_{B;\beta\tilde{l},\tilde{m}}>_{L_2^2}\;
Y_{\alpha,l',m'}^E\; x^{\beta,\tilde{l},\tilde{m}}_E
\nonumber \\ &&
+
2\;\sqrt{l(l+1)}\;\sum_{\beta=e,o}\;
<L_{A;e,l,m},L^{AB}_{h,l',m'}\; L_{B;\beta\tilde{l},\tilde{m}}>_{L_2^2}\;
y_{h,l',m'}^E\; x^{\beta,\tilde{l},\tilde{m}}_E
+
<L_{l,m},L_{l',m'}\; L_{\tilde{l},\tilde{m}}>_{L_2}\;
Y_{l',m'}^{KG}\; (X^{\tilde{l}\tilde{m}})'
\nonumber \\ &&
+
\sum_{\alpha,\beta=e,o}\;
<L_{l,m},L^A_{\alpha,l',m'}\; L_{A;\beta\tilde{l},\tilde{m}}>_{L_2}\;
Y_{\alpha,l',m'}^M\; (X^{\beta,\tilde{l},\tilde{m}}_M)'\; 
]
\nonumber\\
&=:& g^{h,l,m}\;\{
Z_{h,l,m(1)}((y_v^E,q^v_E),(p_v^E,x^v_E),(y_h^E,q^h_E),(p_h^E,x^h_E),
(P^{KG},X_{KG}),(Y^{KG},Q_{KG}))
\nonumber\\
&& +Z_{h,l,m(2)}((y^E,x_E);
(Y^E,X_E),(Y^E,x_E),(y^E,X_E),(Y^{KG},X_{KG}),(Y^M,X_M))
\}
\ea 
and 
\ba \label{3.7c}
&& g^{\alpha,l,m}\; Z_{\alpha,l,m}=
g^{\alpha,l,m}\; \sum_{A=1,2}\; 
(<L^A_{\alpha,l,m},V_A/\omega>_{L_2})_{\alpha=o,e}\;
\nonumber\\
&=&
g^{\alpha,l,m}\; \{
[
\sqrt{l(l+1)}\; p_v^E \; \delta^\alpha_e\; x^{v,l,m}_E
-2(p_v^E\;x^{\alpha,l,m}_E)'
-2(y_{\alpha,l,m}^E\; q^h_E)' 
\nonumber\\ &&
+ 
[-2\sqrt{l(l+1)/2}\; y_{h,l,m}^E\; \delta_\alpha^e
+\sqrt{(l-1)(l+2)/2} Y_{\alpha,l,m} ]q^h_E
\nonumber\\ &&
+
\sqrt{l(l+1)}\; P^{KG}\; \delta_\alpha^e X^{l,m}_{KG}
-
P^M_h \;(X^{\alpha,l,m}_M)'
]
\nonumber\\ &&
+
\sum_{l',m',\tilde{m},\tilde{m}}\;
[
\sqrt{l(l+1)}\;
<L^A_{\alpha,l,m},L_{l',m'}\;L_{A;e,\tilde{m},\tilde{m}}>_{L_2}\;
y_{v,l',m'}^E\; x^{v,\tilde{l},\tilde{m}}_E
\nonumber \\ &&
-
2\;\sum_{\beta=e,o}\; 
<L^A_{\alpha,l,m},L_{l',m'}\;L_{A;\beta,\tilde{m},\tilde{m}}>_{L_2}\;
(y_{v,l',m'}^E\; x^{\beta,\tilde{l},\tilde{m}}_E)'
\nonumber \\ &&
+
2 \sum_{\beta,\gamma=e,o}\; 
<L^A_{\alpha,l,m},
L^B_{\beta,l',m'}\;D_A L_{B;\gamma,\tilde{m},\tilde{m}}>_{L_2}\;
y_{\beta,l',m'}^E\; x^{\gamma,\tilde{l},\tilde{m}}_E
\nonumber \\ &&
+
2\sqrt{l(l+1)}\;\sum_{\beta,\gamma=e,o}\; \delta_\alpha^e\;
<L_{l,m},L^A_{\beta,l',m'}\;L_{A;\gamma,\tilde{m},\tilde{m}}>_{L_2}\;
y_{\beta,l',m'}^E\; x^{\gamma,\tilde{l},\tilde{m}}_E
\nonumber 
\ea

\newpage

\ba 
&&
-
2\sum_{\beta=e,o}\;
<L^A_{\alpha,l,m},L^B_{\beta,l',m'}\;L_{AB;h,\tilde{m},\tilde{m}}>_{L_2}\;
(y_{\beta,l',m'}^E\; x^{h,\tilde{l},\tilde{m}}_E)'
\nonumber \\ &&
+
\sum_{\beta,\gamma=e,o}\;
<L^A_{\alpha,l,m},L^{BC}_{\beta,l',m'}\;
D_A\;L_{BC;\gamma,\tilde{m},\tilde{m}}>_{L_2}\;
Y_{\beta,l',m'}^E\; X^{\gamma,\tilde{l},\tilde{m}}_E
\nonumber \\ &&
+
\sum_{\beta=e,o}\;
<L^A_{\alpha,l,m},L^{BC}_{\beta,l',m'}\;
D_A\;L_{BC;h,\tilde{m},\tilde{m}}>_{L_2}\;
Y_{\beta,l',m'}^E\; x^{h,\tilde{l},\tilde{m}}_E
\nonumber \\ &&
+
\sum_{\gamma=e,o}\;
<L^A_{\alpha,l,m},L^{BC}_{h,l',m'}\;
D_A\;L_{BC;\gamma,\tilde{m},\tilde{m}}>_{L_2}\;
y_{h,l',m'}^E\; X^{\gamma,\tilde{l},\tilde{m}}_E
\nonumber \\ &&
+
<L^A_{\alpha,l,m},L^{BC}_{h,l',m'}\;
D_A\;L_{BC;h,\tilde{m},\tilde{m}}>_{L_2}\;
y_{h,l',m'}^E\; x^{h,\tilde{l}\tilde{m}}_E
\nonumber \\ &&
+
2\;\sum_{\beta,\gamma=e,o}\;
<D_B \;L^A_{\alpha,l,m},L^{BC}_{\beta,l',m'}\;
L_{CA;\gamma,\tilde{m},\tilde{m}}>_{L_2}\;
Y_{\beta,l',m'}^E\; X^{\gamma,\tilde{l},\tilde{m}}_E
\nonumber \\ &&
+
2\;\sum_{\beta=e,o}\;
<D_B \;L^A_{\alpha,l,m},L^{BC}_{\beta,l',m'}\;
L_{CA;h,\tilde{m},\tilde{m}}>_{L_2}\;
Y_{\beta,l',m'}^E\; x^{h,\tilde{l},\tilde{m}}_E
\nonumber \\ &&
+
2\;\sum_{\gamma=e,o}\;
<D_B \;L^A_{\alpha,l,m},L^{BC}_{h,l',m'}\;
L_{CA;\gamma,\tilde{m},\tilde{m}}>_{L_2}\;
y_{h,l',m'}^E\; X^{\gamma,\tilde{l},\tilde{m}}_E
\nonumber \\ &&
+
2\;<D_B \;L^A_{\alpha,l,m},L^{BC}_{h,l',m'}\;
L_{CA;h,\tilde{m},\tilde{m}}>_{L_2}\;
y_{h,l',m'}^E\; x^{h,\tilde{l},\tilde{m}}_E
\nonumber \\ &&
+
\sqrt{l(l+1)}\;
<L^A_{\alpha,l,m},L_{l',m'}\;L_{A;e,\tilde{m},\tilde{m}}>_{L_2}\;
Y_{l',m'}^{KG}\; X^{\tilde{l},\tilde{m}}_{KG}
\nonumber \\ &&
+
\sum_{\beta,\gamma=e,o}\;
\epsilon_{\beta\gamma}
<L^A_{\alpha,l,m},L_{A;\gamma,l',m'}\;L_{\tilde{m},\tilde{m}}>_{L_2}\;
Y_{\gamma,l',m'}^M\; \tilde{X}^{h,\tilde{l},\tilde{m}}_M
\nonumber \\ &&
-
\sum_{\beta=e,o}\;
\epsilon_{\beta\gamma}
<L^A_{\alpha,l,m},L_{l',m'}\;L_{A,\beta,\tilde{m},\tilde{m}}>_{L_2}\;
\tilde{Y}_{h,l',m'}^M\; (\tilde{X}^{\beta,\tilde{l},\tilde{m}}_M)'
]\;
\}
\nonumber\\
&=:& g^{\alpha,l,m} \; \{
Z_{\alpha,l,m(1)}((p_v^E,x^\alpha_E),(y_\alpha^E,q^v_E),
(y_h^E\delta_\alpha^e,q^h_E),(Y_\alpha^E,q^h_E),(P^{KG},\delta_\alpha^e
X_{KG}),(P_h^M,X^\alpha_M))
\nonumber\\
&& +
Z_{\alpha,l,m(2)}((y^E,x_E),(Y^E,X_E),(Y^E,x_E),
(y^E,X_E),(Y^{KG},X_{KG}),(Y^M,X_M))
\}
\nonumber
\ea
where we used the identities (\ref{2.8a}). The notation is that 
$C_{h(n)},\;Z_{\alpha,l,m(n)};\;\alpha\in \{h,e,o\}$ are the collection 
of all terms of order $n=0,1,2$ in the 
perturbations $x,y,X,Y$ and we displayed the pairs of variables 
on which the constraints depend (either a pair of two symmetric, two 
non-symmetric or mixed dgrees of freedom). We see that
$C_{h(1)}=Z_{\alpha,l,m(0)}=0;\;\alpha\in \{h,e,o\}$.
We also note that $C_{h(2)}=\sum_{l,m}\; C_{h,l,m(2)}$ and just like 
$Z_{\alpha,l,m(1)},\; \alpha=h,e,o$ 
the contribution $C_{h,l,m(2)}$ just depends on 
perturbation variables labelled by $l,m$. By contrast, the constraints
$Z_{\alpha,l,m(2)},\; \alpha=h,e,o$ are ``non-local'', i.e. depend 
not only on variables labelled by $l,m$ but in general on 
an infinite number of them because the triangle inequality 
$|\tilde{l}-l'|\le l\le \tilde{l}'+l$ admits an infinite 
number of solutions $\tilde{l},l'$ for any given $l$ e.g.
$\tilde{l}=l'+l,\;l'\in \mathbb{N}$.   
%We write the dependencies of the $C_a,Z_j$ in terms of the pairs 
%of the bilinear terms that occur. 
We display the ``colour'' label 
$v,h,e,o$ for the pairs of variables that occur for 
$C_{h(0)},\;C_{h,l,m(2)},Z_{\alpha,l,m(1)}$ but drop it 
in $Z_{\alpha,l,m(2)}$ although not all possible pairs occur in order
to make the notation not too heavy. 

Note that (\ref{3.7}) is the {\it exact} expression for the 
spatial diffeomorphism constraint, no terms have been dropped. 
We have just written it in terms of the split variables. The explicit 
computation of the coefficients will be carried out in our companion 
papers \cite{NT-SS,NT-RN, N-SA}.

\subsection{General perturbative structure of the Hamiltonian 
constraint}
\label{s3.3}

As emphasised in section \ref{s2.3} it is of considerable computational 
advantage to decompose the polynomial constraint $[\det(m)]^{5/2}\;V_0$ 
displayed in (\ref{2.9}) with respect to the 
perturbations $x,y,X,Y$ which for vanishing potential $U$ is just 
a polynomial of order ten. This is because it remains a 
polynomial rather than an infinite series as long as the 
potential is a polynomial in $\phi$. 
Still, working out all orders
explicitly is a tedious task both algebraically and because one 
needs to perform iterated Clebsch-Gordan decompositionsm, i.e. we 
need the general coefficient ${\sf Tr}(<\prod_{k=0}^N L_k>)$ where $N$ is the 
top polynomial degree that occur, each $L_k$ is a spherical 
harmonic (scalar, vector, tensor) and the trace and expectation value 
indicate contraction of all spherical tensor indices and integration 
on the sphere respectively. We thus 
just display the terms of order zero, one and two. In particular
we need 
\be \label{3.8}
f^v\; C_v= f^v \;<1,[\det(m)]^{5/2}\; V_0 \; \omega^{-6}>_{L_2},\;\;
g^{v,l,m}\; Z_{v,l,m} = 
g^{v,l,m}\;<L_{l,m},[\det(m)]^{5/2}\; V_0 \; \omega^{-6}>_{L_2}
\ee
to those orders and it is clear from section \ref{s2} that 
$C_{v(1)}=Z_{v,l,m(0)}=0$. 

All terms in $[\det(m)]^{5/2} \; V_0$ except for the curvature term
contain two or three factors of 
\ba \label{3.9a}
M &:=& \det(m)
= \frac{1}{3!}\;\epsilon^{\mu\nu\rho}\epsilon^{\mu'\nu'\rho'}\;
m_{\mu\mu'}\; m_{\nu\nu'}\;  m_{\rho\rho'}\;  
\\
%&=&\frac{1}{2}\;\epsilon^{AB}\epsilon^{\mu'\nu'\rho'}\;
%m_{3\mu'}\; m_{A\nu'}\;  m_{B\rho'}\;  
%\nonumber\\
%&=& \frac{1}{2}\epsilon^{AB}\epsilon^{CD}\;
%[m_{33}\; m_{AC}\;  m_{BD}\;  
%-m_{3C}\; m_{A3}\;  m_{BD}\;  
%+m_{3C}\; m_{AD}\;  m_{B3}]
%\nonumber\\
&=& \frac{1}{2}\epsilon^{AC}\epsilon^{BD}\;
[m_{33}\; m_{AB}\;  m_{CD}\;  
-2\;m_{3A}\; m_{3B}\;  m_{CD}]  
=\omega^2\frac{1}{2}\eta^{AC}\eta^{BD}\;
[m_{33}\; m_{AB}\;  m_{CD}\;  
-2\;m_{3A}\; m_{3B}\;  m_{CD}]  
\nonumber
\ea
We also need 
\ba \label{3.9}
M^{\mu\mu'} &:=& \det(m) \; m^{\mu\mu'}
=\frac{1}{2} \;
\epsilon^{\mu\nu\rho}\epsilon^{\mu'\nu'\rho'}\;
m_{\nu\nu'}\;  m_{\rho\rho'}\;  
\nonumber\\
M^{33} &=& 
\frac{1}{2}\epsilon^{AC}\epsilon^{BD}\;
m_{AB}\;  m_{CD}
=\omega^2\frac{1}{2}\eta^{AC}\eta^{BD}\;
m_{AB}\;  m_{CD}
\nonumber\\
M^{3A} 
%&=& 
%=\frac{1}{2} \;
%\epsilon^{BD}\epsilon^{A\nu'\rho'}\;
%m_{B\nu'}\;  m_{D\rho'}\;  
%=\frac{1}{2} \;
%\epsilon^{BD}\epsilon^{AC}\;
%[m_{BC}\;  m_{D3}\;  
%-m_{B3}\;  m_{DC}]  
%=\frac{1}{2} \;
%\epsilon^{AC}\;\epsilon^{BD}\;
%[m_{BC}\;  m_{D3}\;  
%-m_{B3}\;  m_{DC}]  
&=& - \epsilon^{AC}\;\epsilon^{BD}\;m_{3B}\;  m_{CD}  
= - \omega^2\;\eta^{AC}\;\eta^{BD}\;m_{3B}\;  m_{CD}  
\nonumber\\
M^{AB} 
%&=&
%\frac{1}{2} \;
%\epsilon^{A\nu\rho}\epsilon^{B\nu'\rho'}\;
%m_{\nu\nu'}\;  m_{\rho\rho'}\;  
%\nonumber\\
%&=&\frac{1}{2} \;
%\epsilon^{AC}\epsilon^{B\nu'\rho'}\;
%[
%m_{C\nu'}\;  m_{3\rho'}  
%-m_{3\nu'}\;  m_{C\rho'}]  
%\nonumber\\
%&=&\frac{1}{2} \;
%\epsilon^{AC}\epsilon^{BD}\;
%[
%m_{CD}\;  m_{33}-m_{3D}\;  m_{C3}  
%-m_{C3}\;  m_{3D}+m_{33}\;  m_{CD}]  
%\nonumber\\
&=&
\epsilon^{AC}\epsilon^{BD}\;[m_{33}\; m_{CD}-m_{3C}\; m_{3D}]
=\omega^2\eta^{AC}\eta^{BD}\;[m_{33}\; m_{CD}-m_{3C}\; m_{3D}]
\ea

Next we compute $[\det(m)]^3$ times the Ricci scalar 
\be \label{3.10}
R=2\;m^{\mu\rho}\;
[-\partial_{[\mu}\;\Gamma^\nu_{\nu]\rho}+\Gamma^\lambda_{\rho[\mu} \;
\Gamma^\nu_{\nu]\lambda}]
\ee
which after some algebra yields
\ba \label{3.11}
U^E &:=& M^3\; R
\\
&=&
M^{\mu\rho}\; m_{\lambda\mu,\rho}\;
[M\; M^{\nu\lambda}_{,\nu}-\frac{1}{2}\;M_{,\nu}\; M^{\nu\lambda}]
-\frac{1}{2}\; M\; M^{\mu\nu}\; M_{,\mu\nu}
+\frac{3}{4}\; M^{\mu\nu}\; M_{,\mu}\; M_{,\nu}
-\frac{1}{2}\; M\; M^{\mu\nu}_{,\nu}\; M_{,\mu}
\nonumber\\
&&
-M^{\mu\rho}\; M^{\nu\lambda}\; M^{\sigma\tau}\;
\Gamma_{\nu\mu\sigma} \; \Gamma_{\tau\rho\lambda}
\nonumber
\ea
which is manifestly a homogeneous polynomial of order eight. We refrain from
computing its low order expression explicitly as they do not involve 
momenta and just denote them as 
\be \label{3.11a}
\omega^6\;[U_{(0)}^E(q^v_E,q^h_E)
+U_{(1)}^E(q^v_E,q^h_E,x_E,X_E)
+U_{(2)}^E(q^v_E,q^h_E,x_E,X_E)]
\ee
where all three terms are scalars of density weight zero with respect to 
$S^2$.

We discuss the expansion to order two of the various other terms separately
denoting by $[.]_{(n)}$ the homogeneous n-th order contribution of 
$[.]$ and dropping terms of order three or higher
\ba \label{3.12}
&& 
[M^2]\; [m_{\mu\rho}\; m_{\nu\sigma}-\frac{1}{2}\; m_{\mu\nu}
m_{\rho\sigma}]\;[W^{\mu\nu}\; W^{\rho\sigma}]
\nonumber\\
&=&
\sum_{r,s,n\ge 0,r+s+n\le 2}\;
[M^2]_{(r)}\; [m_{\mu\rho}\; m_{\nu\sigma}-\frac{1}{2}\; m_{\mu\nu}
m_{\rho\sigma}]_{(s)}\;[W^{\mu\nu}\; W^{\rho\sigma}]_{(n)}
\ea
We have 
\ba \label{3.13}
M_{(0)} &=& q^v_E\; (q^h_E)^2\; \omega^2
\\
M_{(1)} &=& [m_{33(1)}\;(q^h_E)^2+q^v_E\; q^h_E\;\Omega^{AB}\;m_{AB(1)}]\;  
\omega^2 
\nonumber\\
M_{(2)} &=& [m_{33(1)}\;q^h_E\;\Omega^{AB}\;m_{AB(1)}+
+\frac{1}{2}\;m_{33(0)}\; \eta^{AC}\;\eta^{BD}\;m_{AB(1)}\;  m_{CD(1)}\;  
-q^h_E\;\Omega^{AB}\; m_{3A(1)}\; m_{3B(1)}]\;\omega^2 
\nonumber
\ea
and thus 
\be \label{3.14}
[M^2]_{(0)}=[M_{(0)}]^2,\;
[M^2]_{(1)}=2\; [M_{(0)}]\; [M_{(1)}],\;
[M^2]_{(2)}=[M_{(1)}]^2+2\; [M_{(0)}]\; [M_{(2)}],\;
\ee
Likewise exactly
\ba \label{3.14a}
M^{33} &=&
\omega^2\frac{1}{2}[2\;(q^h_E)^2+2\; (q^h_E)\;\Omega^{AB}\; m_{AB(1)}
+\eta^{AC}\eta^{BD}\;
m_{AB(1)}\;  m_{CD(1)}]
\nonumber\\
M^{3A} &=&- \omega^2\;
[q^h_E\;\Omega^{AB}\;m_{3B(1)}  
+\eta^{AC}\;\eta^{BD}\;m_{3B(1)}\;  m_{CD(1)}]
\nonumber\\
M^{AB} &=&
\omega^2\;\{
\Omega^{AB} \; (q^v_E)\;(q^h_E)
+[\Omega^{AB} \; m_{33(1)}\;(q^h_E)+
+\eta^{AC}\eta^{BD}\;(q^v_E)\; m_{CD(1)}]
\nonumber\\
&& +\eta^{AC}\eta^{BD}\;[m_{33(1)}\; m_{CD(1)}-m_{3C(1)}\; m_{3D(1)}]
\}
\ea
Next we have the exact result 
\ba \label{3.15}
W &:=& m_{\mu\nu}\; W^{\mu\nu}=
m_{33} W^{33}+2\; m_{3A}\; W^{3A}+m_{AB}\; W^{AB}
\nonumber\\
&=& [q^v_E \; p^v_E+q^h_E\; p_h^E]\;\omega
+[m_{33(1)}\; p^v_E\;\omega+q^v_E\; W^{33}_{(1)}+
+q^h_E\; \Omega_{AB}\; W^{AB}_{(1)}+
+\frac{1}{2}\;m_{AB(1)}\; \Omega^{AB}\; p_h^E\;\omega]
\nonumber\\
&&
+[m_{33(1)}\; W^{33}_{(1)}+2\; m_{3A(1)}\; W^{3A}_{(1)}+m_{AB(1)}\; 
W^{AB}_{(1)}]
=: W_{(0)}+W_{(1)}+W_{(2)} 
\ea
and thus up to second order 
\be \label{3.16}
W^2=[W_{(0)}]^2+[2\;W_{(0)}\;W_{(1)}]+[W_{(1)}^2+2\;W_{(0)}\;W_{(2)}]
\ee
Next we have the exact expression
\ba \label{3.17}
T^E &:=& m_{\mu\rho}\;m_{\nu\sigma}\;
W^{\mu\nu}\;W^{\rho\sigma}\;
\nonumber\\
&=&
[m_{33}\; W^{33}]^2
+
4\;[m_{33}\; W^{33}]\;[m_{3A}\; W^{3A}]
+
2\;[m_{3A}\; W^{3A}]^2
\nonumber\\
&&+
2 \; m_{33}\;[m_{AB}\; W^{3A}\; W^{3B}]
+
2 \; W^{33}\;[W^{AB}\; m_{3A}\; m_{3B}]
\nonumber\\
&&+
4\; m_{3A}\; P^{AB}\;m_{BC}\; W^{3C}
+
m_{AC}\; m_{BD}\; W^{AB}\; W^{CD}
\ea
It follows to second order
\ba \label{3.18}
T^E_{(0)}
&=& \omega^2([q^v_E\; p_v^E]^2+\frac{1}{2} [q^h_E\; p_h^E]^2)
\\
T^E_{(1)} &=&
\omega\;[2\;(q^v_E\; p_v^E)\;(q^v_e\; W^{33}_{(1)}+m_{33(1)}\;p_v^E)
+
2\;\omega\;\Omega^{AB}\; m_{AB(1)}\;(p_h^E)^2
+
2\;\Omega_{AB}\; W^{AB}_{(1)}\;(q^h_EE)^2
]
\nonumber\\
T^E_{(2)} &=&
2\;\omega (q^v_E\; p_v^E)\; m_{33(1)}\; W^{33}_{(1)}+
(q^v_e\; W^{33}_{(1)}+\omega\;m_{33(1)}\;p_v^E)^2
\nonumber\\
%(qp+xp+qy+xy)^2
%=(qp)^2+2(qp)(xp+qy+xy)+(xp+qy+xy)^2
%=(qp)^2+2qp(xp+qy)+(2qpxy+(xp+qy)^2)+{xy(xp+qy)+(xy)^2
&& +
4\;\omega\;[q^v_E\; p_v^E]\;[m_{3A(1)}\; W^{3A}_{(1)}]
+
2 \; [q^v_E\; q^h_E]\;\Omega_{AB}\; W^{3A}_{(1)}\; W^{3B}_{(1)}
+
2 \;\omega^2\; [p_v^E\;p_h^E]\; \Omega^{AB}\; m_{3A(1)}\; m_{3B(1)}
\nonumber\\
&& +
4\;\omega\; [q^h_E\; p_h^E]\; m_{3A(1)}\; W^{3C}_{(1)}
\nonumber\\
&& +
\omega^2\; (p_h^E)^2\;
m_{AC(1)}\; m_{BD(1)}\; \Omega^{AB}\; \Omega^{CD}
+(q^h_E)^2\;
\Omega_{AC}\; \Omega_{BD}\; W^{AB}_{(1)}\; W^{CD}_{(1)}
+
4\omega\; q^h_E\; p_h^E\;
m_{AB(1)}\; W^{AB}_{(1)}
\nonumber
\ea
As far as the matter contributions are concerned we note that these 
are up to second order (we consider for concreteness only the uncharged case,
the charged case can be treated analogously)
\ba \label{3.19}
&& 2\; M^2\; T^{KG}
=M^2\; \pi^2
\nonumber\\
&=& [\omega^2\; [M^2]_{(0)}\; (P^{KG})^2]
+2\;\omega\; [M^2]_{(0)}\; P^{KG}\; \pi_{(1)}
+\omega^2\; [M^2]_{(1)}\; (P^{KG})^2]
\nonumber\\
&&
+[+\omega^2\; [M^2]_{(2)}\; (P^{KG})^2+
+\omega\; [M^2]_{(1)}\; P^{KG}\;\pi_{(1)}
+[M^2]_{(0)}\; (\pi_{(1)})^2]
\nonumber\\
&& 2 \;U^{KG} := M^{\mu\nu}\;\phi_{,\mu}\;\phi_{,\nu}
=M^{33}\;(\phi')^2+2\;M^{3A}\; \phi'\; (D_A\phi)+
M^{AB}\; (D_A\phi)\;(D_B\phi)
\nonumber\\
&=&
[M^{33}_{(0)}\;(Q_{KG}')^2]+
[M^{33}_{(1)}\;(Q_{KG}')^2+2M^{33}_{(0)}\; (Q_{KG}')\;(\phi_{(1)}')]
\nonumber\\
&&+
[M_{33(2)}\;(Q_{KG}')^2
+M_{33(0)}\;(\phi_{(1)}')^2+
+2\;M_{33(1)}\;(Q_{KG}')\;(\phi_{(1)}')
+2\;M^{3A}_{(1)}\; (Q_{KG}'\; (D_A\phi_{(1)})
\nonumber\\
&& +
M^{AB}_{(0)}\; (D_A\phi_{(1)})\;(D_B\phi_{(1)})]
\nonumber\\
&& M^2\; U^{KG} = \sum_{r+s\le 2}\; [M^2]_{(r)}\; A_{(s)}
\nonumber\\
&& V(\phi) = V(Q_{KG})+V'(Q_{KG})\; \phi_{(1)}+\frac{1}{2}
V^{\prime\prime}(Q_{KG})\;(\phi_{(1)})^2
\nonumber\\
&& 4\; T^M = m_{\mu\nu}\; E^\mu\; E^\nu
\nonumber\\
&=& 
m_{33}\; (E^3)^2+2\; m_{3A}\; E^3\;E^A+m_{AB}\; E^A\; E^B
\nonumber\\
&=& 
[q^v_E\; (P_h^M)^2]+[m_{33(1)}\;(P_h^M)^2+2\;(q^v_E)\; (P^h_M)\;(E^3_{(1)})]
\nonumber\\
&&+[m_{33(0)} \; (E^3_{(1)})^2+2\; m_{33(1)}\; (P_h^M)\;E^3_{(1)}
+2 m_{3A(1)} \; (P_h^M)\; E^A_{(1)}+(q^h_E)\;\Omega_{AB}\; E^A_{(0)}\;
E^B_{(1)}]
\nonumber\\
&& 4\; U^M = m_{\mu\nu}\; B^\mu\; B^\nu
\nonumber\\
&=& q^v_E\; (B^3_{(1)})^2+q^h_E\;\Omega_{AB}\; B^A_{(1)}\; B^B_{(1)}
\ea
~\\
It remains to compute zeroth, first and second order of 
$\tilde{V}_0:=[\det{m}]^{5/2}\;
V_0$ where for the purposes of this paper it will be sufficient 
to consider 
\be \label{3.20}
C_v:=<1,\tilde{V}_0/\omega^6>_{L_2}=C_{v(0)}+C_{v(2)}+...,\;\;
Z_{v,l,m}:=<L_{l,m},\tilde{v}_0/\omega^6>_{L_2}=Z_{v,l,m(1)}+...
\ee
because we wish to compute the reduced Hamiltonian only up to second 
order for which the solution of $C_a$ to second order and of $Z_j$ 
to first order is required. Accordingly
\ba \label{3.21}
C_{v(0)} &=& C_{v(0)}^E+C_{v(0)}^{KG}+C_{v(0)}^M
\nonumber\\
C_{v(2)} &=& C_{v(2)}^E+C_{v(2)}^{KG}+C_{v(2)}^M
\nonumber\\
Z_{v,l,m(1)} &=& Z_{v,l,m(1)}^E+C_{v,l,m(1)}^{KG}+C_{v,l,m(1)}^M
\ea
and we find to zeroth order 
\ba \label{3.22}
C_{v(0)}^E \; \omega^6 &=&
[M^2]_{(0)}\;[T^E_{(0)}-\frac{1}{2}\; [W^2]_{(0)}]-U^E_{(0)}
\nonumber\\
&=& \omega^6\;[q_v^E\; (q^h_E)^2]^2\;
\{[q^v_E\; p_v^E]^2+\frac{1}{2} [q^h_E\; p_h^E]^2)
-\frac{1}{2}
[q^v_E \; p^v_E+q^h_E\; p_h^E]^2\}
-\omega^6\; U^E_{(0)}
\nonumber\\
2\;C_{v(0)}^{KG} \; \omega^6 &=&
[M^2]_{(0)}\; [\pi^2]_{(0)}+
[M^2]_{(0)}\; U^{KG}_{(0)}
+2\;[M^3]_{(0)}\; V_{(0)}
\nonumber\\
&=&
\omega^6\; [q_v^E\; (q^h_E)^2]^2\; \{(P^{KG})^2+
(q^h_E)^2\;(Q_{KG}')^2+2[q_v^E \; (q^h_E)^2]\; V(Q_{KG})\}
\nonumber\\
4 C_{v(0)}^M \; \omega^6 &=& 
[M^2]_{(0)}\; T^M_{(0)}
=\omega^6\; [q_v^E\; (q^h_E)^2]^2\;q^v_E\;(P_h^M)^2
%\nonumber\\
%Z_{v,l,m(1)}^E &=& ...
\ea
The higher orders will be worked out in our companion papers.

\section{Perturbative construction of the reduced Hamiltonian}
\label{s5}

We first provide the general strategy and then display the details 
for the zeroth, first and second order.

\subsection{Overview}
\label{s5.1}

We follow the general procedure of \cite{pa129}: Adapted to the 
present situation, it consists the following steps:\\
\\
1.\\ 
We denote by $(.)_{(n)}$ the homogenous   
n-th order contribution of $(.)$ with respect to an 
expansion into the perturbations (which are considered of first order):\\  
a. $x_E^{\alpha,l,m},\; y^E_{\alpha,l,m};\; \alpha=v,h,e,o;\; l\ge 1$ \\
b. $X_E^{\alpha,l,m},\; Y^E_{\alpha,l,m};\; \alpha=e,o;\; l\ge 2$\\
c. $X_{KG}^{l,m},\; Y^{KG}_{l,m};\;  l\ge 1$\\
d. $X_M^{\alpha,l,m},\; Y^M_{\alpha,l,m};\; \alpha=e,o;\; l\ge 1$.\\
2.\\
Suppose that one solves the constraints $C_v,C_h,Z_{\alpha,l,m};
\alpha=v,h,e,o$ exactly
for $p^E_v,p^E_h,y^E_{\alpha,lm}$, then that solution 
$\hat{p}^E_v,\hat{p}^E_h,\hat{y}^E_{\alpha,lm}$ can itself be expanded 
into the contributions b.-d. above. 
We write those expansions as
$\hat{p}^E_\alpha=p^E_\alpha(0)+p^E_\alpha(2)+p^E_\alpha(3)+..;\; \alpha=v,h$ 
and
$\hat{y}^E_{\alpha,l,m}=y^E_{\alpha,l,m}(1)+y^E_{\alpha,l,m}(2)+..;\;
\alpha=v,h,e,o$ 
respectively,
where $(.)_{(n)}$ means the homogeneous n-th order contribution of $(.)$ with 
respect to $X,Y$ in b.-d..\\
3.\\
Expand the constraints $C_v,C_h,Z_{\alpha,l,m}$ first with respect to 
all variables a.-d. for general $p^E_\alpha, y^E_{\alpha,l,m}$  
and then in addition with respect to the decomposition of the solution
$p^E_\alpha=\hat{p}^E_\alpha,\; y^E_{\alpha,l,m}=\hat{y}^E_{\alpha,l,m}$.
Denote the n-th order homogeneous contribution with respect to that combined 
expansion by $C_{\alpha(n)},\; Z_{\alpha,l,m(n)}$ where by construction
$C_{\alpha(1)}=Z_{\alpha,l,m(0)}=0$ due to spherical symmetry.\\
4.\\
Solve the symmetric, zeroth order order constraints 
$C_{v(0)}=0, C_{h(0)}=0$ exactly for $p^E_v(0),p^E_h(0)$.
The symmetric, first order order constraints 
$C_{v(1)}\equiv 0, C_{h(1)}\equiv 0$ are equivalent to the statement that 
$p^E_v(1)=p^E_h(1)=0$.\\
5.\\
Solve the unsymmetric first order 
constraints $Z_{\alpha,l,m(1)}=0,\; \alpha=v,h,e,o; l\ge 1,|m|\le l$ for 
$y^E_{\alpha,l,m}(1)$ at $p^E_v=p_v(0),\; p^E_h=p^E_h(0)$.\\
6.\\
Proceeding iteratively, by construction \cite{pa129}, for $n\ge 2$ 
the constraint contribution $C_{\alpha(n)}$ depends linearly on 
the $p^E_\beta(n)$ and polynomially 
on the $p^E_\beta(k),\; y^E_{\beta,l,m}(k)\; k\le n-1$
while 
the constraint contribution $Z_{\alpha,l,m(n)}$ depends linearly on 
the $y^E_{\beta,l',m'}(n)$ and polynomially on the $p^E_\beta(k),\; 
y^E_{\beta,l',m'}(k)\; k\le n-1$. Therefore one can successively 
solve $C_{\alpha(n)}$ for $p^E_\beta(n)$
and $Z_{\alpha,l,m(n)}$ for $y^E_{\beta,l',m'(n)}$.\\
7.\\
In this way, one perturbatively determines the Abelianised form of the 
constraints
\ba \label{5.1}
\hat{C}_\alpha &=& p^E_\alpha+h_\alpha,\;\;
h_\alpha=-\sum_{1\not=n=0}^\infty \; p^E_\alpha(n),\;\; \alpha=v,h
\nonumber\\
\hat{Z}_{\alpha,l,m} &=& y^E_{\alpha,l,m}+h_{\alpha,l,m},\;\;
h_{\alpha,l,m}=-\sum_{n=1}^\infty \; y^E_{\alpha,l,m}(n),\;\;\alpha=v,h,e,o
\ea
8.\\
For the reduced Hamiltonian we are supposed evaluate (\ref{5.1}) in the GPG
$q^v=q_{33}=1,\;q^h=\Omega^{AB} q_{AB}/2=r^2,q^A=q_{3A}=0$. Therefore we 
may solve (\ref{5.1}) already with GPG installed.\\
9.\\
The reduced Hamiltonian is then given for each asymptotic 
end by (\ref{4.a.47}) (we drop constant pre-factors)
\be \label{5.2}
H_{{\rm red}}
=\lim_{r\to\infty}\;\frac{1}{r}\;
[h_v(r)^2+\sum_{l>0,\;|m|\le l}\;
h_{v,l,m}(r)^2]
\ee
which follows from $P^3=p_v+\sum_{l,m}\; y_{v,l,m}\; L_{l,m}$.
Using the expansion of $h_v,h_{v,l,m}$ into the $p_v(n),\;y_{v,l,m}(n)$ 
one can compute $H_{{\rm red}}$ to any desired order of accuracy. The
decay condition on the $y^E_\alpha$ stated in section \ref{s4} in fact 
imply that the $h_{v,l,m}$ contributions in (\ref{5.2}) vanish as 
$r\to\infty$.

\subsection{Zeroth order}
\label{s5.2}

At zeroth order we just need to solve the zeroth order of the 
symmetric parts of the constraints for the zeroth orders 
$p^E_h(0),\;p^E_v(0)$ which are (we drop the label ``0'' for $p^E_h,p^E_v$
and evaluate at GPG) 
\ba \label{5.3}
C_{h(0)} &=& 
-2 (q_{33(0)} P^{33})'-P^{33} \; q_{33(0)}'+q_{AB(0)}'\; P^{AB}
+\pi_{0)}\;\phi_{(0)}'
\nonumber\\
&=& -2(p_v^E\; q^v_E)'-p_v^E\; (q^v_E)'+p^h)E \;(q^h_E)'+P^{KG}\;Q_{KG}'
\nonumber\\
&=& -2(p_v^E)'+2r\; p^h_E+P^{KG}\;Q_{KG}'
\nonumber\\
C_{v(0)} 
&=& \frac{1}{\sqrt{\det(q)_{(0)}}}[
(q_{33(0)}\; P^{33})^2+q_{AC(0)} q_{BD(0)} \;P^{AB}\;P^{CD}
-\frac{1}{2}(q_{33(0)} P^{33}+q_{AB(0)} P^{AB})^2]
\sqrt{\det(q)_{(0)}}\; R[q]_{(0)}
\nonumber\\
&& +\frac{1}{2}[\frac{\pi_{(0)}^2}{\sqrt{\det(q)}_{(0)}}+
\sqrt{\det(q)_{(0)}}\; (q^{33}_{(0)}\; (\phi_{(0)}')^2+2\;V(\phi_{(0)})]
+\frac{1}{2}\sqrt{\det(q)_{(0)}}\; q_{33(0)}\; (E^3_{(0)})^2
\nonumber\\
&=& \frac{1}{\sqrt{q^v_E(q^h_E)^2}}\;[
(q^v_E\; p^v_E)^2+\frac{1}{2}\;(q^h_E)^2 \;(p_h^E)^2
-\frac{1}{2}(q^v_E\; p_v^E+q^h_E\; p_h^E)^2]
\sqrt{\det(q)_{(0)}}\; R[q]_{(0)}
\nonumber\\
&& +\frac{1}{2}[\frac{(P^{KG})^2}{\sqrt{q^v_E\; (q^h_E)^2}}+
\sqrt{q^v_E\; (q^h_E)^2}\; (q^E_v\; (Q_{KG}')^2+2\;V(Q_{KG})]
+\frac{1}{2\sqrt{q^v_E\; (q^h_E)^2}}\; q^v_E\; (P^M)^2
\nonumber\\
&=& \frac{1}{r^2}\;
[(p^v_E)^2+\frac{1}{2}\;r^4 \;(p_h^E)^2
-\frac{1}{2}(p_v^E+r^2\; p_h^E)^2]
\nonumber\\
&& +\frac{1}{2}[\frac{(P^{KG})^2}{r^2}+
r^2\; (q^E_v\; (Q_{KG}')^2+2\;V(Q_{KG})]
+\frac{1}{2\;r^2}\; (P^M)^2
\nonumber\\
&=& 
\frac{1}{2\;r^2}[(p^v_E)^2-2\;r^2 \;p_h^E\; p^E_v]
\nonumber\\
&& +\frac{1}{2}[\frac{(P^{KG})^2}{r^2}+
r^2\; (Q_{KG}')^2+2\;V(Q_{KG})]
+\frac{1}{2\;r^2}\; (P^M)^2
\ea
We solve $C_{h(0)}$ for $p_h^E$ 
\be \label{5.4}
p_h^E=\frac{1}{r^2}\;[2\;(p_v^E)'-I_{KG}];\;I^{KG}_{(0)}:=P^{KG}\;Q_{KG}'
\ee
and insert this into $C_{v(0)}$
\ba \label{5.5}
C_{v(0)} &=&
\frac{1}{2\;r^2}[(p^v_E)^2-2\;r^2 \;p^E_v\; [2\;(p_v^E)'-I^{KG}_{(0)}]
)
\nonumber\\
&& +\frac{1}{2\; r^2}[(P^{KG})^2+
r^4\; (Q_{KG}')^2+2\;V(Q_{KG})]
+\frac{1}{2\;r^2}\; (P^M)^2
\nonumber\\
&=&
\frac{1}{2\;r^2}[\{ (p^v_E)^2-2\;r^2 \;[(p^E_v)^2]'\}+ 
p_v^E\;I^{KG}_{(0)}]
\nonumber\\
&& +\frac{1}{2\; r^2}\{(P^{KG})^2+
r^4\; (Q_{KG}')^2+2\;V(Q_{KG})+(P^M)^2\}
\nonumber\\
&=:&
-\frac{1}{2}\;[\frac{(p^E_v)^2}{r}]'
+ p_v^E\;I^{KG}_{(0)}+E^{KG}_{(0)} + E^M_{(0)}
\ea
where $I^{KG}_{(0)}$ is the symmetric part of the Klein Gordon momentum 
density and $E^{KG}_{(0)},\; E^M_{(0)}$ the symmetric part of the 
Klein Gordon and Maxwell energy density respectively. 

The equation $C_{v(0)}=0$ is solvable in closed form if there is no scalar 
``hair'' (exploiting that $P^M$ is a spatial constant)
\be \label{5.6}
\frac{(p_v^E)^2}{2r}=\hat{M}-\frac{(P^M)^2}{r}
\ee
where $\hat{M}$ is the mass of the black hole and $\sqrt{2}\;P^M$ its 
electric charge.
Indeed in GPG one cane easily check that the information about mass 
and charge resides in the extrinsic curvature part of the initial data.  

In the presence of scalar hair we solve (\ref{5.5}) by the Picard
Lindel\"of method. Let us introduce the abbreviations 
\be \label{5.7}
z:=\frac{p_v^E}{\sqrt{r}},\; a:=I^{KG}_{(0)}\;\sqrt{r},\;
b:=E^{KG}_{(0)} + E^M_{(0)}
\ee
to cast (\ref{5.5}) into the ODE
\be \label{5.8}
z'=a+\frac{b}{z}
\ee
or equivalently into the integral equation
\be \label{5.9}
z(r)=z_0+\int_{r_0}^r\; ds\; [a(s)+\frac{b(s)}{z(s)}]
\ee
(with $z(r_0)=z_0$ an integration constant)
which can be iterated. To solve that iteration we expand (essentially 
an inverse square root of core mass $\hat{M}$ expansion)  
\be \label{5.10}
z(r)=z_0+\sum_{N=0}^\infty\; C_N(r)\;z_0^{-N}
\ee
and compare coefficients. Introducing the abbreviations 
\be \label{5.11}
A(r):=\int_{r_0}^r\; ds \; a(s),\;
B(r):=\int_{r_0}^r\; ds \; b(s)
\ee
one finds by expanding the geometric sums  
\be \label{5.12}
z(r)=A(r)+\frac{B(r)}{z_0}
+\frac{1}{z_0}\;\sum_{M=1}^\infty\;(-1)^M\; 
\sum_{n_1,..,n_M=1}^\infty\;z_0^{-[n_1+..+n_M]}\; 
\int_{r_0}^r\; ds\;b(s)\; \prod_{k=1}^M\; C_{n_k-1}(s)
\ee
i.e. for $N\ge 2$
\ba \label{5.13}
&& C_0=A,\; C_1=B,\;
C_N=\sum_{M=1}^\infty\; (-1)^M\; \sum_{n_1,..,n_M=1}^\infty\;
\delta_{N-1,n_1+..+n_M}\; \zeta_{n_1,..,n_M}
\nonumber\\
\zeta_{n_1,..,n_M} &=&
\int_{r_0}^r\; ds\;b(s)\; \prod_{k=1}^M\; C_{n_k-1}(s)
\ea
where $\zeta_{n_1,..,n_M}$ is completely symmetric.

The hierarchy can be solved iteratively: For each $N$ we have $M\le N-1$ since 
$n_k\ge 1$, hence at most finitely many terms survive in (\ref{5.14}). 
At the same time even for $M=1$ we have $n_1=N-1$ and thus at most $C_{N-2}$ 
appears on the right hand side of (\ref{5.14}). The first few terms are 
\ba \label{5.14}
C_2 &=& -\; \zeta_1=-\int_{r_0}^r\;ds\; b(s)\; A(s)
\nonumber\\
C_3 &=& -\; \zeta_2+\zeta_{1,1}=\int_{r_0}^r\;ds\; b(s)\; [A(s)^2-B(s)]
\nonumber\\
C_4 &=& -\zeta_3+2\;\zeta_{1,2}-\zeta_{1,1,1}
=\int_{r_0}^r\; ds\; b(s)\;[-C_2(s)+2\; A(s)\; B(s)-A(s)^3]
\ea
We see that the coefficients $C_n$ are polynomials in $P^{KG}, Q_{KG}, P^M$
which appear in nested integrals with respect to the radial 
coordinate. Note that for a discharged black hole $z_0$ is simply
$\sqrt{\hat{M}}$ which is a constant of motion when the metric 
does not depend on the momentum conjugate to $\hat{M}$. This means 
that for large $\hat{M}$ the inverse core mass expansion remains a good 
approximation also during time evolution as one expects the perturbation 
contributions to the nested energy integrals to be much smaller than 
$\hat{M}$.

\subsection{First order}
\label{s5.3}

We now consider $p_v^E(0), p_h^E(0)$ to be explicitly known via 
(\ref{5.7}), (\ref{5.9}) and (\ref{5.12}) and insert these as well 
as $p_v^E(1):=0,\; p_h^E(1):=0$ into 
the first order expansions  $Z_{\alpha,l,m(1)}; \alpha=v,h,e,o$. 
Dropping the labels $_{(0)}$ and $_{(1)}$ for $p_\alpha^E(0), 
y^E_{\alpha,lm,}(1)$ respectively for simplicity we find in the 
GPG with $\alpha=e,o$
\ba \label{5.15}
Z_{h,l,m(1)} 
& = & -(2 (q^v_{E}\; y^E_{v,l,m})'+(q^v_{E})'\; y^E_{v,l,m})
+\sqrt{2}\; (q^h_{E})'\; y^E_{h,l,m}+\sqrt{l(l+1)} \; q^v_{E}\; 
y^E_{e,l,m}
+Z^R_{h,l,m(1)}
\nonumber\\
& = & - 2 (y^E_{v,l,m})'
+2\; r\; \sqrt{2}\; y^E_{h,l,m}+\sqrt{l(l+1)} \;y^E_{e,l,m}
+Z^R_{h,l,m(1)}
\nonumber\\
Z_{\alpha,l,m(1)} 
&=&
-(q^h_{E}\; y^E_{\alpha,l,m})'+\sqrt{2l(l+1)}\; q^h_{E}\; 
\delta_\alpha^e\;
y^E_{h,l,m}+Z^R_{\alpha,l,m(1)}
\nonumber\\
&=&
-(r^2\; y^E_{\alpha,l,m})'+\sqrt{2l(l+1)}\; r^2\; \delta_\alpha^e\;
y^E_{h,l,m}+Z^R_{\alpha,l,m(1)}
\nonumber\\
Z_{v,l,m(1)} &=& [q^v_{E} q^h_{E}]^2\;
\{2 (q^v_{E})^2\; p_v^E\; y^E_{v,l,m}
+\frac{1}{\sqrt{2}}
(q^h_{E})^2\; p_h^E\; y^E_{h,l,m}
\nonumber\\
&& +(q^v_{E}\; p_v^E\;+q^h_{E}\; p_h^E)\;
+(q^v_{E}\; y_{v,l,m}^E\;+\frac{1}{\sqrt{2}}\;q^h_{E}\; y_{h,l,m}^E)
\}
+Z^R_{v,l,m(1)}
\nonumber\\
&=&
r^8\;\{
\{2 p_v^E\; y^E_{v,l,m}
+\frac{1}{\sqrt{2}}
r^4\;\; p_h^E\; y^E_{h,l,m}
\nonumber\\
&& -(p_v^E+r^2 p_h^E)\;
+(y_{v,l,m}^E\;+\frac{r^2}{\sqrt{2}}\; y_{h,l,m}^E)
\}
+Z^R_{v,l,m(1)}
\ea
Here the remainder $Z^R_{\alpha,l,m(1)}$ depends 
on $q^v_{E}, q^h_{E}$ polynomially,  
on $p_{v}^E(0), p_{h}^E(0)$ quadratically,  
on $P^{KG}, Q_{KG}, P^M$ quadratically,  
on $x^{\alpha,l,m}_E;\; \alpha=v,h,e,o$ linearly (but is set to zero in GPG),
on $Y_{\alpha,l,m}^E,\;X^{\alpha,l,m}_E;\;\alpha=e,o;\; l\ge 2$ linearly  
and on $Y_{\alpha,l,m}^M,\;X^{\alpha,l,m}_m;\;\alpha=e,o;\; l\ge 1$ linearly.

The system (\ref{5.15}) does not contain derivatives of 
$y^E_{h,l,m}$ and $y^E_{o,l,m}$ decouples from the system. We can therefore 
directly integrate 
\be \label{5.16}
r^2\; y^E_{o,l,m}=Y^E_{o,l,m}+\int_0^r \; ds\;Z^R_{\alpha,l,m(1)}(s)
\ee
and solve $Z_{v,l,m(1)}=0$ algebraically for $y^E_{h,l,m}$ 
\be \label{5.17}
\frac{r^2}{\sqrt{2}} \; p_h^E\; 
y^E_{h,l,m}
=(p_v^E-r^2 \; p_h^E)\; y^E_{v,l,m}+\frac{Z^R_{v,l,m(1)}}{r^8}
\ee
When inserted into the equations 
$Z_{h,l,m(1)}=Z_{e,l,m(1)}=0$ we can cast the remaining system of 
ODE's into the form 
\be \label{5.18}
\left[
\begin{array}{c}
y_{v,l,m}^E \\
y_{e,l,m}^E 
\end{array}
\right]'
=
\left[
\begin{array}{cc}
a & b\\
c & d
\end{array}
\right]
\;\cdot\;
\left[
\begin{array}{c}
y_{v,l,m}^E \\
y_{e,l,m}^E 
\end{array}
\right]
-
\left[
\begin{array}{c}
Z^R_{v,l,m(1)} \\
Z^R_{e,l,m(1)} 
\end{array}
\right]
\ee
for certain known functions $a,b,c,d$ that one can find from 
(\ref{5.15}). The inhomogeneous linear system  
(\ref{5.18}) which we write as $z'=A\cdot z+B$ is easily integrated 
in terms of the holonomy of the matrix valued function $A$
\be \label{5.19}
{\rm Hol}(r)={\cal P}(e^{\int_0^r\; ds\; A(s)})
\ee
where the path ordering symbol $\cal P$ orders the radial dependence 
of polynomials of $A(r)$ with the highst radius to the left. Then 
\be \label{5.20}
z(r)={\rm Hol}(r)\;[\tilde{z}_0+\int_0^r\; ds\; {\rm Hol}^{-1}(s)\; B(s)]\;\;
\tilde{z}_0=
\left[
\begin{array}{c}
Y_{v,l,m}^E \\
Y_{e,l,m}^E 
\end{array}
\right]
\ee
Here $Y^E_{\alpha,l,m},\;\alpha=v,e,o$ are integration constants. 
One may be worried that solving (\ref{5.17}) introduces inverse powers of 
$p^h_E$ and thus $p^v_E$. However, these inverse powers can again be
expanded in terms of powers of the inverse core mass times polynomials 
in $Q_{KG}, P^{KG}, P^M$. We will show in our companion papers that 
these differential equations can be solved explicitly.

\subsection{Second order}
\label{s5.4}

We now consider $y_{\alpha,l,m}(1)$ to be explicitly known 
via (\ref{5.16}), (\ref{5.17}), (\ref{5.20}) and insert 
the expansion 
$p_\alpha^E=p_\alpha(0)+p_\alpha^E(2);\;\alpha=e,o$ and 
$y_{\alpha,l,m}^E=y^E_{\alpha,l,m}(1)+y^E_{\alpha,l,m}(2)$
into $C_{\alpha(2)}=0$ and $Z_{\alpha,l,m(2)}=0$. In 
fact since we just need to second order
$p_v^E=<1,P^{33}>=p_v^E(0)+p_v^E(2)$ and thus 
$[p_v^E]^2=[p_v^E(0)]^2+2\; p-v^E(0)\;p_v^E(2)$
it is sufficient to compute the linear order $y_{\alpha,l,m}(1)$ and 
insert it into 
$C_{v(2)}, C_{h(2)}$
which already allows to extract $p_v^E(2)$. 
We find in the GPG 
\ba \label{5.21}
C_{h(2)} 
&=&-2\; (p_v^E(2)\; q^v_E)'-p_v^E(2) (q^v_E)'+p_h^E(2)\; (q^h_E)'
+\tilde{C}_{h(2)}
\nonumber\\
&=&-2\;(p_v^E(2)\; q^v_E)'+2\;r\; p_h^E(2)\; 
+\tilde{C}_{h(2)}
\nonumber\\
C_{v(2)} 
&=& [q^v_E (q^h_E)^2]^2\;
[
(
q^v_E\; p_v^E(0))\;
(\frac{1}{2}\; q^v_E\; p_v^E(2)-q^h_E\; p^h_E(2))
+
q^v_E\; p_v^E(2))\;
(\frac{1}{2}\; q^v_E\; p_v^E(0)-q^h_E\; p^h_E(0))
] 
+\tilde{C}_{v(2)}
\nonumber\\
&=& r^8\; 
[p_v^E(2)(p_v^E(0)-r^2\; p_h^E(0))
-r^2\; p_h^E(2)\; p_v^E(0)]
+\tilde{C}_{v(2)}
\ea
Here $\tilde{C}_{v(2)}, \tilde{C}_{h(2)}$ depend quadratically 
on $y_{\alpha,l,m(1)}^E,\; x^{\alpha,l,m}_E;\; \alpha=v,h,e,o$,
on $Y_{\alpha,l,m}^E,\; X^{\alpha,l,m}_E;\; \alpha=e,o$,
on $Y^{KG},\; X_{KG}$, 
on $Y_{\alpha,l,m}^M,\; X^{\alpha,l,m}_M;\; \alpha=e,o$,
on $Q_{KG}, P^{KG}, P^M$ and 
polynomially on $q^\alpha_E,\; \alpha=v,h,e,o$ but we assume their gauge fixed 
values and set $x^{\alpha,l,m}_E=0$. Similar as for the zeroth order 
we solve the first equation 
in (\ref{5.21}) algebraically for $p_{h(2)}$ and insert into 
the second. The result is a single linear inhomogeneous ODE for 
$p_v^E(2)$ which can be solved by standard methods. One has to devide 
by $p_v^E(0)$ in an intermediate step which upon an inverse 
mass expansion can be written again in terms of just polynomial fields
to arbitary precision. The details are contained in our companion 
papers \cite{NT-SS, NT-RN, N-SA}.

\section{Perturbative structure of the irreducible mass}
\label{s6}

In this section we show that the same perturbative scheme that one 
applies to the reduced Hamiltonian can be employed in order to compute 
the irreducible mass pertubatively whose dynamics we consider as a measure 
for mass loss due to Hawking radiation as we argued in section \ref{s0}.
In the first subsection we recall some of the notions associated with 
horizons and the irreducible mass, see \cite{14a} and references therein 
for the rather extensive literature on the subject. 
In the second we compute the 
irreducible mass perturbatively. This is based partly on 
\cite{PertApparentHorizon} but here we do this directly in the 
Hamiltonian framework.

\subsection{Horizons, expansions and irreducible mass}
\label{s6.1}

We begin with some elementary definitions.\\
\\
Consider a globally hyperbolic spacetime $(M,g)$ and a Cauchy surface 
$\Sigma$ in it. Let $n$ be the future oriented timelike unit normal of 
$\Sigma$. Let $S\subset \Sigma$ be a closed, oriented 
2-surface in $\Sigma$ without
boundary $\partial S=\emptyset$ and $s$ be the spacelike unit normal 
of $S$ pointing outward from $S$ and tangential to $\Sigma$. Hence 
$g(n,s)=0$ at $S$. We note that if we are just given $\Sigma, S$, 
then $n,s$ are known only at $\Sigma,S$ respectively and thus 
the covariant derivatives of $n,s$ with respect to $\nabla$, the 
torsion-free covariant differential compatible with $g$, are 
only computable in directions tangential to $\Sigma,S$ respectively.
The tensor $q=g+n\otimes n$ on $\Sigma$ 
has the property $q(n,u)=0$ for every 
tangent vector $u$ of $\Sigma$ and the tensor 
$h=q-s\otimes s$ on $S$ 
has the property $h(n,v)=h(s,v)=0$ for every tangent vector $v$ of $S$.   
 
On $S$ we can define the future (from $\Sigma$) and outward respectively 
inward (from $S$) oriented null vectors $l_+=n+s,\; l_-=n-s$. We can now 
construct the affinely parametrised 
null geodesics starting from $S$ with initial tangent 
$l_\pm$ respectively. This defines two null geodesic congruences 
$C^\pm_S$ and thus 3-manifolds in $M$. Such a geodesic carries 
parameters $s, y^A,\; A=1,2$ where $s$ is the affine parameter and $y$ are 
coordinates on $S$. Thus $C^\pm_S$ is an embedded 3 manifold with 
local coordinates $(s,y)\mapsto c^\pm_y(s)$ where $c^\pm_y(s)$ is the geodesic 
with initial data $c^\pm_y(0)=Y(y),\; 
\dot{c}^\pm_y(0)=l_\pm(Y(y))$ and 
$Y:U\subset \mathbb{R}^2 \to S$ is an embedding of $S$.  
The tangential vectors to $C^\pm_S$ are 
$\partial^\pm_A=\frac{\partial c^\mu_y(s)}{\partial y^A}\; \partial_\mu
=:e_{A\pm}^\mu \partial_\mu$ and 
$\partial^\pm_s=\frac{\partial c^\mu_y(s)}{\partial y^A}\; \partial_\mu
=:l_\pm$. Note that by definition $\nabla_{l_\pm} l_{\pm}=0$ by 
definition of an affinely parametrised geodesic, i.e. $l_\pm$ 
at $C^\pm_S$ is just the parallel transport of the initial 
$l_\pm$ at $S$.
 
Since these vector fields are known on all
of $C^\pm_S$ we can take covariant derivatives of those in directions 
of $C^\pm_S$. Thus we have access to $\nabla_u v$ on $C^\pm_S$ 
where $u,v$ are in the span of $e_{A\pm},l_\pm$. We note that 
$\nabla_{l_\pm}\;g(l_\pm,l_\pm)=2g(l_\pm,\nabla_{l_\pm} l_\pm)=0$ 
i.e. the quantity $g(l_\pm,l_\pm)$ is constant along every geodesic 
and since it is zero initially it follows $g(l_\pm,l_\pm)=0$ on all
of $C^\pm_S$. Then
\be \label{6.1}
\nabla_{l_\pm}\; g(e_{A,\pm},l_\pm)=
g(\nabla_{l_\pm} e_{A\pm},l_\pm) 
=g(\nabla_{e_{A\pm}} l_\pm,l_\pm) 
=\frac{1}{2} \nabla_{e_{A\pm}}\;g(l_\pm,l_\pm)=0
\ee 
as $[\partial^\pm_A,\partial^\pm_s]=0$. Thus also 
$g(e_{A,\pm},l_{\pm})$ is constant along every geodesic and since 
it vanishes initially ($e_{A,\pm}$ is tangential to $S$ initially and 
$l_\pm$ is normal to $S$ initially) whe have $g(e_{A\pm},l_\pm)=0$ 
everywhere on $C^\pm_S$.

The vectors $e_{A\pm}$ are geodesic deviation vectors i.e. they carry
information about the deviation of nearby geodesics as we move 
infinitesimally in direction of $y^A$ within the congruence. The vector 
$\nabla_{l_\pm} e_{A\pm}=\nabla_{e_{A\pm}} l_\pm$ thus contains the 
information of how the deviation vectors expand, shear and rotate along
the ``fluid'' with fluid tangent $l_\pm$.
It has no components in direction of $l_\pm$ by (\ref{6.1}) 
hence the full information about the geodesic deviation is contained in the 
quantity
\be \label{6.2}
\kappa^\pm_{AB}:=g(e_{A\pm},\nabla_{e_{B\pm}} l_\pm)
\ee
We also define on $C^\pm_S$ the objects
\be \label{6.3}
h^\pm_{AB}:=g(e_{A\pm},e_{B,\pm}),\; h_\pm^{AC}\; h^\pm_{CB}=\delta^A_B
\ee
Then 
\be \label{6.4a}
\theta_\pm:=h_\pm^{AB} \;\kappa^\pm_{AB},\;
\sigma^\pm_{AB}:=\kappa^\pm_{(AB)},\;
\omega^\pm_{AB}:=\kappa^\pm_{[AB]},\;
\ee
are respectively called expansion, shear and rotation of the congruence
(rotation vanishes because $l_\pm$ is also normal to $C^\pm_S$ and 
$l_\pm$ is explicitly hypersurface orthogonal with $C^\pm_S$ as integral 
manifold so $\nabla_{[\mu} l_{\pm\nu]}=\alpha_{[\mu}\; l_{\pm\nu]}$ by Frobenius 
theorem). 
Using the defintion of the Riemann tensor one can compute 
$\nabla_{l_\pm} \theta_\pm$ which leads to Raychaudhuri's eqaution
\cite{12a}.\\ 
\\
We note that $\theta^\pm$ evaluated at $S$ only requires information 
available at $S$, i.e. we need not know anything about the actual 
geodesic congruence $C^\pm_S$ away from $S$. Nevertheless the above 
theory is useful as it equips us with a geometric interpretation 
of $\theta_\pm$ familiar from hydrodynamics: If $\theta^\pm>0/<0$ then
a volume element that flows with the fluid (here: a light ray) starting from
$S$ expands/contracts along the flow lines. 
In flat space always $\theta^+>0,\;\theta^-<0$ for a sphere $S$ (light leaves 
the sphere outwards/inwards). In a general spacetime one can have 
both $\theta_-,\; \theta_+\le 0$.  
\begin{Definition} \label{def6.1} ~\\
Consider a globally hyperbolic spacetime $(M,g)$ and a Cauchy surface 
$\Sigma$ in it.\\
i.\\
A closed, orientable 2-surface in $\Sigma$
$S\subset \Sigma$ without boundary $\partial_\Sigma S=\emptyset$ 
is called trapped if $\theta_+=0$.\\
ii.\\
A trapped region in $\Sigma$ is a closed subset $T\subset \Sigma$ such that 
$S:=\partial_\Sigma T$ is trapped.\\    
iii.\\
The trapped surface in $\Sigma$ defined by the total trapped 
region (closure of union of all trapped regions) 
is called the apparent horizon $A_\Sigma$ of $\Sigma$.
\end{Definition}
\begin{Definition} \label{def6.2} ~\\
Consider a globally hyperbolic spacetime $(M,g)$ and a foliation 
${\cal F}=\cup_{\tau\in \mathbb{R}}\; \Sigma_\tau$ of $M$ by Cauchy surfaces 
$\Sigma_\tau$.\\
i.\\
If $\tau\mapsto S_\tau\subset \Sigma_\tau$ is a one parameter family 
of trapped surfaces then ${\cal S}:=\cup_\tau S_\tau$ is 
called a trapping horizon.\\
ii.\\
Let $A_\tau:=A_{\Sigma_\tau}$ be the apparent horizon of 
$\Sigma_\tau$. Then ${\cal A}_{{\cal F}}:=\cup_\tau \; A_\tau$ 
is called the apparent horizon of $\cal F$. 
\end{Definition}
If $M$ is asymptotically flat then by definition 
it has a conformal completion $\hat{M}$ which in particular is equipped 
with future null infinity ${\cal S}_+$. The manifold $M$ is embedded into 
$\hat{M}$ via some $\;\varphi: M\to \hat{M}$ and 
$B:=\varphi^{-1}(\varphi(M)\cap[\hat{M}-J_-({\cal S}_+)])$ is called the 
black hole region. Its boundary ${\cal H}:=\partial_M B$ is called the event 
horizon of $M$. If $\Sigma$ is a Cauchy surface then 
$H_\Sigma:={\cal H}\cap \Sigma$  
is called the event horizon of $\Sigma$. Given a foliation
$\cal F$ with corresponding Cauchy leaves $\Sigma_\tau$, the classical 
black hole 
area theorem states that for all $\tau_1\le \tau_2$ we have 
${\sf Ar}[H_{\Sigma_{\tau_1}}]  
\le {\sf Ar}[H_{\Sigma_{\tau_2}}]$ when the Einstein equations and
suitable matter conditions (inequalities for the energy momentum tensor)
hold. One also shows that every trapped region lies in $B$ and therefore 
every trapped surface in any $\Sigma$, in particular the apparent
horizon in $\Sigma$, lies in $B$. This means that with respect to a foliation,
while the area of the event horizon can only grow within classical GR, 
the area of the apparent horizon can both shrink (e.g. radiation emission)
and grow (e.g. radiation absorption).  
\begin{Definition} \label{def6.3} ~\\
A (part of a) 
trapping horizon ${\cal S}$ is called dynamical horizon, trapped tube or 
isolated horizon respectively if $\cal S$ is a spacelike, timelike and 
null 3-manifold respectively.
\end{Definition}
The advantage of the various notions of trapping and apparent horizons
over the event horizon is that they are by construction local in nature
both spatially and temporally while the event horizon is a so called
``teleological'' construct requiring us to know the entire future development 
of a spacetime which is neither practical nor operational: after all 
an astronomer detects a black hole as the limited region of space from 
which no light can escape
and within her limited life time. Furthermore an astronomer will
measure the time development of that region with respect to a certain 
notion of time, i.e. a certain foliation. This makes the apparent horizon 
the ideal and physically motivated notion of a black hole. It is often 
objected that an apparent horizon is foliation dependent while 
the event horizon is an absolute notion, however, 
when viewed as a necessary part of the definition of an observer the foliation
dependence is actually physically well motivated. 
\begin{Definition} \label{def6.4} ~\\
Given a foliation $\cal F$ of a globally hyperbolic $(M,g)$ by Cauchy 
surfaces $\Sigma_\tau$ the irreducible mass at time $\tau$ is defined as 
\be \label{6.4}  
[M_{{\sf irr}}(\tau)]^2:={\sf Ar}[A_\tau]
\ee
i.e. the square root of the apparent horizon.
\end{Definition}
We slightly abuse here the terminology as the irreducible mass squared is
usually defined as the area of the event horizon rather than the apparent
horizon.\\
\\
We now have to provide a concrete formula for $\theta_+(\tau)$ and 
$M_{{\sf irr}}(\tau)$. Given a foliation $\cal F$ we introduce ADM 
coordinates $\tau,x^a;\; a=1,2,3$. We restrict attention to trapped 
surfaces of spherical topology 
and thus have embeddings $Y_\tau:\; S^2\to \sigma_\tau\subset\sigma,\; 
y\mapsto Y_\tau(y)$ 
and $E_\tau:\; S^2\to S_\tau\subset M=\mathbb{R}\times \sigma;\;
y\mapsto (\tau, Y_\tau(y))$. The future oriented timelike unit normal
to $\Sigma_\tau$ has components $n^\tau=\frac{1}{N},\; n^a=-\frac{N^a}{N}$
where the usual metric components are $g_{\tau\tau}=-N^2+q_{ab} N^a N^b,\;
g_{\tau a}=q_{ab} N^b, \; g_{ab}=q_{ab}$ with $a,b,c=1,2,3$. The vectors 
$T_A:=\frac{\partial E}{\partial y^A}$ are tangential to $S_\tau$, therefore
(in this section $\mu,\nu,..=0,1,2,3$ while $a,b,c,..=1,2,3$)
\be \label{6.5}
\tilde{s}_\mu:=-\frac{1}{2}
\epsilon_{\mu\nu\rho\sigma} n^\nu \; T_A ^\nu\; T^\nu_B\;
\epsilon^{AB}
\ee
is co-normal to and outgoing fom $S_\tau$ and normal to $n$. Therefore 
up to normalisation $\tilde{s}^\mu=g^{\mu\nu} \tilde{s}_\nu$ is the
spacelike unit normal to $S_\tau$. Explicitly with $\epsilon_{\tau abc}
=\epsilon_{abc}$ and $Y(\tau,y):=Y_\tau(y)$ and not displaying the 
$\tau$ dependence 
\be \label{6.6}
\tilde{s}_a=\frac{1}{N}\;\hat{s}_a,
\tilde{s}_\tau=\frac{1}{N} \hat{s}_\tau,
\hat{s}_a:= 
\frac{1}{2} \epsilon_{abc}\;\epsilon^{AB} Y^b_{,A} Y^c_{,B},\;
\hat{s}_\tau:=N^a \hat{s}_a
\ee
Thus 
\be \label{6.7}
\hat{s}^\tau=\frac{1}{N^2}[-\hat{s}_\tau+N^a \hat{s}_a]=0,\;
\hat{s}^a=\frac{1}{N^2}[N^a \hat{s}_\tau -N^a N^b \hat{s}_b]+q^{ab} \hat{s}_b
=q^{ab} \hat{s}_b
\ee
It follows $g(\hat{s},\hat{s})=q_{ab} \hat{s}^a \hat{s}^b=
q^{ab}\hat{s}_a \hat{s}_b$ thus 
\be \label{6.8}
s_a=\frac{\hat{s}_a}{\sqrt{q^{cd} \hat{s}_c \hat{s}_d}},\; 
s_\tau=N^a \hat{s}_a
\ee
is the properly normalised ouward oriented outgoing spacelike unit co-normal 
from $S_\tau$. It follows for the corresponding outgoing future oriented 
null normal $l:=l_+$
\be \label{6.9}
l^\tau=n^\tau=\frac{1}{N},\;
l^a=n^a+s^a=-\frac{N^a}{N}+q^{ab}\; s_b
\ee
Next, the pull-back metric on $S^2$ is given by (recall 
$h=g+n\otimes n-s\otimes s$)  
\be \label{6.10}
h_{AB}(\tau;y):=
[Y_\tau^\ast g]_{AB}(y)=
Y^\mu_{\tau,A}\; Y^\nu_{\tau,B} \;
g_{\mu\nu}(Y_\tau(y))
\ee
and thus 
\be \label{6.11}
{\sf Ar}[A_\tau]=\int_{S^2}\;d^2y\; \sqrt{\det(h(\tau;y))}
\ee
Note the identity 
\be \label{6.12}
h^{\mu\nu}=g^{\mu\nu}+n^\mu n^\nu-s^\mu s^\nu=h^{AB} Y^\mu_{,A} Y^\nu_{,B}
\ee
as one can quickly show by contacting with the co-basis 
$n_\mu, s_\mu, g_{\mu\nu} Y^\nu_{,A}$. Thus by definition
\ba \label{6.13}
\theta_\pm &=& h^{\mu\nu} \; \nabla_\mu l_{\pm \nu}
=
h^{AB}\; Y^\mu_{,A}\; Y^\nu_{,B} [\nabla_\mu l_{\pm \nu}]  
\nonumber\\
&=& 
h^{AB}\; Y^\mu_{,A}\; Y^\nu_{,B} [(\nabla_\mu n_\nu) \pm (\nabla_\mu s_\nu)]
\nonumber\\
&=& 
h^{AB}\;  [Y^\mu_{,A}\; Y^\nu_{,B} K_{\mu\nu} \pm 
Y^\mu_{,A}\; Y^\nu_{,B}\; (D_\mu s_\nu)]
\ea
where we have used that the spatial-spatial projection of $\nabla_\mu n_\nu$ 
is the extrinsic curvature $K_{\mu\nu}$ of $\Sigma_\tau$ and 
$D$ is the torsion free covariant derivative compatible with $q=g+n\otimes n$
acting on spatial tensors (i.e. whose contraction with $n$ of any index 
vanishes). Now 
$n_\tau=g_{\tau \mu} n^\mu=-N,\; n_a=g_{a\mu} n^\mu=0$ and thus 
$n_\mu Y^\mu_{,a}=-N Y^\tau_{,A}=0$ compatible with $Y^\tau=\tau$. 
It follows 
\be \label{6.13a}
\theta_\pm=
h^{AB}\;  [Y^a_{,A}\; Y^b_{,B} K_{ab} \pm 
Y^a_{,A}\; Y^b_{,B}\; (\nabla s)_{ab}]
\ee
By working out these expressions explicitly in ADM coordinates and using 
(\ref{6.7}), (\ref{6.8}) one 
finds
\be \label{6.14}
K_{ab}=\frac{1}{2N}\;[\partial_\tau q_{ab}-[{\cal L}_{\vec{N}} q]_{ab}],\;
[\nabla s]_{ab}=[D s]_{ab}=\partial_a s_b-\Gamma^c_{ab}(q) s_c
\ee
where $\Gamma(q)$ is constructed from $q_{ab},\; q^{ab},\; q^{ac} q_{cb}=
\delta^a_b$. Thus
\be \label{6.15}
\theta_\pm=
h^{AB}\;  [Y^a_{,A}\; Y^b_{,B} K_{ab} \mp 
Y^a_{,A}\; s_b \; D_a\; Y^b_{,B}]
\ee
where we exploited $Y^a_{,A} s_a=0$. Next one verifies 
that 
\be \label{6.16}
h^{AB} Y^a_{,A} Y^b_{,B} 
=q^{ab}-s^a s^b
\ee
by checking with the co-basis $s_a, q_{ab} Y^b_{,A}$. Therefore 
\be \label{6.17}
h^{AB} Y^a_{,A} Y^b_{,B} K_{ab} 
=[q^{ab}-s^a s^b]\; K_{ab}=-s_c\; s_d\; [q^{ac} q^{bd}-q^{ab} q^{cd}]\; K_{ab}
=-s_a s_b \frac{p^{ab}}{\sqrt{\det(q)}}
\ee
where we used normalisation $q^{ab} s_a s_b=1$ and the definition of 
the ADm momentum $p^{ab}$ conjugate to $q_{ab}$. Accordingly the final 
formula reads
\be \label{6.18}
\theta_\pm=-s_a s_b \frac{P^{ab}}{\sqrt{\det(q)}}
\mp Y^a_{,A}\; s_b \; D_a\; Y^b_{,B}
\ee
which expresses the expansion explicitly in terms of ADM data $(q,p)$ and
the embedding function $Y$ with $s_a=s_a(Y,q)$ via (\ref{6.7}), (\ref{6.8})
considered as also defined by these.

The time derivative of ${\sf Ar}[S_\tau]$ is given by 
\be \label{6.19}
\frac{d}{d\tau}\;{\sf Ar}[A_\tau]=\frac{1}{2}
\int_{S^2}\;d^2y\; \sqrt{\det(h(\tau;y))}\; h^{AB}\; 
[\frac{d}{d\tau} h_{AB}]
\ee
and with the vector fields 
$\xi(Y_\tau(y));=\frac{\partial Y(\tau,y)}{\partial \tau},\;
T_A(Y_\tau(y));=\frac{\partial Y(\tau,y)}{\partial y^A}$
tangential to the apparent horizon ${\cal A}_{{\cal F}}$
\be \label{6.20}
[\frac{d}{d\tau} h_{AB}]=
\xi^\mu\;\partial_\mu g(T_A,T_B)\circ Y
=2\xi^\mu\;g(T_{(A},\nabla_\mu T_{B)})
=2\;g(Y_{(A},\nabla_\xi T_{B)})
=2\;g(T_{(A},\nabla_{T_{B)}} \xi)
\ee
whence 
\be \label{6.21}
\frac{d}{d\tau}\;{\sf Ar}[A_\tau]=
\int_{S^2}\;d^2y\; \sqrt{\det(h(\tau;y))}\; \theta_\xi,\;
\theta_\xi=
h^{AB}\; 
Y^\mu_{,A} Y^\nu_{,B} [\nabla_\mu \xi_\nu]
\ee
This would vanish for $Y$ the embedding of the apparent horizon if 
we had $\xi\propto l_+$ but generically it is not because $\xi$ is 
generically not even null. 

The extreme cases are that $\xi\propto n, s, l_+$ which means that 
the apparent horizon is a timelike, spacelike or null surface 
(trapped tube, dynamical horizon, isolated horizon). 
Now $\theta_\pm=\theta_n\pm \theta_s$ and 
if $\theta_-<0$ as one usually assumes, from $\theta_+=0$ we get 
$\theta_s=-\theta_n>0$. Then either 
$\theta_\xi\propto \theta_n<0$,
$\theta_\xi\propto \theta_s>0$ or
$\theta_\xi\propto \theta_{l_+}=0$.

\subsection{Constructing the apparent horizon in GPG}
\label{s6.2}

To construct the apparent horizon in GPG we proceed as in 
\cite{PertApparentHorizon} and assume that it has sperical topology. 
Then the embedding function takes the explicit form 
\be \label{6.22}
Y^\tau(\tau,y)=\tau,\; Y^3(\tau,y)=\rho(\tau,y), Y^A(\tau,y)=y^A
\ee
The function $\rho$ is called the {\it radial profile}. Then 
\be \label{6.23}
\hat{s}_a=\frac{1}{2}\epsilon_{abc} \epsilon^{BC} Y^b_{,B} Y^c_{,C}
=\left\{ \begin{array}{cc}
-\rho_{,A} & a=A=1,2\\
1 & a=3
\end{array}
\right.
\ee
Hence $\hat{s}_3>0$ correctly implements outward orientation. The normalised 
components are 
\be \label{6.24}
s_a=f\;\hat{s}_a,\;
f=\frac{1}{\sqrt{q^{ab}\hat{s}_a \hat{s}_b}}=
\frac{1}{\sqrt{1+q^{AB} \rho_{,A} \rho_{,B}}},\;
\ee
where we used the GPG $q_{33}=q^{33}=1,\; q_{3A}=q^{3A}=0,\;
q_{AB}=r^2\Omega_{AB}+X_{AB},\; q^{AC} q_{CB}=\delta^A_B$.  
Thus $s_a$ is determined entirely by $\rho$. The profile function $\rho$ 
must then solve the {\it trapping equation}
\be \label{6.25}
-\theta_+=\frac{P^{ab}}{\sqrt{\det(q)}} s_a s_b+h^{AB} s_a\; 
D_{Y_{,A}} Y^a_{,B}=0
\ee
where 
\ba \label{6.26}
h_{AB} &=& q_{ab} Y^a_{,A} Y^b_{,B}=\rho_{,A}\;\rho_{,B}+q_{AB},\;
\det(h) h^{AB}=\epsilon^{AC} \epsilon^{BD} h_{CD},\;
=\rho^A \rho^B+\det(q) q^{AB},\; \rho^A=\epsilon^{AB}\; \rho_{,B}
\nonumber\\
\det(h) &=&
\frac{1}{2} \epsilon^{AC} \epsilon^{BD} h_{AB} h_{CD}
=\det(q) \; f^{-2}
\ea
and $P^{33}=P^3, P^{3A}=\frac{1}{2} P^A,\; P^{AB}=\frac{1}{2} P^0 
\Omega^{AB} +Y^{AB}$ where $P^3, P^A, P^0$ are themselves functions 
of $X_{AB}, Y^{AB}$ and the physical matter fields upon solving 
the constraints. Then the trapping condition reads
\be \label{6.28} 
-\theta_+=\frac{p^{ab}}{\sqrt{\det(q)}} s_a s_b+h^{AB} s_a\; 
[Y^a_{,AB}+\Gamma^a_{bc}(q) Y^b_{,A} Y^c_{,B}] =0
\ee
where $Y^a_{,AB}=\delta^a_3 \rho_{,AB}$. It is a non-linear second 
order PDE for the profile function on the sphere.\\
\\
To solve it exactly for general $X,Y$ is too complicated. However, since 
also $p^{ab}$ is only known perturbatively, it is well motivated to 
compute the profile function also only perturbatively. Thus
we consider 
\be \label{6.29}
\rho=\sum_{n=0}\; \rho_n
\ee
where $\rho_n$ is a monomial in $X,Y$ and matter perturbations of order 
$n$. Then one inserts (\ref{6.29}) into (\ref{6.28}) and expands 
all $p,q,s,Y_{,A}$ in terms of the perturbations $X_{AB}, Y^{AB}$ etc,
extracts the terms of order $n$ and aims for a hierarchy of equations 
that one can iteratively solve in closed form.     

We begin 
with the zeroth order and consider all $X,Y$ dependence vanishing. 
Thus ($R=2M,\;\kappa=1/2$)
\ba \label{6.30}
q_{ab} 
&=& \delta_a^3\delta_b^3+r^2\;\Omega_{AB}\; \delta^A_a \delta^B_b
\nonumber\\
\frac{P^{ab}}{\omega} &=& 2\sqrt{R\;r}
\delta^a_3\delta^b_3+\frac{1}{2}\sqrt{R/r^3}\;
\Omega^{AB}\; \delta_A^a \delta_B^b
\nonumber\\
s_a &=& f\; \hat{s}_a,\; \hat{s}_a=\delta^3_a-\rho_{,A} \delta^A_a,\;
f=[1+r^{-2} \Omega^{AB} \rho_{,A} \rho_{,B}]^{-1/2}
\nonumber\\
h_{AB} &=& \rho_{,A} \rho_{,B}+r^2\Omega_{AB},\;
\det(h)=r^4\;\omega^2\; f^2,\;
\det(h)\; h^{AB}=\rho^A\rho^B+r^2\omega^2\Omega^{AB},\; 
\rho^A=\epsilon^{AB}\;\rho_{,A}
\ea
The Christoffel symbols in the GPG at $X=0$  
were already computed in section \ref{s4}
\be \label{6.31} 
\Gamma^3_{33}=\Gamma^3_{3A}=\Gamma^A_{33}=0,\;
\Gamma^A_{B 3}=\frac{1}{r}\delta^A_B,\;
\Gamma^3_{AB}=-r\Omega_{AB},\;
\Gamma^A_{B C}=\Gamma^A_{B C}(\Omega)
\ee
Then
\ba \label{6.32}
&& 
h^{AB}\; s_c\; [Y^c_{,AB}+\Gamma^c_{ab} Y^a_{,A} Y^b_{,B}]
=f\; h^{AB}\; 
[\rho_{,AB}+\Gamma^3_{ab} Y^a_{,A} Y^b_{,B}
-\rho_{,C}\;\Gamma^C_{ab} Y^a_{,A} Y^b_{,B}]
\nonumber\\
&=& f\; h^{AB}\; 
[\rho_{,AB}+\Gamma^3_{CD} Y^C_{,A} Y^D_{,B}
-\rho_{,C}\;(2\Gamma^C_{3 D} Y^3_{,(A} Y^D_{,B)}
+\Gamma^C_{DE} Y^D_{,A} Y^E_{,B})]
\nonumber\\
&=& f\; h^{AB}\; 
[\rho_{,AB}+\Gamma^3_{AB} 
-\rho_{,C}\;(2\Gamma^C_{3 (A}\; \rho_{,B)} +\Gamma^C_{AB})]
\nonumber\\
&& P^{ab} s_a s_b=f^2\; [(Rr)^{1/2}+\frac{1}{4}\; \Omega^{AB} 
\rho_{,A} \; \rho_{,B}
\ea
We make the Ansatz $\rho=\rho_0=$const. then $f=1, h_{AB}=q_{AB}=
r^2\omega_{AB},\; s_a=\delta_a^3$ and 
\be \label{6.33}
\theta_+
=[\frac{P^{33}}{\sqrt{\det(q)}}+q^{AB} \Gamma^3_{AB}]_{r=\rho_0}
=\frac{2}{r^2}[\sqrt{Rr}-r]_{r=\rho_0}
=0
\ee
has the unique solution
\be \label{6.34}
\rho_0=R
\ee
as expected.\\
\\
We will assume inductively that $\rho_0,\; \rho_1,\; ..,\rho_{n-1},\; n\ge 1$
have been already computed. We write the unperturbed trapping equation 
in the form (recall $\det(q)=\omega^2 r^4+\det(X)$ and all quantities 
are evaluated at $r=\rho$)
\ba \label{6.35}
0 &=& 
[P^{33}-2 P^{3A} \rho_{,A}+P^{AB} \rho_{,A} \rho_{,B}]\;
[1+q^{AB}\;\rho_{,A}\;\rho_{,B}]^{1/2}\;[\rho^2\omega]\;
[1+\frac{\det(X)}{\rho^4 \omega^2}]^{1/2}  
\nonumber\\
&& 
+[\rho^A\rho^B+\det(q) q^{AB}]\;
[\rho_{,AB}+\Gamma^3_{AB}-\rho_{,C}(\Gamma^C_{AB}+2\Gamma^C_{3(A} \rho_{,B)})]
\ea
To capture the full n-th order dependence of this expression, all quantities 
that depend on $\rho$ need to be Taylor expanded up to n-th order in 
$\Delta=\rho-\rho_0$ around $\rho=\rho_0=R$, for example 
$P^{33}(\rho)=P^{33}(\rho_0)+
P^{33\prime}(\rho_0)\;\Delta+ 
\frac{1}{2}P^{33\prime\prime}(\rho_0)\;\Delta^2+...$ and 
the 2nd order contribution in $\Delta^2$ is given by $\rho_1^2+2\rho_0\rho_2$
etc.

Let us denote the 
four factors in the first term of (\ref{6.35}) by 
$A, B, C, D$ from left to right and the  
two factors in the second term of (\ref{6.35}) by 
$E,F$ from left to right. We isolate all terms of order $n$ by expanding each 
factor to order $n$. In the resulting sum of terms, which is now 
a monomial of order $n$, we want to isolate all terms that contain 
$\rho_n$. These are contained in the following contribution 
\be \label{6.36}
A_n B_0 C_0 D_0+    
A_0 B_n C_0 D_0+    
A_0 B_0 C_n D_0+    
A_0 B_0 C_0 D_n
+E_n F_0+E_0 F_n
\ee
where $A_n,A_0$ is the n-th and 0-th order contribution respectively of 
$A$ etc. We now consider the indivial terms $A_n,..,F_n$ and isolate 
the terms that contain $\rho_n$:
\be \label{6.37}
A=P^{33}-2 P^{3A} \rho_{,A}+P^{AB} \rho_{,A} \rho_{,B}
\ee
Since $\rho_{,A}=\rho_{1,A}+..$ is already of first order and $P^{rA}$
has no zeroth order perturbation, the only term 
in $A_n$ that contains $\rho_n$ is $P^{33\prime}_0(\rho_0)\; \rho_n$ and 
$A_0=P^{33}_0(\rho_0)$ where the subscript $0$ of $P^{33}_0$
means that we first expand
$P^{33}$ in terms of the pertubations $X,Y$ and then take the zeroth order 
term of that. The resulting function ist still to be expanded in terms 
of $\rho-\rho_0$ and we note the corresponding derivatives by a prime.   
\be \label{6.38} 
B=[1+q^{AB}\;\rho_{,A}\;\rho_{,B}]^{1/2}\;
\ee
For the same reason this term has no $\rho_n$ contribution in $B_n$
and $B_0=1$.
\be \label{6.39}
C=\rho^2\omega
\ee
The $\rho_n$ contribution to $C_n$ is $2\rho_0\rho_n\omega$ and 
$C_0=\rho_0^2 \omega$.
\be \label{6.40}
D=[1+\frac{\det(X)}{\rho^4 \omega^2}]^{1/2}  
\ee
There is no $\rho_n$ contribution to $D_n$ because $\det(X)$ is a second 
order perturbation and thus $D_0=1$.
\be \label{6.41}
E=\rho^A\rho^B+\epsilon^{AC}\epsilon^{BD}[\rho^2 \Omega_{CD}+X_{CD}(\rho)] 
\ee
As $X_{CD}$ is already of first order, the $\rho_n$ contribution to $E_n$
is $2\rho_0 \rho_n \omega^2 \Omega^{AB}$ and $E_0=\rho_0^2 \omega^2 
\Omega^{AB}$.
\be \label{6.42}
F=\rho_{,AB}+\Gamma^3_{AB}-\rho_{,C}(\Gamma^C_{AB}+2\Gamma^C_{3(A} \rho_{,B)})
\ee
The $\rho_n$ contribution to $F_n$ is $
\rho_{n,AB}+\Gamma^{3\prime}_{0;AB}(\rho_0)\; \rho_n
-\rho_{n,C}\Gamma^C_{0;AB}(\rho_0)$ and
$F_0=\Gamma^3_{AB}(\rho_0)$. 

It follows that the n-th order perturbation equation can be written in the 
form 
\ba \label{6.43}
G_n 
&=&
[P^{33\prime}_0(\rho_0)\rho_0^2 \omega + P^{33}_0(\rho_0) 2\rho_0\omega
+2\omega^2 \Omega^{AB} \rho_0 \Gamma^3_{0;AB}(\rho_0)+
\omega^2\rho_0^2 \Gamma^{3\prime}_{0;AB}(\rho_0)]\; \rho_n
\nonumber\\
&&
+\omega^2 \rho_0^2 \Omega^{AB}(\rho_{n,AB}-\Gamma^C_{0;AB}(\rho_0) \rho_{,C})
\nonumber\\
&=&
[P^{33}_0(r) r^2 \omega 
+\omega^2 \Omega^{AB} r^2 \Gamma^3_{0;AB}(r)]'_{r=\rho_0}\; \rho_n
\nonumber\\
&&
+\omega^2 \rho_0^2 \Omega^{AB}(\rho_{n,AB}-\Gamma^C_{0;AB}(\rho_0) \rho_{,C})
\nonumber\\
&=&
\omega^2 \; \rho_0^2 [\Omega^{AB} \; D_A\; D_B - 1]\rho_n
=\omega^2 \; \rho_0^2 [\Delta_{S_2}-1]\; \rho_n
\ea
where we used (\ref{6.33}) and that $\Gamma^C_{0,AB}=\Gamma^C_{AB}(\Omega)$
is the Christoffel symbol of the sphere metric independent of $r=\rho_0$
to write the last relation in terms of the sphere Laplacian. The
term $G_n$ is the complete n-th oder contribution to (\ref{6.35}) except for 
the terms that contain $\rho_n$. It thus contains the $\rho_m,\; m\le m-1$
and their derivatives polynomially which already have been solved for.
It remains to expand 
\be \label{6.44}
G_n=\sum_{l\ge |m|}\; G_n^{l,m} \; L_{l,m},\;
\rho_n=\sum_{l\ge |m|}\; \rho_n^{l,m} \; L_{l,m},\;\;
\Rightarrow\;\;
\rho_n^{l,m}=-\frac{G_n^{l,m}}{l(l+1)+1}
\ee
This proves that the radial profile $\rho$ of the apparent horizon can 
be solved for to arbitary order in the perturbations in closed form.\\

\subsection{Expansion of the irreducible  mass squared}
\label{s6.3}

Having computed the radial profile $\rho$ of the apparanet horizon we 
can compute the irreducible mass squared perturbatively as follows. The
non-perturbative expression is 
\ba \label{6.45}  
{\sf Ar} &=&
\int\; d^2y\; \sqrt{\det(h)}
=\int\; d^2y\; \sqrt{\det(q)}\;\sqrt{1+q^{AB}\; \rho_{,A}\; \rho_{,B}} 
\nonumber\\
&=& \int\; d\Omega\; \rho^2\; \sqrt{1+\frac{\det(X)}{\rho^4\omega^2}}\;
\sqrt{1+q^{AB}\; \rho_{,A}\; \rho_{,B}} 
\ea
This expression can be systematically expanded to any order in the 
perturbations. To second order
\ba \label{6.46}  
{\sf Ar}
&=&\int\; d\Omega\; [\rho_0^2+2\rho_0 \rho_1+\rho_1^2
+2\rho_0\rho_2]\;[1+\frac{1}{2}(\frac{\det(X)(\rho_0)}{\rho_0^4\omega^2}\;
+q^{AB}(\rho_0)\; \rho_{1,A}\; \rho_{1,B})] 
\nonumber\\
&=&\int\; d\Omega\; \{[\rho_0^2]+[\rho_1^2
+2\rho_0\rho_2+\frac{\rho_0^2}{2}(\frac{\det(X)(\rho_0)}{\rho_0^4\omega^2}\;
+q^{AB}(\rho_0)\; \rho_{1,A}\; \rho_{1,B})]\} 
\nonumber\\
&=:& {\sf Ar}_0+{\sf Ar}_2
\ea
where the linear term has dropped out because it contains no $l=0$ mode.
Thus the mass itself to second order is 
\be \label{6.47}
M_{{\sf irr}}={\sf Ar}^{1/2}_0\;[1+\frac{1}{2}\frac{{\sf Ar}_2}{{\sf Ar}_0}]
\ee
which to zeroth order just reproduces $M$ while to second oder is 
a functional quadratic in $X,Y$. In a Fock representation of $X,Y$ we expect 
quantum fluctuations of the irreducible mass, non-trivial dynamics and even
violations of positivity inequalities \cite{Fewster}.

\section{Quantum Fields in a BHWHT spacetime}
\label{s7}

We consider free quantum fields on the spherically symmetric 
BHWHT spacetime $(M,g)$ with mass parameter
$M$ or Schwarzschild radius $R=2M$. The line element is given by 
\be \label{7.1}
ds^2=-(1-\frac{R}{|z|})\; d\tau^2+2\sqrt{\frac{R}{|z|}}\;d\tau \; dz
+dz^2+z^2\; d\Omega^2
\ee
and $M=\mathbb{R}^4\cup \mathbb{R}^4$ with coordinates 
$\tau\in\mathbb{R},\; \Omega=(y^1,y^2)\in \mathbb{S}^2,\;
z\in \mathbb{R}$ and radial coordinates $r=z,\; z>0;\; \bar{r}=-z,\;
z<0$. Hence (\ref{7.1}) is the Schwarzschild solution in ingoing/outgoing
Gullstrand-Painlev\'e coordinates for $z>0/<0$. The singularity is 
at $r=\bar{r}=0$. However, causal geodesics can be continued across 
it and this spacetime is foliated by $\tau=$const.
spacelike hypersurfaces which are Cauchy surfaces and define the 
simultaneity proper time surfaces of free falling timelike observers 
that fall all the way from past timelike infinity in a past universe 
towards future timelike infinity in a future universe. Those Cauchy 
surfaces extend all the way from the spatial infinity of the past 
universe to spatial infinity of the future universe. Together 
both universes therefore define a globally hyperbolic spacetime 
if one allows singularity crossing. 
That spacetime is the common domain of dependence of all those 
hypersurfaces.  Global hyperbolicity is very important 
for constructing quantum field theories and sticking to only one universe the
free falling synchronous hypersurfaces form a foliation but none of its 
leaves is a Cauchy surface as they stop at the singularity. If needed, 
we can consider two regular spacetimes glued at the cylinder surface 
$r=\bar{r}=l\ll R$ and with the solid cylinder cut out as a regularisation 
step for what follows. More details are given in appendix \ref{sd}.\\
\\
The spacetime metric (\ref{7.1}) naturally appears in our perturbative 
scheme to compute the reduced Hamiltonian and the black hole apparent 
horizon at second order and enters the Regge-Wheeler and Zerilli equations.
It therefore motivates a natural class of Fock representations and 
therefore plays a fundamental role also for higher order contributions 
to the reduced Hamiltonian which we will treat by standard methods of 
perturbative QFT. In the first subsection, we give a brief introduction
to QFT in general CST. Then we specialise to CST equipped with a Killing 
vector field which is not necessarily everywhere timelike but 
such that the constant Killing time hypersurfaces are everywhere spacelike
so that the time dependence of the wave equation obeyed by the quantum 
field can be separated off. After that we specialise even further to CST 
with spherical symmetry so that even the angular dependence can be separated 
off and the wave equation reduces from a PDE to an ODE of second order.
In this case one can gain important information on the modes of the 
quantum field using the Wronskian identities and without explicitly 
solving the wave equation. Finally we discuss some of the details of 
the wave equation for the concrete CST given by (\ref{7.1}) and outline 
the applications that we have in mind with regard to particle production 
and Hawking radiation.
        
\subsection{Elements of QFT in CST}
\label{s7.1}

Consider a bosonic Quantum Field Theory (QFT) in 
globally hyperbolic Curved Spacetime
(CST) $(M,g)$. The classical, real valued,
free spacetime fields are subject to a linear wave equation 
of the form 
\be \label{7.2}
\Box\Phi=U\Phi,\; \Box=g^{\mu\nu} \nabla_\mu\;\nabla_\nu 
\ee
where $U$ is a real valued potential function (e.g. a 
position dependent mass term, it does not depend on 
$\Phi$). This equation 
is either an Euler-Lagrange equation derived from some 
Lagrangian or from the correponding Hamiltonian formulation. 

Let $V$ be the 
vector space of real valued solutions of (\ref{7.2}) that vanish
sufficiently fast at spatial infinity. Given a Cauchy 
surface $\Sigma$ in $M$ consider the anti-symmetric 
bilinear form on $V$ defined by
\be \label{7.3} 
B(u,v):=
\int_{\Sigma} \; d\Sigma_\mu\;\{
u\; [\nabla^\mu v]-
[\nabla^\mu u] \; v\}
\ee
where $d\Sigma_\mu=\frac{1}{3!}|\det(g)|^{1/2}\epsilon_{\mu\nu\rho\sigma} 
dX^\nu\wedge  dX^\rho\wedge  dX^\sigma$ is the volume element defined by
$g$. It is not difficult to see that the 3-form defined by the integrand 
of (\ref{7.3}) is closed which is why $B$ is independent of the choice 
of $\Sigma$.  

Next we consider the complexification $V_{\mathbb{C}}$ of linear 
combinations $w=u+iv,\; u,v\in V$ and consider the sesqui-liear form 
on $V_{\mathbb{C}}$ defined by 
\be \label{7.4}
<w,w'>:=-i \; B(\overline{w},w')
=-i\;
\int_{\Sigma} \; d\Sigma_\mu\;\{
\overline{w}\; [\nabla^\mu w']-
\overline{[\nabla^\mu w]} \; w'\}
\ee
Decomposing $w,w'$ into real and imaginary parts one sees that the 
sesqui-linear form is still independent of $\Sigma$, however,
it is not positive semi-definite and does not 
equip all of $V_{\mathbb{C}}$ with a Hilbert space structure. 

Suppose that we find a subspace $V_+\subset V_{\mathbb{C}}$ such that 
$<.,.>$ restricted to $V_+$ is positive semi-definite. Then automatically
$<.,>$ restricted to $\overline{V_+}$ is negative semi-definite because 
\be \label{7.5}
<\overline{w},\overline{w}'>=-i\;B(w,\overline{w}')=i\;B(\overline{w}',w)
=-<w',w>
\ee
If moreover $V_+,V_-:=\overline{V}_+$ are orthogonal with respect to 
$<.,.>$ and $V_+\oplus V_-=V_{\mathbb{C}}$ then 
$(V_{\mathbb{C}},<.,.>)$ carries a {\it Krein structure} i.e. an orthogonal 
decomposition 
\be \label{7.5a}
V_{\mathbb{C}}=V_+ \oplus V_- 
%\oplus V_0
\ee
such that $<.,..>$ restricted to $V_+, V_-$ is respectively 
is postive semi-definite and negative semi-definite respectively
(we refrain from the usual separate treatment of zero norm vectors 
for convenience0. 
Then $(V_\pm,\; \pm\;<.,.>)$ is a pre Hilbert space whose completion 
(after moding by null vectors) is a 
Hilbert space ${\cal H}_\pm$. Equivalently, $(V_{\mathbb{C}},\;(.,.))$
is a pre Hilbert space with inner product
\be \label{7.6}
((w_+,w_-),(w'_+,w'_-))=<w_+,w'_+>-<w_-,w'_->  
\ee
with completion ${\cal H}={\cal H}_+ \oplus {\cal H}_-$ 
Then ${\cal H}_\pm$ are orthogonal subspaces of 
${\cal H}$ with corresponding self-adjoint projections $P_\pm$ i.e. 
$P_+ + P_-=1_{{\cal H}},\; P_\pm^2=P_\pm=P_\pm^\ast,\;P_+ P_-=0$ where 
the adjoint is with respect to $(.,.)$.

Consider the anti-self adjoint operator 
\be \label{7.7}
J:=-i(P_+ - P_-),\; J^\ast =-J,\; J^2=-1_{{\cal H}},\; 
P_\pm=\frac{1}{2}(1_{{\cal H}}\mp\; i\; J)
\ee
It preserves the real vector space $V$: By assumption we have 
$\overline{P_+ V_{\mathbb{C}}}=P_- V_{\mathbb{C}}$ whence 
$\overline{P_- V_{\mathbb{C}}}=P_+ V_{\mathbb{C}}$. Therefore  
for $w,w'\in V_{\mathbb{C}}$ we have using $<w,w'>
=-<\overline{w'},\overline{w}>$ 
\ba \label{7.8}
&& (w,\overline{P_+ w'})
=(P_- w,\overline{P_+ w'})
=-<P_- w,\overline{P_+ w'}>
=<P_+ w',\overline{P_- w}>
=(P_+ w',\overline{P_- w})
\nonumber\\
&=& (w', P_+ \overline{P_- w})
=<w', \overline{P_- w}>
=-<P_- w,\overline{w'}>
=(P_- w,\overline{w'})
=(w,P_-\overline{w'})
\ea
i.e. $\overline{P_\pm w}=P_\mp \overline{w}$. Thus for $u\in V$
\be \label{7.8a}
\overline{J u}=i(\overline{P_+ u}-\overline{P_- u})=-i(P_+-P_-)\;\bar{u}=
J\;u
\ee
Moreover $iJ P_\pm =\pm\;P_\pm$ i.e. ${\cal H}_\pm$ are eigenspaces of 
$iJ$ with 
eigenvalues $\pm 1$. Finally 
\be \label{7.9}
B(u, J\;v)=-i [B(u,P_+ v)-B(u,P_-u)]=
<u,P_+ v>-<u,P_- v> 
=(u,P_+ v)+(u,P_- v)=(u,v)
\ee
is positive semidefinite definite on $V$ and for $u,v\in V$   
\ba \label{7.9b}
&&B(J\;u, J\;v)
=(J\;u,P_+ v)+(J\;u,P_- v)
=-(u,J\; P_+ v)+(u,J\; P_- v)
\nonumber\\
&=& i(u,P_+ v)-i(u,P_- v)
=B(u,P_+ v)+B(u,P_- v)=B(u,v)
\ea
One calls $B$ a symplectic structure, $J: V\to V, J^2=-1_V$ a complex 
structure, a Kaehler structure if $B(J., J.)=B(.,.)$, a positive 
Kaehler structure if $B(.,J.)$ is positive semidefinite. One can 
reverse the argument and start from a complex structure on $V$ which is 
positive Kaehler with respect to $B$ and then arrives at a Krein structure 
on $V_\mathbb{C}$ such that $V_-=\overline{V_+}$ and 
$(w,w')=<w,(P_+-P_-)w'>=B(\overline{w},J w')$.

The classical, real valued field $\phi$ is an element of $V$ and thus 
\be \label{7.8b}
\phi=[P_+\phi]+[P_-\phi]=:A+A^\ast
\ee
is a decomposition into annihilation and creation parts. If $w\in V_+$
set 
\be \label{7.9c}
A(w):=<w,A>=-i B(\overline{w},\Phi)
\ee
Consider a foliation of $M$ with corresponding lapse and shift functions 
such that $\Sigma$ is one of its leaves i.e. a $\tau=$const. hyper surface. 
Then recalling that $g^{\tau\mu}=-\frac{1}{N^2}[\delta^\mu_\tau-N^a 
\delta^\mu_a]=-\frac{1}{N} n^\mu,\; \epsilon_{\tau abc}=-\epsilon_{abc}$
\ba \label{7.10}
A(w) &=&
%-i\int_\sigma\; d^3x\; N\; \sqrt{\det(q)}\;g^{\tau\mu}[
%\bar{w}\;\Phi_{,\mu}  
%-\bar{w}_{,\mu}\;\Phi]  
=i\int_\sigma\; d^3x\; N\; \sqrt{\det(q)}\;g^{\tau\mu}[
\bar{w}\;\Phi_{,\mu}  
-\bar{w}_{,\mu}\;\Phi]  
\\
&=& -i\int_\sigma\; d^3x\;\sqrt{\det(q)}\; 
[\bar{w}\;\;[\nabla_n\Phi]  
-\overline{\nabla_n w};\Phi]  
=-i\int_\sigma\; d^3x\; 
[\bar{w}\;\pi 
-\sqrt{\det(q)}\overline{\nabla_n w}\;\phi]  
\ea
where $\pi=\sqrt{\det(q)}\;[\nabla_n;\Phi]_{|\Sigma}$ is the momentum 
conjugate to $\phi=\Phi_{|\Sigma}$. It follows with $Q=[\det(q)]^{1/2}$
\ba \label{7.11}
&& \{A(w),[A(w')]^\ast\}=
\int\;d^3x\int\;d^3y\;\{
[\bar{w}\pi-Q [\overline{\nabla_n w}] \phi](x),
[w'\pi-Q [\nabla_n w'] \phi](y)\}
\nonumber\\
&=& -\int\;d^3x\;Q\;
[\bar{w}\; [\nabla_n w']-\overline{\nabla_n w}\;w']
\nonumber\\
&=& =-i <w,w'>
\ea
so that the canonical commutation relations (CCR) are 
\be \label{7.11b}
[A(w),[A(w')]^\ast]=<w,w'>\; 1
\ee
confirming the roles of $A,A^\ast$ as annihilation and creation operator 
valued distributions in potential Fock representations.

\subsection{CST with spacelike Killing time hypersurfaces}
\label{s7.2}

The question is of course, given $B$, how to obtain either the Krein or 
complex structure 
with the additional properties mentioned, and how much freedom there 
is in choosing them. In the case that is of interest here, namely that there 
is a Killing vector field $\partial_\tau$ such that the $\tau=$const. 
surfaces are spacelike Cauchy surfaces, the following construction may be 
applied. Note that this is more general than the stationary case 
in which the Killing vector field is supposed to be everywhere 
timelike. In fact this does not hold for our $\partial_\tau$ 
where the GPG time $\tau$ defines the free falling foliation 
with synchronous $\tau=$const. surfaces. We consider the Hamiltonian 
\be \label{7.12}
H=\int_\sigma\; d^3x\; [\frac{N}{2}(\frac{\pi^2}{Q}+Q\;[q^{ab} 
\phi_{,a}\phi_{,b}+U\phi^2])+\pi N^a\phi_{,a}]
\ee
where $N,N^a,q_{ab}$ are not explicitly $\tau$ dependent. The Hamiltonian 
equations of motion
\be \label{7.13}
\dot{\phi}=\{H,\phi\}=N\; \frac{\pi}{Q}+N^a \phi_{,a},\;\;
\dot{\pi}=\{H,\pi\}=(N\; Q\; q^{ab}\;\phi_{,a})_{,b}-N Q U \phi+(N^a \pi)_{,a}
\ee
reproduce the Euler Lagrange equations (\ref{7.2}) since 
from the first relation in (\ref{7.13}) $\pi=Q\nabla_n \phi$ thus
\ba \label{7.14}
U\;\phi=\Box\phi &=&
|\det(g)|^{-1/2}\;[g^{\mu\nu}\;|\det(g)|^{1/2} \phi_{,\mu}]_{,\nu}  
\nonumber\\
&=&
[N\; Q]^{-1}\;\{
[g^{\mu\tau}\;N\;Q \phi_{,\mu}]_{,\tau}  
+[g^{\mu a}\;N\;Q \phi_{,\mu}]_{,a}]\}
\nonumber\\
&=&
[N\; Q]^{-1}\;\{
-[Q \;\nabla_n \phi]_{,\tau}  
+[Q N^a\;\nabla_n \phi]_{,a}]+[N Q q^{ab} \phi_{,a}]_{,b}\}
\nonumber\\
&=&
[N\; Q]^{-1}\;\{
-\dot{\pi}
+[\pi N^a]_{,a}]+[N Q q^{ab} \phi_{,a}]_{,b}\}
\ea  
~\\
The idea is now to construct a system of complex solutions of (\ref{7.2})
whose normalisable span defines the space $V_+$ of the Krein structure, i.e.
$V_+,<.,.>$ is a pre Hilbert space such that $V_-=\overline{V_+}$ and 
$V_+\perp V_-$. To do this we write out the d'Alembertian explicitly
\be \label{7.14b}
\Box W=\frac{1}{N\;Q} \;
\{-[Q \;\nabla_n W]_{,\tau}  
+[Q N^a\;\nabla_n W]_{,a}]+[N Q q^{ab} W_{,a}]_{,b}\};\;\;
\nabla_n=N^{-1}(\partial_\tau-N^a\partial_a)
\ee
which is still the general expression. Now we exploit that $N,N^a,q_{ab}$
do not depend on $\tau$ and thus can separate off the $\tau$ dependence 
in $w$
\be \label{7.15}
W_\omega(\tau,x):=e^{i\omega\tau}\; w_\omega(x),\; \omega\in \mathbb{C}
\ee
For $\omega\in \mathbb{R}$ the coprresponding solutions are called 
{\it modes}, for $\Im(\omega)\in \mathbb{R}_+-\{0\}$ ring down or  
{\it quasi-normal modes}. It follows
\be \label{7.16}
0=\frac{1}{NQ}\;[\omega^2\;\frac{Q}{N}\; w_\omega
+i\omega\; (\frac{Q}{N} N^a\partial_a+\partial_a \frac{Q}{N} N^a)w_\omega 
+\partial_a(N Q (q^{ab}-\frac{N^a N^b}{N^2})\partial_b w_\omega
-N\;Q\;U\; w_\omega]
\ee
This suggests to introduce the operators 
\be \label{7.17}
E:=\frac{Q}{N}\;1,\;
A:=\partial_a\; N^a\; E+E\; N^a\;\partial_a,\;
B:=-\partial_a(N Q (q^{ab}-\frac{N^a N^b}{N^2})\partial_b 
+N\;Q\;U\;1
\ee
where $E,B$ and $A$ are formally symmetric and anti-symmetric 
respectively on the auxiliary Hilbert space
$\mathfrak{H}=L_2(\sigma,d^3x)$. Formally each of them maps scalars into 
scalar densities of weight one on $\sigma$ because $Q$ carries density
weight one and $N,N^a$ are scalars and vectors respectively 
of density weight zero.
Accordingly 
\be \label{7.18}
[\omega^2\; E+i\omega\; A-B]\; w_\omega=0
\ee
We focus on real $\omega$. Taking the complex conjugate of (\ref{7.18}) we 
see that $w_\omega^\ast$ solves the same equation as a soultion 
$w'_{-\omega}$ would. Next 
we compute the $\mathfrak{H}$ inner product of (\ref{7.18}) with 
$w'_{\omega'}$  
\ba \label{7.19}
&& <w_\omega,\;B\;w'_{\omega'}>_{\mathfrak{H}}=
<w_\omega,[(\omega')^2\; E+i\omega'\;A]\; w'_{\omega'}>_{\mathfrak{H}}=
\nonumber\\
&=& <B\;w_\omega,\;w'_{\omega'}>_{\mathfrak{H}}=
<[(\omega)^2\; E+i(\omega)\;A] w_{\omega},\;w'_{\omega'}>_{\mathfrak{H}}
\ea
where we used the symmetry of $B$. It follows from the symmetry of 
$E,i\;A$ that
\be \label{7.20} 
(\omega-\omega')\;
[(\omega+\omega')\;<w_\omega,\;E\; w'_{\omega'}>_{\mathfrak{H}}
+i\;<w_\omega,\;A\; w'_{\omega'}>_{\mathfrak{H}}]=0
\ee
Now we compute on the other hand the inner product with respect to 
the sesqui-linear form (\ref{7.4})
\ba \label{7.21}
<W_\omega,W'_{\omega'}>
&=& -i\int_\sigma\;\; d^3x\; Q\; 
[\overline{W_\omega}\; (\nabla_n W'_{\omega'})-\overline{\nabla_n W_\omega}\;
W'_{\omega'} ]
\nonumber\\
&=&
-i\;e^{i(\omega'-\omega)\tau}\; 
\int_\sigma\;\; d^3x\; E\; 
[\overline{w_\omega}\; ((i\omega'-N^a \partial_a) w'_{\omega'})
-\overline{(i\omega-N^a \partial_a) w_\omega)}\; w'_{\omega'}
\nonumber\\
&=&
e^{i(\omega'-\omega)\tau}\; 
\int_\sigma\;\; d^3x\; \overline{w_\omega}\; 
[(\omega+\omega') E+i(E\; N^a\; \partial_a+\partial_a\; N^a\; E)]\;w'_{\omega'}
\nonumber\\
&=&
e^{i(\omega'-\omega)\tau}\; 
[(\omega+\omega')\;<w_\omega,\;E\; w'_{\omega'}>_{\mathfrak{H}}
+i\;<w_\omega,\;A\; w'_{\omega'}>_{\mathfrak{H}}]=0
\ea
which coincides up to the $\tau$ dependent phase with the square bracket 
in (\ref{7.20}). Thus we get from (\ref{7.20}) and (\ref{7.21}) that 
\be \label{7.22}
<W_{\omega},W'_{\omega'}>\propto\delta^{(1)}(\omega,\omega')
\ee
i.e. the solutions are orthonormal with respect to the 
inner product (\ref{7.4}) and the label $\omega$ in the sense of 
delta distributions. If the sign of the proportionality factor in 
(\ref{7.22}) is positive or negative respectively 
then the solution lies in $V_+$ or $V_-$ respectively. When $V_+$ 
coincides with the subspace of solutions corresponding to 
$\omega\in\mathbb{R}_+$ one calls $V_+$ the positive frequency subspace 
but in general $V_+$ can also contain negative frequency (e.g. for 
the potential barrier underlying the Klein paradox \cite{Klein}) 
and thus should better be called the 
positive inner product subspace.
In particular, if $W_\omega \in V_+$ then $W_\omega^\ast=W'_{-\omega},\;
w'_{-\omega}=w_\omega^\ast$, hence $V_+ \perp \overline{V_+}$ and since 
by the properties of the scalar product $<.,.>$ we have 
$<\bar{W},\bar{W}'>=-<W',W>$ it follows that $V_-=\overline{V_+}$.

\subsection{Further reduction of PDE to ODE in sperically symmetric CST}
\label{s7.3}

Accordingly 
we have found a possible Krein structure in the present situation 
once we know the space of solutions of (\ref{7.18}). In general
this is a complicated second order PDE. However, in the presence 
of further symmetries as is the case for our spherically symmetric 
background or an axisymmetric background (rotating black hole), 
a further separation of variables is possible which then 
turns (\ref{7.18}) into a second order ODE in just the radial 
variable $z$.
Namely in the case of spherical symmetry we have the general 
structure 
$q_{ab}=\gamma^2(z) \delta^3_a \delta^3_b+\rho(z)^2\Omega_{AB} 
\delta^A_a\delta^B_b,\; Q=\gamma\;\rho^2\sqrt{\det(\Omega)},\;
N=N(z), N^a=N^3(z)\;\delta^a_3$ 
\ba \label{7.23}
&& E=\sqrt{\det(\Omega)}\;e,\; e=\frac{\gamma\rho^2}{N},\
A=\sqrt{\det(\Omega)}\; a,\; 
a=(e\; N^3\; \partial_z+\partial_z\; N^3\; e),\;
\nonumber\\
&& B = \sqrt{\det(\Omega)}\;
[b+c\; \Delta+N\gamma\rho^2\; U],\;
b=-\partial_z\; N\;\gamma\rho^2[\gamma^{-2}-[\frac{N^3}{N}]^2]\;\partial_z,\;
c=-N\gamma
\ea
where $\Delta$ is the Laplacian on the sphere.
Then the separation Ansatz $w_\omega(z,y)=w_{\omega,l,m}(z)\; L_{l,m}(y)$ 
yields the ODE
\be \label{7.24}
[\omega^2\; e+i\omega\; a-b_l]\; w_{\omega,l,m}=0,\;
b_l=b-l(l+1) \; c+N\gamma\rho^2 U
\ee
and we have with $\mathfrak{h}=L_2(\mathbb{R},\;dz)$
\be \label{7.25}
<W_{\omega,l,m},\;W'_{\omega',l',m'}>
=e^{i(\omega'-\omega)\tau}\; \delta_{l,l'}\;\delta_{m,m'}\;
[(\omega+\omega')\;
<w_{\omega,l,m},\;e\;w'_{\omega',l,m}>_{\mathfrak{h}}
+i\;<w_{\omega,l,m},\;a\;w'_{\omega',l,m}>_{\mathfrak{h}}]
\ee
and 
\be \label{7.26}
(\omega-\omega')\;  
[(\omega+\omega')\;
<w_{\omega,l,m},\;e\;w'_{\omega',l,m}>_{\mathfrak{h}}
+i\;<w_{\omega,l,m},\;a\;w'_{\omega',l,m}>_{\mathfrak{h}}]
=0
\ee
One can find out more about the normalisation of those solutions without 
actually computing them explicitly. These follow from the {\it Wronskian
identities}. Let us abbreviate 
\ba \label{7.27}
&& g(z)=e\; N^3,\; f(z)=
N\;\gamma\rho^2[\gamma^{-2}-[\frac{N^3}{N}]^2],\;
u_l=-l(l+1) \; c+N\gamma\rho^2 U,\;
\nonumber\\
&&
[\omega^2\; e+i\omega\; (g\;\partial_z+\partial_z\; g)+
\partial_z \; f\; \partial_z-u_l]\;w_{\omega,l,m}=0 
\ea
In the previous section and in (\ref{7.25}) and (\ref{7.26}) 
we have implicitly assumed that $ia,\;a=
g \partial_z+\partial_z g,\;
b=\partial_z f\partial_z$ are symmetric operators on $\mathfrak{h}=
L_2(\mathbb{R},dz)$ i.e. that no boundary terms are picked up when 
integrating by parts. In what follows we will drop this assumption, 
i.e. no assumptions about boundary terms will be needed. The 
inner product between two solutions is defined by
(we drop the labels $l,m$ as the solutions are rigorously 
orthogonal for $l\not=\tilde{l}, m\not=\tilde{m}$) 
\be \label{7.27b}
<W_\omega,\tilde{W}_{\tilde{\omega}}>
=e^{i(\tilde{\omega}-\omega)\tau}\int\; dz\;\{
[\overline{w_\omega}\; 
[(e\;\tilde{\omega}+i\; g\partial_z)\;\tilde{w}_{\tilde{\omega}}]
+\overline{[(e\;\omega+i\; g\partial_z)\;w_\omega]}
\;\tilde{w}_{\tilde{\omega}}
\}
\ee
By construction, (\ref{7.27}) is time independent, therefore taking 
the time derivative of (\ref{7.27}) we find  
\be \label{7.27a}
0=(\tilde{\omega}-\omega)\;\int\; dz\;\{
[\overline{w_\omega}\; 
[(e\;\tilde{\omega}+i\; g\partial_z)\;\tilde{w}_{\tilde{\omega}}]
+\overline{[(e\;\omega+i\; g\partial_z)\;w_\omega]}
\;\tilde{w}_{\tilde{\omega}}
\}
\ee
which means that the integral in (\ref{7.27}) must be proportional 
to $\delta^{(1)}(\tilde{\omega},\omega)$.
On the other hand we have the {\it Green identity}, using that 
$e,f,g$ are real valued
\ba \label{7.28}  
&&\overline{w_\omega}\;[b\; \tilde{w}_{\tilde{\omega}}]
-\overline{b\; w_\omega}\;\tilde{w}_{\tilde{\omega}}
=\{
\overline{w_\omega}\;f\; \tilde{w}'_{\tilde{\omega}}
-\overline{w_\omega}'\;f\; \tilde{w}_{\tilde{\omega}}
\}'
\nonumber\\
&=&
-\overline{w_\omega}\;
[(\tilde{\omega}^2\; e\;+i\tilde{\omega}\;a-u_l)\; \tilde{w}_{\tilde{\omega}}]
+
\overline{[(\omega^2\; e\;+i\omega\;a-u_l)\; w_\omega]}
\;\tilde{w}_{\tilde{\omega}}]
\nonumber\\
&=&
(\omega-\tilde{\omega})\;(\omega+\tilde{\omega}) 
\;\overline{w_\omega}\;e\;\tilde{w}_{\tilde{\omega}}
-i\;(
\omega\; [a\;\overline{w_\omega}]\;\tilde{w}_{\tilde{\omega}}
+\tilde{\omega}\;\overline{w_\omega}]\;[a\;\tilde{w}_{\tilde{\omega}}])
\nonumber\\
&=&
(\omega-\tilde{\omega})\;[I-i\;g\;(
\overline{w_\omega}\;\tilde{w}'_{\tilde{\omega}}
-\overline{w_\omega}'\;\tilde{w}_{\tilde{\omega}})
-i\;(
\omega\; [a\;\overline{w_\omega}]\;\tilde{w}_{\tilde{\omega}}
+\tilde{\omega}\;\overline{w_\omega}]\;[a\;\tilde{w}_{\tilde{\omega}}])
\nonumber\\
&=&
(\omega-\tilde{\omega})\;I
-i\{
\omega\;[(2g\;\overline{w_\omega}'
+g'\;\overline{w_\omega})\;\tilde{w}_{\tilde{\omega}}
+g(\overline{w_\omega}\;\tilde{w}'_{\tilde{\omega}}
-\overline{w_\omega}'\;\tilde{w}_{\tilde{\omega}})]
+\tilde{\omega}\;[\overline{w_\omega}\;
(2g\;\tilde{w}'_{\tilde{\omega}}+g'\;\tilde{w}_{\tilde{\omega}})
-g(\overline{w_\omega}\;\tilde{w}'_{\tilde{\omega}}
-\overline{w_\omega}'\;\tilde{w}_{\tilde{\omega}})]
\}
\nonumber\\
&=&
(\omega-\tilde{\omega})\;I
-i\{
\omega\;[g\;(
\overline{w_\omega}'\;\tilde{w}_{\tilde{\omega}}
+\overline{w_\omega}\;\tilde{w}'_{\tilde{\omega}})
g'\;\overline{w_\omega})\;\tilde{w}_{\tilde{\omega}}]
+\tilde{\omega}\;[
g\;(\overline{w_\omega}\;\tilde{w}'_{\tilde{\omega}}
+\overline{w_\omega}'\;\tilde{w}_{\tilde{\omega}})
+g'\;\overline{w_\omega}\;\tilde{w}_{\tilde{\omega}}]
\}
\nonumber\\
&=&
(\omega-\tilde{\omega})\;I
-i\{(\omega+\tilde{\omega})\; g\; 
\overline{w_\omega}\;\tilde{w}_{\tilde{\omega}}\}'
\ea
where $I$ is the integarnd of the integral in (\ref{7.27}).
We obtain the {\it Wronskian identity}
\be \label{7.29}
W'(w_\omega,\tilde{w}_{\tilde{\omega}})
:=\frac{d}{dz}\;
[
\overline{w_\omega}\;f\; \tilde{w}'_{\tilde{\omega}}
-\overline{w_\omega}'\;f\; \tilde{w}_{\tilde{\omega}}
+i\;
(\omega+\tilde{\omega})\; g\; 
\overline{w_\omega}\;\tilde{w}_{\tilde{\omega}}
]
=(\omega-\tilde{\omega}) \; I
\ee
which holds for arbitrary solutions $w_\omega,\; \tilde{w}_{\tilde{\omega}}$
with the same $l,m$ labels and frequencies $\omega,\tilde{\omega}$ 
respectively. The term in the square bracket on the left hand side 
is the (generalised) 
Wronskian of the two solutions. \\
\\
Formula (\ref{7.29}) has two applications:\\ 
1.\\
If we integrate (\ref{7.29}) over $\mathbb{R}$ we find 
\be \label{7.30}
<W_\omega,\;\tilde{W}_{\tilde{\omega}}>
=e^{i(\tilde{\omega}-\omega)\;\tau}
\frac{1}{\omega-\tilde{\omega}}\;
\int\; dz\; W'(w_\omega,\tilde{w}_{\tilde{\omega}})
\ee
In the present application of this formula we anticipate singular behaviour
of the solution at $z=0,\pm R$ hence we interprete the r.h.s. of (\ref{7.30}) 
as the 
{\it principal value}
\ba \label{7.31}
&& \int\; dz\; W'(w_\omega,\tilde{w}_{\tilde{\omega}})
:=\lim_{l\to 0+}
[\int_{-\infty}^{-R-l}
+\int_{-R+l}^{-l}
+\int_l^{R-l}+
\int_{R+l}^\infty]
\; dz\; 
W'(w_\omega,\tilde{w}_{\tilde{\omega}})
\nonumber\\
&=& [W(w_\omega,\tilde{w}_{\tilde{\omega}})]_{-\infty}^\infty
+\lim_{l\to 0+}\; (
[W(w_\omega,\tilde{w}_{\tilde{\omega}})]_{-R-l}^{-R+l}
+[W(w_\omega,\tilde{w}_{\tilde{\omega}})]_{-l}^{+l}
+[W(w_\omega,\tilde{w}_{\tilde{\omega}})]_{R-l}^{R+l}
)
\ea
Thus the inner product between two solutions can be expressed in terms of 
their values and first derivatives at both spatial infinities plus a 
term that is exactly given by the {\it discontinuities} of the Wronskian at 
$z=0,\pm R$. Indeed, as the coefficients of the second order 
ODE have singularities, we expect singularities of the second derivatives 
compatible with discontinuities of the first derivative. 

Since the right hand side of (\ref{7.30}) does not vanish 
and becomes singular for $\omega=\tilde{\omega}$ we conclude that the 
solutions are not normalisable in the strict sense but in the generalised 
sense, i.e. the inner product will be proportional to 
$\delta^{(1)}(\omega,\tilde{\omega})$.

To read off the normalisation constant $\kappa_{\omega,l,m}$ 
suppose that the discontinuity vanishes.
Then note that far out at infinity the solutions obey the flat space 
wave equations and thus will display a radial dependence corresponding
to radial plane waves $e^{i\;\pm \omega z}/|z|$ while 
$e,f$ grow as $z^2$. Then we make use of the 
distributional identity
\be \label{7.31z}
\lim_{z\to \infty} 
\frac{\sin(z(\omega-\tilde{\omega})}{\pi(\omega-\tilde{\omega})}       
=\lim_{z\to \infty} \int_{-z}^z\; \frac{dz'}{2\pi}\; e^{iz'(\omega-
\tilde{\omega})}=\delta(\omega,\tilde{\omega})
\ee
The positive solution subspace is then selected to be the span of 
the $W_{\omega,l,m},\; \kappa_{\omega,l,m}>0$. To actually compute 
$\kappa_{\omega,l,m}$ we must of course gain sufficient knowledge 
on the solution and its first derivative near $z=-\infty$ once we specify 
those data near $z=\infty$ or vice versa paying attention to the 
singularities. This is a non-trivial task as one 
has to compute the influence of the curvature and its singularity 
at $z=0$ as we follow the solution from $z=\infty$ to $z=-\infty$.
We expect that 
methods from the theory of Heun functions \cite{Heun} and the rich literature 
on the solutions of singular second order linear ODE's \cite{Titchmarsh}
can be employed.\\
\\
2.\\
For $\omega=\tilde{\omega}$ we see that the Wronskian is a constant. 
This leads to {\it Wronskian relations} between the solutions and their 
first derivatives at the two spatial infinities. Moreover for 
$w_\omega=\tilde{w}_\omega$ we find the constant
\be \label{7.31a}
f(\bar{w}_\omega\; w'_\omega-\bar{w}_\omega\; w'_\omega)     
+2i\omega |w_\omega|^2=:2i c_\omega
\ee
Using the WKB decomposition $w_\omega=m_\omega\; e^{i\alpha_\omega}$ into 
modulus and phase we 
obtain 
\be \label{7.31b}
m_\omega^2\;[ f\; \alpha'+\omega\; g]=c_\omega
\ee
which allows to determine the phase exactly in terms of the modulus. 
The differential equation for $w_\omega$ can be split into real and 
imaginary part and vanishing of the imaginary part is equivalent
to (\ref{7.31b}). For the real part we find, using (\ref{7.31a}) 
and using the abbreviation $A:=m_\omega^2$
\be \label{7.31c}
[e\omega^2-u_l]\; A^2+\frac{1}{2}\; f' \; A\; A'
-\frac{1}{f}(c_\omega^2-\omega^2\; g^2\;A^2)+f(\frac{1}{2} A\; 
A^{\prime\prime}-\frac{1}{4} (A')^2)=0
\ee
Unsurprisingly, this equation has a similar structure as the one for the 
modulus of the wave function in cosmology (both are obtained 
from the WKB Ansatz) whose iterative solution leads 
to the adiabatic vacua \cite{Fulling}. For $c_\omega=0$ one 
can transform this into a first order Riccati equation for $B=A'/A$.

\subsection{Details for the GPG background}
\label{s7.4}

We discuss some of the details of the required steps where we 
keep a close analogy with the treatment of the potential barrier 
problem of QED and the resolution of the Klein paradox (Schwinger effect, 
superradiant scattering)  discussed in \cite{Klein}. The analogy is 
only very rough: In contrast to (\ref{7.1}) the external electric field
of the potential barrier problem is everywhere finite and has at most 
jump discontinuities. However 
both problems share the feature that the external potential is asymmetric 
under reflection of $z\to -z$ because in the past Kruskal 
universe the metric 
is given by ingoing GP coordinates while in the future universe it is given 
by outgoing ones.\\
\\
In what follows we define past (P) and future (F) by GP time $\tau$ and 
forget about 
the past and future labels attached to the Kruskal universes: 
Every synchronous 
$\tau=$const. surface crosses both Kruskal universes and has a ``left'' 
end (L) at spacelike infinity of the future 
Kruskal universe and a ``right'' end (R) in the past Kruskal universe. 
Thus the $P,L$ labels have very different meaning. Thus in what follows:\\
P means $\tau\to -\infty$,\\ 
F means $\tau\to +\infty$,\\
L means $z\to -\infty\;\;\Rightarrow\;\; \bar{r}\to +\infty$,  \\
R means $z\to +\infty\;\;\Rightarrow\;\; r\to +\infty$.

An observer in the infinite past P located a L can emit 
spherical waves into the spacetime which therefore must travel
towards smaller values of $\bar{r}$ or larger values of $z$. Such a mode
is described by $\frac{e^{-ikz}}{|z|}, k>0$ (we define the 
radial momentum operator by 
$p=i\frac{x^a}{r}\;\partial_a=i\frac{d}{dz}$ so that 
$[p,z]=i$ and $p\; e^{-ikz}=k\; e^{ikz}$).
Since at $(P,L)$ the spacetime metric is flat and the potential $U$ vanishes
we have for a temporal dependence $e^{i\omega\tau}$ the dispersion relation 
$k^2=\omega^2$, hence $k=|\omega|$.  
That wave will be transmitted and reflected by the curvature and will at 
$F$ reach both ends $L,R$. The transmitted part travels to $R$ 
at $z=\infty$ and is described by $^{-ikz}$, the reflected part travels to 
$L$ and is described by $e^{ikz}$. Thus we define the modes 
\be \label{7.32}
W^{P,L}_{\omega,l,m}:=N^{P,L}_{\omega,l,m}\;\frac{e^{i\omega\tau}}{|z|}\;
\left\{ \begin{array}{cc}
e^{-i|\omega|z} + R^{P,L}_{\omega,l,m} \; e^{i|\omega|z} & z\to -\infty\\
T^{P,L}_{\omega,l,m} \; e^{-i|\omega|z} & z\to \infty
\end{array}
\right.
\ee
where $N,R,T$ are called normalisation, reflection and transmission 
coefficients. 

In complete analogy, we can consider emission from (P,R) of waves travelling 
towards smaller $r$ i.e. smaller $z$ described by modes $e^{i|\omega| z}$. 
Now the transmitted mode travels towards L and the reflected towards R again.
Thus we define the modes
\be \label{7.33}
W^{P,R}_{\omega,l,m}:=N^{P,R}_{\omega,l,m}\;\frac{e^{i\omega\tau}}{|z|}\;
\left\{ \begin{array}{cc}
T^{P,R}_{\omega,l,m} \; e^{i|\omega|z} & z\to -\infty\\
e^{i|\omega|z} + R^{P,R}_{\omega,l,m} \; e^{-i|\omega|z} & z\to \infty
\end{array}
\right.
\ee
  
An observer in the infinite future F located at L can receive waves that 
travel to larger values of $\bar{r}$ i.e. smaller values of $z$. Those waves 
are the result of a superposition of two waves coming from (P,L) and (P,R)
respectively. The wave from (P,R) looks like a a wave transmitted from (F,L)
if we travel to the past while the wave from (P,L) looks like 
a wave reflected from (F,L) if we travel to the past. This gives the modes
\be \label{7.34}
W^{F,L}_{\omega,l,m}:=N^{F,L}_{\omega,l,m}\;\frac{e^{i\omega\tau}}{|z|}\;
\left\{ \begin{array}{cc}
e^{i|\omega|z} + R^{F,L}_{\omega,l,m} \; e^{-i|\omega|z} & z\to -\infty\\
T^{F,L}_{\omega,l,m} \; e^{i|\omega|z} & z\to \infty
\end{array}
\right.
\ee
    
In complete analogy, an observer in F located at R can receive waves 
travelling towards increasing $r$ i.e. larger values of $z$ and 
is described by $e^{-i|\omega|z}$
which is the result of a superposition of waves from P from both ends where 
the wave from L and R looks as transmitted and reflected respectively from 
(F,R) when followed towards the past. Thus we define the modes
\be \label{7.35}
W^{F,R}_{\omega,l,m}:=N^{F,R}_{\omega,l,m}\;\frac{e^{i\omega\tau}}{|z|}\;
\left\{ \begin{array}{cc}
T^{F,R}_{\omega,l,m} \; e^{-i|\omega|z} & z\to -\infty\\
e^{-i|\omega|z} + R^{F,R}_{\omega,l,m} \; e^{i|\omega|z} & z\to \infty
\end{array}
\right.
\ee
For both choices of $P,F$, the modes labelled by $\ast\in \{L,R\},\;
\omega\in \mathbb{R},\;0\le  |m|\le l\in \mathbb{N}_0$ are complete 
generalised orthonormal bases of the Krein Hilbert space 
but given $\ast,l,m$ only for $\omega\in 
\Delta^{P/F}_{+,\ast,l,m}$ correspond to positive norm solutions with respect 
to the 1-particle inner product (\ref{7.4}) which selects 
the annihilation operators $A^{P/F,\ast}_{\omega,l,m}=
<W^{P/F,\ast}_{\omega,l,m},\Phi>,\; \omega\in
\Delta^{P/F}_{+,\ast,l,m}$. Thus there will be particle production 
in the sense that the observers at P,F do not agree on what the zero 
particle (vacuum state) is when the sets 
$\Delta^{P/F}_{+,l,m}:=\Delta^{P/F}_{+,L,l,m}\cup 
\Delta^{P/F}_{+,R,l,m}$ do not coincide. 

These sets will be in bijection with $\mathbb{R}_+$, i.e. we have 
bijections $b^{P/F}_{l,m}:\;\mathbb{R}_+\to   
\Delta^{P/F}_{+,l,m}$. Given $\omega\in \mathbb{R}_+,l,m$ we define 
the Fock space ${\cal H}^{P/F}$ as the excitations by the modes 
$b^{P/F}_{l,m}(\omega),l,m$ of the Fock vacuum $\Omega^{P/L}$.
If these Fock spaces are the same, that is, when $\Omega^F$ constructed 
as the excitation by particle pairs of $\Omega^P$ is normalisable 
in ${\cal H}^P$, then the unitary S-matrix $S$ defined by the matrix 
elements 
\be \label{7.36}
<\psi^F_\alpha, \psi^P_\beta>=:
<\psi^F_\alpha, S \psi^F_\beta>  
=<\psi^P_\alpha, S \psi^P_\beta>  
\ee
will be non-trivial.

We note that the four sets of modes $W^{P/F,L/R}_{\omega,l,m}$ with 
asymptotics given in (\ref{7.32}) - (\ref{7.35}) allow us to compute 
their inner products via (\ref{7.30}) and (\ref{7.31}) in terms 
of the twelve complex numbers 
$N^{P/F,L/R}_{\omega,l,m},\;R^{P/F,L/R}_{\omega,l,m},\;     
T^{P/F,L/R}_{\omega,l,m}$. The modulus $|N^{P/F,L,R}_{\omega,l,m}|$ will 
be fixed by the condition that the corresponding 
$\kappa^{P/F,L/R}_{\omega,l,m}$ equals $\pm 1$. 

%Further information is obtained using the self-adjoint reformulation 
%that has no first derivative.     
Further analysis has to take the concrete details of (\ref{7.1}) and 
the potential into account. One finds explicitly
\ba \label{7.37}
&&N=1,\; N^3=\sqrt{\frac{R}{|z|}},\;
e=z^2,\; g=e N^3,\; f=e(1-[N^3]^2), U=\frac{(1-s^2)\;R}{|z|\; z^2}, 
\nonumber\\
&& u_l=e(\frac{l(l+1)}{z^2}+U)=l(l+1)+(1-s^2)\;(N^3)^2
\ea
where $s=0,1,2$ is the spin of the field (scalar, vector, tensor). 

We see that $f$ grows as $z^2$ while $g$ only as $|z|^{3/2}$ and thus 
its contribution to the Wronskian vanishes at infinity since by 
(\ref{7.32}) - (\ref{7.35}) the solutions decay as $|z|^{-1}$. It 
follows (drop labels $l,m$) 
\be \label{7.38}
\lim_{z\to \pm\infty}\;
W(w_\omega,\tilde{w}_{\tilde{\omega}})(z)
=\lim_{z\to \pm\infty}\; 
[\overline{|z|\; w_\omega}\; (|z|\tilde{w}_{\tilde{\omega}})'   
-(\overline{|z|\; w_\omega})'\; (|z|\tilde{w}_{\tilde{\omega}})]
\ee
It is not difficult to see that equality of (\ref{7.38}), which
holds in the 
absence of discontinuities, at 
$z=\pm \infty$ for $\omega=\tilde{\omega}$, $w_\omega=\tilde{w}_\omega
=w^{P/L,L/}_{\omega,l,m}$ gives the continuity equations
\be \label{7.39}
1=|R^{P/L,L/R}_{\omega,l,m}|^2+|T^{P/L,L/R}_{\omega,l,m}|^2
\ee

\subsection{Particle production and Hawking radiation in BHWHT spacetime}
\label{s7.5}

Assuming that these interesting challenges can be met, one then would 
have access to QFT on a BHWHT spacetime i.e. the quantum field 
is given by  
\be \label{7.38a}
\Phi=\sum_{l,m}\int_{\Delta_{+,l,m}}\; d\omega\;
\{A_{\omega,l,m} \; W_{\omega,l,m}+c.c.\};\;\;
\Delta_{+,l,m}=\{\omega\in \mathbb{R};\;\kappa_{\omega,l,m}>0\}
\ee
This has several aspects. 
One may construct modes that look like those of Minkowski space for 
$z=r\to \infty, \tau\to -\infty$ in the past universe or for 
$z=-\bar{r}\to -\infty, \tau\to +\infty$ in the future universe.
This selects corresponding Fock vacua $\Omega_P, \Omega_F$ and
particle number operators $N_P=A_P^\ast A_P, A_P\Omega_P=0$ and  
$N_F=A_F^\ast A_F, A_f\Omega_F=0$ and  
one may study the particle content of $F$ particles in the 
$P$ state $<\Omega_P, N_F \Omega_P>$ 
or vice versa which maybe interpreted as particle creation effect due to 
curvature and singularity. Or one may consider only the SS portion of 
the past universe (or the MSS portion of the future universe) and 
construct a Fock representation with respect to the 
Schwarzschild time $t$ foliation in the SS (or MSS) portion.
One can expand the spacetime field $\Phi$ in SS with respect to 
both sets of mode systems and compute $A_{SS}$ in terms of $A_P, A_P^\ast$
(Bogol'ubov coefficients). This would be the analog of the 
Unruh effect with the free falling GP observer and stationary 
SS observer respectively 
playing the role of the inertial and accelerated observer respectively.
As usual, the selection of, say, the P vacuum is not unique and may wonder 
whether there exists a choice for which the 2-point function 
$<\Omega_P, \Phi(\tau,x)\;\Phi(\tau',x')\Omega_P>$ has the short distance 
singularity structure of Minkowski space (Hadamard state \cite{Fulling}) 
which has 
an elegant reformulation in terms of the wave front set of this 
two point function \cite{Radzikowski}.

\section{Quantisation and Backreaction}
\label{s8}

In this section we discuss further challenges of quantum black hole
perturbation theory without going into much detail.\\
\\
Let us summarise:\\
The concrete expression for the physical or reduced Hamiltonian
which depends only on the true $X,Y,P,Q$ degrees of freedom
cannot be provided in closed form but can be approximated in terms 
of two distinct perturbative hierarchies: First, the deviation from 
spherical symmetry as parametrised by the perturbations $X,Y$ which 
are considered as of first order. Second, an inverse core mass expansion.
The first expansion arises from the split of the degrees of freedom 
into spherically symmetric and spherically non-symmetric sets and is 
a simple canonical transformation on the phase space. This step 
is in principle non-perturbative and can be performed exactly in closed 
form if one uses an equivalent polynomial form of the constraints. 
The expansions come into play upon solution of the constraints for the 
momenta: First, as the constraints depend non-linearly on the momenta, 
we get an infinite set of non-linear constraints which are challenging 
to solve non-perturbatively as is known already 
from finite dimensional algebraic geometry \cite{AlgGeo}. However,
exploiting the first perturbative hierarchy this becomes feasible as 
shown in \cite{pa129} provided one can solve the zeroth order equations 
of that scheme. It is this assumption which triggers the second 
perturbative hierarchy: The zeroth order constraints (one for each value 
of the radius) are still 
non linear in the symmetric momenta but almost decoupled. The coupling 
is via the fact that the constraints depend on radial derivatives. 
There would be no second hierarchy necessary if one could solve that 
ODE system exactly. However, unless there is no scalar hair, this 
is not the case and thus we must release the inverse core mass expansion 
as a second approximation scheme of the Picard-Lindel\"of type. 
Here the core mass is identified as 
an integration constant in that system of zeroth order differential 
equations which reduces to the Schwarzschild mass in absence of scalar
hair (in which case the core mass equals the ADM mass). From here on 
at each order in the first perturbative scheme one just has to solve 
linear ODE systems for which uniqueness and existence results as well as 
efficient approximation schemes are available. The inverse mass 
expansion continues to play an important role also in those
higher orders.

At any order with respect to both expansions, the 
approximate expression for $H_{{\rm red}}$ depends non-polynomially
on the core mass $\hat{M}$ and polynomially on the fields $Q,P,X,Y$
in terms of nested radial integrals and multiple sums over the 
harmonic labels $l,m$. Its quantisation is therefore challenging
in two aspects: First, it is not clear that the operators
$\hat{M}^n,\; n\in \mathbb{Z}$ can be densely defined. This question 
has been answered affirmatively in \cite{TTNewBasis}. Second, consider
the energy density $e_{KG}(r)$ of the symmetric Klein Gordon field
which is a quadratic function of $P^{KG}$ and a polynomial function 
of $Q_{KG}$ (that polynomial has degree two unless the scalar field 
is self-interacting). Then one can envisage a Fock quantisation of the 
total Klein Gordon Energy $E_{KG}=\int_0^\infty\; ds \; e_{KG}(s)$. 
However, the physical Hamiltonian depends on nested radial 
integrals of the form 
\be \label{8.1}
\int_0^\infty\; ds\; e_{KG}(s)\;\int_0^s\; dt\; e_{KG}(t)
\ee
and not on polynomials of $E_{KG}$. The energy density integral with 
a sharp radial cut-off could be problematic in QFT 
%\cite{SharpCut} 
and perhaps must be 
regularised by smoothening the cut-off. Upon removal of the regulator,
the result could retain a regulator dependence which must be controlled.\\
\\
Next we come to the question of backreaction beyond the backreaction  
on the spacetime causal structure discussed in section \ref{s0}.
To avoid misunderstanding, the notion of backreaction used in this work 
is {\it not} that of 
the semiclassical Einstein equations $G(g):=<T(g)>_g$ where 
$G(g)$ is the classical Einstein tensor, $T(g)$ the quadratic form 
corresponding to 
the matter energy momentum tensor and $<.>_g$ is a positive linear functional 
on the Weyl algebra $\mathfrak{A}_g$
of the (free) spacetime matter fields which are supposed   
to solve the (free) Heisenberg equations generated by (the quadratic part of) 
the matter Lagrangian in the background metric $g$. The reasons are manifold:
First of all, the semiclassical Einstein equations treat $g$ as a classical 
field rather than a quantum field. Second, they are supposed to be 
diffeomorphism covariant rather than singling out a preferred gauge condition
and thus do not distinguish between true and gauge degrees of freedom.
This has the following consequences: Not only is it a complicated task 
to find a self-consistent metric that solves these equations because the 
state $<.>_g$ depends on $g$ as well, it also is inconsistent with 
the Bianchi identity $\nabla_g \cdot G\equiv 0$ if one does not carefully 
supplements the chosen (normal) ordering of $T(g)$ by suitable counter 
terms which for example is responsible for the trace anomaly in cornformally 
invariant theories \cite{15}. By contrast, our notion of backreaction 
is in the context of quantum gravity in which we quantise the metric as 
well after reducing the phase space so that only the true degrees of freedom 
are quantised, the issue of gauge invariance has been solved from the 
outset and no Bianchi identies have to be obeyed. In fact it is not even 
a priori clear what one would mean by a ``quantum Bianchi identity''. 
A possible interpretation of this term would be to try to construct a
quadratic form corresponding the classical objects $\nabla_g, G(g)$
where $g$ is to be replaced by the quantum metric and then to consider 
$\nabla_g\cdot G(g)$. Due to non-commutativity of operator valued 
distributions involved, this is unlikely to vanish identically as a 
quadratic form, at best we expect it to vanish to zeroth order in $\hbar$
when computing the expectation value of the quadratic form with respect 
to semiclassical (low fluatuation) states.\\ 
\\ 
Having clarified this, consider first the classical 
theory. Given the Hamiltonian $H=H(P,Q,X,Y)$ we can in principle solve 
the classical Hamiltonian equations of motion given initial data. 
To simplify the notation let $R=(P,Q),\; Z=(Y,X)$ with initial data
$(R_0,Z_0)$. Consider first an arbitrary function $R(\tau)$ and solve 
the Hamiltonian equations of motion 
$\dot{Z}^R_{Z_0}(\tau)=\{H,Z\}_{R=R(\tau),Z=Z^R_{Z_0}(\tau)},\;
Z^R_{Z_0}(0)=Z_0$ where $H(\tau,Z)=H(Z(R(\tau),Z)$ is considered as an
explicitly time dependent Hamiltonian. 
The solution will be a function 
$Z^R_{Z_0}(\tau)$ which depends on the chosen $R(\tau)$. 
It can be obtained by Picard-Lindel\"of iteration
$Z^R_{Z_0}(\tau)=Z_0+\int_0^\tau\; ds;\;
\{H,Z\}_{R=R(\tau),Z=Z^R_{Z_0}(s)}$ and thus depends on nested integrals 
of $R(\tau)$. Now we 
set $R(\tau)=R_{R_0,Z_0}(\tau)$ and solve the Hamiltonian equations 
of motion $\dot{R}_{R_0,Z_0}(\tau)=\{H,Z\}_{R=R_{R_0,Z_0}(\tau),
Z=Z^{R_{R_0,Z_0}}_{Z_0}(\tau)},\;R_{R_0,Z_0}(0)=R_0$. This 
takes the form of an integro differential equation which contains 
memory and friction terms and whose parametric dependence on $Z_0$ is 
considered as ``noise'' or ``environment''. 
If one has only statistical knowledge about 
$Z_0$ one can consider it as a random variable leading to 
a Langevin type \cite{Langevin}
of equation or one can average over the $Z_0$ dependence 
which leads to corresponding ``master'' equations. If one has specific 
knowledge
about the initial state $Z_0$ and can motivate a certain specific $Z_0$ 
as given, one may try to obtain an effective Hamiltonian 
which governs the integro-differential equation.

In the quantum theory we see that in the statistical approach we enter the 
regime of open quantum systems, decoherence, entanglement, partial tracing
and non-equilibrium statistical mechanics \cite{Decoherence} which leads to 
master equations 
of Lindbladt type for the ``statistical operator'', i.e. the density
matrix that replaces $Z_0$. In the effective approach 
of the quantum theory one tries to solve the exact Schr\"odinger or Heisenberg 
equations following the above idea of the classical theory 
of solving first the equations for the $X,Y$ sector and after that for the 
$Q,P$ sector. To do this one assumes that the physical 
Hilbert space is a tensor 
product ${\cal H}={\cal H}_f\otimes {\cal H}_s$ where ${\cal H}_f,\; 
{\cal H}_s$ respectively describe the ``fast'' $X,Y$ and ``slow'' 
$Q,P$ degrees of freedom. The terminnology comes from the well known 
Born-Oppenheimer-Approximation (BOA) scheme \cite{BOA}
that was invented for molecular 
physics and for which the huge differences in the masses of electrons and 
nuclei leads to a natural hierarchy of time scales of these two types 
of degrees
of freedom so that the fast particles move essentially as if the slow 
particles were at rest with only a weak (adiabatic) correction governed
by an adiabatic parameter $\epsilon$ (the quotient of electron and nuclei
masses). Thus the BOA works very well for systems with 1. finitely many 
fast degrees of freedom, 2. finitely many slow degrees of freedom,
such that 3. the interaction between the two subsystems depends only on
a commuting set of slow degrees of freedom (typically $Q$), such that 
4. a natural adiabatic parameter $\epsilon$ exists and such that 
5. the Hilbert space representation is of tensor product type as above.

In the application to black hole perturbation theory the first 
4 assumptions all fail
when incorporating matter ``hair'' in addition to Maxwell fields
such as the Klein-Gordon or fermion fields. First, assumption 3. is violated
as the interaction involves both $P,Q$ variables which do not commute in 
quantum theory. A natural extension of BOA which relaxes 3. is space 
adiabatic perturbation theory (SAPT) which however still rests on the validity 
of assumptions 1., 2., 4. while 5. is trivially satisfied for finite 
dimensional systems. In \cite{ST} SAPT was applied to the simpler case 
of quantum cosmology for which the slow sector has only finitely many degrees 
of freedom, hence assumption 2. holds. The relaxation of assumption 1. 
is non-trivial because assumption 5. is now no longer granted as the 
Hilbert space representation of the fast sector depends on the slow sector 
variables and the corresponding are not necessarily unitarily equivalent. 
However, the obstacles can be overcome perturbatively. Already in quantum 
cosmology also assumption 4. is violated because there is no mass hierarchy 
between e.g. the homogeneous modes of a scalar inflaton field and its 
non-homogeneous perturbations. This obstacle can be overcome as follows:
The adiabatic effective decoupling of the slow and fast sectors rests 
on the Weyl quantisation of the slow system in terms of the variables 
$Q,P':=\epsilon P$ where $\epsilon$ is the adiabatic parameter.
Now even if the reduced Hamiltonian $H(X,Y,Q,P)$ does not have a natural 
such parameter, using that $H$ depends polynomially on $P$ we can multiply
$H$ by $\epsilon^N$ where $N$ is the smallest positive integer such that 
$H_\epsilon(X,Y,Q,P'):=\epsilon^N\;H(X,Y,Q,P=\frac{P'}{\epsilon})$ has no 
negative powers of $\epsilon$. Then SAPT can be formally applied 
to $H_\epsilon$ which 
allows to compute an effective Hamiltonian $H_{\epsilon,n,k}(Q,P')$ for each 
energy band $n$ of the fast sector and $k$ is the order of the adibatic 
expansion. At the end we set $\epsilon:=1$ which can be considered as a 
kind of analytic extension. In this way the adiabatic parameter just 
serves to organise the adiabatic perturbation and is removed at the end.
It is of course not granted that resulting effective Hamiltonians 
$H_{n,k}=H_{\epsilon=1,n,k}$ converge in any sense.

This still leaves condition 2. which in contrast to quantum cosmology 
is violated for quantum black holes: The symmetric degrees of freedom 
are just spherically symmetric and not homogeneous and 
thus define an effectively 1+1 dimensional
field theory on the radial z-axis rather than a 1+0 dimensional mechanical 
system. In order to apply SAPT this calls for the Weyl quantisation of 
a field theory which is problematic \cite{WQI}. As a first step one may 
apply the following regularisation of the problem: Let $b_I,\; I=0,1,2,..$ 
be an on ONB of $\mathfrak{h}=L_2(\mathfrak{R},dz)$ and consider the 
conjugate variables 
$Q_I=<b_I,Q>_{\mathfrak{h}},\; P_I=<b_I,Q>_{\mathfrak{h}}$. We now 
expand $H(X,Y,Q,P)$ into the $Q_I,P_I$ using completeness 
$Q=\sum_I\; Q_I\; b_I$ and for given ``mode cut-off'' $0<M<\infty$ 
set to zero all instances of $Q_I,P_I,\; I>M$ which maybe 
called a truncation. The resulting Hamiltonian $H_M(X,Y,\{Q_I,P_I\}_{I=0}^M)$
can then be treated by SAPT methods as above and one may study in which way 
one can remove the cut-off $M$.

\section{Observation, radiation energy and flux}
\label{s9} 

Since we have a concrete expression for a physical Hamiltonian 
$H[P,Q,X,Y]$ at our disposal which depends only on the true degrees of 
freedom $P,Q,X,Y$ we can write it in the form 
\be \label{9.1}
H[P,Q,X,Y]=\int_\sigma \; d^3x\; h[P,Q,X,Y;x)
\ee 
where $h$ can be interpreted as the energy volume density observed by 
the free falling observers. From the way it is constructed 
perturbatively in terms of nested radial integrals using the 
inverse core mass expansion, it is spatially 
non-local, explicitly spatially but not explicitly 
$\tau$ dependent as the system is conservative and 
at any finite order in perturbation theory a polynomial in the canonical 
fields and their spatial derivatives up to a finite order $N$. 
Since $h$ therefore depends generically spatially non-locally on the fields
we used a square bracket notation with respect to the dependence 
of $h$ on them and a round bracket notation with respect to the explicit 
$x$ dependence.  

To compactify the notation we write $R=(Q,X),\; Z=(P,Y)$. The 
variation of (\ref{9.1}) is then given by 
\be \label{9.2}
[\delta H][R,Z]=\sum_{k=0}^N\; \int_\sigma\;d^3y \int_\sigma\;\;d^3x\;
\{r^{a_1..a_k}(y,x)\;[\delta R]_{,a_1..a_k}(x)
+z^{a_1..a_k}(y,x)\;[\delta Z]_{,a_1..a_k}(x)\}
\ee
where the integral kernels displayed are themselves nested integrals over 
polynomials in $R,Z$ and their spatial derivatives up to order $N$. 
They can be found by collecting in the expression for 
$\delta H$ all instances of e.g. 
$[\delta R]_{,a_1..a_k}(z)$ where $z$ is  
one of the integration variables in the nested integrals and then relablelling 
integration variables $z\leftrightarrow x$. The Hamitonian equations of 
motion are obtained by integrating by parts, dropping boundary terms 
\be \label{9.3}
\dot{R}(x)=\sum_{k=0}^N\; (-1)^k\; \int\; d^3y\; 
[r^{a_1..a_k}(y,x)]_{,a_1..a_k},\;\;
-\dot{Z}(x)=\sum_{k=0}^N\; (-1)^k\; \int\; d^3y\; 
[z^{a_1..a_k}(y,x)]_{,a_1..a_k},\;\;
\ee
where the multiple partial derivatives are with respect to the argument 
$x$. 

Let now $S(\tau)\subset \sigma$ be a compact region in $\sigma$ 
then the object
\be \label{9.4}
E_{S(\tau)}=\int_{S(\tau)} \; d^3x\; h[R(\tau),Z(\tau);x)
\ee 
is the energy content of the region $S(\tau)$ where the $\tau$ dependent 
fields satisfy the above Hamiltonian equations of motion, e.g.
$\dot{X}=\{H,X\}$. Then $P_{S(\tau)}=\frac{d}{d\tau} E_{S(\tau)}$ 
is the total power emitted/absobed from $S(\tau)$. To obtain an 
energy flux area current one usually resorts to the Lagrangian framework 
and extracts the on shell conserved Noether current from the symmetry 
of the Lagrangian under foliation time translations, 
formally corresponding to the conserved time component of the canonical energy
momentum tensor. We could proceed 
like this here as well by formally computing the Legendre transform 
of the Hamiltonian but since we expect $H$ to depend higher than quadratically
on $Z$ and since the relation between momenta and velocities 
will involve integro differential equations one will have trouble to 
compute the Legendre transform sufficiently explicitly. 

Fortunately this 
is not necessary. All that is needed is the Hamiltonian. To see 
how this works in a familiar setting consider the Hamiltonian density 
$h=\frac{1}{2}[\vec{E}^2+\vec{B}^2]$ of free Maxwell theory 
on Minkowski space where 
$\vec{B}=\vec{\nabla}\times \vec{A}$ is the magnetic field of the 
vector potential $\vec{A}$ which has canonical brackets with the 
electric field $\vec{E}$. The equations of motion resulting 
from $H=\int\;d^3x\; h$ are the familiar vacuum Maxwell equations 
$\dot{B}=\vec{\nabla}\times \vec{E},\;
\dot{E}=-\vec{\nabla}\times \vec{B}$. Consider a compact 
region $S\subset \mathbb{R}^3$ and its energy content
$E(S)=\int_S\; d^3x\; h$ evaluated on a solution of Maxwell's equations.
Then 
\be \label{9.5}
\dot{E}(S)=
-\int_S\; d^3x\;\vec{\nabla}\cdot[\vec{E}\times\vec{B}] 
=-\int_{\partial S}\; d\Sigma_a\;[\vec{E}\times\vec{B}]^a
\ee
which correctly yields the Poynting vector 
$\vec{J}=\vec{E}\times \vec{B}$ as energy current area density. In other 
words with $J^\tau=h$ the 4-vector $J^\mu$ is conserved on shell.  
   
We proceed analogously in the present more complicated situation. From 
(\ref{9.2}) we have for a time independent set $S$
\be \label{9.6}
\dot{E}(S)
=\sum_{k=0}^N\; \int_S\;d^3y \int_\sigma\;\;d^3x\;
\{r^{a_1..a_k}(y,x)\;[\dot{R}]_{,a_1..a_k}(x)
+z^{a_1..a_k}(y,x)\;[\dot{Z}]_{,a_1..a_k}(x)\}
\ee
We integrate the spatial derivatives successively by parts. In the 
course of this process we may pick up boundary terms or not 
depending on the distributional properties of the integral kernels 
$r,z$. Typically, if e.g. $r^{a_1..a_k}(y,x)$ depends on $l\le k$ nested 
integrals, then the first $l$ integrations by parts do not generate 
any boundary terms but each integration by parts removes one nested integral.
Thus rearranging terms by performing integrations 
by parts that do not generate boundary terms one may rewrite (\ref{9.6}) as 
\be \label{9.7}
\dot{E}(S)
=\int_S\;d^3y \int_\sigma\;\;d^3x\;
\{\hat{r}(y,x)\;\dot{R}(x)
+\hat{z}(y,x)\;\dot{Z}(x)\}
+\sum_{k=1}^N\; \int_S\;d^3x\;
\{\hat{r}^{a_1..a_k}(x)\;[\dot{R}]_{,a_1..a_k}(x)
+\hat{z}^{a_1..a_k}(x)\;[\dot{Z}]_{,a_1..a_k}(x)\}
\ee
where the second term is now ultra-local. We now perform the remaining 
integrations by parts and obtain a bulk term and a surface term.
The bulk term is given by (\ref{9.6}) 
with all integrations by parts performed and all boundary terms dropped,
that is 
\be \label{9.8}
0=\sum_{k=0}^N\; (-1)^k\;\int_S\;d^3y \int_\sigma\;\;d^3x\;
\{r^{a_1..a_k}_{,a_1..a_k}(y,x)\;[\dot{R}](x)
+z^{a_1..a_k}_{,a_1..a_k}(y,x)\;[\dot{Z}](x)\}
\ee
where the equations of motion (\ref{9.3}) were used. The surface term is 
given by 
\be \label{9.9}
-J^a:=\sum_{k=1}^N\;\;\sum_{l=1}^{k}\;(-1)^{l+1}\;
\int_{\partial S}\;d\Sigma_{a_1}\;
\{\hat{r}^{a_1..a_k}_{,a_2..a_l}(x)\;[\dot{R}]_{,a_{l+1}..a_k}(x)
+\hat{z}^{a_1..a_k}_{,a_2..a_l}(x)\;[\dot{Z}]_{,a_{l+1}..a_k}(x)\}
\ee
where it is understood that the spatial derivatives in (\ref{9.9}) are 
simply absent when the indices are out of range. 

The on shell conserved energy Noether
current can now be determined: it is given by $J^\mu$ with 
$J^\tau=h$ and $J^a$ as in (\ref{9.9}) where $\dot{R},\dot{Z}$ are 
to be replaced by (\ref{9.3}). In particular, given a 
surface element $s$, $\int_s\; d\Sigma_a\; J^a$ is the power 
flowing through $s$. Using appropriate solutions of the field equations
one can now compute the classical radiation power through any surface 
similar as in usual second order perturbation theory \cite{Chandrasekhar}.

Similarly, in the quantum theory, one can define at second order, 
the grey body factors
$\sigma_{\omega,l,m}$  
of the radiation in the usual way for each bosonic 
degree of freedom \cite{Page}
\be
\label{9.10}
\lim_{r\to \infty} \omega_\beta[\int_{s(r)}\; d\Sigma_a J^a]=:
-\sum_{l,m}\int\; d\omega\; \frac{\sigma_{\omega,l,m}}{
\exp(\beta\;\omega)-1}
\ee
where $s(r)$ is a round sphere at radius $r$ 
and $\beta$ is the inverse temperature, expected to scale as $1/M$, of 
a corresponding KMS state $\omega_\beta$ which we consider as restriction 
of the GP vacuum state to the SS portion in a similar way as 
the Unruh KMS state arises by restriction of the Minkowski vacuum state
\cite{10} to the Rindler wedge.
The 
form of (\ref{9.10}) in fact has been confirmed for second order 
perturbation theory using mode functions of the type discussed in section 
\ref{s7} but restricted to an asymptotic SS region. More in detail,
there one works with the turtoise coordinate $z=r_\ast=r+R\ln(r/R-1),\;r>R=2M$
rather than the GPG coordinate $z$ 
and observes that the potential that enters the Schr\"odinger type equation
rewritten in terms of $z$ 
vanishes at both $z=\pm \infty$ being positive in between with a maximum 
around the photon sphere $r=3M$. Hence the mode function problem 
becomes a regular quantum mechanical scattering problem, see 
\cite{Chandrasekhar}
for all the details. One can work with either SS null coordinates 
$u=t-z,\;v=t+z$ or corresponding  
Kruskal null coordinates $V,U$ and the KMS state here arises by restricting 
the Kruskal vacuum state of one full Kruskal universe (SS, BH, WH, MSS;
Hartle-Hawking state) to the SS portion (the quantum field is written in 
terms of $u,v$ mode functions defining the Boulware state 
which are then expanded into $U,V$ mode functions 
defining the hartle-Hawking state which gives rise 
to non-vanishing Bogol'ubov coefficients in (\ref{9.10})).\\
\\
Remarks:\\
1.\\
It should be noted that there is some debate which Noether current 
to use for the gravitational radiation 
in order to derive the radiation formulae, often the 
Landau-Lifshitz pseudo tensor is argued for \cite{LL}, or one proceeds as 
sketched above, perturbs the Lagrangian to second 
order and then computes the canonical (Noether) energy momentum tensor. 
In the present situation 
we have a natural candidate as derived above and it would be interesting to 
compare the different formalisms at least at second order.\\
2.\\
Note that the Schwarzschild stationary and Gullstrand-Painleve\'e 
free falling clocks 
at distances far away from the black hole tick at the same rate but 
they have a radius dependent off-set. Far away from the black hole,
during a short amount of time the GP observer is barely picking up speed 
if previously at rest and thus the radius is approximately constant 
during observation. An astronomer on Earth can be argued to be rather 
free falling towards a black hole rather than being stationary as one 
cannot prevent Earth from being attracted to the black hole. However
for both observers the time passed during observation is the same
to high accuracy.\\
3.\\
When the set $S(\tau)$ itself is time dependent then the radiation 
formula for $\frac{d}{d\tau}\; E(S(\tau))$ 
must be corrected by the term that takes the 
time change of $S(\tau)$ into account. The application would be a trapped 
region $S(\tau)$ with $\tau$ dependent profile function 
$\rho_\tau:\; S^2\; \to \mathbb{R}_+$ and 
$S(\tau)=\{(r,\Omega),\; r\le \rho_\tau(\Omega),\; \Omega\in S^2\}$.
Then the correction term for the time dependence of 
$E(S(\tau))=\int\; d\Omega\; 
\int_0^{\rho_\tau(\Omega)}\;dr\; h$ is given 
by
\be \label{9.11}
\int_{S_2}\; d\Omega\; 
\dot{\rho}_\tau(\Omega)\; h(r=\rho_\tau(\Omega),\Omega) 
\ee
which due to the dependence of $\rho_\tau$ on the gravitational 
radiation variables $X,Y$ is by itself a rather complicated functional 
of the true degrees of freedom.

\section{Conclusion and Outlook}
\label{s10}

To define interacting, gauge invariant black hole perturbations of geometry 
and matter is a complicated topic for which many conceptual and 
technical questions have to be answered. In this paper we have attempted 
at a concrete proposal. The basic idea is to divide the problems into several 
steps. The first step consists in disentangling gauge invariance 
from perturbation theory already in the classical theory. Thus one 
first constructs the non-perturbative reduced phase space (true degrees 
of freedom) and physical 
Hamiltonian and then perturbs it directly in terms of the gauge invariant 
perturbations which are defined as those true degrees of freedom 
which are non (sherically or axi) symmetric. 
In order that one has access to both the black hole 
interior and the exterior at the same time puts restrictions on 
the choice of the true degrees of freedom (equivalent to the choice of 
a congruence of observers) and therefore we have opted for the 
Gullstarnd-Painlev\'e gauge. The second step consists 
in quantising the true degrees of freedom, both the symmetric  
and non symmetric ones, in suitable representations of the canonical 
commutation relations which are such that the physical Hamiltonian,
perturbed to the desired order of accuracy is at least a well-defined 
quadratic form. 

Since the dependence of the perturbed Hamiltonian 
on the symmetric true degrees of freedom typically 
is non-polynomial while it depends 
polynomially on the symmtric true degrees of freedom, 
one has to use different quantisation techniques for these two sets 
of degrees of freedom. This observation has already been made 
in quantum cosmology where one uses a so-called ``hybrid'' approach
\cite{Gomar, Hybrid} and uses a Narnhofer-Thirring type representation
\cite{LQC} 
for the symmetric (homogeneous) degrees of freedom while the non-symmetric 
ones are treated in a Fock representation. In \cite{FockQG} we have recently 
shown that one can in fact also use Fock representations for the 
non-polynomial dependence of the Hamiltonian if one carefully chooses 
the dense domain of the quadratic form, for instance as the span of the 
excitations of a coherent state concentrated on a 3-metric of Euclidian
signature. This stresses once more the importance of  
states that describe non-degenerate quantum metrics as emphasised 
in \cite{NonDeg,U3} in a different context.

In this paper we only have proposed a possible framework but of course the 
real task to describe black hole evaporation is still to be done. 
In our companion papers \cite{NT-SS,NT-RN,N-SA} 
we perform first steps.
These consist in showing that our approach reproduces the known 
classical 2nd order results due to Regge-Wheeler, Zerilli and Moncrief
\cite{24,25} outside the Horizon in the Einstein-Maxwell sector after 
one translates our reduced Hamiltonian into the Schwarzschild coordinates 
used in \cite{24,25}. However, this is just a consistency check. The real 
virtue of our method is that it enables to construct the reduced 
Hamiltonian also to higher than 2nd orders without the necessity to 
change the gauge invariant (true) degrees of freedom when increasing 
the order and thus 
to describe self-interactions among the symmetric and non-symmetric 
true degrees of freedom respectively as well as interaction (backreaction)
between them. This requires a better understanding of the Fock 
representation that is suggested by the 2nd order part of the 
reduced Hamiltonian. We have started this investigation in the present paper 
but the construction of the mode functions in a black hole white hole 
transition spacetime that we considered is an interesting mathematical 
challenge in itself and we certainly must know more of their properties
before we can proceed. This is even not under full control in an asymptotic 
end of a black hole. However, once this is done, one can study the quantum 
dynamics of interesting measures of evaporation such as the quantum 
area of the apparent horizon for which we have given a perturbative 
formula in the present paper and which can be quantised by 
the tools provided in \cite{FockQG}.
              
The methodology of the present manuscript can be readily applied also to 
cosmology or rotating black holes. But even for spherically symmetric 
black holes there is a huge amount of work to be done whose steps 
we described rather concretely in the present work.

\begin{appendix}

\section{Reduced phase space and gauge fixing of constraints with 
spatial derivatives}
\label{sa}

The reduced phase space of constrained system is conveniently obtained 
by solving the constraints for suitable momenta, imposing gauge fixing
conditions on the conjugate configuration degrees of freedom and 
determining the values of the smearing functions of the constraints 
from the corresponding stability conditions of the gauge imposed. 
The physical Hamiltonian is then determined by those three solutions 
of the constraints, gauge fixing conditions and stability conditions.

In this section 
we review that there is an important difference between A. constraints 
depending 
purely algebraically on the canonical fields and B. constraints involving 
spatial derivatives of these. Namely in case A. the reduced phase space is 
smaller by at least one canonical pair than in case B. This section serves 
as a preparation to understand in later sections 
why even in spherically symmetric 
vacuum GR there are two rather than one Dirac observable and why 
this is not in conflict with Birkhoff's theorem.\\
\\
To understand this in non-technical terms, consider a 1+1 dimensional 
field theory on $\mathbb{R}^2$ with canonically conjugate 
fields $(q(x),p(x)),\;
x\in\mathbb{R}$ and the following constraints
\be \label{a.1}
A.\;\;\;C(x)=p(x),\; 
B.\;\;\;C(x)=p'(x)
\ee
where $(.)'=\frac{d}{dx}$. 

For a field theory it is not sufficient just 
to state that $\{p(x),q(y)\}=\delta(x,y)$, a complete characterisation 
of the phase space must also specify the space of functions to which 
$q,p$ belong which among other things involves their decay behaviour at 
spatial infinity. One of the conditions is that the symplectic 
structure 
\be \label{1.1a}
\Omega=\int_{-\infty}^\infty\; [\delta p](x)\wedge [\delta q](x)
\ee
converges where $\delta$ is the functional exterior derivative. For instance 
we could impose that both $\delta q, \delta p$ decay as $1/x$ which allows 
both $q,p$ to asymptote to fixed values $q_\pm,p_\pm$ at $x=\pm \infty$
which are not variable on the phase space. If on the other hand $p_\pm$ 
is considered a variable on the phase space and $p=p_\pm+O(1/x)$ then 
$q_\pm$ must not be a variable on the phase space and we need the 
integral of $\delta q$ to converge which either requires a stronger fall
off condition on $\delta q$, say as $1/x^2$ or an asymptotic 
parity condition, e.g. that $\delta q= 
\delta c/x$ in leading order where $c$ is another variable on the phase 
space. The leading order then vanishes when defined as a principal value 
integral. This specification of the decay behaviour has also consequences 
for the treatment of the constraints.       

We smear the constraints with test functions $f$ which are treated as
constants on the phase space ($\delta f(x)\equiv 0$), that is, we consider 
$C(f):=\int_\mathbb{R}\; dx \; f(x)\; C(x)$. The exterior derivative 
of $C(f)$ enters the computation of the Poisson brackets
\be \label{a.1b}
A.\;\;\;[\delta C(f)]=\int\; dx\; f(x)\; [\delta p](x),\;\;\;
B.\;\;\;[\delta C(f)]=-\int\; dx\; f'(x)\; [\delta p](x)-[\delta B(f)],\;\;\;
B[f]:=-[f(x) p(x)]_{x=-\infty}^\infty
\ee
For model A. the functional $C(f)$ is functionally differentiable without 
any condition on $f$ while for model B. it is functionally differentiable 
if and only if $f_+ \delta p_+ - f_- \delta p_-=0$ where 
$f_\pm=f(\pm \infty)$. This is automatically the case if $p_\pm$ do not 
vary on the phase space (in which case we can drop $B(f)$ altogether) 
or if e.g. $f$ decays at both infinities. We can however
define for both models 
\be \label{a.1c}
A.\;\;\; H(f):=C(f),\;\;\;
B.\;\;\; H(f):=C(f)+B(f)
\ee
which are functionally differentiable with no condition on $f$ in both 
cases no matter what the decay behaviour of $p$ is. 
In contrast to model A. in model B. the functional $H(f)$ is different 
from the functional $C(f)$ unless $B(f)=0$. We call canonical 
transformations generated by $H(f)$, with $f$ such that 
$B(f)=0$, {\it gauge transformations} because $C$ is the constraint
and not $H$. We call canonical 
transformations generated by $H(f)$ with $f$ such that $B(f)\not=0$ 
{\it symmetry transformations}. For model A. there is no difference between 
the two because $B(f)\equiv 0$. 

For both models the unconstrained phase space is infinite dimensional. The 
constraint surface is the kernel of the constraint $C(x)=0$ for all $x$.
The reduced phase space is the constraint surface with points identified 
which lie on the same gauge orbit. In the present case, the reduced phase 
space is very simple to compute. In both models $p$ is gauge invariant 
because $C(f)$ depends only on $p$. 

In model A. any $f$ corresponds to a 
gauge transformation and the gauge transformation of q is 
\be \label{a.1d}
[\Delta_f q](x)=\{C(f),q(x)\}=f(x)
\ee
We are allowed to pick $f$ from the same function space that $q$ belongs 
to and see that $q(x)$ can be gauged to any value, say $q_\ast(x)=0$ for 
all $x$.
As the constraint requires $p(x)=0$ for all $x$ we see that the reduced 
phase space is the single point $p\equiv q\equiv 0$. The gauge 
$q_\ast=0$ is also complete, i.e. there are no non-trivial 
stability transformations that preserve $q=0$ as $\{H(f),q(x)\}=f(x)=0$
imposes $f(x)=f_\ast(x)=0$ for all $x$. 

In model B. we have to be more careful. The constraint $C(x)=0$ now only 
imposes that $p(x)=M$ is a spatial integration 
constant but it is allowed to be considered a 
variable on the phase space. Thus we have in particular $p_\pm =M$. Thus 
$B(f)=-M(f_+- f_-)$ on the constraint surface. Thus the weakest condition 
we can impose for $f$ to define a gauge transformation is that $f_+=f_-$. 
For $f_+\not=f_-$ we obtain a symmetry transformation. Then a possible
condition on the decay behaviour of $q$ is that $\delta q$ decays as 
$[\delta c]/x$ in leading order so that we have 
odd parity conditions at infinity. The 
gauge transformation of $q$ is  
\be \label{a.1e}
[\Delta_f q](x)=\{H(f),q(x)\}=-f'(x)
\ee
To gauge a given $q(x)$ to zero we must solve $\Delta_f q=-q$ i.e.
$f'=q$ which is solved by 
\be \label{a.1f}
f_q(x)=f_- + \int_{-\infty}^x\; dy\; q(y)
\ee
However, unless 
\be \label{a.1g1}
Q:=\int_{-\infty}^\infty\; dy\; q(x)
\ee
equals zero, the function (\ref{a.1f}) does not correspond to a gauge 
transformation because $f_+\not= f_-$. 
It follows that $q(x)$ cannot 
be gauged to zero for all $x$. Indeed for a gauge transformation we 
have 
\be \label{a.1g}
[\Delta_f Q]=\{H(f),Q\}=-\; \int_{-\infty}^\infty\; dy f'(y)=-[f_+ - f_-]=0
\ee
whence $Q$ is gauge invariant, i.e. a Dirac observable. We require 
the function space space of $q$ to be such that $Q$ is well defined. 
Pick $\infty>L$. We can 
then set for $x\in (-\infty,L]$ 
\be \label{a.1h}
f_q(x)=f_- + Q - \int_x^\infty\; dy\; q(y)
\ee
and interpolate smoothly between $f_q(L)$ and $f_-$ in $x\in (L,\infty)$. 
As this can be done for any $L<\infty$ we see that $q(x)$ for $x<\infty$ 
is pure gauge. Note that $q(x)=0$ for $x<\infty$ does not fix the gauge 
completely because we have $f_+=f_-$ possibly non vanishing. Equivalently,
let $w(x)$ be a fixed weight function belonging to the function space of $q$ 
with $\int_\mathbb{R}\; dx \;w(x)=1$. Then we can gauge $q(x)$ to 
$q_\ast(x)=Q\; w(x)$ for 
all $x$ because the required gauge transformation is now given by 
\be \label{a.1i}
f_q(x)=f_- + \int_{-\infty}^x\; dy\; [q(y)-Q \; w(y)]
\ee
which satisfies $f_q(\infty)=f_-=f_q(-\infty)$. Again the constant 
$f_-$ is unspecified corresponding to a residual gauge freedom which 
is also clear from the stability condition $\delta_f (q-Q w)=-f'=0$ i.e.
the solution (\ref{a.1i}) is only unique up to adding a constant 
$f_q\to f_q+c$ which corresponds to shifting $f_-$.
%$f=f_\ast=\kappa=const.$

On the other hand for a symmetry transformation we 
have 
\be \label{a.1j}
[\Delta_f Q]=\{H(f),Q\}=-\; \int_{-\infty}^\infty\;dy f'(y)=-[f_+ - f_-]\not=0
\ee
Since $\{p(x),Q\}=1$ for any $x$ and as we have $p(x)=M$ on the 
constraint surface
it follows that (\ref{a.1j}) has the canonical generator 
\be \label{a.1k}
H=-(f_+ - f_-)M=:\kappa\; M=B(f)
\ee
called the physical Hamiltonian. The more systematic way to obtain $H$ 
is the condition that for any function $F=F(M,Q)$ we require 
\be \label{a.1l}
\{H,F\}=\{H(f),F\}_{f=f_\ast,C=0,q=q_\ast}=-(f^\ast_+ - f^\ast_-)\;
\frac{\partial F}{\partial Q}
\ee
where $f=f^\ast$ is the general solution of the stability condition 
\be \label{a.1m}
\{H(f),[q-q_\ast](x)\}=\{H(f),q(x)\}-\{H(f),Q\}\;w(x)
=-f'(x)+[f_+ - f_-]\;w(x)=0
\ee
which has the solution 
\be \label{a.1n}
f^\ast(x)=f_- +(f_+ - f_-)\int_{-\infty}^x\; dy\; w(y)
\ee
depending on two free parameters $f_+,f_-$. Thus the general solution 
of the stability condition is a symmetry transformation when 
$f_+\not= f_-$ and it is only sensitive on the single parameter 
$\kappa=-(f_+ - f_-)$. Thus the systematic analysis reproduces 
(\ref{a.1k}).\\
\\
To summarise, the innocent looking spatial derivative of canonical 
variables in constraints has drastic 
consequences on the reduction of the system: the reduced phase space 
is augmented by canonical pairs and there is a residual 
transformation freedom 
even after all possible gauge freedom has been exploited and the 
gauge has been maximally fixed.
This residual transformation freedom parameter finds its way into the  
physical Hamiltonian on which it depends linearly. That parameter therefore
can be interpreted as the clicking rate of the clock that measures time.

\section{Consequences for Black Hole physics}
\label{sb}

For spherically symmetric spacetimes, the existence of a {\it pair}
of Dirac observables rather than just the black hole mass has been 
discovered, to the best knowledge of the author, for the first time in 
\cite{Kastrup} which provides a complete quantum theory of  
spherically symmetric black holes in terms of complex connection variables.
In \cite{KucharSS} a similar analysis is performed in terms of ADM variables.
In view of the previous section and as we will review below, the 
origin of the second Dirac observable 
is due to the fact that the constraints of spherically symmetric gravity 
involve momentum derivatives. This means that the constraints cannot 
be used to completely gauge away all configuration degrees of freedom and 
that there is an integration constant when solving the constraints for 
the momenta. These degrees of freedom are the analogs of $M,Q$ above.

How can this be reconciled with Birkhoff's theorem which states that every
spherically symmetric vacuum solution is static and completaely characterised 
by a {\it single} degree of freedom, namely the mass of the black hole      
\cite{HE}? As we will show, the freedom corresponding to 
$Q$ can be considered as associated with a 1-parameter family of
temporal diffeomorphism corresponding to the choice of the time coordinate 
which is supposed to coincide with the Cartesian time coordinate at 
spatial infinity. If one considers this freedom as gauge degree of 
freedom as it is customary in the Lagrangian framework, then indeed 
one can discard of $Q$. On the other hand, in the Hamiltonian 
framework one is instructed not to consider that freedom as gauge.

We repeat here from the main text why this is significant:
Suppose there would be no variable conjugate to the mass at all. 
Then, because the mass has 
vanishing Poisson brackets with the matter 
and the 
gravitational multiplole degrees of freedom, it would have vanishing Poisson 
brackets with the reduced Hamiltonian, which is a constant of motion, 
and thus would be a constant of motion 
itself. It would be present also after all Hawking processes have ceased and 
thus could also be called {\it remnant} mass. If on the other hand we respect 
the existence of that Dirac observable conjugate to the mass, then 
the gauge fixing condition must be consistent with its existence. This
can be granted, for instance, when the gauge fixing condition keeps a one 
parameter freedom 
which is able to capture the existence of the conjugate variable, let us 
call it $Q$. 
%Being conjugate to a mass or energy, it has the interpretation of a physical
%time variable. Then $Q$
can also find its way into the reduced Hamiltonian and now both $Q,M$
can change in time and in particular and we can have 
{\it backreaction} from the multipole and matter degrees of freedom to $Q,M$.

We will now give a self-contained exposition of how the time variable comes 
into existence. The mechanism at work is a 1-parameter family of purely 
temporal 
diffeomorphisms that we use to pull-back the Schwarzschild metric away from
the
GPG. This leaves the radial coordinate intact but changes lapse, shift and 
the radial-radial component of the spatial metric. For instance,
we can choose that 
1-parameter family such that the pulled back metric deviates from the 
exact GPG only in an arbitrarily small neighbourhood of the core $r=0$, 
say in a Planck length neighbourhood which is the region of spacetime in 
which we do not trust classical GR anyway, in fact from a strictly classical 
point of view the point of view, $r=0$ or a neighbourhood of it should be 
cut out 
from the physical spacetime. The parameter $Q$
is directly 
determined by that temporal diffeomoprphism. If one considers that 
diffeomorphism as a {\it gauge} transformation then $Q$
would be 
considered as a gauge degree of freedom. If one considers that diffeomorphism 
as a {\it symmetry} transformation then $Q$ would 
be considered as a Dirac observable. It is the first point of view that is 
taken in Birkhoff's theorem coming from a Lagrangean point of view,
thereby explaining why one only has the mass 
parameter as an observable. It is the second point of view which is taken 
coming from the Hamiltonian point of view. 

We thus follow the Hamiltonian
path in order to keep the possibility open that also the remnant mass can 
change dynamically. In the main text 
and appendix \ref{se} we show that one can also have $Q$ 
existent without that it leaves a trace in the physical Hamiltonian 
by exploiting that the expression that defines $Q$ requires regularisation
which introduces the required one parameter freedom without implementing 
it into the gauge fixing condition. Still evaporation is not excluded because 
$2M$ is not the event horizon or apparent horizon when there is radiation
present.

\subsection{The reduced phase space of spherically symmetric 
vacuum GR}
\label{sb.1}
     
Proceeding to the details, following the notation of section \ref{s4}
in the purely spherically symmetric sector we have the following 
two constraints prior to any gauge fixing of $2\;q_0=\Omega^{AB} q_{AB},
q_3=q_{33}$ with conjugate momenta 
$p^0=P^{AB}\Omega_{AB}/\omega,\;p^3=P^{33}/\omega$ 
\be \label{b.1}
v_3:=V_3/\omega=p^0\; q_0'+p^3\; q_3'-2(q_3 p^3)',\;
v_0:=V_0\sqrt{\det(q)}/\omega^2=
\frac{1}{2}(q_3 p^3)^2-(q_0 p^0)(q_3 p^3)-\det(q) R/\omega^2
\ee
By introducing $p^0=(p^3 q_3'+2[p^3]'q_3)/q_0'$ into $V_0$ relying 
on $q_0'>0$ as $\sqrt{q_0}$ should be the radial coordinate up to a 
radial diffeomorphism, we have 
\be \label{b.2}
v_0=-\frac{q_3 \;q_0^{3/2}}{q_0'}\;[\frac{q_3 [p^3]^2}{2\;q_0^{1/2}}]'-
\det(q) R/\omega^2
\ee
By working out the Ricci scalar explicitly for the non-vanishing components 
$q_{33}, q_{AB}=q_0/2\;\Omega_{AB}$ one finds after a longer calculation
\be \label{b.3}
v_0= 
-4\;\frac{q_3 \;q_0^{3/2}}{q_0'}\;\{\frac{q_3 [p^3]^2}{4\; q_0^{1/2}}
+q_0^{1/2}[1-(\frac{\sqrt{q_0}'}{\sqrt{q_3}})^2]\}'
\ee
In the exact GPG $q_0=r^2, q_3=1, p^3=2\sqrt{2Mr}$ we see that the curly 
bracket 
is just twice the black hole mass $M$. 
Without imposing any gauge, let us call this 
quantity $m$. Then a non-trivial calculation confirms that we can pass 
to new canonical coordinates 
\be \label{b.4}
\gamma=\sqrt{q_{33}},\;p_\gamma=-2\gamma p^3,\;
\delta=\sqrt{q_0},\;p_\delta=2\delta\; p^0
\ee
and after that 
\be \label{b.5}
m=\frac{p_\gamma^2}{16\delta}+\delta(1-[\frac{\delta'}{\gamma}]^2),\;
p_m=\frac{\gamma\;p_\gamma}{2\delta\Phi},\;\Phi=1-\frac{m}{\delta},\;
X=\delta,\;
p_X=p_\delta-(\gamma\; p_\gamma'+m' p_m)/\delta'
\ee
which enables us to express the constraints in the equivalent form
\be \label{b.6}
\tilde{v}_3=p_m\; m' + p_X\; X',\; \tilde{v}_0=m'
\ee
which simply enforce $p_X=0,\; m'=0$. 

The transformation (\ref{b.4}) is easy but (\ref{b.5}) is non-trivial to check,
see \cite{Kastrup, KucharSS}.
A short-cut is as follows:\\
The spatial diffeomorphism constraint $v_3=\delta'\;p_\delta-\gamma\;p_\gamma'$
identifies 
$\delta,p_\gamma$ radial scalars and  
$p_\delta,\gamma$ as radial scalars of density weight one. Thus 
$p_\delta/\delta',\gamma/\delta'$ as radial scalars. For any radial scalar
$F$ the function $\hat{F}(s)=F(r=\delta^{-1}(s))$ is spatially diffeomorphism
invariant for any value of $s$ relying on $\delta:\mathbb{R}\to 
\mathbb{R}; r\mapsto \delta(r)=s$ to be a diffeomorphism (if one 
wants to consider one asymptotic end only, one just has to resrict
$r$ to positive or negative values respectively). Explicitly
one has 
\be \label{b.6a}  
\hat{F}(s)=\int_{-\infty}^\infty\; dr\; \delta'(r)\; \Delta(\delta(r),s)\;
F(r)
\ee
where $\Delta$ is the $\delta$ distribution.
This  makes it possible to compute the Poisson brackets among 
$\hat{\gamma}(s),\hat{P}_\gamma(s)$ and with $v_3$. That computation 
shows that $\{\hat{P}_\gamma(s),\hat{\gamma}(\tilde{s})\}=
\delta(s,\tilde{s})$ and that they have vanishing Poisson brackets with 
$v_3$. The function $\hat{m}(s)$ can now be epressed just in terms of 
these
\be \label{b.6b}
\hat{m}(s)=\frac{\hat{P}_\gamma^2(s)}{16s}+s\;[1-\hat{\gamma}^{-2}(s)]
\ee
and the Hamiltonian constraint becomes $\hat{m}'(s)=0$.
As (\ref{b.6b}) has no derivatives with respect to $s$ we have 
$\{\hat{m}(s),\hat{m}(\tilde{s})\}=0$. To construct a momentum $\hat{P}_m$
conjugate to $\hat{m}$ we solve (\ref{b.6b}) for 
$\hat{P}_\gamma(s)=4\;s\;\sigma \sqrt{\frac{\hat{m}(s)}{s}
+\hat{\gamma}^{-2}(s)-1}$ with $\sigma=\pm 1$ and 
plug it into the symplectic potential 
$\hat{\Theta}=-\int\; ds\;\hat{\gamma}\;\delta\hat{P}_\gamma$. Then we take 
the functional exterior derivative 
\be \label{b.6c}
\hat{\Omega}=\delta\hat{\Theta}=
-\int\; ds\;\delta\hat{\gamma}\;\wedge\;\delta\hat{P}_\gamma= 
-\int\; ds\;[\frac{\partial\hat{P}_\gamma}{\partial \hat{m}}\;
\delta\hat{\gamma}]\;\wedge\;\delta{m}
\ee
so that $\delta\hat{p}_m$ at fixed $\hat{m}$ must be 
$-\frac{\partial\hat{P}_\gamma}{\partial \hat{m}}\;\delta\hat{\gamma}$.
This is indeed solved by 
\be \label{b.6d}
\hat{P}_m=2\sigma\frac{\sqrt{1-\hat{\gamma}^2\hat{\Phi}}}{\hat{\Phi}},\;\;
\hat{\Phi}=1-\frac{\hat{m}}{s}
\ee
modulo addition of a function that just depends on $\hat{m}$. 
It is instructive to check by hand that $\hat{P}_m,\hat{m}$ are conjugate.
For this one needs to plug (\ref{b.6b}) into (\ref{b.6d}) which yields
\be \label{b.6e}
\hat{P}_m=\frac{\hat{\gamma}}{2}
\frac{[\hat{P}_\gamma/s]}{\hat{\gamma}^{-2}-
\frac{1}{16}[\hat{P}_\gamma/s]^2}
\ee
~\\ 
The two constraints $p_X(r)=0$ and $m'(r)=0$ respectively 
bring us exactly into the situation of models A and B of the previous
section. 
In terms of these canonical coordinates and setting $p=m,q=-p_m$
we have the symplectic structure 
\be \label{b.8}
\Theta=\int\; dx\; 
[p_{X}(x)\delta X(x)+p(x)
%[p_{\tilde{X}}(x)\delta \tilde{X}(x)+\tilde{p}(x)\delta \tilde{q}(x)]
\delta q(x)]
\ee
Note that our aim is to solve the original constraints for $p^0, p^3$ which 
are encoded in $p_X, m$ and not $p_X, p_m$ which is why we switched the 
roles of momentum and configuration coordinates with respect to $m$.

By the results of the previous section, the coordinatisation of the 
reduced phase of the equivalent set of constraints 
%$p_{\tilde{X}}(x)=0,\; \tilde{p}'(x)=0$ 
$p_X(x)=0,\; p'(x)=0$ 
is now very transparent: 
%$\tilde{X}$ 
$X$ is pure gauge 
while 
%$\tilde{q}$ 
$q$ must carry one degree of freedom $Q$. Also 
%$p_{\tilde{X}}(x)=0$ 
$p_X(x)=0$
fixes 
%$p_{\tilde{X}}$ 
$p_X$ 
completely while 
%$\tilde{p}'(x)=0$ 
$p'(x)=0$
retains one parameter $P$ (integration constant)
%$P:=\tilde{p}(-\infty)=m(0)$ 
%P=p(-\infty)$
as unconstrained. The reduced phase space of the system is thus 
2-dimensional encoded by $Q,P$ and the reduced Hamiltonian is $H=\kappa P$
up to a constant $\kappa$.
Thus up to a constant $\kappa$,
the physical Hamiltonian is the ADM mass or energy as one would have expected.
The conjugate variable $Q$ thus plays the role of an intrinsic time variable
whose ticking rate is given by $\kappa$ which can be chosen to be any value 
(reparametrisation of the coordinate time).

\subsection{Gauge Fixings consistent with the existence of Q}
\label{sb.2}

The task left over to do is to pick a suitable set of gauge fixings 
which give rise to $Q$ i.e. which yields 
\be \label{b.9}
Q=\int_{-\infty}^\infty\; dr\; q(r)
%\tilde{q}(x)=
=-\int_{-\infty}^\infty\; dr\; p_m(r)
\ee
We have the freedom to subtract from $Q$ an arbitrary function 
of $P$ because that will change the reduced symplectic potential 
$\Theta=P\delta Q$ only by a total differential, hence we require 
\be \label{b.9b}
Q=-\int_0^\infty\; dr\; [p_m(r)-f_P(r)]
\ee
Without specifying $q_{33}=\gamma^2, q_0=\delta^2$ for the moment, we 
can solve $m'(r)=0$ as $m(r)=P$ and solve for $p_m$ in terms of 
$\gamma, \delta, P$ which yields  
\be \label{b.10}
p_m=\frac{\gamma p_\gamma}{2\delta\Phi}
=\frac{2\gamma}{\Phi}\;\frac{p_\gamma}{4\delta}
=\frac{2\gamma\sigma}{\Phi}\;\sqrt{[\frac{\delta'}{\gamma}]^2-\Phi},\;\;
\Phi=1-\frac{P}{|\delta|}
\ee
where $\sigma$ is the sign freedom left over from solving 
the quadratic equation $m=P$ for $P_\gamma$ and we note that when solving 
that equation one finds that the argument of the square root in (\ref{b.10})
is constrained to be non-negative.

We note that in the exact GPG we would choose $\delta=r,\; \gamma=1$ which 
would give $[\frac{\delta'}{\gamma}]^2-\Phi=\frac{P}{|r|}\ge 0$ constraining
$P$ to be positive. In order to give $Q$ an independent value
that gauge must be relaxed by a one parameter family of gauges.
Among the many possible choices, we will discuss here two simple
possibilities.\\
\\
{\bf Range restriction}\\
This is motivated by section \ref{sd.3} where we construct a non-singular
spacetime by gluing a black hole and a white hole spacetime
along a cylinder $|r|=l<R=2M$. We then have for $q=-p_m,\;f_P=0$
\be \label{b.10a}
Q=\int_{|r|\ge l} \; dr\;q(r)
=\int_{-l}^l\; dr\; q(r)+\int_\mathbb{R}\; dr\; q(r)
\ee
The second integral in (\ref{b.10a}) is ill-defined as it stands
and must be defined by a limiting procedure (principal value). 
(Alternatively we can simply drop it because it just depends on $M$ 
thus drops out of the symplectic structure).
The
calculation is carried out in section \ref{sd.6}. It can be given
any value and thus can be used to define $Q$ for $l\equiv 0$, see
that section for details. Here we use that regularisation
freedom to make the second integral vanish. The first integral
then gives
\be \label{b.10b}
Q=-4\sigma\int_0^l\; dr\; \frac{\sqrt{R/r}}{1-R/r}
=-4\sigma\; R\;
[2\sqrt{\frac{l}{R}}-\ln
(\frac{1+\sqrt{\frac{l}{R}}}{1-\sqrt{\frac{l}{R}}})]
\ee
Since the reduced Hamiltonian is $H=\kappa\;M, \kappa>0$ and $M,Q$
are conjugate this gives the equations of motion
$\dot{Q}=\kappa>0,\; \dot{M}=0$, hence the variable $Q$ is eventually positive
which means that $\sigma.0$ so that for large times $\tau$ and
with $c=\kappa/4$
\be \label{b.10c}
\frac{c\tau}{R}=
-2\sqrt{\frac{l}{R}}+\ln
(\frac{1+\sqrt{\frac{l}{R}}}{1-\sqrt{\frac{l}{R}}})
\ee
Since this diverges with $\tau$, $l(\tau)\to R-$. Thus for large
$\tau$
\be \label{b.10d}
l=R\; {\sf th}^2(\frac{c\tau}{R})
\ee
i.e. the solution approaches exponentially fast the {\it Einstein
-- Rosen bridge solution} $|r|\ge R$ with two asymptotic ends. Thus
for late times the scond parameter $l$ gets frozen to $R$ and 
the metric depends only on a single parameter $R$. However, in this 
implementation, the interior of the black hole is cut out from the 
spacetime. We will therefore not consider this possibility further
in what follows.\\
\\
{\bf Local deviation from exact GPG}\\
We generalise the gauge to $\delta=r, 
\gamma=\sqrt{1+\Delta}$
where $\Delta$ will be further specified below and which will be non-vanishing
only in compact subsets of $\mathbb{R}^+$ and carries the information 
about $Q$. In particular $\Delta>-1$ in order that the 
metric stays non-degenerate. As $P$ is constrained to be a 
constant, this still imposes $P>0$. The motivation for that 
particular generalisation within the Hamiltonian analysis is that the 
spatial diffeomorphism constraint generates radial reparametrisations and 
therefore we can always choose $\delta=r$.
On the other hand the relation between temporal spacetime 
diffeomorphisms and the gauge transformations generated by the Hamiltonian 
constraint is more tricky: These two notions only coincide {\it on shell} 
i.e. when the equations of motion (e.g. vacuum Einstein equations) hold.
We will not violate those equations at all, we simply pull back the exact
GPG form of the Schwarzschild metric by a temporal diffeomorphism which 
encodes $\Delta$ and that pulled back metric then still solves the Einstein
equations. In particular we will will not at all contradict Birkhoff's theorem
because in contrast to the Hamiltonian picture, 
in the Lagrangian picture one considers {\it all} diffeomorphisms as a 
gauge transformation. 

In that parametrised GPG 
(PGPG -- by $Q$) the spatial line element reads 
\be \label{b.11}
q_{33}\;dr^2+q_0\; \Omega_{AB}\; dy^A\; dy^B= 
\gamma^2\; dr^2+r^2\; \Omega_{AB}\; dy^A\; dy^B
\ee
Then (\ref{b.10}) becomes 
\be \label{b.12}
p_m
=\frac{2\sigma}{\Phi}\;\sqrt{1-\gamma^2\Phi},\;\;
\ee
In order that the square root be well defined it is sufficient to require that 
$\gamma^2\Phi\le 1$ for $|r|>P$. This will be in particular the case for 
$\gamma^2\le 1$ for all $r$. We will choose $\Delta\not=0$ for some subset 
of $[-P,P]$ which thus satisfies this requirement. Accordingly the integral 
of $p_m$ approaches for large $|r|$ the function $|r|^{-1/2}$ and thus would 
diverge. We thus use the freedom to add to $Q$ a function that depends just 
on $P$ which does not change the reduced symplectic structure and makes 
the integral defining $Q$ converge. The natural choice is to subtract from
(\ref{b.12}) its value in the exact GPG (i.e. for $\gamma=1$), 
i.e. we define    
\be \label{b.13}
Q
=-2\sigma\int_{-\infty}^\infty\; 
\frac{\sqrt{1-\gamma^2\Phi}-\sqrt{1-\Phi}}{\Phi} 
=4\sigma\int_0^\infty\; 
\frac{\Delta}{\sqrt{1-[1+\Delta]\Phi}+\sqrt{1-\Phi}}
\ee
where we focus on even functions $\Delta(r)=\Delta(-r)$ so that 
we can restrict to the positive axis.
The integral is now confined to the support of $\Delta$ which will lie 
in the region where $\Phi<0$. We will chose $\Delta\ge 0$ there so that 
we obtain the interpretation
\be \label{b.14}
\sigma={\rm sgn}(Q)
\ee
Let now $w\ge 0$ be a function of compact support in $[0,P]$ such that 
$\int_0^\infty\; dr\;w =1$ then we pick $\Delta$ such that 
\be \label{b.15}
\frac{w\;|Q|}{2}=\frac{\Delta}{\sqrt{1-[1+\Delta]\Phi}+\sqrt{1-\Phi}}
\ee
This can be solved for $\Delta$ and yields either $\Delta=w=0$ or 
\be \label{b.16}
\Delta=|Q|\;w\;[\sqrt{1-\Phi}-|Q|\;w\;\Phi/4]
\ee
which is manifestly non negative as $\Phi<0$ in the support of $w$. 
If we want the support of $w$ to be independent of the value of $P$ we
can restrict it to the interval $[0,\ell],\; l=\epsilon \ell_P]$ 
with $\epsilon\le 1$ 
because a Planck size black hole mass is believed to be outside of regime 
of classical GR and within classical GR it is well motivated to cut out 
the region $r\le \ell_P$ from the physical manifold. Alternatively we 
may pick $\ell={\sf min}(\epsilon\ell_P,P/2)$. Then with 
$\chi_{[0,\ell]}$ the characteristic function of that interval
we make the Ansatz 
\be \label{b.17}
w(r)=\chi_{[0,\ell]} f(r),\; f\ge 0 
\ee
so that 
\be \label{b.18}
\Delta=|Q|\;\chi_{[0,\ell]}\;
[\frac{f}{r^{1/2}}\;\sqrt{P}+|Q|\;\frac{f^2}{r}\;(P-r)/4]
\ee
This is regular at $r=0$ e.g. for the choice $f(r)=h\sqrt{r}$ for some 
height amplitude $h$ and yields  
\be \label{b.18a}
\Delta=|Q|\;\chi_{[0,\ell]}\;
[h\;\sqrt{P}+|Q|\;h^2\;(P-r)/4]
\ee
The height $h$ is fixed by the requirement
\be \label{b.19}
\int_0^\infty\; dr\; w=h\; \int_0^{\ell}\; dr\;
\sqrt{r}=\frac{2h}{3} [\ell]^{3/2}=1
\ee
so that we finally obtain 
\be \label{b.19b}
\Delta=\frac{|Q|}{\ell}\;\chi_{[0,\ell]}\;
[\frac{3}{2}\;\sqrt{\frac{P}{\ell}}+\frac{9}{16}\;
\frac{|Q|}{\ell}\;\;\frac{P-r}{\ell}]
\ee
in particular
\be \label{b.20}
\Delta(0)=
\frac{3}{2}\;z^{1/2}+\frac{9}{16}\;z,\;\;z=\frac{P \; Q^2}{[\ell]^3}
\ee
which is also the maximal value that $\Delta$ can take for this particular 
gauge. Note that the range of $Q$ is all of $\mathbb{R}$ which is compatible 
with the equation of motion $\dot{Q}=\kappa$ that follows from the reduced 
Hamiltonian $H=\kappa P$. The non-differentiable step function can be mollified 
to obtain a smooth function which would yield qualitatively similar 
formulae.\\
\\
This proves that a suitable gauge for $\gamma$ exists that produces a 
given value of $Q$ independent of the value of $P$ and that
deviates from the exact GPG only very {\it locally}, i.e. 
$\gamma^2=1$ except for a 
neighbourhood of zero of 
at most of Planck size behind the event horizon.

\subsection{Relation between existence of Q and temporal diffeomorphisms}
\label{sb.3}

To relate a deviation from the strict GPG 
 to a spacetime diffeomorphism we write the Schwarzschild metric in exact GPG 
\be \label{4.a1}
ds^2=-[1-\frac{2M}{|r|}]\;d\tau^2+2\sqrt{\frac{2M}{|r|}}\;dr\;dt\; + dr^2 +
r^2\; \Omega_{AB}\; dy^A\; dy^B
\ee
and pull it back by a temporal diffeomorphism
\be \label{4.a2}
r=\rho(\tilde{\tau},\tilde{r}):=\tilde{r},\; \tau=T(\tilde{\tau},\tilde{r})
\ee
and rewriting (\ref{4.a1}) in terms of the coordinate $\tilde{\tau}$. That 
reparametrised metric still solves the Einstein equations, no matter 
what the function $T$ is, as long as $\partial_{\tilde{t}} T>0$, since 
then we have 
just carried out a diffeomorphism (\ref{4.a2}) as 
$\det(\partial(\tau,r)/\partial(\tilde{\tau},\tilde{r}))=
\partial_{\tilde{t}}T(\tilde{\tau},\tilde{r})$. That the 
pulled back metric still solves the Einstein equations 
is obvious from its tensorial character but can of course also 
be verified by hand.

The spatial 
part of the metric pulled back by this diffeomorphism starting from 
GPG $q_3=:=\gamma^2=1,q_0=r^2$ becomes 
\be \label{4.a3}
\tilde{q}_3(\tilde{\tau},\tilde{r})=
\tilde{\gamma}^2(\tilde{\tau}, \tilde{r})=1
-\Phi(\tilde{r})\;[\partial_{\tilde{r}}\;T(\tilde{\tau},\tilde{r})]^2
+2\sqrt{1-\Phi(\tilde{r})}\;[\partial_{\tilde{r}}\;T(\tilde{\tau},\tilde{r})],
\tilde{q}_A=0, \;\tilde{q}_0(\tilde{t},\tilde{r})=\;\tilde{r}^2   
\ee
where $\Phi(r)=1-2M/|r|$. It is therefore still flat in regions where 
$\partial_{\tilde{r}}\;T(\tilde{t},\tilde{r})=0$. If we compare 
(\ref{4.a3}) with $q_3=:\gamma^2=:1+\Delta, m:=2M$ from the previous 
subsection we obtain the relation 
\be \label{4.a4}
\frac{|Q|w}{2}=\partial_{\tilde{r}}\;T(\tilde{t},\tilde{r}) 
\ee
where $w$ has compact support in an at most Planck size neighbourhood of 
the origin
and $Q$ is the afore mentioned second Dirac observable which can in principle 
be an arbitary function of coordinate time $\tilde{t}$ and which is 
canonically conjugate to the mass $m$. As the physical Hamiltonian is just 
$m$ up to a constant, $Q$ is actually linear in $\tilde{\tau}$ on shell. 

The GPG lapse $\alpha=1$ and shift $\beta=\sqrt{1-\Phi}$ become upon 
pull-back
\be \label{4.a4b}
\tilde{\alpha}^2-\tilde{\beta}^2=
\Phi[\partial_{\tilde{\tau}}\;T(\tilde{\tau},\tilde{r})]^2,
\;\;\tilde{\beta}\tilde{\gamma}=[\sqrt{1-\Phi}\;
-\Phi\partial_{\tilde{r}}\;T(\tilde{t},\tilde{r})]\; 
\partial_{\tilde{\tau}}\;T(\tilde{\tau},\tilde{r})
\ee
This can be combined with (\ref{4.a3}) to 
\be \label{4.a4c}
(\tilde{\alpha}\tilde{\gamma})^2=
(\partial_{\tilde{\tau}}\;T(\tilde{\tau},\tilde{r}))^2
\ee
Here we have employed a general parametrisation of spherically symmetric 
spacetimes in 
coordinates $\tau,r$ given by
$g_{\tau\tau}=-\alpha^2+\beta^2,\; 
g_{\tau r}=\beta\gamma,\; g_{rr}=\gamma^2,\;
g_{AB}=\delta^2\; \Omega_{AB}$ and in the {\it radial gauge} $\delta=r$
chosen here the Einstein equations are equivalent to
(we drop the tilde again)
\be \label{4.a5}
1-\gamma^{-2}+\frac{\beta^2}{\gamma^2\alpha^2}=1-\frac{2M}{r},\;\;
\partial_t M=\partial_r M=0,\;\;
\partial_t \gamma=\frac{\beta}{\gamma\alpha}\;\partial_r[\gamma\alpha]
\ee
which can be combined and integrated to parametrise the metric as a function
of $\gamma$ 
\be \label{4.a6}
\beta=\sigma\alpha\sqrt{1-\gamma^2\Phi},\;
(\gamma\alpha)(\tau,r)=[\partial_\tau\;\hat{T}](\tau)-\sigma\;
\partial_\tau
[\int_0^r\;\;ds\;[\frac{\sqrt{1-\gamma^2\Phi}-\sqrt{1-\Phi}}{\Phi}](\tau,s)]
\ee
where $\hat{T}(\tau)$ is an arbitrary function of time and $\sigma$ a sign
which determines whether we consider the out/ingoing patch ($\sigma=\pm 1$).
If we compare this with the definition of $Q$ in the previous section and with 
(\ref{4.a4c}) we find 
\be \label{4.a7}
(\gamma\alpha)(\tau,r)=[\partial_\tau\; \hat{T}](\tau)+
[\partial_\tau Q]\;\frac{1}{2} [\int_{-\infty}^r\; ds\; w(s)],\;\;
\lim_{r\to\infty} \alpha\gamma=\partial_\tau[\hat{T}(\tau)+Q/2]=\partial_\tau T,
\;\; \beta\gamma\to \sigma\alpha\gamma\sqrt{1-\gamma^2\Phi}
\ee
As $\gamma\to 1$ at infinity, we see that $\alpha,\beta$ approach their 
exact lapse and shift value value in GPG up to a 
pure time
reparametrisation. By the result of the previous section 
$\kappa:=\alpha(\infty)-\alpha(-\infty)=\dot{Q}$ agrees with the 
Hamiltonian equations of motion.

Thus the physical meaning of the 
time function $Q$ has been
worked out: It is canonically conjugate to $m$ and its clicking rate 
at infinity coincides both with the asymptotic lapse value and the 
asymptotic clicking rate of the temporal diffeomorphism. The temporal 
diffeomorphism obeys $T'=|Q| w$ and 
$\dot{T}=\dot{\hat{T}}+\partial_\tau[\int\; ds T']$ and 
this PDE system is solved by 
\be \label{4.a8}
T(\tau,r)=\hat{T}(\tau)+\frac{Q(\tau)}{2}\;\int_{-\infty}^r\; ds w(s)
\ee
This diffeomorphism is generically not an asymptotic identity
even if $\hat{T}(\tau)=\tau$ unless $Q=0$ 
and thus should not be considered as a gauge transformation
but rather a symmetry transformation in agreement with the Hamiltonian
distinction between symmetry and gauge reviewed in the previous section.\\
\\
The choice of a one parameter set of gauge fixings consistent with 
$Q$ given in the previous subsection that deviate from 
the exact GPG only {\it locally} in a neighbourhood of the 
origin $r=0$ of at most Planck size has the advantage 
that it is not observaable from the outside of the black hole. 
On the other hand, it makes the 
analysis of mode functions 
in such a spacetime very hard. In the next section we therefore 
consider another one parameter set of gauge fixings consitent with $Q$ 
which has a {\it non-local} effect on the spacetime metric and which has an 
intuitively quite appealing interpretation in terms of the energy 
of timelike observers called generalised Painlev\'e 
Gullstrand coordinates.

\section{Generalised Gullstrand Painlev\'e Coordinates}
\label{sd} 

We review here the theory of radial timelike geodesics in Schwarzschild 
spacetime with mass $M$ \cite{26}. 
These define a one parameter set $e\mapsto C^\pm_e$ 
of congruences $C^\pm_e$ of free falling observers that fill the spacetime
starting (ending) at timelike infinity and ending (starting) 
at the singularity for the ingoing (outgoing) congruence $C^-_e$ ($C^+_e$)
respectively. 
The parameter $e\ge 1$ has the physical interpretation of the special 
relativistic energy per unit mass at spatial infinity (i.e. 
$e=[1-(v/c)^2]^{-1/2}$ if $v$ is the velocity at spatial infinity).
For each congruence $C^\pm_e$ the geodesics 
fill an asymptotic end and the white (black) hole 
region of the Kruskal extension respectively.
The radial geodesics $c^\pm_{e,\rho,\Omega}\in C^\pm_e$ 
are labelled, besides the angular direction $\Omega$, by a parameter 
$\rho\in \mathbb{R}$ that labels the range of the affine parameter 
$\tau$ along the geodesic where 
$\tau\in (-\infty,\rho)$ for $c^-_{e;\rho,\Omega}$ and  
$\tau\in (\rho,\infty)$ for $c^+_{e;\rho,\Omega}$ and at $\tau=\rho$ 
the geodesic intersects the singularity $r=0$. 

The set of synchronous points 
$\Sigma^\pm_{e,\tau}=\{c^\pm_{e;\rho,\Omega}(\tau);\; \pm (\tau-\rho)\ge 0,\;
\Omega\in S^2\}$ defines the leaf of a foliation of the white (black) hole 
and asymptotic region by spacelike hypersurfaces. However, since they start 
(end) at the singularity, none of them is a Cauchy surface, i.e. 
there exist inextendible causal curves not intersecting them. Thus one cannot
use them for the initial value formulation. One could use a segment of the 
singularity to turn them into Cauchy surfaces but then different leaves of 
the foliation would not be disjoint. One could use two asymptotic ends 
in the {\it same} Kruskal spacetime and join say the two ingoing geodesics 
from the two ends that hit the same point of the singularity in the black 
hole region to form Cauchy surfaces but these intersect and do not form a 
foliation. However, one can join an ingoing geodesic in the part of
a {\it past}
Kruskal spacetime covering the Schwarzschild (SS) and black hole (BH) region
with an outgoing geodesic in a 
{\it different} i.e. {\it future} Kruskal spacetime
covering the mirror Schwarzschild (MSS) and white hole (WH) region 
that hit the same point of the singularity. We may then consider 
the geodesics $c_{e;\rho,\Omega}$ with 
$c_{e;\rho,\Omega}(\tau):=c^\pm_{e;\rho,\Omega}(\tau)$ for 
$\pm(\tau-\rho)\ge 0$ and their synchronous hypersurfaces 
$\Sigma_{e,\tau}=\{c^\pm_{e;\rho,\Omega}(\tau);\; \rho\in \mathbb{R},\;
\Omega\in S^2\}$. The resulting {\it black hole white hole transition} 
spacetime $M$
consisting of the four pieces SS, BH, WH, MSS is then foliated by the 
$\Sigma_{e,\tau}$ which are Cauchy surfaces for $(M,g)$, where 
$g$ is the extension of the Schwarzschild metric just outlined, is 
then globally hyperbolic (but singular) in the usual sense in the 
future and past Kruskal patches. We also consider a non singular wormhole 
regularisation of that spacetime. 
 
Therefore $(M,g)$ is an interesting spacetime to study when analysing 
questions such as black hole white hole transition (BHWHT) and singularity 
resolution in quantum gravity as we explore both 
interior and exterior regions of spactime. Moreover, the free falling observer 
congruences labelled by $e$ define a natural 1-parameter family of 
gauge fixing conditions for black hole spacetimes and 
corresponding preferred generalised Gullstrand Painlev\'e (GGP) coordinates.
A parameter
like $e$ is motivated by the result of the previous appendix 
because the constraints depend on spatial derivatives giving rise to 
two Dirac observables even in vacuum. These will be related to the 
mass $M$ and $e$ as we will see in the course of this appendix which 
are natural observables to consider in spherically symmetric
spacetimes. Finally, 
each foliation $\tau\mapsto \Sigma_{e,\tau}$ makes it possible to use 
the machinery of QFT in CST and to define 1-particle inner products etc. 
to study Hawking radiation etc.  

\subsection{Radial timelike geodesics in spherically symmetric vacuum
spacetimes}
\label{sd.1}

We consider the exterior static region of a spherically symmetric black hole 
with mass $M>0$, usual Schwarzschild coordinates 
$t\in \mathbb{R},r>2M,\Omega=(\theta,\phi)\in S^2$ and
line element
\be \label{d.1}
ds^2=-\Phi \; dt^2+\Phi^{-1}\; dr^2+r^2\; d\Omega^2,\;\;
\Phi=1-\frac{2M}{r}
\ee
which has a timelike 
Killing vector field $\xi=\partial_t=\delta^\mu_t \partial_\mu$.

A radial $\Omega=$const., timelike $g(u,u)<0;\;u=\dot{c}$ geodesic 
$\tau\to c(\tau)$ with affine parameter $\tau$ i.e. $\nabla_u u=0$ 
obeys $\nabla_u g(u,u)=0$ and $\nabla_u (g(u,\xi)=0$ providing two 
constants of motion $K:=g(u,u)$ amd $e:=-g(u,\xi)$. As usual we may 
fix $K=-1$ by rescaling the affine parameter so that 
\be \label{d.2} 
-1=-\Phi(r(\tau))\;\dot{t}(\tau)^2+\Phi^{-1}(r(\tau))\; \dot{r}(\tau)^2,\;
e=\Phi(r(\tau))\; \dot{t(\tau)}
\ee
It follows with $R:=2M$
\be \label{d.3}
\dot{r}^2=e^2-\Phi=e^2-1+\frac{R}{r}
\ee
We are interested in geodesics that extend all the way to spatial 
infinity $r=+\infty$ which requires that 
\be \label{d.4}
e^2\ge 1
\ee
The geodesic labelled by $e$ has an outgoing and ingoing branch corresponding 
to the choice of square root of (\ref{d.3})
\be \label{d.5}
\dot{r}=\pm\; \sqrt{e^2-1+\frac{R}{r}}
\ee
Although the coordinate system is a priori only defined for $r>R$, equation 
(\ref{d.5}) is meaningful for $r\in \mathbb{R}_+$. We note that 
\be \label{d.6}
u_t=g_{t\mu} u^\mu=-\Phi u^t=-\Phi \dot{t}=-e,\;
u_r=g_{r\mu} u^\mu=\Phi^{-1} u^r=\Phi^{-1} \dot{r}=\pm \Phi^{-1}\;
\sqrt{e^2-\Phi}=:\pm\; f'(r),\; u_\theta=u_\phi=0
\ee
which means that $u_\mu=-\nabla_\mu \tau^\pm_e$ where 
\be \label{d.7}
\tau^\pm_e:=e\; t\mp f(r),
f'(r)=\Phi^{-1}(r)\; \sqrt{e^2-\Phi(r)}
\ee
Thus 
$\nabla_{[\mu} u_{\nu]}=0$ so that the geodesics are hypersurface 
orthogonal, forming a foliation by $\tau^\pm_e=$const. hypersurfaces.

The coordinates $(\tau=\tau_e^\pm,r)$ are called out(in)
going generalised Gullstrand 
Painlev\'e
(GGP) coordinates. The line element in terms of them 
is obtained from 
\be \label{d.8}
t(\tau,r)=e^{-1}[\tau \pm f(r)],\;f'(r)=\sqrt{e^2-\Phi}\;\Phi^{-1}
\;\; \Rightarrow \;\;
ds^2=-e^{-2} \Phi\; d\tau^2 \mp 2 \;e^{-2}\; \sqrt{e^2-\Phi}\;d\tau\;dr
+e^{-2}\; dr^2+r^2\;d\Omega^2
\ee
which does not require to solve for $f(r)$ explicitly. The line element
(\ref{d.8}) is no longer static but still stationary. It is
easy to check that the outgoing (ingoing) future oriented (with respect to 
$\tau$) unit timelike 
geodesics with congruence parameter $e'\ge 1$ in the outgoing (ingoing) 
version of the 
line element (\ref{d.8}) have full range $r\in \mathbb{R}^+$
but intersect $r=0$ at finite $\tau$. The future oriented ingoing (outgoing)
geodesics in the outgoing (ingoing) version 
of (\ref{d.8}) on the other hand are confined to $r>2M$.     

The ADM data of 
(\ref{d.8}) are 
\be \label{d.8a}
q_{rr}=e^{-2},\;q_{AB}=r^2\Omega_{AB},\; q_{rA}=0,\;
N^a=\mp \delta^a_r \sqrt{e^2-\Phi},\;
N=1
%g_{ta}=q_{ab} N^b=q_{rr} N^r=\mp e^{-2} \sqrt{e^2-\Phi}, 
%g_{tt}=-N^2+q_{ab} N^a N^b=-N^2+e^{-2}(e^2-\Phi}=-N^2+1-e^{-2} \Phi
\ee
which gives the future oriented timelike unit normal to the $\tau=$const. 
leaves $n=N^{-1}(\partial_\tau-N^a\partial_a)=
\partial_\tau \pm \sqrt{e^2-\Phi} \partial_r$. Note that the vector 
field $\partial_\tau$ in these $(\tau,r)$ coordinates 
is a Killing vector field but it is not everywhere 
timelike and nowhere orthogonal to the $\tau=$const. foliation.\\ 
\\  
On the other hand, Gaussian or synchronous coordinates are characterised 
by unit lapse squared and vanishing shift. One obtains them most easily 
from (\ref{d.8}) by computing the function $r=a(\tau,\rho)$ 
\be \label{d.9}
ds^2=e^{-2}\; (-\Phi \mp 2\;\dot{a}\; \sqrt{e^2-\Phi}+\dot{a}^2)\; d\tau^2
+2\;a'\;e^{-2}\; (\dot{a}\mp \sqrt{e^2-1-\Phi})\;d\tau\;d\rho
+[a']^2\; d\rho^2+a^2\;d\Omega^2
\ee
with $\dot{a}=\partial_\tau a,\;a'=\partial_\rho a$. The shift vanishes iff 
\be \label{d.10}
\dot{a}=\pm\sqrt{e^2-\Phi}
\ee
It follows
\be \label{d.11}
ds^2=-d\tau^2+e^{-2}\;(a')^2\; d\rho^2+a^2\;d\Omega^2
\ee
We call the integration constant $\rho=\rho^\pm_e$ 
in (\ref{d.10}) and find 
\be \label{d.12}
\pm(\tau-\rho)=\int\; da\; [e^2-\Phi(a)]^{-1/2}
\ee
which shows that $\dot{a}=-a'$ hence without further calculation
\be \label{d.13}
ds^2=-d\tau^2+e^{-2}[e^2-\Phi(a)]\; d\rho^2+a^2\;d\Omega^2
\ee
with $a$ implicitly determined by (\ref{d.12}). 

To actually determine $a$ we have to treat 
the case $e^2=1$ seperately. The integral is elementary  
\be \label{d.14}
\pm(\tau-\rho)=\frac{2}{3}\; a^{3/2}\; R^{-1/2}
\;\;\Leftrightarrow\;\;
r=a(\tau,\rho)=[\pm\frac{3}{2} \sqrt{R}(\tau-\rho)]^{2/3}
\ee
valid for $\pm(\tau-\rho)>0$. Then (\ref{d.13}) simplifies 
\be \label{d.15b}
ds^2=-d\tau^2+\frac{R}{a}\; d\rho^2+a^2\;d\Omega^2
\ee
For $e^2>1$ we introduce the quantity
\be \label{d.13b}
z:=\sqrt{(e^2-1)\frac{a}{R}}
\ee
then 
\be \label{d.14b}
\pm(\tau-\rho)
=R\;(e^2-1)^{3/2}\;h(z),\; h(z)=[z\sqrt{z^2+1}-\ln(z+\sqrt{z^2+1})],\;
h'(z)=2\frac{z^2}{\sqrt{z^2+1}}
%=\int\; dr\; [e^2-\Phi(r)]^{-1/2}
%=\int\; dr\;r^{1/2} [(e^2-1)\;r+1]^{-1/2}
%=R\;\int\; dx\;r^{1/2} [(e^2-1)\;x+1]^{-1/2}\; x=R/r
%=R\;(e^2-1)^{3/2})\;\int\; dy\;y^{1/2} [y+1]^{-1/2}\; y=(e^2-1) x
%=2\;R\;(e^2-1)^{3/2})\;\int\; dz\;z^2 [z^2+1]^{-1/2}\; z=y^{1/2}
%\int\; dz\;z^2 [z^2+1]^{-1/2}\; z=y^{1/2}
%=R\;(e^2-1)^{3/2})\;[z\sqrt{z^2+1}-\ln(z+\sqrt{z^2+1})]
%[.]'=(z w)'=w+z^2/w-(1+z/w)/(z+w)=w+z^2/w-1/w=(w^2+z^2-1)/w=2z^2/w
\ee
which determines $r=a(\tau,\rho)$ implicitly. Since $h$ is monotonouly 
increasing and $h(0)=0$, again the range of $\tau,\rho$ is constrained 
by $\pm (\tau-\rho)>0$.

Note that the metric coefficients $g_{\tau\tau},g_{\rho\rho}, 
g_{AB},\; A,B=1,2$ in (\ref{d.13}) only depend on $a$ and thus 
$\tau-\rho$.
Therefore $\xi_e:=\partial_\tau+\partial_\rho$ is a Killing vector field 
$[{\cal L}_{\xi_e} g]_{\mu\nu}=\xi_e^\sigma\; g_{\mu\nu,\sigma}=0$ 
with norm $g(\xi_e,\xi_e)=-1+1-\frac{\Phi}{e^2}=-\frac{\Phi}{e^2}$ which 
is timelike for $r=a>M$. Thus $\xi_e$ must coincide for $a>R$ with 
$\xi=\partial_t$ up to a constant, which can be confirmed. Evaluating 
the norm at spatial infinity $a=\infty$ we find $\xi_e=e^{-1} \xi$. 

Note also that the change between GGP coordinates $(\tau,r)$ and 
synchronous coordinates $\tilde{\tau},\rho$ with 
$\tau=\tilde{\tau},r=a(\tilde{\tau},\rho)$ gives 
$\partial_\tau=\partial_{\tilde{\tau}}+
([\partial_\tau b(\tau,r)]_{r=a(\tau,\rho)})_{\tau=\tilde{\tau}}\;
\partial_\rho$
where $b(\tau,a(\tau,\rho))=\rho$ inverts $r=a(\tau,\rho)$ for 
$\rho=b(\tau,r)$ at fixed $\tau$. Thus 
$0=b_{,\tau}(\tau,r=a)+b_{,r}(\tau,r=a)\;a_{,\tau},\;
1=b_{,r}(\tau,r=a)\;a_{,\rho}$ and since $a_{,\rho}+a_{,\tau}=0$ it follows 
$\partial_{\tau}=\partial_{\tilde{\tau}}+\partial_\rho$. Thus while 
$\tilde{\tau}=\tau$ there is a non-trivial transformation between 
$\partial_\tau,\partial_{\tilde{\tau}}$ as vector fields when changing 
from GGP to synchronous coordinates. This also explains why 
$\partial_\tau$ is a KVF but not hypersurface orthogonal while
$\partial_{\tilde{\tau}}$ is no KVF but hypersurface orthogonal.
With this clarification out of the way we keep the notation $\partial_\tau$
for both coordinate systems but have to remember the difference between 
the roles that $\partial_\tau$ plays in them.

We consider the geodesic congruence with congruence parameter $e'$ and
geodesic tangent $u=\partial_s=
u^\tau \partial_\tau+u^\rho\partial_\rho,\;
u^\tau=d\tau/ds,\;u^\rho=d\rho/ds$ in the coordinates 
$\tau=\tau^\pm_e,\rho=\rho^\pm_e$. Then 
\be \label{d.15}
-e'=g(u,\xi)=e[-d\tau/ds+(1-\Phi/e^2)\;d\rho/ds],\;\;
-1=g(u,u)=\frac{d\tau}{ds}^2+[1-\frac{\Phi}{e^2}]\;[\frac{d\rho}{ds}]^2 
\ee
These have two solutions. The outgoing solution for $\tau-\rho>0$ and the 
ingoing solution for $\tau-\rho<0$ respectively correspond to $\rho=$const.
and $d\tau/ds=\frac{e'}{e}$ which can be seen from the fact that 
$a(\tau,\rho)$ is monotonously increasing and respectively decreasing 
with increasing $\tau$ thanks to the monotonocity of $h$ in (\ref{d.14}).
For the geodesic congruence $e'=e$ we see that $\tau=s$ 
coincides with the proper time along the geodesics. 

Thus the out(in) going geodesic conguence
with $e'=e$ becomes especially simple in out(in)going synchronous 
coordinates $\tau=\tau^\pm_e,\rho=\rho^\pm_e$, they are just the lines 
$\rho=$const. and $s\mapsto\tau=s$ and are valid for $\tau>\rho$ and
$\tau<\rho$ respectively. All geodesic observers are synchronised 
on the $\tau=$const. 
hypersurface $\Sigma^\pm_{e,\tau}=\{\pm(\tau-\rho)>0,\;\Omega\in S^2\}$.
The hypersurfaces are mutually disjoint and cover one exterior region and 
the white (black) hole region of the Kruskal extension. Since 
$1-\Phi(a)/e^2>0$ for all $a\in \mathbb{R}_+$ and $e^2\ge 1$ the 
hypersurfaces $\tau=$const have intrinsic metric of manifestly positive 
signature and are thus spacelike. The vector field $\partial_\tau$ is 
everywhere timelike and in fact the future oriented timelike unit normal 
to the hypersurfaces, however, it is not a Killing vector field and therefore 
for the geodesic observer the metric is eigentime $\tau$ dependent both 
in the exterior and interior region. Yet, the observer (in the ideal 
limit of vanishing spatial extension) feels no tidal forces and thus 
considers herself in an inertial frame.

\subsection{Black Hole White Hole Transition}
\label{sd.2}

The geodesic congruence $C^+_e,C^-_e$ so constructed thus determines 
a spacelike foliation of the WH and MSS region or BH and SS region respectively.
However, none of the leaves $\Sigma^\pm_{e,\tau}$ 
of the foliation is a Cauchy surface for those parts of the 
Kruskal spacetime because $\rho$ is not allowed to 
take full range $\mathbb{R}$, rather it is restricted by $\pm(\tau-\rho)>0$.  
The obvious idea to turn them into Cauchy surfaces is to consider 
a gluing of a past SS and BH part of one Kruskal spacetime with a future 
WH and MSS part of another Kruskal spacetime along the singularity $a=0$.
Accordingly we consider at fixed $e$ coordinates $\tau,\rho\in \mathbb{R}$ 
and the metric 
\be \label{d.16}
ds^2=-d\tau^2+[1-\frac{\Phi(a)}{e^2}]\;d\rho^2+a^2\; d\Omega^2
\ee
where $a(\tau,\rho)$ is the function implicitly defined by 
\be \label{d.17}
|\tau-\rho|=
\left\{ \begin{array}{cc}
R\;(e^2-1)^{3/2})\;h(z),\; h(z)=[z\sqrt{z^2+1}-\ln(z+\sqrt{z^2+1})],\;
z=[(e^2-1)a/R]^{1/2} & e^2>1\\
\frac{2}{3} \frac{a^{3/2}}{R^{1/2}} & e^2=1
\end{array}
\right.
\ee
The metric (\ref{d.16}) is singular at the singularity
``hypersurface'' $\tau=\rho$ 
where $a=0$.
In the $(\tau,\rho)$ diagramme the 
geodesics $\rho=$const. cross the singularity at the proper time 
$\tau=\rho$ smoothly. The radial coordinate $a(\tau,\rho)\to 0$ 
is also continuous at $\tau=\rho$ but 
the radial velocity $da/d\tau\to \mp \infty$ as $\tau\to \rho\mp$. 
The $\tau=$const. surfaces also intersect the singularity transversally
in the $(\tau,\rho)$ diagrammme. Nevertheless they are geometrically 
tangential: The reason for this is that the co-normals $n,n'$ to the 
$\tau=$const. and the $\tau-\rho=0$ surface respectively are given by
$n_\mu=\nabla_\mu\tau=\delta^\tau_\mu,\;  
n'_\mu=\nabla_\mu(\tau-\rho)=\delta^\tau_\mu-\delta^\rho_\mu$, hence the
normal is given by 
$n^\mu=g^{\mu\nu} n_\nu=-\delta^\mu_\tau$ and  
$n^{\prime\mu}=g^{\mu\nu} n'_\nu=-\delta^\mu_\tau+\frac{a}{R} \delta^\mu_\rho$
respectively, thus $n=n'$ at $a=0$.   

We can return to GGP coordinates but need two 
radial coordinates, i.e. a past radial 
coordinate $r$ and a future radial coordinate $\bar{r}$ which are related 
to $\tau,\rho$ by 
\be \label{d.18}
r=a(\tau,\rho);\; t<\rho;\; 
\bar{r}=a(\tau,\rho);\; t>\rho
\ee
These can be combined into a single coordinate 
\be \label{d.19}
z=-{\rm sgn}(\tau-\rho) a(\tau,\rho)=
\left\{ 
\begin{array}{cc}
-\bar{r} & \tau>\rho\\
r & \tau<\rho
\end{array}
\right.
\ee 
which like $\rho$ takes full range in $\mathbb{R}$. Then $dz/d\tau <0$
(thus $z$ is monotonous along the geodesic) 
taking its minimal value $-\infty$ at $\tau=\rho$ and its maximal 
value $-\sqrt{e^2-1}$ at $\tau=\pm \infty$. 
Note that if we use the same 
$\tau,\theta,\phi$ coordinates for the whole spacetime, therefore the 
radial geodesics $\rho,\theta,\phi=$const. change to opposite direction when 
passing through the singularity.

\subsection{Non-singular spacetime}
\label{sd.3}

In terms of the synchronous coordinates, the surface $a(\tau,\rho)=0$ is 
three-dimensional but in terms of the Cartesian coordinates
$x^a=r\Omega^a,\;\bar{x}^a=\bar{r}\Omega^a,\;a=1,2,3$ which vanish at the 
singularity, it is just a one dimensional line. This can can be seen
also by considering 
the surface $\tau=\rho+\epsilon,\;\epsilon\not=0$ which has induced line element
\be \label{d.20}
ds^2=-\frac{\Phi(a)}{e^2}\;d\rho^2+a^2\; d\Omega^2
\ee
which for $\epsilon\to 0$ results in $R>a\to 0$ so that $0< -\Phi(a)\to 
+\infty$ while $a\to 0$ so that (\ref{d.19}) formally has signature 
$(1,0,0)$. 

Following ideas about non-singular and 
wormhole spacetimes such as \cite{NonSing,Wormhole} we may exclude the 
singularity in an ad hoc manner by simply restricting the range 
of $r,\bar{r}$ to $(l,\infty),\; l>0$ and perform the gluing for each 
$\tau$ at 
$r=\bar{r}=l,\; \theta=\bar{\theta},\phi=\bar{\phi}$ 
or at 
$r=\bar{r}=l,\; \theta=\pi-\bar{\theta},\phi=\pi+\bar{\phi}$. 
which now has the the topology of $\mathbb{R}\times S^2$. 
In contrast to \cite{NonSing} and similar to \cite{Wormhole} 
this still defines a {\it vacuum solution} 
for any $r,\bar{r}>l$ i.e. the energy momentum tensor 
vanishes. This has the the following 
mild disadvantage: The geodesic 
$\rho=$const. in the region $\pm (\tau-\rho)>0$ hits the value 
$r=a(\tau,\rho)=l$ 
at a value $\tau^\pm_l(\rho)=\rho\pm \epsilon_l$ where 
$a(\tau^\pm_l(\rho),\rho)=l$. Thus the geodesic that starts at $\rho=$const.
in the  $\tau<\rho$ region cannot continue as the geodesic with the same
value of $\rho$ in the $\tau>\rho$ region if the affine parameter is to 
be continuous. Rather the geodesic parameter must change to $\rho'$ 
where $\tau^+_l(\rho')=\tau^-_l(\rho)$. In the first gluing option 
the geodesic then continues in the opposite direction, in the second
gluing option it continues into the same 
direction (remember that the geodesics are in/outgoing respectively). 
If one wishes to interpret this in a $\tau,z$ diagramme 
in which the angular dimension is suppressed, then 
it appears as if the geodesic jumps between $z=\pm l$. However, 
if we consider a three dimensional diagramme in which we depict the angular 
dependence by circles then we should consider two copies of $\mathbb{R}^3$
covered by coordinates $\tau,r,\varphi$ and 
$\bar{\tau},\bar{r},\bar{\varphi}$ from which we cut out the solid 
cyinders $0\le r\le l$ and  $0\le \bar{r}\le l$ respectively. We then 
glue the surfaces $r=l$ and $\bar{r}=l$ of the cylinders at either 
$\tau=\bar{\tau},\varphi=\bar{\varphi}$ or    
$\tau=\bar{\tau},\varphi=\pi+\bar{\varphi}$. A radial inward geodesic 
$\varphi=$const. starting in the first copy then hits $r=l$ at some 
$\tau$ and continues either 
as the geodesic $\bar{\varphi}=\varphi=$const. into the opposite or
as the geodesic $\bar{\varphi}=\pi+\varphi=$const. into the same direction. 
In both cases the geodesic is continuous because of the prescription 
in which we identified the points. In the second gluing option also 
the first derivative of the geodesic is continuous. W.l.g. consider 
the ingoing radial geodesic in 1-direction $\tau\mapsto (r_\rho(\tau),0,0)$
with $\tau\in (-\infty,\tau_l),\; r_\rho(\tau_l)=l$. Then it continues as 
the outgoing geodesic $(\bar{r}_{\bar{\rho}}(\tau),0,0)$ 
with $\tau\in (\tau_l,\infty),\; r_{\bar{\rho}}(\tau_l)=l$ 
in the first option and as  
$(-\bar{r}_{\bar{\rho}}(\tau),0.0)$ in the second. We have by construction
$-\dot{r}_{\rho}(\tau_l)=\dot{\bar{r}}_{\bar{\rho}}(\tau_l)>0$. This 
spacetime is therefore geodesically complete with respect to the 
oberservers in these congruences and in that sense singularity 
free. In the first option an observer considers herself as ``bounced'' 
off $r=l$ when entering the second universe while in the second option 
she considers herself as ``gone through'' the cylinder. 

Note that
no causal geodesic can stay on the cylinder surface $r=l$ as it is spacelike.
The $r=$const. surfaces
are timelike/null/spacelike for $r>/=/< R$ as maybe seen 
easiest from (\ref{d.21}) while the hypersurfaces $\tau=const.$ are 
always spacelike. Therefore the cylinder surfaces $r=l$ are for $l<R$ 
certainly spacelike and thus causal geodesics must cross it transversally. 
The cylinder surface replaces the singularity by a spacelike surface 
with coordinates $\tau,\theta,\phi$ and thus has the topology 
$\mathbb{R}\times S^2$. It maybe disturbing that the time $\tau$ here 
serves as a coordinate on a spacelike hypersurface 
but we can interpret it as the point of 
eigentime at which geodesic observers cross the gluing cylinder between the 
universes.      

One may consider the introduction of $l$ also as a regularisation of the 
singular spacetime which maybe used to construct QFT in CST. 
This is relevant in the construction of mode systems (solutions 
of Klein Gordon like equations) and 1-particle inner products which 
rely on the presence of Cauchy surfaces such as the leaves of this 
BHWHT foliation. From that perspective the corresponding wave equations 
for a function $f(\tau,z,\Omega)=e^{i\omega\tau} f_\omega(z,\Omega)$
become stationary Schr\"odinger type of eigenvalue equations for $f_\omega$
in a singular potential as $l\to 0$ as we have shown in section \ref{s7}.

\subsection{Causal structure and Penrose diagramme}
\label{sd.4}

In order to understand the causal structure of this singular BHWHT spacetime 
we consider the simpler case $e^2=1$ for which we can write the line element 
in terms of $\tau,z,\theta,\phi$ using (\ref{d.8}), (\ref{d.20})
\be \label{d.21}
ds^2=-\Phi(|z|)\;d\tau^2+2\sqrt{\frac{R}{|z|}}\; d\tau\;dz
+dz^2+z^2 d\Omega^2
\ee
Its radial null geodesics are determined by  
\be \label{d.22}
ds^2=-|z|^{-1}\;
(d\tau\;[1+|z|^{-1/2}]+dz)\;(d\tau\;[1-|z|^{-1/2}]-dz)\;
=0
\ee
where we switched to $\hat{\tau}=\tau/R,\; \hat{z}=z/R$ and removed the hat 
again. We use $z\in \mathbb{R}$ as a parameter so that we get two 
types of null geodesics
\be \label{d.23}
\frac{d\tau}{dz}=\mp\; \frac{|z|^{1/2}}{|z|^{1/2}\pm 1}
\ee
For the upper sign, $d\tau/dz$ is everywhere regular, at 
$z=\pm\infty$ taking 
the value $-1$, at the two horizons $z=\pm 1$ taking the value $-1/2$ and 
at the singularity the value $0$. For the lower sign we have four  
kinds of null geodesics, namely those that are stuck in either of
the intervals $(1,\infty),\;(-1,1),\;(-\infty,-1)$ and those that 
are stuck at the horizions $|z|=1$. 
For the first interval
the null geodesic starts at $z=1$ with $d\tau/dz=+\infty$ moving to 
$z=\infty$ with $d\tau/dz=1$.     
For the third interval
the null geodesic ends at $z=-1$ with $d\tau/dz=+\infty$ having moved from 
$z=-\infty$ with $d\tau/dz=1$. For the second interval the geodesic 
starts at $z=1$ in the infinite past with $d\tau/dz=-\infty$,
passes through $z=0$ with $d\tau/dz=0$ and ends at $z=-1$ in the infinite 
future with $d\tau/dz=-\infty$. Accordingly, in the $\tau,z$ diagramme 
the lightcone structure is as follows: For $|z|>1$ there are ingoing and 
outgoing light rays (i.e. moving to smaller and larger $|z|$), 
for $z=1$ there is one ingoing one and one that is tangential to the horizon,
for $z=-1$ there is one outgoing one and one that is tangential to the 
horizon, for $0<z<1$ there are only ingoing lightrays (trapped region),  
for $-1<z<0$ there are only outgoing lightrays (anti trapped region) and 
for $-\infty<z<-1$ there are both ingoing and outgoing lightrays. 
This is of course exactly the BHWHT spacetime structure.
The discussion also shows that the restriction to $z>0$ ($z<0$) respectively 
covers precisely an ingoing (outgoing) GP spacetime or equivalently 
an advanced (retarded) Finkelstein spacetime (covered by 
$v,r$ or $u,r$ coordinates respectively with $v=t+r_\ast,\;u=t-r_\ast$
on the SS portion where $t$ is SS time and $r_\ast$ is the turtoise 
coordinate, i.e. there is always the ingoing (outgoing) null geodesic
$v=$const. ($u=$const.)). 

It is helpful to construct the corresponding Penrose diagramme which 
can be done analytically in the case $e^2=1$ in terms of  
Kruskal coordinates. From (\ref{d.7}) we have 
\be \label{d.24}
\tau=\bar{t}-f(\bar{r})=t+f(r),\;
f(r)=R(2y+\ln(\frac{y-1}{y+1})),\;y=\sqrt{\frac{r}{R}}
\ee
for Schwarzschild coordinates $r,\bar{r}>1$ and Schwarzschild asymptotic 
times $t,\bar{t}$ in the SS and MSS region respectively. In terms of the 
null coordinates $v=t+r_\ast,\; u=t-r_\ast,\; r_\ast=r+R\ln(\frac{r}{R}-1)$
and analogously for the barred quantities we set
\be \label{d.25}
V:=e^{\frac{v}{2R}},\; 
U:=-e^{-\frac{u}{2R}},\; 
\bar{V}:=-e^{\frac{\bar{v}}{2R}},\; 
\bar{U}:=e^{-\frac{\bar{u}}{2R}},\; 
\ee
By substituting $r_\ast/R=y+\ln(y^2-1)$ and for $t,\bar{t}$ according to 
(\ref{d.24}) one finds with $\kappa=\tau/(2R)$
\be \label{d.25b}
V=e^{\kappa+\frac{y^2}{2}-y}\;(y+1),\;
U=-e^{-\kappa+\frac{y^2}{2}+y}\;(y-1),\;
\bar{V}=-e^{\kappa+\frac{\bar{y}^2}{2}+\bar{y}}\;(\bar{y}-1),\;
\bar{U}=e^{-\kappa+\frac{\bar{y}^2}{2}-\bar{y}}\;(\bar{y}+1)
\ee
The choice of signs is here uniquely determined by the 
continuity requirement 
that at the singularity $y=\bar{y}=0$ we have $\bar{V}=V,\bar{U}=U$.
Thus for $y,\bar{y}\in \mathbb{R}^+, \; \tau\in \mathbb{R}$ we have 
$U,\bar{V}\in \mathbb{R},\; V,\bar{U}\in \mathbb{R}^+$ and the 
Kruskal relations
\be \label{d.26}
U\;V=-e^{y^2}(y^2-1),\;\;\bar{U}\;\bar{V}=-e^{\bar{y}^2}(\bar{y}^2-1)
\ee
which are bounded from above by $+1$. The SS, BH, WH, MSS regions are 
respectively covered by $V>0>U$, $U,V>0$, $\bar{U},\bar{V}>0$, 
$\bar{U}>0>\bar{V}$ separated respectively by the BH horizon $y=1$,
the singularity $y=\bar{y}=0$ and the WH horizon $\bar{y}=1$ in 
chronological order. We introduce compactified null coordinates 
\be \label{d.27}
\hat{v}={\sf arctan}(V),\; 
\hat{u}={\sf arctan}(U),\; 
\hat{\bar{v}}={\sf arctan}(\bar{V}),\; 
\hat{\bar{u}}={\sf arctan}(\bar{U}),\; 
\ee
with $\hat{v},\hat{\bar{u}}\in (0,\frac{\pi}{2}),\;
\hat{u},\hat{\bar{v}}\in (-\frac{\pi}{2},\frac{\pi}{2})$ and finally
\be \label{d.28}
\hat{t}:=\left\{ \begin{array}{cc}
\hat{v}+\hat{u} & {\sf SS, BH}\\
\pi-(\hat{\bar{v}}+\hat{\bar{u}}) & {\sf WH, MSS}
\end{array}
\right.
\;\;,\;\;
\hat{x}:=\left\{ \begin{array}{cc}
\hat{v}-\hat{u} & {\sf SS, BH}\\
\hat{\bar{v}}-\hat{\bar{u}} & {\sf WH, MSS}
\end{array}
\right.
\ee
It follows that in SS $-\pi/2\le \hat{u}\le 0 \le \hat{v}$ we have 
\be \label{d.29}
0\le 2\hat{v}=\hat{t}+\hat{x}\le \pi,\; 
-\pi\le 2\hat{u}=\hat{t}-\hat{x}\le 0,\;0\le \hat{x}\le \pi
\;\;\Rightarrow\; {\sf max}(-\hat{x},\hat{x}-\pi)
\le \hat{t}\le {\sf min}(\hat{x},\pi -\hat{x})
\ee
in BH $0\le \hat{u},\hat{v}\le \pi/2$ and 
$U\;V=\frac{
\cos(\hat{v}-\hat{u})-\cos(\hat{v}-\hat{u})}
{\cos(\hat{v}-\hat{u})+\cos(\hat{v}-\hat{u})}\le 1$ i.e. 
$\cos(\hat{v}+\hat{u})\ge 0$ i.e. $\hat{t}\le \pi/2$ and 
\be \label{d.30}
0\le 2\hat{v}=\hat{t}+\hat{x}\le \pi,\; 
0\le 2\hat{u}=\hat{t}-\hat{x}\le \pi,\;0\le \hat{t}\le \pi/2
\;\;\Rightarrow\; -\hat{t}\le \hat{x} \le \hat{t}
\ee
in WH $0\le \hat{\bar{u}},\hat{\bar{v}}\le \pi/2$ and 
$\bar{U}\;\bar{V}=\frac{
\cos(\hat{\bar{v}}-\hat{\bar{u}})-\cos(\hat{\bar{v}}-\hat{\bar{u}})}
{\cos(\hat{\bar{v}}-\hat{\bar{u}})+\cos(\hat{\bar{v}}-\hat{\bar{u}})}\le 1$
i.e. $\cos(\hat{\bar{v}}+\hat{\bar{u}})\ge 0$ 
i.e. $\hat{\bar{u}}+\hat{\bar{v}}\le \pi/2$ 
i.e. $\hat{t}\ge \pi/2$ and 
\be \label{d.31a}
0\le 2\hat{\bar{v}}=\pi-\hat{t}+\hat{x}\le \pi,\; 
0\le 2\hat{\bar{u}}=\pi-\hat{t}-\hat{x}\le \pi,\;\pi/2\le \hat{t}
\;\;\Rightarrow\; \hat{t}-\pi\le \hat{x} \le \pi-\hat{t}
\ee
and in MSS $-\pi/2\le \hat{\bar{v}}\le 0\le \hat{\bar{u}}\le \pi/2$
and $-\pi\le \hat{x} \le 0$ and
\be \label{d.31b}
-\pi\le 2\hat{\bar{v}}=\pi-\hat{t}+\hat{x}\le 0,\; 
0\le 2\hat{\bar{u}}=\pi-\hat{t}-\hat{x}\le \pi,
\;\;\Rightarrow\; {\sf max}(-\hat{x},\hat{x}+\pi)
\le \hat{t}\le {\sf min}(2\pi+\hat{x},\pi -\hat{x})
\ee
It is not difficult to see that in the $\hat{x},\hat{t}$ diagramme 
SS is a diamond with corners $b_P=(0,0),\;
i^-_P=(\pi/2,-\pi/2),\; i^0_P=(\pi,0),i^+_P=(\pi/2,\pi/2)$,
BH is a triangle with corners $b_P,i^+_P,i_F^-=(-\pi/2,\pi/2)$,
WH is a triangle with corners $b_F=(0,\pi),i_F^-, i^+_P$ and 
MSS is a diamond with corners $i^0_F=(-\pi,\pi),\; i^-_F,\;b_F,\;
i^+_F=(-\pi/2,3\pi/2)$. Here the subscripts refer to past and future 
Kruskal portions. The singularity is the line betwwen $i^-_F, i^+_P$,
the BH horizon is the line between $b_P,i^+_P$, the WH horizon is 
the line between $i^-_F,b_F$. 
Past and future null infinity in SS are 
the lines between $i^-_P,i^0_P$ and $i^0_P,i^+_P$ respectively while  
past and future null infinity in MSS are 
the lines between $i-_F,i^0_F$ and $i^0_F,i^+_F$ respectively.
All other diagonal lines are at $r=R$ or $\bar{r}=R$ respectively. 
The points $b_P,b_F$ are the bifurcation points in the full past and future 
Kruskal spacetimes. 

This BHWHT spacetime can be extended indefinitely to the 
future and the past by gluing identical pieces along the $r=R,\bar{r}=R$ 
lines. Or we can complete it by a Minkowski part of spacetime both in the 
past and the future by adding the points $I^-_P=(-\pi/2,-3\pi/2)$ and 
$I^+_F=(\pi/2,5\pi/2)$ respectively and adding the triangles with corners 
$I^-_P,i^-_P,i^-_F$ and 
$I^+_F,i^+_F,i^+_P$ respectively. The vertical lines between $I^-_P,I^+_F$ 
and the singularity then represent $r=0$ during formation of the black hole 
and evaporation of the white hole respectively. In this completed 
spacetime the free falling hypersurfaces are still Cauchy surfaces 
and we can complete the foliation in the 
Minkowski regions by segments along past and future null infinity 
between $i^0_P,I^-_P$ and $i^0_F,I^+_F$ and Cauchy surfaces in the 
Minkowksi parts. In the completed spacetime past null infinity in 
the past part and and future null infinity in the future part got 
extended by the Minkowski parts and we have two spacelike infinities 
$i^0_P,i^0_F$ and one past and future timelike infinity $I^-_P, I^+_F$
respectively.

We now explore the radial timelike 
geodesics and the free falling orthogonal foliation they 
generate. For a geodesic $\rho=$const. we are interested in the limits 
$\tau\to \pm\infty$. 
For $\tau\to +\infty$ we eventually enter the region 
$\tau>\rho$ covered by 
$\bar{y}=(\bar{r}/R)^{1/2}=[\frac{3}{2}(\tau-\rho)/R]^{1/3}$
which grows as $[\tau/R]^{1/3}$. Hence even $\bar{y}^2$ grows slower than 
$\tau/R$ and the behaviour of $\bar{V},\bar{U}$ is governed by 
$e^{\pm \tau/(2R)}$. Thus $\bar{V}\to -\infty, \bar{U}\to 0$ hence
$\hat{\bar{v}}\to -\pi/2,\;\hat{\bar{u}}\to 0$ i.e. 
$\hat{t}=\pi-\hat{\bar{v}}-\hat{\bar{u}}\to 3\pi/2,
\;\hat{x}=\hat{\bar{v}}-\hat{\bar{u}}\to -\pi/2$ 
i.e. the geodesic ends up in $i^+_F$.
For $\tau\to -\infty$ we eventually enter the region 
$\tau<\rho$ covered by 
$y=(r/R)^{1/2}=[-\frac{3}{2}(\tau-\rho)/R]^{1/3}$
which grows as $[-\tau/R]^{1/3}$. Hence even $y^2$ grows slower than 
$-\tau/R$ and the behaviour of $V,U$ is governed by 
$e^{\pm \tau/(2R)}$. Thus $V\to 0, \bar{U}\to -\infty$ hence
$\hat{v}\to 0,\;\hat{u}\to -\pi/2$ i.e. 
$\hat{t}=\hat{v}+\hat{u} \to -\pi/2,\;\hat{x}=\hat{v}-\hat{u}
\to \pi/2$ i.e. the geodesic ends up in $i^-_P$.

For the $\tau=$const. slices we are interested in $\rho\to \pm \infty$. 
For $\rho\to \infty$ we eventually enter the region $\tau-\rho<0$ 
covered by $y=[-\frac{3}{2}(\tau-\rho)/R]^{1/3}$ which grows as 
$[\rho/R]^{1/3}$ and the behaviour of $V,U$ is governed by $e^{y^2/2}$.
Thus $V\to +\infty, U\to -\infty$ i.e. $\hat{v}\to \pi/2,\; 
\hat{u}\to-\pi/2$ hence $\hat{t}\to 0,\; \hat{x}\to \pi$ i.e. we end up 
in $i^0_P$.  
For $\rho\to -\infty$ we eventually enter the region $\tau-\rho>0$ 
covered by $\bar{y}=[\frac{3}{2}(\tau-\rho)/R]^{1/3}$ which grows as 
$[-\rho/R]^{1/3}$ and the behaviour of $V,U$ is governed by $e^{\bar{y}^2/2}$.
Thus $\bar{V}\to -\infty, \bar{U}\to \infty$ i.e. $\hat{\bar{v}}\to 
-\pi/2,\; 
\hat{\bar{u}}\to \pi/2$ hence $\hat{t}\to \pi,\; \hat{x}\to -\pi$ i.e. we 
end up in $i^0_F$.  
  
Thus the following geometric picture emerges: All geodesics start at 
$i^-_P$ and end in $i^+_F$ as $\tau$ grows, all leaves start in 
$i^0_F$ and end in $i^0_P$ as $\rho$ grows. The geodesic labelled by 
$\rho$ intersects at $\tau=\rho$ the singularity in a point 
while the hypersurface 
labelled by $\tau$ intersects the singularity at coordinate label $\rho=\tau$
in a sphere. One can work out $d\hat{t},d\hat{x}$ explicitly in terms 
of the differentials $d\tau,d\rho$ using the coordinate 
transformation between these coordinates by the 
same technique as below for $T,X$ coordinates. One then shows 
by computing $\frac{d\hat{t}}{d\hat{x}}$ that in the 
Penrose diagramme all $\tau=$const. surfaces coordinatised by $\rho$
are tangent to 
the horizontal line representing the singularity ($\rho=\tau$) 
and have inclination 
of +45 degrees at the spatial infinities ($\rho=\pm \infty$)
while the geodesics 
$\rho=$const. coordinatised by $\tau$ intersect the 
singularity line at 90 degrees ($\tau=\rho$) and have inclination +45 degrees 
at the timelike infinities ($\tau=\pm \infty$). 
\begin{figure}[hbt]
\includegraphics[width=13.5cm,height=10cm]{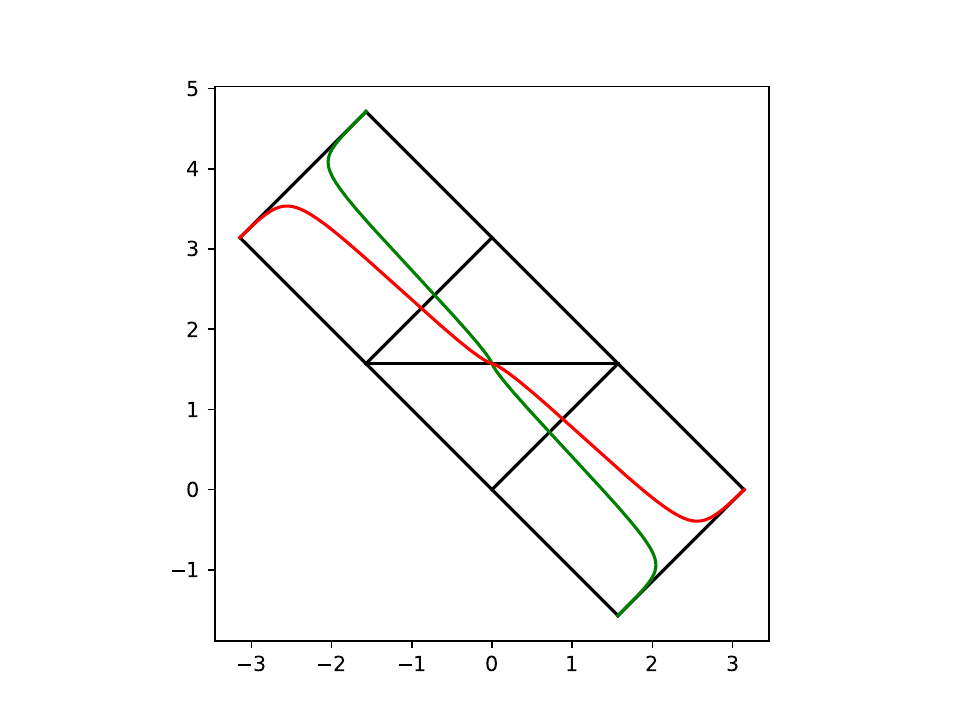}
\caption{Penrose diagramme of a globally hyperbolic region 
of a black hole -- white hole transition with timelike geodesic labelled
by $\rho=c$ (green) and Cauchy hypersurface labelled by $\tau=c$ for 
a constant $c$. 
The
geodesic starts in $i^-_P$ and ends in $i^+_F$ (bottom to top)
intersecting the 
singularity $r=0$ at proper time $\tau=c$ vertically. 
The Cauchy surface starts at 
$i^0_F$ and ends in $i^0_P$ (left to right) intersecting the singularity 
at spatial coordinate $\rho=c$ horizontally. The spacetime can be extended 
indefinitely to the future and past by gluing identical regions 
along the 45 degree lines $r=2M$ of the boundaries of the black hole and 
white hole regions respectively. Or it can be completed by adding a 
triangular Minkowski region in the past and the future along the 45 lines 
$r=2M$ between $i^-_P,i^+_F$ and the line $r=0$ respectively.}
\end{figure}
\begin{figure}[hbt]
\includegraphics[width=13.5cm,height=10cm]{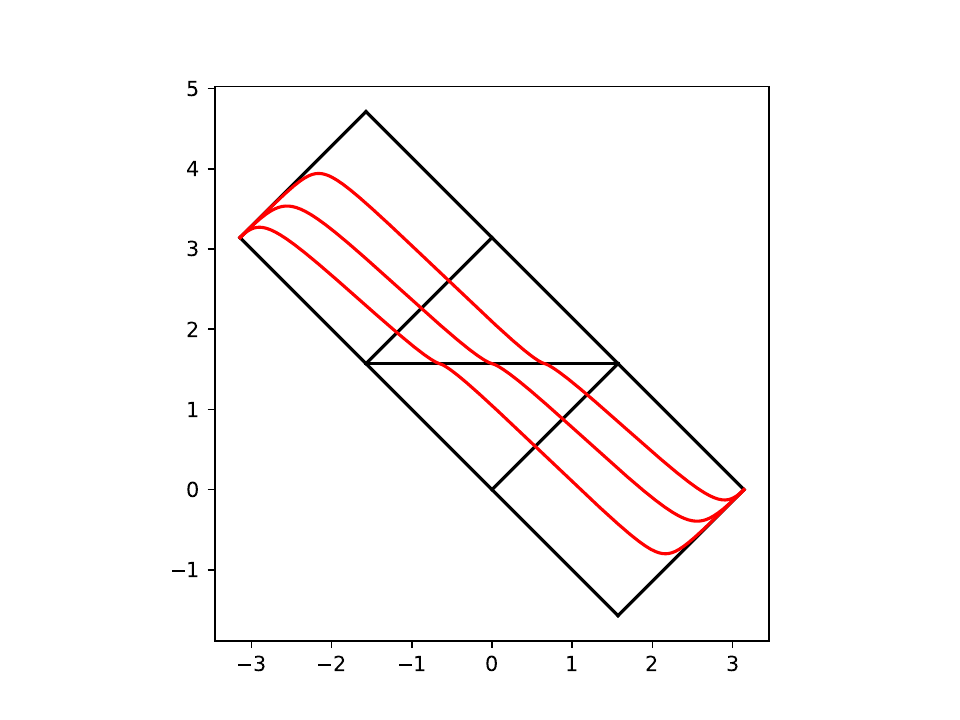}
\caption{Foliation of a globally hyperbolic portion of a 
BHWHT spacetime by synchronous proper time free falling Cauchy 
surfaces. The portion is the common domain of dependence of all
leaves of the foliation.}
\end{figure}
\begin{figure}[hbt]
\includegraphics[width=13.5cm,height=10cm]{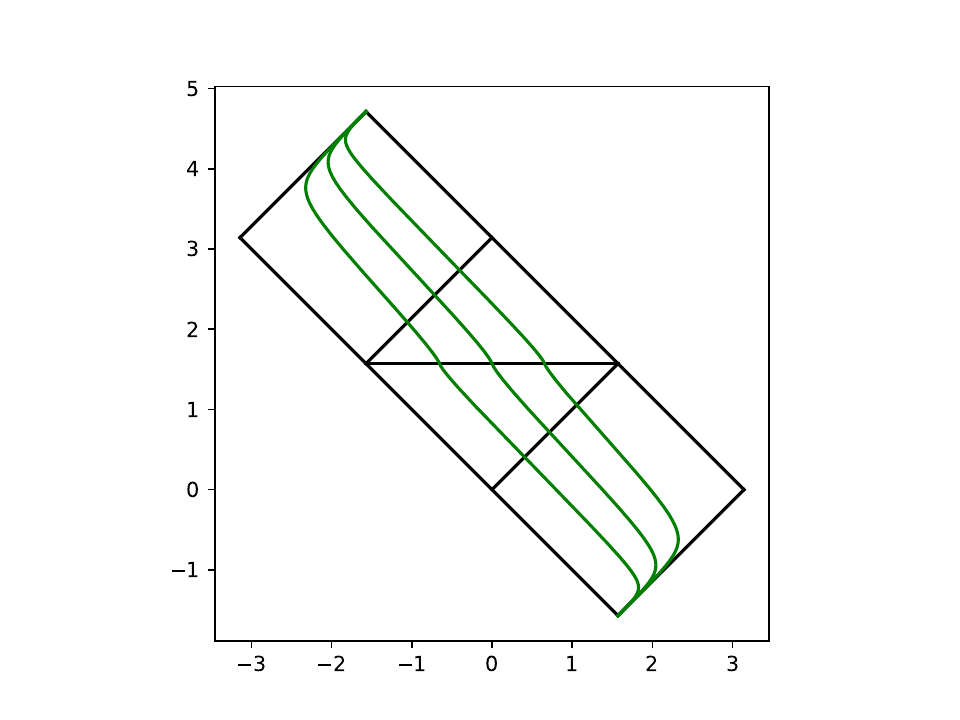}
\caption{Congruence of free falling timelike observers in a globally hyperbolic
portion of a BHWHT spacetime.}
\end{figure}
In order to determine the finer details of this intersection we introduce 
the coordinates $T,X,\bar{T},\bar{X}$ defined by 
\be \label{d.30b}
2T=V+U,\;2X=V-U,\;2\bar{T}=-(\bar{V}+\bar{U}),\;2\bar{X}=\bar{V}-\bar{U} 
\ee
Thus $\bar{X}=X,\; \bar{T}=-T$ at the singularity. Since 
$U\;V=T^2-X^2=1=\bar{V}\;\bar{U}=\bar{T}^2-\bar{X}^2$ at the singularity and
$U=\bar{U},V=\bar{V}>0$ it follows $T=\sqrt{1+X^2},\; 
\bar{T}=-\sqrt{1+\bar{X}^2}$ describes the singularity curves in terms of 
these coordinates. These have inclinations 
$dT/dX=\frac{X}{\sqrt{1+X^2}},\;
d\bar{T}/d\bar{X}=-\frac{\bar{X}}{\sqrt{1+\bar{X}^2}}$ 
respectively. 

We now compute for 
$y=y(\tau,\rho),\;\bar{y}=\bar{y}(\tau,\rho)$ respectively and 
$\tau<\rho,\;\tau>\rho$ respectively
using  
\be \label{d.31}
d y^2/2=\frac{1}{2R y}(d\rho-d\tau),\;
d \bar{y}^2/2=\frac{1}{2R \bar{y}}(d\tau-d\rho),\;
d [(y\pm 1) \; e^{y^2/2-\mp y}]=y^2 \; e^{y^2/2\mp y}\;dy
\ee
the differentials
\ba \label{d.32}
dV &=& \frac{V}{2R}\; (d\tau + \frac{1}{y+1}(d\rho-d\tau))
= \frac{V}{2R(y+1)} (y\; d\tau + d\rho)
\nonumber\\
dU &=& \frac{U}{2R}\; (-d\tau+\frac{1}{y-1}(d\rho-d\tau)
= \frac{U}{2R(y-1)}\; (-y \; d\tau+d\rho)
\nonumber\\
d\bar{V} &=& \frac{\bar{V}}{2R}\; (d\tau + \frac{1}{\bar{y}-1}
(d\tau-d\rho))
= \frac{\bar{V}}{2R(\bar{y}-1)} (\bar{y}\; d\tau - d\rho)
\nonumber\\
d\bar{U} &=& \frac{\bar{U}}{2R}\; (-d\tau+\frac{1}{\bar{y}+1}(d\tau-d\rho)
= \frac{\bar{U}}{2R(\bar{y}+1)}\; (-\bar{y} \; d\tau-d\rho)
\ea
It follows
\ba \label{d.33}
2dT &=& \frac{e^{y^2/2}}{2R}\;[e^{\kappa-y}\;(d\rho+y d\tau)
-e^{-\kappa+y}\;(d\rho-y d\tau)]
\nonumber\\
2dX &=& \frac{e^{y^2/2}}{2R}\;[e^{\kappa-y}\;(d\rho+y d\tau)
+e^{-\kappa+y}\;(d\rho-y d\tau)]
\nonumber\\
2d\bar{T} &=& \frac{e^{\bar{y}^2/2}}{2R}\;[-e^{\kappa+\bar{y}}\;(d\rho-
\bar{y} d\tau)
+e^{-\kappa-\bar{y}}\;(d\rho+\bar{y} d\tau)]
\nonumber\\
2d\bar{X} &=& \frac{e^{\bar{y}^2/2}}{2R}\;[e^{\kappa+\bar{y}}\;(d\rho-
\bar{y} d\tau)
+e^{-\kappa-\bar{y}}\;(d\rho+\bar{y} d\tau)]
\ea
This gives 
\ba \label{d.34}
\frac{dT}{dX} &=& 
\frac{
{\sf sh}(\kappa-y)\;d\rho+y\;{\sf ch}(\kappa-y)\;d\tau
}
{
{\sf ch}(\kappa-y)\;d\rho+y\;{\sf sh}(\kappa-y)\;d\tau
}
\nonumber\\
\frac{d\bar{T}}{d\bar{X}} &=& 
\frac{
-{\sf sh}(\kappa+\bar{y})\;d\rho+\bar{y}\;{\sf ch}(\kappa+\bar{y})\;d\tau
}
{
{\sf ch}(\kappa+\bar{y})\;d\rho-\bar{y}\;{\sf sh}(\kappa+\bar{y})\;d\tau
}
\ea
This enables us to conveniently compute the inclinations. For geodesics
$d\rho=0$
\be \label{d.35}
\frac{dT}{dX}={\sf coth}(\kappa-y),\;\;
\frac{d\bar{T}}{d\bar{X}}=-{\sf coth}(\kappa+\bar{y})
\ee
and for hypersurfaces $d\tau=0$
\be \label{d.36}
\frac{dT}{dX}={\sf th}(\kappa-y),\;\;
\frac{d\bar{T}}{d\bar{X}}=-{\sf th}(\kappa+\bar{y})
\ee
Thus in the $T,X$ and $\bar{T},\bar{X}$ diagramme respectively the 
geodesics and hypersurfaces have everywhere anti-reciprocal inclinations.

Since at $y=0$ we have $2 X=2{\sf sh}(\kappa)$ and 
at $\bar{y}=0$ we have $2 \bar{X}=2{\sf sh}(\kappa)$ it follows that 
the singularity inclination is 
${\sf th}(\kappa)$ in terms of $T,X$ and 
$-{\sf th}(\kappa)$ in terms of $\bar{T},\bar{X}$ respectively. It follows 
that the hypersurfaces are tangent to the singularity while the 
geodesics are transversal. Thus in these coordinates, a $\tau=$const. 
hypersurface can be described in terms of the $T,X$ coordinates until 
it intersects the singularity tangentially 
$(T={\sf ch}(\kappa),X={\sf sh}(\kappa)$. Then it continues from the 
tangential point 
$(\bar{T}=-{\sf ch}(\kappa),\bar{X}={\sf sh}(\kappa)$ in terms of 
$\bar{T},\bar{X}$ coordinates.

If one wants to avoid the jump by $-2\sqrt{1+X^2}$ between $T$ and 
$\bar{T}$ at the singularity we may substitute $T,\bar{T}$ by 
$T-\sqrt{1+X^2},\; \bar{T}+\sqrt{1+\bar{X}^2}$ which maps the singularity 
to the common line $T=\bar{T}=0$.  

\subsection{Relation between different GGP coordinates}
\label{sd.5}

A GGP coordinate system is determined by two parameters $M,e$. It determines 
a geodesic congruence $C_e$ in the BHWHT spacetime which by construction 
is isometric to two copies of two complementary halves 
(namely SS, BH and WH, MSS parts) of Kruskal spacetime where 
the latter carries a single parameter $M$. From this point of view the 
additional parameter $e$ is redundant and can be fixed to any desired value 
by a spacetime diffeomorphism. To change between two different values, say
$e,e'$ we relate them through the Schwarzschild time 
\be \label{d.36b}
\tau^\pm_e=e \; t\mp \int\; dr \;\frac{\sqrt{e^2-\Phi(r)}}{\Phi(r)}
\ee
which means 
\be \label{d.37}
\frac{1}{e}[\tau^\pm_e \pm\int\; dr \;\frac{\sqrt{e^2-\Phi(r)}}{\Phi(r)}]
=t=
\frac{1}{e'}[\tau^\pm_{e'} \pm\int\; dr \;\frac{\sqrt{(e')^2-\Phi(r)}}{\Phi(r)}]
\ee
This corresponds to a temporal diffeomorphism (consider $e,e'>0$) 
\be \label{d.38}
\tau^\pm_{e'}=\frac{e'}{e} \tau_e^\pm \pm [e']\;\int \; dr\; 
\Phi(r)^{-1}[\sqrt{1-e^{-2}\;\Phi(r)}-\sqrt{1-(e')^{-2}\;\Phi(r)}]
\ee
while the radial coordinate is unchanged. Noticing that 
$g^e_{rr}=q^e_{rr} =e ^{-2}$ we see that the right hand side of (\ref{d.38})
is exactly the second Dirac observable conjugate to the mass as derived from 
the Hamiltonian formulation. This shows that in the Lagrangian formulation
the second Dirac observable is considered as a gauge degree of freedom.

\subsection{Dirac observable conjugate to mass in GGP}
\label{sd.6}

We evaluate the Dirac observable conjugate to the mass $M$ (dropping 
inessential constants)
\be \label{d.40}
Q=\int_\mathbb{R}\; dz\; \delta'\; \frac{\sqrt{1-\gamma^2\Phi}}{\Phi},\;
q_{ab} \; dx^a\;dx^b=\gamma^2(z) dz^2+\delta(z)^2 d\Omega^2,\;\;
\Phi=1-\frac{R}{|\delta(z)|},\; R=2M
\ee
in the GGP gauge $\delta(z)=z, \; \gamma(z)=e^{-2},\; e^2\ge 1,\;e=$const.
This gives 
\be \label{d.41}
Q=2\;\int_{\mathbb{R}^+}\; dr\;\frac{\sqrt{1-\gamma^2\Phi(r)}}{\Phi(r)},\;
\Phi(r)=1-\frac{R}{r}
\ee
The integral (\ref{d.41}) is ill-defined as it stands due to a 
singularity at $r=R$ and $r=\infty$ while it is regular at $r=0$. 
It therefore needs a more detailed definition. We regularise it with 
three parameters $c<R,d<R,L>2\;R$ 
\be \label{d.41b}
\frac{Q_{c,d,L}}{2}=
\int_0^{R-c}\; dr\;\frac{\sqrt{1-\gamma^2\Phi}}{\Phi},\;
+\int_{R+d}^L\; dr\;\frac{\sqrt{1-\gamma^2\Phi}}{\Phi},\;
\ee
and eventually take $c,d,L^{-1}\to 0+$. Note that $\gamma^2\le 1$. 
We treat the case $\gamma^2=1$ separately from the case $\gamma^2<1$.\\
\\
{\bf Case $\gamma^2=1$:}\\
We have
\ba \label{d.42}
\frac{Q_{c,d,L}}{2} 
&=& \sqrt{R}\;[
\int_0^{R-c}\; dr\;\frac{\sqrt{r}}{r-R}\;
+\int_{R+d}^L\; dr\;\frac{\sqrt{r}}{r-R}
]
\nonumber\\
&=& R\;[
\int_0^{1-c/R}\; dx\;\frac{\sqrt{x}}{x-1}\;
+\int_{1+d/R}^{L/R}\; dx\;\frac{\sqrt{x}}{x-1}
]
\nonumber\\
&=& 2\;R\;[
\int_0^{\sqrt{1-c/R}}\; dy\;\frac{y^2}{y^2-1}\;
+\int_{\sqrt{1+d/R}}^{\sqrt{L/R}}\; dy\;\frac{y^2}{y^2-1}
]
\nonumber\\
&=& 2\;R\;[\sqrt{L/R}-\sqrt{1+d/R}+\sqrt{1-c/R}
-\int_0^{\sqrt{1-c/R}}\; dy\;\frac{1}{1-y^2}\;
+\int_{\sqrt{1+d/R}}^{\sqrt{L/R}}\; dy\;\frac{1}{y^2-1}
]
\nonumber\\
&=& 2\;R\;[\sqrt{L/R}-\sqrt{1+d/R}+\sqrt{1-c/R}
-\frac{1}{2}\;[\ln(\frac{1+y}{1-y})]_0^{\sqrt{1-c/R}}
+\frac{1}{2}\;
[\ln(\frac{y-1}{y+1})]_{\sqrt{1+d/R}}^{\sqrt{L/R}}
]
\nonumber\\
&=& 2\;R\;[\sqrt{L/R}-\sqrt{1+d/R}+\sqrt{1-c/R}
\nonumber\\
&& -\frac{1}{2}\;[\ln(\frac{[1+y]^2}{1-y^2})]_{y=\sqrt{1-c/R}}
-\frac{1}{2}\;[\ln(\frac{y^2-1}{[y+1]^2})]_{y=\sqrt{1+d/R}}
+\frac{1}{2}\;
[\ln(\frac{y-1}{y+1})]_{y=\sqrt{L/R}}
]
\ea
with $x=r/R=y^2$. The terms $\sqrt{1-c/R}-\sqrt{1+d/R},\;
[\ln(y+1)]_{\sqrt{1-c/R}}^{\sqrt{1+d/R}},\;
\ln((y+1)/(y-1))_{y=\sqrt{L/R}}$ vanish independently of how 
we take the limit $c,d,L^{-1}\to 0$. Thus, up to those terms, 
(\ref{d.42}) becomes 
\be \label{d.43}
\frac{Q_{c,d,L}}{2}=2\;R\;[\sqrt{L/R}-\frac{1}{2}\ln(\frac{d}{c})]
\ee
{\bf Case $\gamma^2<1$}\\
We have with $a^2:=\frac{\gamma^2}{1-\gamma^2}$
\ba \label{d.44}
\frac{Q_{c,d,L}}{2}
&=& 2\;R\;\sqrt{1-\gamma^2}\;[
\int_0^{\sqrt{1-c/R}}\; dy\;
\frac{y^2\sqrt{y^2+a^2}}{y^2-1}
+\int_{\sqrt{1+d/R}}^{\sqrt{L/R}}\; dy\;
\frac{y^2\sqrt{y^2+a^2}}{y^2-1}
]
\\
&=& 2\;R\;\sqrt{1-\gamma^2}\;[
\int_0^{\sqrt{1-c/R}}\; dy\;
\{\sqrt{y^2+a^2}
+\frac{\sqrt{y^2+a^2}}{y^2-1}\}
+\int_{\sqrt{1+d/R}}^{\sqrt{L/R}}\; dy\;
\{\sqrt{y^2+a^2}
+\frac{\sqrt{y^2+a^2}}{y^2-1}\}
]
\nonumber
\ea
We have the elementary integrals
\ba \label{d.45}
&& \int\; dy\; \sqrt{y^2+a^2}=
\frac{1}{2}[y\sqrt{y^2+a^2}+\ln(y+\sqrt{y^2+a^2})]
\nonumber\\
&& \int\; dy\; \frac{1}{\sqrt{y^2+a^2}}=
\ln(y+\sqrt{y^2+a^2})
\nonumber\\
&& 
\int\; dy\; \frac{\sqrt{y^2+a^2}}{y^2-b^2}
=\ln(y+\sqrt{y^2+a^2})
+\int\; dy\; \frac{1}{\sqrt{y^2+a^2}}\;[\frac{y^2+a^2}{y^2-b^2}-1]
\nonumber\\
&=&
\ln(y+\sqrt{y^2+a^2})
+(a^2+b^2)\;
\int\; dy\; \frac{1}{\sqrt{y^2+a^2}}\;\frac{1}{y^2-b^2}
\ea
We have
\be \label{d.46}
\frac{d}{dy}{\sf arth}(h\frac{y}{\sqrt{y^2+a^2}})
=\frac{d}{dy}{\sf arcoth}(h\frac{y}{\sqrt{y^2+a^2}})
=-\frac{1}{\sqrt{y^2+a^2}}\;\frac{ha^2}{h^2-1}\;
\frac{1}{y^2-\frac{a^2}{h^2-1}}
\ee
To match this to last integral in (\ref{d.45}) we pick 
$a^2/(h^2-1)=b^2,\; h^2=1+a^2/b^2$ so that 
$ha^2/(h^2-1)
%=a^2/b\sqrt{a^2+b^2}/(a^2/b^2)
=b\sqrt{a^2+b^2}$.
Then the argument 
of the hyperbolic function becomes 
$\sqrt{1+a^2/b^2}y/\sqrt{y^2+a^2})$ which must take values in 
$(-1,1)$ and $\mathbb{R}-[-1,1]$ respectively for hyperbolic 
tangens and cotangens 
respectively. For the tangens function this 
implies due to $b=1$ in our case that $|y|<1$ while $|y|>1$ for 
the cotangens function.  
%1>\frac{(1+a^2/b^2)y^2}{y^2+a^2}
%y^2+a^2>(1+a^2/b^2)y^2
%a^2>a^2/b^2 y^2
%b^2>y^2
Assembling these findings we have with $b^2=1$ in our case
\ba \label{d.47}
&& \frac{Q_{c,d,L}}{2}
= 2\;R\;\sqrt{1-\gamma^2}\;
\{
\frac{1}{2}[y\sqrt{y^2+a^2}
+3\;\ln(y+\sqrt{y^2+a^2})]_0^{\sqrt{1-c/R}}
\nonumber\\
&& +\frac{1}{2}[y\sqrt{y^2+a^2}
+3\;\ln(y+\sqrt{y^2+a^2})]_{\sqrt{1+d/R}}^{\sqrt{L/R}}
\nonumber\\
&& 
-\sqrt{a^2+1}\;
[{\sf arth}(\sqrt{a^2+1}\frac{y}{\sqrt{y^2+a^2}})]_0^{\sqrt{1-c/R}}
-\sqrt{a^2+1}\;
[{\sf arcoth}(\sqrt{a^2+1}\frac{y}{\sqrt{y^2+a^2}})]_{\sqrt{1+d/R}}^{\sqrt{L/R}}
\}
\ea
where we may use 
\be \label{d.48}
{\sf arth}(z)=\frac{1}{2}\ln(\frac{1+z}{1-z}),\;
{\sf arcoth}(z)=\frac{1}{2}\ln(\frac{z+1}{z-1}),\;
%
%th(x)=y, (e^x+e^{-x})y=e^x-e^{-x}, (e^{2x}+1)y=e^{2x}-1
%e^{2x}(1-y)=1+y,e^{2x}=(1+y)/(1-y), x=1/2\ln((1+y)/(1-y))
%
\ee
for $|z|<1, |z|>1$ respectively. Up to terms that vanish no matter how we 
take $c,d,L^{-1}\to 0$ the first two terms in (\ref{d.47}) maybe combined 
into 
\be \label{d.49}
2\;R\;\sqrt{1-\gamma^2}\;
\frac{1}{2}[y\sqrt{y^2+a^2}+3
\ln(y+\sqrt{y^2+1})]_0^{\sqrt{L/R}}
\ee
where the contribution from $y=0$ vanishes. The last two terms in (\ref{d.47})
are 
with $z(y)=\sqrt{a^2+1}\frac{y}{\sqrt{y^2+a^2}}$
\be \label{d.50}
2\;R\;\sqrt{1-\gamma^2}\;(-\frac{1}{2}\sqrt{1+a^2})
\ln(
[\frac{1+z}{1-z}]_{y=\sqrt{1-c/R}}\;
[\frac{1-z}{1+z}]_{y=0}\;
[\frac{z+1}{z-1}]_{y=\sqrt{L/R}}\;
[\frac{z-1}{z+1}]_{y=\sqrt{1+d/R}})
\ee
We have 
\be \label{d.51}
\frac{z+1}{z-1}=\frac{\sqrt{1+a^2}\;y+\sqrt{y^2+a^2}}
{\sqrt{1+a^2}\;y-\sqrt{y^2+a^2}}
\ee
which for $y=0$ equals $-1$ and for $y=\sqrt{L/R}$ equals
\be \label{d.52}
\frac{\sqrt{1+a^2}+\sqrt{1+a^2 R/L}}
{\sqrt{1+a^2}-\sqrt{1+a^2 R/L}} \to 
\frac{\sqrt{1+a^2}+1}
{\sqrt{1+a^2}-1}  
\ee
which is finite no matter how $c,d,L^{-1}\to 0$. For $y=\sqrt{1+d/R}$
(\ref{d.51}) becomes
\ba \label{d.53} 
&&\frac{\sqrt{1+a^2}\;\sqrt{1+d/R}+\sqrt{1+a^2+d/R}}
{\sqrt{1+a^2}\;\sqrt{1+d/R}-\sqrt{1+a^2+d/R}}
=\frac{[\sqrt{1+a^2}\;\sqrt{1+d/R}+\sqrt{1+a^2+d/R}]^2}
{(1+a^2)\;(1+d/R)-(1+a^2+d/R)}
\nonumber\\
&=& \frac{[\sqrt{1+a^2}\;\sqrt{1+d/R}+\sqrt{1+a^2+d/R}]^2}{a^2\;d/R}
\to \frac{4(1+a^2)^2}{a^2\; d/R}
\ea
while for 
$y=\sqrt{1-c/R}$ (\ref{d.51}) becomes
\ba \label{d.54} 
&& \frac{\sqrt{1+a^2}\;\sqrt{1-c/R}+\sqrt{1+a^2-c/R}}
{\sqrt{1+a^2}\;\sqrt{1-d/R}-\sqrt{1+a^2-d/R}}
=\frac{[\sqrt{1+a^2}\;\sqrt{1-d/R}+\sqrt{1+a^2-d/R}]^2}
{(1+a^2)\;(1-c/R)-(1+a^2-d/R)}
\nonumber\\
&=& \frac{[\sqrt{1+a^2}\;\sqrt{1-c/R}+\sqrt{1+a^2-c/R}]^2}{-a^2\;c/R}
\to -\frac{4(1+a^2)^2}{a^2\; c/R}
\ea
Hence the last two terms in (\ref{d.47}) approach 
\ba \label{d.55}
&& 2\;R\;\sqrt{1-\gamma^2}\;(-\frac{1}{2}\sqrt{1+a^2})
\ln(
\frac{4(1+a^2)^2}{a^2\; c/R}\;
\frac{\sqrt{1+a^2}+1}{\sqrt{1+a^2}-1}\;  
\frac{a^2\; d/R}{4(1+a^2)^2})
\nonumber\\
&=&
2\;R\;\sqrt{1-\gamma^2}\;(-\frac{1}{2}\sqrt{1+a^2})
\ln(
\frac{d}{c}\;
\frac{\sqrt{1+a^2}+1}{\sqrt{1+a^2}-1}) 
\ea
Altogether with $a^2=\gamma^2/(1-\gamma^2),1+a^2=1/(1-\gamma^2))$
\be \label{d.56}
\frac{Q_{c,d,L}}{2}
=
\frac{2\;R}{\sqrt{1+a^2}}\;\{
\frac{1}{2}[y\sqrt{y^2+a^2}+3
\ln(y+\sqrt{y^2+a^2})]_{y=\sqrt{L/R}}
-\frac{1}{2}\sqrt{1+a^2}\;
\ln(\frac{d}{c}\;\frac{\sqrt{1+a^2}+1}{\sqrt{1+a^2}-1}) 
\}
\ee
In the limit $\gamma\to 1-$ this becomes 
\be \label{d.57}
\frac{Q_{c,d,L}}{2}
=
2\;R\;\{
\frac{1}{2}\sqrt{L/R}
-\frac{1}{2}\ln(\frac{d}{c})
\}
\ee
which differs by a factor $1/2$ from the result (\ref{d.43})
i.e. the integral and the limit $\gamma\to 1-$ do not commute.\\
\\
Yet, we find for all values of $\gamma$ that 
\be \label{d.58}
\frac{Q_{c,d,L}}{4R}=g(\gamma,L/R)-\frac{1}{2} \ln(d/c)
\ee
where $g(\gamma,L/R)=Q_{c,d,}/(4R)+\frac{1}{2}\ln(d/c)$ is given 
explicitly in (\ref{d.43}) and (\ref{d.56}) for $\gamma^2=1,\gamma^2<1$ 
respectively and diverges as $\sqrt{L/R}$ as $L\to \infty$. 
The fact that (\ref{d.43}) is not the limit of (\ref{d.56}) as 
$\gamma\to 1-$ suggests to consider two different strategies:\\
\\
{\bf Strategy 1:}\\
We consider the exact GPG $\gamma^2\equiv 1$, i.e. $\gamma$ is not a dynamical 
variable. Then we pick the following limit $c,d,L^{-1}\to 0$ in 
(\ref{d.43})
\be \label{d.59}
\ln(\frac{d}{c})=2\sqrt{L/R}+\frac{Q}{2R},\;\;
c=e^{-L/R}
\ee
Then the large $L$ behaviour of $d$ is $d\propto e^{2\sqrt{L/R}-L/R}$, 
thus both $c,d$ decay exponentially in $L/R$.
Thus the value of $Q$ comes about simply because of the ambiguity in the 
{\it principal value regularisation} of the integral defining it. However,
$Q$ is not a parameter on which the spatial metric depends, it is
the variable conjugate to $M$ but its  existence has no further 
consequences for the theory. The advantage of this strategy is that it 
yields a consistent picture, i.e. the integral defining $Q$ is actually 
able to produce that value while the exact GPG is imposed, 
without introducing additional observable consequences. 
Furthermore, it agrees 
with the Kantowski-Sachs picture that we review in the next subsection 
which also yields two Dirac observables one of which is $M$ and the 
other one is related to a time rescaling freedom $\kappa$ which in the 
GPG is also arises, however, not as a Dirac observable but rather as a 
residual gauge freedom in choosing the physical Hamiltonian, see 
appendix \ref{sa} and \ref{sb}.\\
\\
{\bf Strategy 2:}\\
We consider $0<\gamma^2<1$ as a dynamical variable. 
Then we pick the following limit $c,d,L^{-1}\to 0$ in (\ref{d.56})
\be \label{d.59b}
\frac{1}{2}\ln(\frac{d}{c})=g(\gamma,L/R)-\frac{\zeta}{2}\;
{\sf arth}(\gamma),\;
c=e^{-L/R}
\ee
for $\zeta>0$ some numerical constant.
Then the large $L$ behaviour of $d$ is $d\propto 
e^{\sqrt{L/R}-L/R}$ 
and thus both $c,d$ decay exponentially
in $L/R$. The limit $L\to \infty$ then yields
\be \label{d.60}
\frac{Q}{R}=\zeta {\sf arth(\gamma)}
\ee
Now the physical Hamiltonian of sperically symmetric vacuum 
gravity is just $R=2M$ up to a constant which yields the equations 
of motion $M=$const. and $\dot{Q}=$const. Thus $Q$ diverges 
linearly in $\tau$. Thus $\gamma={\sf th}(Q/(\zeta R)$ approaches 
exponentially fast (the faster the smaller $\zeta$) the value $\pm 1$ from 
below/above. Thus $e^2=1/\gamma^2$ approaches the value $e^2=1$ exponentially 
fast from above. I.e. the {\it generalised} GP coordinates become 
{\it dynamically} exponentially 
fast the {\it exact} GP coordinates.

This conclusion is of course dependent on the choice of the finite part
in the regularisation (\ref{d.59}) which is a regularisation 
ambiguity. It is motivated by the desire to reconcile the fact that 
in the Lagrangian picture the parameter $e$ is a choice of gauge 
that can be removed by a temporal diffeomorphism while in the 
Hamiltonian picture it is a function of the Dirac observables 
$M,Q$ and thus cannot be gauged away. Thus the only way to bring 
both pictures into agreement with Birkhoff's theorem that there is only 
one physical degree of freedom in the Lagrangian picture is to ensure that 
in the Hamiltonian picture the additional degree of freedom {\it 
dynamically} settles to the value it can be assigned to in the Lagrangian 
picture. This can be viewed as a temporal diffeomorphism as well but 
that diffeomorphism in the Hamiltonian picture is a symmetry transformation.
This requirement still does not fix the finite part of $f(\gamma)$
(\ref{d.59}) uniquely,
any bijection $f:(-1,1)\to \mathbb{R};\; \gamma\mapsto f(\gamma)$ 
with the property that $\lim_{\gamma\to \pm\; 1\mp}=\pm \infty$ will 
do such as 
$f(\gamma)=\frac{1+\gamma}{1-\gamma}$. However, the faster $\gamma^2\to 1$
dynamically, the faster the black hole becomes truly static observationally
no matter which picture is used. 

Note that these conclusions hold only in the strictly sperically symmetric 
vacuum case. With the presence of gravitational perturbations and matter,
the physical Hamiltonian will be of the form $H=M+H_1(M,Q)$ where 
$H_1$ contains the information about perturbations and matter and which 
will depend on both $M,Q$ when expanding about the backgrond metric 
parametrised by $M,e$ and thus $M,Q$. This means that $M$ is no longer 
a constant of motion and that $Q$ is not necessarily diverging which 
means that $\gamma^2$ does not necessarily become unity as time progresses.
In this case we must use the SAPT framework \cite{SAPT,ST} to capture 
the corresponding quantum backreaction.\\ 
\\
We will not follow the second strategy in the present paper because 
it also requires to revisit the decay behaviour of the fields and 
the whole boundary structure analysis that leads to the reduced Hamiltonian
as developed in section \ref{s4} and which may lead to some restriction 
on the freedom to choose $f(\gamma)$.
However, the advantage of the second strategy is that it offers the possibility 
to change the integration constant $M$ dynamically, a possibility that 
one may want to keep in mind for future investigations.

\section{Kantowski-Sachs spacetimes}
\label{se}

The purpose of this section is to provide the link with the substantial 
amount of work \cite{LQG-BH} that has been devoted to the interior 
of quantum black holes. In particular we show that there is 
no contradiction between the presence of two independent and canonically 
conjugate Dirac observables on which the metric depends non-trivially
and Birkhoff's theorem: While the reduced or physical phase space 
is manifestly two-dimensional, one of the degrees of freedom corresponds 
to a time rescaling of the spacetime coordinates which is considered a 
gauge transformation in the Lagrangean formulation but certainly not 
in the Hamiltonian formulation (by definition a Dirac observable is 
gauge invariant). This is a general phenomenon for cosmological models
as has been pointed out in \cite{AshtekarSamuel}.   

Kantowski-Sachs (KS) spacetimes are homogeneous and spherically symmetric 
rather 
than isotropic spacetimes described by the line element
\be \label{e.1}
ds^2=-D(T)^2\;dT^2\;+A(T)^2\; dX^2\; +B(T)^2\; d\Omega^2
\ee
where $d\Omega^2$ is the standard line element of the round sphere metric 
$\Omega_{EF}$ while $X\in [-K/2,K/2]$ is a KS ``radial'' 
coordinate with spatial cut-off $K\in \mathbb{R}_+$ and $A,B,D$ 
are functions of KS time $T\in \mathbb{R}$ only. We consider all 
coordinates $T,X,\theta,\varphi$ dimension-free while $D,A,B$ have dimension 
of length. Alternatively we may want to introduce dimensionful coordinates 
$\hat{T}=L\; T,\; \hat{X}=L X$ where $L$ is some unit of length 
and $\hat{A}=\frac{A}{L},\; \hat{B}=\frac{B}{L},\;\hat{D}=\frac{D}{L}$ 
become dimensionless. 
 
To obtain the Hamiltonian description 
of these models we identify the ADM variables 
\be \label{e.2}
N=D,\; N^a=0,\; 
q_{ab}
=
A^2\;\delta^X_a\;\delta^X_b+B^2\;
\delta^E_a\; \delta^F_b \Omega_{EF},\;
=:
q_X\;\delta^X_a\;\delta^X_b+q_S\;
\delta^E_a\; \delta^F_b \Omega_{EF},\;
\ee
and 
\be \label{e.3}
k_{ab}=\frac{1}{2N}[\dot{q}_{ab}-({\cal L}_{\vec{N}} q)_{ab}]
=\frac{1}{D}[A\dot{A}\delta^X_a\;\delta^X_b+B\;\dot{B}
\delta^E_a\; \delta^F_b \Omega_{EF}]
=:k_X\;\delta^X_a\;\delta^X_b+k_S\;
\delta^E_a\; \delta^F_b \Omega_{EF},\;
\ee
whence
\ba \label{e.4}
p^{ab} &=& \sqrt{\det(q)}[q^{ac}\; q^{bd}-q^{ab}\; q^{cd}]\;k_{cd}
\nonumber\\
&=& A\; B^2 \sqrt{\det(\Omega)}\;[
\delta^a_X\delta^b_X\;(A^{-4} k_X- A^{-2}(A^{-2} k_X+2\;B^{-2} k_S))
\nonumber\\
&& +\delta^a_E\delta^b_F \Omega^{EF}\; 
(B^{-4} k_S- B^{-2}(A^{-2} k_X+2\;B^{-2} k_S))
]
\nonumber\\
&=:& \sqrt{\det(\Omega)}\;
(p^X\;\delta_X^a\;\delta_X^b+p^S\;
\delta_E^a\; \delta_F^b \Omega^{EF})
\ea
Here $q_X, q_S, k_X, k_S, p^X, p^S$ do not depend on the spatial coordinates
and $k_X, k_S$ have dimension of length so that 
$p^X, p^S$ are dimension free. 

We rescale the Einstein-Hilbert action by $\frac{1}{K}$ and 
pull back the symplectic potential
\ba \label{e.3b}
&& G_N\;\Theta=\frac{1}{K}\;\int\; d^3x\;\; p^{ab}\;[\delta q]_{ab}
=\frac{1}{K}\int\; dX\;\int\; d\theta\;d\varphi\;\sqrt{\det(\Omega)}\; 
(2A\; p^X\;[\delta A]
+4B\; \;p^S\; [\delta B])
\nonumber\\
&=& 4\pi\;
(2A\; p^X\;[\delta A]
+4B\; p^S\; [\delta B])
\ea
where $G_N$ is Newton's constant which has dimension of length squared 
in units in which $\hbar=1$. We define 
\be \label{e.4b}
p_A=2\;A\; p^X,\; p_B=4\; B\; p^S
\ee
which have dimension of length. Accordingly we have the non-vanishing 
Poisson brackets
\be \label{e.5}
\{p_A,A\}=\{p_B,B\}=g_N,\; \; g_N=\frac{G_N}{4\pi}
\ee
Next we compute the constraints 
\ba \label{e.6}
&& G_N\; C_a(N^a)=-\frac{2}{K} \int\; d^3x\; N^a D_b p^b_a\equiv 0,\;
\nonumber\\
&& G_N\; C(N)=\frac{1}{K}\;\int\; d^3x\;N\; 
[\det(q)]^{-1/2}(p^{ab}\; p_{ab}-[p^a_a]^2)-
\det(q)]^{1/2} R(q)]
\nonumber\\
&=& \frac{4\pi\;D}{A\; B^2}\;
(\frac{1}{8}\;[A\;p_A]^2-\frac{1}{4}\;[A\; p_A]\;[B\;p_B]-2 [A\;B]^2)
\ea
where the results of section \ref{s2} were used, in particular that 
$R[\Omega]=2$. The appearance of (\ref{e.5}) suggests to 
transform from $A,B>0$ to $x:=\ln(A/L),\; y:=\ln(B/L)\in \mathbb{R}$ and 
to introduce $p_x:=L^{-2}\; A p_A,\; p_y=L^{-2} \; B\; p_B$. Then
with $\tilde{D}=\frac{D\; L^2}{A\; B^2}$ 
\be \label{e.7}
\{p_x, x\}=\{p_y, y\}=g=\frac{g_N}{L^2},\;\;
C(N)=\frac{\tilde{D}}{g}\;
(\frac{1}{8}\;p_x^2-\frac{1}{4}\;p_x\;p_y-2 e^{2(x+y)})
=:\frac{\tilde{D}}{g}\;\tilde{C}
\ee
~\\
We can now develop three equivalent descriptions of the system:\\
1. The reduced phase space description in terms of a physical 
Hamiltonian and true degrees of freedom.\\
2. The description in terms of non-relational Dirac observables.\\
3. The description in terms of relational Dirac observables.

\subsection{Reduced phase space description}
\label{se.1}

Since $\tilde{C}$ is linear in momentum $p_y$ we choose $y$ as a clock
and rewrite the constraint as 
\be \label{e.8}
\tilde{D}\; \tilde{C}=\hat{D}\;\hat{C},\;\;
\hat{D}=-\frac{\tilde{D} p_x}{4},\;\;
\hat{C}=p_y+h,\;\;h=8\;\frac{e^{2(x+y)}}{p_x}-\frac{p_x}{2}
\ee
We impose the explicitly time dependent gauge fixing condition 
\be \label{e.9}
\hat{G}_T:=y-T
\ee
It is preserved in time on the constraint surface $\hat{C}=0$ iff 
(note the distinction betwee the total and explicit time derivative)
\be \label{e.9b}
\frac{d}{dT} \hat{G}_T=\{\hat{D}\hat{C}/g,\hat{G}_T\}
+\frac{\partial}{\partial T} \hat{G}_T=\hat{D}-1=0
\ee
which fixes $\hat{D}_\ast=1$. Thus the gauge degrees of freedom are 
$y,p_y$ while the true degrees of freedom are $x,p_x$. The reduced 
Hamiltonian is defined for functions $F$ depending only on $x,p_x$
by 
\be \label{e.10}
\{H,F\}:=\{\hat{D}\hat{C}/g,F\}_{\hat{D}=\hat{D}_\ast,y=T,p_y=-h}
=\hat{D}_\ast/g\{h,F\}_{y=T}=\{h_{y=T},F\}
\ee
hence 
\be \label{e.11}
H=H_T=
\frac{1}{g}\;
(8\;\frac{e^{2(x+T)}}{p_x}-\frac{p_x}{2})
\ee
which is explicitly time dependent. 

We now solve the resulting equations of motion
\be \label{e.12}
\dot{x}(T)=\{H_T(x,p_x),x\}_{x=x(T),p_x=p_x(T)},\;
\dot{p}_x(T)=\{H_T(x,p_x),p_x\}_{x=x(T),p_x=p_x(T)}
\ee
We note that 
\be \label{e.13}
\frac{d}{dT} H_T=\frac{\partial}{\partial T} H_T=
g^{-1}\; 16\; e^{2(x+T)} \; p_x^{-1}
\ee
while 
\be \label{e.14}
\frac{d}{dT} p_x=-16\; e^{2(x+T)} \; p_x^{-1}
\ee
Hence 
\be \label{e.15}
E_T(x,p_x):=H_T(x,p_x)+ g^{-1} p_x
=\frac{1}{g}\;
(8\;\frac{e^{2(x+T)}}{p_x}+\frac{p_x}{2})
\ee
is a constant of motion $E_T=\epsilon$
on the trajectories. As a function of 
the reduced phase space it is explicitly time dependent. 
Next, combining (\ref{e.14}) and (\ref{e.15}) we have 
\be \label{e.16}
\dot{p}_x=-2\; g\; E_T+p_x
\ee
which is solved by
\be \label{e.17} 
p_x(T)=\kappa\; e^T+2\;g\;E_T 
\ee
where $\kappa$ is an integration constant. 
It follows
\be \label{e.18}
16\; e^{2(x+T)}=p_x(2\;\;E_T-p_x)=-c\;e^T\;(2\;E_T+c\; e^T)\;\;\Rightarrow
\;\;
e^{2\;x(T)}=-\frac{\kappa}{16}\;(\kappa+2\;g\;E_T\; e^{-T})
\ee
which provides the general and explicit solution. Since (\ref{e.18}) 
is positive, for a solution parametrised by $\epsilon,\kappa$,
we must necessarily have $\epsilon\kappa<0$ and 
the range of $T$ becomes confined to the 
set 
\be \label{e.19}
2\;g\; \frac{|\epsilon|}{|\kappa|} \; e^{-T}>1 
\ee
We may also 
combine (\ref{e.17}) and (\ref{e.18}) into the statement that 
\be \label{e.20}
c_T(x,p_x):=-16\;\frac{e^{2x+T}}{p_x}
\ee
is an explicitly time dependent function on the phase space which is 
a constant $c_T=\kappa$ on a trajectory. 

In terms of $E_T, c_T$ the 
description of this dynamical system is therefore especially convenient.
The reduced Hamiltonian and true degrees of freedom are given by 
\be \label{e.21}
H_T=-\frac{1}{g}\;(g\;E_T+c_T\; e^T),\;\;
p_x=c_T\; e^T+2\;g\;E_T,\;
e^{2x}=-\frac{c_T}{16}\;(c_T+2\;g\;E_T\; e^{-T})
\ee
where we denoted objects with explicit time dependence with subscript $T$.
The inversion (\ref{e.21}) is given by (\ref{e.15}) and (\ref{e.20})
which yields
\be \label{e.21b}
\{E_T,c_T\}(x,p_x)
=\frac{1}{g}\;
\{8\;\frac{e^{2(x+T)}}{p_x}+\frac{p_x}{2},(-16)\;\frac{e^{2x+T}}{p_x}\}
=\{\frac{p_x}{2\;g},(-16)\;\frac{e^{2x+T}}{p_x}\}
=(-16)\;\frac{e^{2x+T}}{p_x}=c_T
\ee
To interpret $E_T,c_T$ geometrically, we express $D,A,B$ in terms of them. 
We have 
\ba \label{e.22}
[\frac{B}{L}]^2 &=& e^{2y}=e^{2T}
\nonumber\\
{[}\frac{A}{L}]^2 &=& e^{2x}
=-\frac{c_T^2}{16}\;(1+ 2\;g\;\frac{E_T}{c_T}\;e^{-T})
\nonumber\\
D &=& 
\frac{A\; B^2\;\tilde{D}}{L^2}
=-4\;\frac{A\; B^2\;\hat{D}_\ast}{L^2\; p_x}
=-4\;\frac{A\; B^2}{L^2\; (c_T\; e^T+2\;g\;E_T)}
\nonumber\\
D^2 
&=& 16\; L^2\;e^{4T}\; 
(-\frac{c_T^2}{16}\;(1+ 2\;g\;\frac{E_T}{c_T}\;e^{-T})
\frac{1}{(c_T\; e^T+2\;g\;E_T)^2}
\nonumber\\
&=& -L^2\;e^{2T}\; 
\frac{1}{1+2\;g\frac{E_T}{c_T} e^{-T}}
\ea
This suggests to introduce the new coordinates and functions 
\be \label{e.23}
r:=L\; e^T,\; t:=L\;X,\;
M_T:=-L^{-1}\;\frac{E_T}{c_T}
\ee
In terms of these the line element takes the form ($g_N=L^2 g$)
\be \label{e.24}
ds^2=-\frac{1}{\frac{2\; g_N\;M_T}{r}-1}\; dr^2
+[\frac{c_T}{4}]^2\;(\frac{2\; g_N\;M_T}{r}-1)\; dt^2
+r^2\; d\Omega^2
\ee
This is precisely the interior Schwarzschild solution with the roles 
of $r,t$ of being spatial and temporal coordinates switched since 
in the range (\ref{e.19}) i.e. $r<2 g_N M_T$ the coefficients 
of $dr^2$ and $dt^2$ are negative and positive respectively. 
The metric depends on the explicitly time $T$ dependent functions on the 
reduced phase space given by $M_T,c_T$ which are conjugate up to 
a factor of $L^{-1}$ 
\be \label{e.24b}
\{M_T, c_T\}=L^{-1}
\ee
as follows from (\ref{e.21}). Clearly $M_T$ which is a positive constant on 
solutions is nothing but the mass of the black hole while $c_T$ 
which is a dimensionless constant on solutions is nothing but a rescaling 
freedom of $t$.\\ 
\\
The interesting point is that although 1. the reduced Hamiltonian $H_T$, 2.
the mass $M_T$ and 3. the rescaling freedom $c_T$ 
are explicitly time $T$ dependent and although $M_T, c_T$ are canonically
conjugate coordinates of the {\it 2-dimensional} reduced phase space, 
nevertheless on solutions $M_T,c_T$ are in fact time $T$ independent.   

\subsection{Non-relational Dirac observables}
\label{se.2}
  
We consider the full phase space with conjugate pairs $(x,p_x),(y,p_y)$ and 
the constraint in the form $\hat{C}=p_y+h(x,p_x,p_y)$. We note that 
$x,y$ appear only in the combination $x+y$ in $h$. Thus $x-y$ is cyclic     
and therefore $p_x-p_y$ is gauge invariant, i.e. a Dirac observable.
Since $\hat{C}$ is trivially a Dirac observable also 
\be \label{e.25}
E(x,p_x,y):=g^{-1}(p_x-p_y)+\hat{C}=g^{-1}[8\frac{e^{2(x+y)}}{p_x}
+\frac{p_x}{2}]
\ee
is a Dirac observable. In the gauge $y=T$ it coincides with 
$E_T(x,p_x)$. Correspondingly we conjecture that 
\be \label{e.26}
c(x,p_x,y):=-16 \; \frac{e^{2x+y}}{p_x}
\ee
is a second independent Dirac observable because in the gauge $y=T$ 
it coincides with $c_T(x,p_x)$. This is readily confirmed 
\be \label{e.27}
\{\hat{C},c\}=-16\; g^{-1}\{8\;\frac{e^{2(x+y)}}{p_x}+p_y-\frac{p_x}{2},
\frac{e^{2x+y}}{p_x}\}
=-\frac{16}{p_x\; g}\;\{p_y-\frac{p_x}{2},e^{2x+y}\}=0
\ee
This provides an independent interpretation of the true degrees of 
freedom $M_T, c_T$ of the previous subsection: 
They correspond to the Dirac observables 
$M:=L^{-1} E/c, c$ evaluated on the gauge cut $y=T$. In particular they are 
canonically conjugate 
\be \label{e.28}
\{M,c\}=L^{-1}
\ee
Being Dirac observables, they have trivial ``evolution'' with respect to 
$\hat{C}$ by construction. This is equivalent to the statement that 
$E_T, c_T$ are constants of motion with respect to the reduced Hamiltonian
because 
\be \label{e.29}
\frac{d}{dT} E_T=
\frac{\partial}{\partial T} E_T+\{H_T, E_T\}
=\{\hat{C},E\}_{y=T}
\ee
and similar for $c$.

\subsection{Relational Dirac observables}
\label{se.3}
  
The relational Dirac observables corresponding to a function 
$F$ of the true degrees of freedom $x,p_x$ are given by 
the explicit formula
\be \label{e.30}
O_F(T):=\sum_{n=0}^\infty\; \frac{(T-y)^n}{n!}\; \{\hat{C},F\}_{(n)}
\ee
where $\{\hat{C},F\}_{(0)}=F,\; 
\{\hat{C},F\}_{(n+1)}=\{\hat{C},\{\hat{C},F\}\}$ is the iterated Poisson 
bracket. The direct evaluation of the infinite series is quite 
non-trivial. 

However, we may avoid the direct evaluation whenever 
we have a complete set of Dirac observables at our disposal as follows:
Suppose that $(x,p_x)$ are the true degrees of freedom, $(y,p_y)$ the 
gauge degrees of freedom, the constraints 
are given in the form $\hat{C}=p_y+h(x,p_x,y)$, the gauge fixing 
condition are given in the form $G=y-k(T)$ with $T$ dependent constants
$k(T)$  
and $D$ a complete set of Dirac observables,
i.e. their Hamiltonian vector fields are linearly independent on 
the constraint surface $\hat{C}=0$. Now 
for any function $F=F(x,p_x,y,p_y)$ on the full phase space we have the 
identity (see \cite{22} and references therein)
\be \label{e.31}
O_F(T):=[e^{\{s\cdot \hat{C},.\}}\cdot F]_{s=k(T)-y}=
F(O_x(T),O_{p_x}(T),T,-h(O_x(T),O_{p_x}(T),T))
\ee
Applied to the system of Dirac observables $O_D(T)=D$ we thus find the 
relations
\be \label{e.32}
D(x,p_x,y,-h(x,p_x,y))=D(O_x(T),O_{p_x}(T),T,-h(O_x(T),O_{p_x}(T),T))
\ee
which can be solved algebraically for $(O_x(T), O_{p_x}(T))$.

Applied to our system and using the Dirac observables of the 
previous section we find with $Q:=O_x(T),\; P:=O_{p_x}(T)$
\ba \label{e.33}
E(x,p_x,y) &=& 
\frac{1}{g}(8\frac{e^{2(x+y)}}{p_x}+\frac{p_x}{2}
=\frac{1}{g}(8\frac{e^{2(Q+T)}}{P}+\frac{P}{2})
\nonumber\\
c(x,p_x,y) &=& 
-16 \frac{e^{2x+y}}{p_x}
=-16 \frac{e^{2Q+T}}{P}
\ea
which can be solved for
\ba \label{e.34}
P &=& p_x+16\; \frac{e^{2(x+y)}}{p_x}(1-e^{T-y})
\nonumber\\
e^{2Q} &=& e^{2x}\;e^{-(T-y)}(1+16\frac{e^{2(x+y)}}{p_x^2}(1-e^{T-y}))
\ea
This maybe Taylor expanded in powers of $y-T$ thus providing 
explicit formulae for the iterated Poisson brackets. In particular 
the zeroth order gives $Q=x, P=p_x$ as it should be. 

The physical Hamiltonian, i.e. the Dirac observable, 
that drives the evolution of the relational observables is 
given by \cite{22}
\be \label{e.35}
H(x,p_x,y)=O_h(T)=h(x=O_x(T),p_x=O_{p_x}(T),y=T) 
\ee
and coincides with $H_T(x,p_x)$ at the gauge cut $y=T$.

\end{appendix}

\end{document}